\documentclass[]{article}

\usepackage[english]{babel}

\usepackage[utf8]{inputenc} 
\usepackage[autostyle]{csquotes}
\usepackage[dvipsnames]{xcolor}
\usepackage{booktabs}       

\usepackage{microtype}      
\usepackage{lipsum}
\usepackage[colorinlistoftodos]{todonotes}

\usepackage{afterpage} 
\usepackage{etex}

\usepackage[stable]{footmisc} 
\usepackage[section]{placeins} 

\usepackage{graphicx} 
\usepackage{multirow} 
\usepackage{longtable}
\usepackage{xltabular} 

\usepackage{bigints}
\usepackage[left=25.4mm, right=25.4mm, top=25.4mm, bottom=25.4mm]{geometry} 

\usepackage{yfonts} 
\usepackage{tikz}
\usepackage{tikz-cd}

\tikzset{State/.style={circle, draw, very thick, scale = 0.8}}
\tikzset{NodeRectangle/.style={rectangle, draw, very thick, scale = 0.8}}
\tikzset{Link/.style={-{Stealth[scale=1]},ultra thick}}
\tikzset{QuantumTrajektory/.style={Circle-Diamond[open], ultra thick,dashed}}
\usetikzlibrary{decorations.pathreplacing,calligraphy} 

\usetikzlibrary{arrows.meta}

\usetikzlibrary[arrows]

\makeatletter
\tikzset{MyLoopClockWise_2/.style =  {to path={
  \pgfextra{}
  [looseness=12,min distance=5mm,max distance=5mm,out=300,in=240]
  \tikz@to@curve@path},font=\sffamily\small
  }} 
\tikzset{MyLoopCounterClockWise_2/.style =  {to path={
  \pgfextra{}
  [looseness=12,min distance=5mm,max distance=5mm,out=240,in=300]
  \tikz@to@curve@path},font=\sffamily\small
  }}
\tikzset{MyLoopClockWise_4/.style =  {to path={
  \pgfextra{}
  [looseness=12,min distance=3mm,max distance=3mm,out=300,in=240]
  \tikz@to@curve@path},font=\sffamily\small
  }} 
\tikzset{MyLoopCounterClockWise_4/.style =  {to path={
  \pgfextra{}
  [looseness=12,min distance=3mm,max distance=3mm,out=240,in=300]
  \tikz@to@curve@path},font=\sffamily\small
  }}
\tikzset{MyLoopClockWise_3/.style =  {to path={
  \pgfextra{}
  [looseness=12,min distance=4mm,max distance=4mm,out=300,in=240]
  \tikz@to@curve@path},font=\sffamily\small
  }} 
\tikzset{MyLoopCounterClockWise_3/.style =  {to path={
  \pgfextra{}
  [looseness=12,min distance=4mm,max distance=4mm,out=240,in=300]
  \tikz@to@curve@path},font=\sffamily\small
  }}
\tikzset{MyLoop_Special_1/.style =  {-{Stealth[scale=1]},ultra thick, to path={
  [looseness=1, min distance=8mm,max distance=8mm,out=180,in=180]
  \tikz@to@curve@path},font=\sffamily\small
  }}
 \tikzset{MyLoop_Special_2/.style =  {-{Stealth[scale=1]},ultra thick, to path={
  [looseness=1, min distance=8mm,max distance=8mm,out=0,in=0]
  \tikz@to@curve@path},font=\sffamily\small
  }}
 \tikzset{MyLink_Special_3/.style =  {-{Stealth[scale=1]},ultra thick, to path={
  [looseness=0, min distance=2mm,max distance=2mm,out=45,in=135]
  \tikz@to@path},font=\sffamily\small
  }}
 \tikzset{MyLoop_UpperLeft/.style =  {-{Stealth[scale=1]},ultra thick, to path={
  [looseness=1, min distance=4mm,max distance=4mm,out=105,in=165]
  \tikz@to@path},font=\sffamily\small
  }}
 \tikzset{MyLoop_LowerLeft/.style =  {-{Stealth[scale=1]},ultra thick, to path={
  [looseness=1, min distance=4mm,max distance=4mm,out=195,in=255]
  \tikz@to@path},font=\sffamily\small
  }}
 \tikzset{MyLoop_LowerRight/.style =  {-{Stealth[scale=1]},ultra thick, to path={
  [looseness=1, min distance=4mm,max distance=4mm,out=-15,in=-75]
  \tikz@to@path},font=\sffamily\small
  }}
 \tikzset{MyLoop_UpperRight/.style =  {-{Stealth[scale=1]},ultra thick, to path={
  [looseness=1, min distance=4mm,max distance=4mm,out=15,in=75]
  \tikz@to@path},font=\sffamily\small
  }}
 \tikzset{MyLoop_UpperLeft_1/.style =  {-{Stealth[scale=1]},ultra thick, to path={
  [looseness=1, min distance=4mm,max distance=4mm,out=60,in=120]
  \tikz@to@path},font=\sffamily\small
  }}
\tikzset{MyLoop_UpperLeft_2/.style =  {-{Stealth[scale=1]},ultra thick, to path={
  [looseness=1, min distance=4mm,max distance=4mm,out=150,in=210]
  \tikz@to@path},font=\sffamily\small
  }}
 \tikzset{MyLoop_LowerRight_1/.style =  {-{Stealth[scale=1]},ultra thick, to path={
  [looseness=1, min distance=4mm,max distance=4mm,out=-60,in=-120]
  \tikz@to@path},font=\sffamily\small
  }}
 \tikzset{MyLoop_LowerRight_2/.style =  {-{Stealth[scale=1]},ultra thick, to path={
  [looseness=1, min distance=4mm,max distance=4mm,out=-30,in=30]
  \tikz@to@path},font=\sffamily\small
  }}
\makeatletter

\usepackage{cancel}

\usepackage{amsmath}
\usepackage{amssymb}
\usepackage{amstext}
\usepackage{amsthm}
\usepackage{amsfonts}       
\usepackage{nicefrac}       
\usepackage{enumerate}
\usepackage{bbold}

\usepackage{extarrows}
\usepackage{esvect}
\usepackage{MnSymbol}

\usepackage{relsize}
\usepackage{array}
\usepackage{dsfont}

\usepackage{changes}
\usepackage{caption}
\usepackage{subcaption}	
\usepackage{float}

\usepackage{pdfpages} 

\usepackage[style=alphabetic,backend=biber]{biblatex}
\usepackage{hyperref}       

\hypersetup{colorlinks=true,linkcolor=blue,citecolor=blue, linkbordercolor=blue, citebordercolor=blue, linktoc=all}

\usepackage{nomencl}
\makenomenclature

\newtheorem{thm}{Theorem} 
\newtheorem{rem}[thm]{Remark} 
\newtheorem*{rem*}{Remark}
\newtheorem{MyDef}[thm]{Definition} 
\newtheorem{lemma}  [thm]{Lemma}  

\newtheorem{claim}[thm]{Claim}

\definechangesauthor[name=Bernd, color=red]{Bernd}
\newcommand{\G}{\Gamma} %
\newcommand{\g}{\gamma} %
\newcommand{\e}{\mathrm{e}} %
\newcommand{\N}{\mathbb{N}} 
\newcommand{\Z}{\mathbb{Z}} 
\newcommand{\R}{\mathbb{R}} 
\newcommand{\C}{\mathbb{C}} 
\newcommand{\Id}{\mathbb{1}} 

\newcommand{\BL}{\mathcal{BL}} %
\newcommand{\System}{\mathcal{S}} %
\newcommand{\CMean}{\mathcal{C}} %
\newcommand{\Edge}{\mathcal{E}} %
\newcommand{\LL}{\mathcal{L}} %
\newcommand{\B}{\mathcal{B}} %

\newcommand{\Prob}{\mathcal{P}}
\newcommand{\ProbStates}{\mathcal{PV}} 

\newcommand{\Zen}{\mathcal{Z}}

\newcommand{\MM}{\mathcal{M}}

\newcommand{\Null}{\mathcal{N}} %
\newcommand{\NP}{\mathcal{N}_{\bP} } %
\newcommand{\NPzero}{\mathcal{N}_{\bP_0} } %

\newcommand{\Order}{\mathcal{O}}
\newcommand\smallO{
  \mathchoice
    {{\scriptstyle\mathcal{O}}}
    {{\scriptstyle\mathcal{O}}}
    {{\scriptscriptstyle\mathcal{O}}}
    {\scalebox{.7}{$\scriptscriptstyle\mathcal{O}$}}
  }

\renewcommand{\d}{\text{ d}} %

\newcommand{\Span}{\text{span }}

\newcommand{\Kern}{\text{kern }}

\renewcommand{\Re}{\text{Re }}
\renewcommand{\Im}{\text{Im }}
\newcommand{\Image}{\text{image }}

\newcommand{\Fin}{\text{Fin}}

\newcommand{\bc}{\boldsymbol{c}}
\newcommand{\bd}{\boldsymbol{d}}

\newcommand{\bw}{\boldsymbol{w}}

\newcommand{\bX}{\boldsymbol{X}}
\newcommand{\bY}{\boldsymbol{Y}}
\newcommand{\bZ}{\boldsymbol{Z}}

\newcommand{\bp}{\boldsymbol{p}}
\newcommand{\bP}{\boldsymbol{P}}

\newcommand{\be}{\boldsymbol{e}}
\newcommand{\bE}{\boldsymbol{E}}

\newcommand{\bN}{\boldsymbol{N}}
\newcommand{\bT}{\boldsymbol{T}}

\newcommand{\bzero}{\boldsymbol{0}}
\newcommand{\bone}{\boldsymbol{1}}

\newcommand{\bPF}{\bP^{[F]}}
\newcommand{\bXF}{\bX^{[F]}}
\newcommand{\bPM}{\bP^{[M]}}

\newcommand{\PF}{P^{[F]}}
\newcommand{\GF}{\G^{[F]}}

\newcommand{\Sgap}{S_g \left(\GF \right)}
\newcommand{\JGF}{J\left(\GF\right)}

\newcommand{\bPt}{\bP\left(t \,\bigl|\, \bP_0\right)}

\newcommand{\bPFt}{\bP^{F}\left(t \,\bigl|\, \bP_0^{[F]}\right)}

\newcommand{\bPInftyF}{\bP_\infty^{F}\left(\bP_0^{[F]} \right)}
\newcommand{\state}[1]{#1}

\newcommand{\lme}{l^1_\text{m.e.}(\Omega, Q)}
\newcommand{\ltb}{l^1_\text{l.t.b.}(\Omega, \G)}

\newcommand{\NN}[1]{N^{(\state{#1}) } }
\renewcommand{\t}[1]{T^{(\state{#1}) } }

\tikzset{Vertex/.style={circle, draw, very thick, scale = 0.7}}

\tikzset{State/.style={circle, draw, very thick, scale = 0.8}}
\tikzset{ProbVector/.style={draw, color=white}}
\tikzset{NodeRectangle/.style={rectangle, draw, very thick, scale = 0.8}}
\tikzset{Link/.style={-{Stealth[scale=1]},ultra thick}}

\tikzstyle{arrow} = [thick,->,>=stealth]
\tikzstyle{Edge} = [ultra thick,-]

\newcount\colveccount
\newcommand*\colvec[1]{
        \global\colveccount#1
        \begin{pmatrix}
        \colvecnext
}
\def\colvecnext#1{
        #1
        \global\advance\colveccount-1
        \ifnum\colveccount>0
                \\
                \expandafter\colvecnext
        \else
                \end{pmatrix}
        \fi
}

\definecolor{MyYellow}{rgb}{1.0, 0.75, 0.0}

\definecolor{CrimsonGlory}{rgb}{0.75, 0.0, 0.2}
\definecolor{DarkRed}{rgb}{0.55, 0.0, 0.0}

\definecolor{MyBlue}{RGB}{0, 91, 180}   
\definecolor{MyRed}{RGB}{230, 0.0, 30}

\definecolor{ElectricPurple}{rgb}{0.75, 0.0, 1.0}
\definecolor{DarkGreen}{rgb}{0.0, 0.4, 0.0} 
\definecolor{MyOrange}{RGB}{235, 100, 2}

\definecolor{MyGreen}{rgb}{0.0, 0.5, 0.0} 
\definecolor{MyPurple}{RGB}{141, 0, 255} 
\definecolor{MyCyan}{rgb}{0, 0.67, 0.67} 

\definecolor{TUDa_0c}{RGB}{137, 137, 137}
\definecolor{TUDa_0d}{RGB}{83, 83, 83}

\definecolor{TUDa_1a}{RGB}{93, 133, 195} 
\definecolor{TUDa_1b}{RGB}{0, 90, 169}   
\definecolor{TUDa_1c}{RGB}{0, 78, 138}   
\definecolor{TUDa_1d}{RGB}{36, 53, 114}  

\definecolor{TUDa_2a}{RGB}{0, 156, 218}
\definecolor{TUDa_2b}{RGB}{0, 131, 204}
\definecolor{TUDa_2c}{RGB}{0, 104, 157}
\definecolor{TUDa_2d}{RGB}{0, 78, 115}

\definecolor{TUDa_3a}{RGB}{80, 182, 140} 

\definecolor{TUDa_4a}{RGB}{175, 204, 80} 
\definecolor{TUDa_4b}{RGB}{153, 192, 0} 
\definecolor{TUDa_4c}{RGB}{127, 171, 22} 
\definecolor{TUDa_4d}{RGB}{106, 139, 55} 

\definecolor{TUDa_8a}{RGB}{238, 122, 52} 
\definecolor{TUDa_8b}{RGB}{236, 101, 0}  
\definecolor{TUDa_8c}{RGB}{204, 76, 3}   
\definecolor{TUDa_8d}{RGB}{169, 73, 19}  

\definecolor{TUDa_9a}{RGB}{233, 80, 62} 
\definecolor{TUDa_9b}{RGB}{230, 0, 26} 
\definecolor{TUDa_9c}{RGB}{185, 15, 34} 

\definecolor{TUDa_10a}{RGB}{201, 48, 142} 
\definecolor{TUDa_10b}{RGB}{166, 0, 132} 
\definecolor{TUDa_10c}{RGB}{149, 17, 105}  
\definecolor{TUDa_10d}{RGB}{115, 32, 84}  

\definecolor{TUDa_11a}{RGB}{128, 69, 151} 
\definecolor{TUDa_11b}{RGB}{114, 16, 133} 
\definecolor{TUDa_11c}{RGB}{97, 28, 115}  
\definecolor{TUDa_11d}{RGB}{76, 34, 106}  

\title{Long-term behavior of the master equation on a countable network and approximation methods of the (stationary) solutions via finite subsystems in the thermodynamic limit}
\author{$\text{Bernd Michael Fernengel}^{(1)}, \text{Thilo Gross}^{(1)},  \text{Wolfram Just}^{(2)} $  \\ 
(1) HIFMB Oldenburg,  \\ 
(2) Institute of Mathematics University of Rostock
}
\date{\today}

\begin{document}

\maketitle
\tableofcontents
\section{Abstract} \label{Chapter_Abstract}

The Master equation on directed networks - also called the differential Chapman-Kolmogorov equation - is a linear differential equation, which describes the probability evolution in a discrete system. While this is well understood, if the underlying  graph is finite, the mathematics required for the treatment of  a network with countable many nodes is way more complicated and advanced. 
In this paper we provide criteria for the rates of the system, which makes it possible to \emph{approximate} the solution by \emph{finite subsystems} in the \emph{thermodynamic limit}. 
By writing the phase space as a direct sum of \emph{stationary} states and states which vanish in the time limit, we give a new proof of when the time limit for an countable, infinite dimensional system exists and when it can be interchanged with the limit of large systems. We give sufficient criteria, when these two limits commute and demonstrate on various examples, what happens when these criteria are violated and only one of these limits exists.

\section{ Introduction} \label{Chapter_SettingUpTheSystemAndTheNeedOfFiniteTruncation}

Master equations describe the time evolution of probabilities in a system. For a discrete state space, this dynamic comes from the transition between different states, which means that the system can be modeled as a directed, weighted graph, where the directed links model the transitions and the non-negative weights the transition rate. 

The master equation can be viewed a a differential version of the Chapman-Kolmogorov equation and is usually formulated as the following initial value problem

\begin{equation}  \label{Equation_MasterEquation_components} \tag{MEq1}
\begin{aligned}
\frac{d}{d t} P^{(i)}(t) 
&=
\mathlarger{\sum}\limits_{{j\in \Omega}\atop {j\neq i}} \; \big(\,  P^{(j)}(t) \, \gamma_{j\to i} - P^{(i)}(t) \, \gamma_{i \to j} \, \big) \\
P^{(i)}(t=0) &= P^{(i)}_0, 
\end{aligned}
\end{equation}

or in matrix-vector notation 

\begin{equation} \label{MasterEquation} \tag{MEq2}
\begin{aligned}
\dot{\bP}(t) &= \,\G \, \bP(t)  \\
\bP(t=0) &= \bP_0 
\end{aligned}
\end{equation}

with  the so-called \emph{generator matrix}

\begin{equation} \label{MasterEquation_Generator} \tag{Gen}
\begin{aligned}
 \G^{(i, j)}&= \, \begin{cases}
 \g_{j\to i}  &\text{  , for } i \neq j \\
  -\mathlarger{\sum}\limits_{k \in \Omega}^{} \g_{j\to k} &\text{  , for } i = j\, .
\end{cases}
\end{aligned}
\end{equation}

From a mathematical point of view, natural questions to ask are whether unique solutions of the master equations exist, if the solutions possesses the desired properties of a probability distribution (being non-negative and normalized) and under which condition the long-term behavior $t \to \infty $ exists.

For a finite system, these questions can be answered: Not only does the solution exist for all times and satisfies all properties of a probability distribution, but it is also possible to compute not only the transition matrix $\e^{t \, \G} $ via its power series, but also the stationary solution, for which an explicit expression in terms of the transition rates exists \cite{mirzaev2013laplacian, haag2017modelling, fernengel2022obtaining}. 

Whereas for finite dimensional systems one can rely on any standard textbook of ordinary differential equation (see for example \cite{konigsberger2006analysis, fischer2004mathematik, van1992stochastic, honerkamp2012statistical}), the literature is considerably more sparse in the countable, infinite dimensional case, where results in the physical literature are often merely states without a formal proof (see \cite{grimmett2020probability,norris1998markov}). This is due to the fact that one has to deal with operators on the non-reflexive Banach space $ l^1(\N) $.

Since equation \eqref{Equation_MasterEquation_components} contains a sum over the state space $ \Omega $ as an index set, it is crucial to require $ \Omega $ to be \emph{countable} infinite, which is implicitly assumed throughout this paper.

This paper is meant to close this gap, or at least part of it. In our research we tried to combine aspects of different theories, such as Markov chains \cite{bremaud2013markov, douc2018markov, privault2013understanding}, operator theory \cite{fischer2004mathematik, werner2011funktionalanalysis}, C0-semigroups \cite{engel2000one, batkai2017positive}, and infinite graphs \cite{diestel2024graph}. 

In addition to the questions about master equations listed above, we are interested in how the countable, infinite dimensional case is linked to the (well understood) finite dimensional one, that is, whether there is a \emph{finite subsystem} capturing the essential dynamics of the infinite system. This is not only handy when transferring the formulas for the stationary solution from the \emph{finite} dimensional -  to the \emph{infinite} dimensional case, but also when it comes to numerical approximations. This means that we have two different limits, the \emph{time limit} and the limit of the \emph{system size}, which we call \emph{thermodynamic} limit. Or more more formally: Given a countable, infinite system $ \Omega $ and a finite subsystem $F \subseteq \Omega$, $ |F| < \infty $, there are four possible limits one can consider

\begin{itemize}
\item[-] The time limit $ \lim\limits_{t \to \infty } \bP^{F}\left(t \right)  $
\item[-] The limit of the system size $ \lim\limits_{|F| \to \infty } \bP^{F}\left(t \right)  $
\item[-] The time limit, followed by the limit of the system size:  $ \lim\limits_{|F| \to \infty }\lim\limits_{t \to \infty }  \bP^{F}\left(t \right)  $
\item[-] The limit of the system size, followed by the time limit:  $ \lim\limits_{t \to \infty } \lim\limits_{|F| \to \infty }   \bP^{F} \left(t \right)  $ 
\end{itemize}

The following figure illustrates serves as an illustration:

\begin{figure}[H]
\begin{center}
    \scalebox{0.6}{
    \begin{tikzpicture} [scale = 1]
    \node (tLessThanInfty)      at (0,1.5)  { \Large $ t< \infty $} ;
    \node (OmegaLessThanInfty)  at (-3.1,0) { \huge $    F \subseteq \Omega  \atop |F| < \infty  $} ;
    \node (OmegaEqualsInfty)    at (-3.1,-3.3)   { \Large $ |\Omega|= \infty $} ;
    \node (tEqualsInfty)        at (4.5,1.5)   {\Large $ ``t= \infty" $} ;
    \draw[->, ultra thick]  (-2.0 , 1.8) -- (-2.0,-4.3) ; 
    \draw[->, ultra thick]  (-4.7 , 1.0) -- ( 6.0, 1.0);  
    \draw( 7.0 , 1.0) node[font=\large] {time limit} ;
    \draw(-2.1 ,-4.6) node[font=\large] {thermodynamic } ;
    \draw(-2.1 ,-5.0) node[font=\large] {limit} ;
    %
    \node (bPFt)    at (0,0)  {$\Large \bP^{F}(t) $} ;
    \node(bPFinfty) at (4.7,0)  {$ \Large \bP^{F}_\infty $};
    \node(bPt)      at (0,-3.4) {$ \Large\bP(t)  $};
    \node(bPinfty)  at (4.3,-3.4) {$ \Large  \bP_\infty  $};
    \node(bP)       at (4.7,-2.7) {$ \Large \bP_*   $};
    \draw[very thick] (3.0, -2.6) -- (5.9, -3.25) ; 
    \path[Link, DarkGreen] (bPFt) edge node[above]  {$ t \to \infty $} (bPFinfty);
    \path[Link, MyPurple] (bPt) edge node[above]  {$ t \to \infty $} (bPinfty);
    \path[Link, MyBlue] (bPFt) edge node[left]  {$ |F| \to \infty $} (bPt);
    \path[Link, MyOrange] (bPFinfty) edge node[right]  {$ |F| \to \infty $} (bP);
    \end{tikzpicture}
    }
\caption{Illustrating the time limit and the thermodynamic limit.  }
\label{GeneralRoadMapToPaper} 
\end{center}
\end{figure}
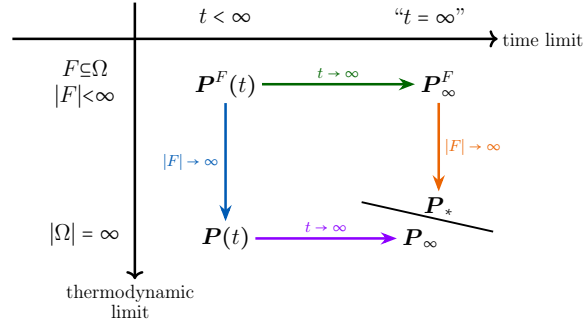

Note that it is apriori not clear if all of the limits are well defined and if this diagram commutes. 

Before defining in detail, what we mean by the notation $ \bP^{F} $ or the thermodynamic limit $|F| \to \infty $ we want to illustrate on two simple examples, that the diagram depicted in figure \ref{GeneralRoadMapToPaper} is not necessarily commutative. 

We start by looking at figure \ref{Fig_Network_FlowingInACircle}, where the probability flows from left to right, but is `reshuffled' . It can be shown - and we hope the readers believe us at this point, details are in chapter \ref{Chapter_ExampleOfANetworkWithoutDetailedBalance} - that for suitable rates the time limit for the countable, infinite dimensional system exists, and can be approximated by the sequence of increasing subsets $F_n = \{-n, \dots, n\} $. However, when a different sequence of increasing subsets in chosen, e.g. $\widetilde{F}_n = \{-n, \dots n+1\} $ is chosen, then this limit does not exist, since $\lim\limits_{n \to \infty} \lim\limits_{t \to \infty} \bP^{F}(t) = \lim\limits_{n \to \infty} \bE_{n+1} = \lightning $ .

\begin{figure}[H]

\begin{center}
    \scalebox{0.7}{
        \begin{tikzpicture}[scale=2]
        \centering
        \node[State](-4)  at (-4,0) [circle,draw] {$\dots$} ;
        \node[State](-3)  at (-3,0) [circle,draw] {$-3$} ;
        \node[State](-2)  at (-2,0) [circle,draw] {$-2$} ;
        \node[State](-1)  at (-1,0) [circle,draw] {$-1$} ;
        \node[State](0)   at (0,0) [circle,draw] {$0$} ;
        \node[State](1)   at (1,0) [circle,draw] {$1$} ;
        \node[State](2)   at (2,0) [circle,draw] {$2$} ;
        \node[State](3)   at (3,0) [circle,draw] {$3$} ;
        \node[State](4)   at (4,0) [circle,draw] {$\dots$} ;
        \path[Link] (-4) edge node[above] {} (-3);
        \path[Link] (-3) edge node[above] {} (-2);
        \path[Link] (-2) edge node[above] {} (-1);
        \path[Link] (-1) edge node[above] {} (0);
        \path[Link] (0) edge node[above] {} (1);
        \path[Link] (1) edge node[above] {} (2);
        \path[Link] (2) edge node[above] {} (3);
        \path[Link] (3) edge node[above] {} (4);
        \path[Link, bend left] (1) edge node[above] {} (-1);
        \path[Link, bend left] (2) edge node[above] {} (-2);
        \path[Link, bend left] (3) edge node[above] {} (-3);
        \path[Link, bend left] (4) edge node[above] {} (-4);
        \end{tikzpicture}
    }
    \caption*{full network } 
\end{center}

\begin{minipage}{0.49\textwidth}
\begin{center}
    \scalebox{0.6}{
    \begin{tikzpicture}[scale=2]
    \node[State](-1)   at (-1,0) [circle,draw] {$-1$} ;
    \node[State](0)   at (0,0) [circle,draw] {$0$} ;
    \node[State](1)   at (1,0) [circle,draw] {$1$} ;
    %
    \path[Link] (-1) edge node[above] {} (0);
    \path[Link] (0) edge node[above] {} (1);
    %
    \path[Link, bend left] (1) edge node[above] {} (-1);
    \end{tikzpicture}
    }
    \caption*{subnetwork  $F_1$} 
\end{center}
\end{minipage}\begin{minipage}{0.49\textwidth}
\begin{center}
    \scalebox{0.6}{
    \begin{tikzpicture}[scale=2]
    \node[State](-1)   at (-1,0) [circle,draw] {$-1$} ;
    \node[State](0)   at (0,0) [circle,draw] {$0$} ;
    \node[State](1)   at (1,0) [circle,draw] {$1$} ;
    \node[State](2)   at (2,0) [circle,draw] {$2$} ;
    %
    \path[Link] (-1) edge node[above] {} (0);
    \path[Link] (0) edge node[above] {} (1);
    \path[Link] (1) edge node[above] {} (2);
    %
    \path[Link, bend left] (1) edge node[above] {} (-1);
    \end{tikzpicture}
    }
    \caption*{subnetwork  $\widetilde{F}_1$} 
\end{center}
\end{minipage}
\begin{minipage}{0.49\textwidth}
\begin{center}
    \scalebox{0.6}{
    \begin{tikzpicture}[scale=2]
    \node[State](-2)   at (-2,0) [circle,draw] {$-2$} ;
    \node[State](-1)   at (-1,0) [circle,draw] {$-1$} ;
    \node[State](0)   at (0,0) [circle,draw] {$0$} ;
    \node[State](1)   at (1,0) [circle,draw] {$1$} ;
    \node[State](2)   at (2,0) [circle,draw] {$2$} ;
    %
    \path[Link] (-2) edge node[above] {} (-1);
    \path[Link] (-1) edge node[above] {} (0);
    \path[Link] (0) edge node[above] {} (1);
    \path[Link] (1) edge node[above] {} (2);
    %
    \path[Link, bend left] (1) edge node[above] {} (-1);
    \path[Link, bend left] (2) edge node[above] {} (-2);
    \end{tikzpicture}
    }
    \caption*{subnetwork  $F_2$} 
\end{center}
\end{minipage}\begin{minipage}{0.49\textwidth}
\begin{center}
    \scalebox{0.6}{
    \begin{tikzpicture}[scale=2]
    \node[State](-2)   at (-2,0) [circle,draw] {$-2$} ;
    \node[State](-1)   at (-1,0) [circle,draw] {$-1$} ;
    \node[State](0)   at (0,0) [circle,draw] {$0$} ;
    \node[State](1)   at (1,0) [circle,draw] {$1$} ;
    \node[State](2)   at (2,0) [circle,draw] {$2$} ;
    \node[State](3)   at (3,0) [circle,draw] {$3$} ;
    %
    \path[Link] (-2) edge node[above] {} (-1);
    \path[Link] (-1) edge node[above] {} (0);
    \path[Link] (0) edge node[above] {} (1);
    \path[Link] (1) edge node[above] {} (2);
    \path[Link] (2) edge node[above] {} (3);
    \path[Link, bend left] (1) edge node[above] {} (-1);
    \path[Link, bend left] (2) edge node[above] {} (-2);
    \end{tikzpicture}
    }
    \caption*{subnetwork  $\widetilde{F}_2$} 
\end{center}
\end{minipage}
\begin{minipage}{0.49\textwidth}
\begin{center}
    \scalebox{0.6}{
    \begin{tikzpicture}[scale=2]
    \node[State](-3)   at (-3,0) [circle,draw] {$-3$} ;
    \node[State](-2)   at (-2,0) [circle,draw] {$-2$} ;
    \node[State](-1)   at (-1,0) [circle,draw] {$-1$} ;
    \node[State](0)   at (0,0) [circle,draw] {$0$} ;
    \node[State](1)   at (1,0) [circle,draw] {$1$} ;
    \node[State](2)   at (2,0) [circle,draw] {$2$} ;
    \node[State](3)   at (3,0) [circle,draw] {$3$} ;
    \path[Link] (-3) edge node[above] {} (-2);
    \path[Link] (-2) edge node[above] {} (-1);
    \path[Link] (-1) edge node[above] {} (0);
    \path[Link] (0) edge node[above] {} (1);
    \path[Link] (1) edge node[above] {} (2);
    \path[Link] (2) edge node[above] {} (3);
    \path[Link, bend left] (1) edge node[above] {} (-1);
    \path[Link, bend left] (2) edge node[above] {} (-2);
    \path[Link, bend left] (3) edge node[above] {} (-3);
    \end{tikzpicture}
    }
    \caption*{subnetwork  $F_3$} 
\end{center}
\end{minipage}\begin{minipage}{0.49\textwidth}
\begin{center}
    \scalebox{0.6}{
    \begin{tikzpicture}[scale=2]
    \node[State](-3)   at (-3,0) [circle,draw] {$-3$} ;
    \node[State](-2)   at (-2,0) [circle,draw] {$-2$} ;
    \node[State](-1)   at (-1,0) [circle,draw] {$-1$} ;
    \node[State](0)   at (0,0) [circle,draw] {$0$} ;
    \node[State](1)   at (1,0) [circle,draw] {$1$} ;
    \node[State](2)   at (2,0) [circle,draw] {$2$} ;
    \node[State](3)   at (3,0) [circle,draw] {$3$} ;
    \node[State](4)   at (4,0) [circle,draw] {4} ;
    \path[Link] (-3) edge node[above] {} (-2);
    \path[Link] (-2) edge node[above] {} (-1);
    \path[Link] (-1) edge node[above] {} (0);
    \path[Link] (0) edge node[above] {} (1);
    \path[Link] (1) edge node[above] {} (2);
    \path[Link] (2) edge node[above] {} (3);
    \path[Link] (3) edge node[above] {} (4);
    \path[Link, bend left] (1) edge node[above] {} (-1);
    \path[Link, bend left] (2) edge node[above] {} (-2);
    \path[Link, bend left] (3) edge node[above] {} (-3);
    \end{tikzpicture}
    }
    \caption*{subnetwork  $\widetilde{F}_3$} 
\end{center}
\end{minipage}

\caption{Example of a network, where the time limit of the \emph{infinite} dimensional system exists, but not the thermodynamic limit of the stationary solutions of \emph{finite} subsystems. }
\label{Fig_Network_FlowingInACircle}
\end{figure}

While the previous argument was based on the topology of the network, the other case is based on a suitable choice of rates. We consider the network depicted in figure \ref{Fig_Network_LinearChainOnN0_OneTrappingState}, which is a linear chain with one open end on one side, and a trapping state on the other. For $p \in (\frac{1}{2}, 1 ) $, every state $ n \in \N_{\geq 1} $ can be shown to be transient, that is there is a non-zero chance that a trajectory stating at a state $n\in \N$ does not return. Informally speaking, there are two attracting `forces', one is the state $0$, the other is infinity. When we consider now a finite network $F_n = \{0, \dots, n\} $, the attracting force at infinity vanishes, meaning that the limiting solution of the finite subnetwork is $\bE_0$. This mean that the limit of stationary solutions of \emph{finite} systems exists $ \lim\limits_{n \to \infty} \lim\limits_{t \to \infty} \bP^{F_n}(t) =  \lim\limits_{n \to \infty} \bE_0 = \bE_0 $, while the \emph{time} limit of the \emph{infinite} dimensional systems does \emph{not} exist, only the pointwise limit $\bP(t) \xlongrightarrow[\text{pointwise}]{t \to \infty} (p_*, 0, 0, 0, \dots) $, for some $p_* \in (0, 1)$, which is not a probability measure.

\begin{figure}[H]
\begin{center}
\begin{tikzpicture}[scale=2]
\node[State](0)   at (0,0) [circle,draw] {0} ;
\node[State](1)   at (1,0) [circle,draw] {1} ; 
\node[State](2)   at (2,0) [circle,draw] {2} ;
\node[State](3)   at (3,0) [circle,draw] {3} ;
\node[State](4) at (4,0) [circle,draw] {4} ;
\node[State](5)   at (5,0) [circle,draw] {$\dots$} ;
%
\path[Link,bend left] (1) edge node[above] {$p \, q $} (2);
\path[Link,bend left] (2) edge node[above] {$p \, q^2$} (3);
\path[Link,bend left] (3) edge node[above] {$p \, q^3 $} (4);
\path[Link,bend left] (3) edge node[above] {} (4);
\path[Link,bend left] (4) edge node[above] {} (5);
%
\path[Link,bend left] (5) edge node[below] {} (4);
\path[Link,bend left] (4) edge node[below] {$ (1-p) \, q^4  $} (3);
\path[Link,bend left] (3) edge node[below] {$ (1-p) \, q^3  $} (2);
\path[Link,bend left] (2) edge node[below] {$ (1-p) \, q^2  $} (1);
\path[Link] (1) edge node[below] {$(1-p) \, q $} (0);
\end{tikzpicture}
\caption{The linear chain with an open end on one side, and a trapping state on the other \\ for $q \in (0, 1)$ and $p \in (\frac{1}{2}, 1)$.  }
\label{Fig_Network_LinearChainOnN0_OneTrappingState} 
\end{center}
\end{figure}
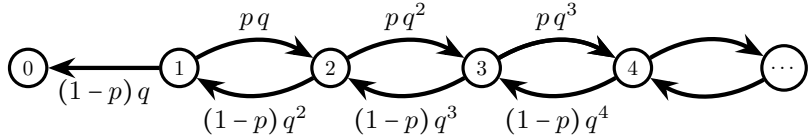

We hope that these - hand waving - arguments are enough to convince the reader that the interchanging of the two limits is no trivial procedure, but requires careful handling. 

The outline of the paper is the following: 

While section \ref{Chaper_FromMarkovChainsToMasterEquations} is a reminder of the connection between master equations and Markov chains, we prove the existence and uniqueness of the solution of the \emph{infinite} dimensional master equation $\bP(t )$ in section \ref{Chaper_TheSolutionOfTheInfiniteDimensionalMasterEquation}.

After the definition of a \emph{finite subsystem} in section \ref{Chapter_TheThermodynamicLimitOfTheMasterEquations}, we show that the solution of the countable, infinite dimensional master equation can be approximated by the solution  $ \bP^{F}(t) $ of finite subsystems in the thermodynamic limit $ \bP^{F}(t ) {\color{TUDa_1b} \xlongrightarrow{|F| \to \infty } } \bP(t ) $. 

Chapter \ref{Chapter_MarkovChains} summarizes equivalent conditions for steady states of - both discrete-time and continuous-time - Markov chains, in the language that we require later on, in particular showing that the existence of stationary solutions is closely related to the Markov chain being \emph{positive recurrent}. 

In section \ref{Chapter_TheLongTermBehaviorOfAnInfiniteDimensionalMasterEquation} we show that the solution of the countable, infinite dimensional master equation converges to a stationary probability distribution, if the associated system is both irreducible and positive recurrent. 

While neither the time limit  nor the limit of the system size is necessarily uniform, it is shown in section \ref{Chapter_TheThermodynamicLimitOfTheStationarySolutionsOfTheMasterEquations} that if the stationary solutions of the finite subsystems converges, it will converge to a stationary solution of the countable, infinite dimensional system. Moreover, if the stationary solutions are of a special form, we can show that the thermodynamic limit of the stationary solution for finite systems exists $ \bP^{F}_\infty(\bPF_0) {\color{TUDa_8b} \xlongrightarrow{|F| \to \infty } } \bP_\infty(\bP_0) $ and hence the the diagram of figure \ref{GeneralRoadMapToPaper} commutes. In particular, the long-term behavior of a countable, infinite dimensional master equation of an irreducible system is not only well defined, but the explicit form of this stationary solution can also be recovered by looking at stationary solutions of corresponding finite subsystems and letting the system size tend to infinity. 

In each of the following sections \ref{Chapter_AFirstNonTrivialExample_TheLinearChainOnN}, \ref{Chapter_ASecondExample_TheLinearChainOnZ}, \ref{Chapter_TheInfiniteDimensionalHypercube} and \ref{Chapter_ExampleOfANetworkWithoutDetailedBalance} an example is being discussed in order to demonstrate the implications of our theorems. Chapter \ref{Chapter_DiscussionAndConclusion} summarizes our work and points to possible generalizations.


\section{The solution of the countable, infinite dimensional master equation}\label{Chaper_TheSolutionOfTheInfiniteDimensionalMasterEquation}

In the following section we study the solution of the countable, infinite dimensional master equation, by first showing that the existence of a unique solution in the Banach space $l^1(\Omega) $, which is something that is not guaranteed for a infinite, countable state space $\Omega$. We continue by showing that the solution has the desired properties of a probability vector (namely being non-negative and normalized to one) for all times $t \geq 0$. 

For \emph{irreducible} networks the solution operator can be shown to be - element-wise - strictly positive, resulting in - at most - one stationary solution. 


We note, that the content of this section is already known in the literature. Countable systems of differential equations are dealt with in \cite{bellman1947boundedness} and \cite{batkai2017positive}, while the existence of uniform semigroups and bounded generators is stated in \cite{grimmett2020probability} (chapter 6.10 "uniform semigroups, Theorem (10) ") without a proof. 

We decided to include this section nevertheless, both for the convenience of the reader and to have a self-contained description. 


\subsection{Existence and uniqueness of the solution 
}\label{Section_ExistenceAndUniquenessOfTheSolution}

In this section we will prove that the master equation \eqref{Eq_MasterEq} has a unique solution in $l^1(\Omega) $, by following the arguments of \cite{bellman1947boundedness} and \cite{batkai2017positive}. Let the generator matrix $\G \in \R^{\Omega \times \Omega } $ be as in equation \eqref{Eq_GeneratorMatrix}. The master equation itself can be written in two equivalent ways, namely a \emph{differential}- and an \emph{integral} form:

\begin{minipage}{0.35\textwidth}
\begin{equation*} 
\begin{aligned}
\frac{d}{d t} \bP(t) 
&=
\G \, \bP(t) \\ 
\bP(t=0) 
&=
\bP_0
\end{aligned}
\end{equation*}
\end{minipage} $ \hspace*{-10mm} \iff $ \begin{minipage}{0.65\textwidth}
\begin{equation}\label{Eq_MasterEquation_IntegralForm} \tag{master eq. integral form}
\begin{aligned}
\bP(t) 
=
\bP_0 + \int_0^t \G \, \bP(\tau) \d \tau. 
\end{aligned}
\end{equation}
\end{minipage}

The \emph{existence} and \emph{uniqueness} of the solution is shown in details in lemma \ref{Lemma_ExistenceAndUniquenessOfTheMasterEquation}. In the following, we outline the basic structure of the proof, using and adaptive version of the theorem of Picard-Lindelöf. 


First, we define a sequence iteratively, namely

\begin{equation*}
\begin{aligned}
\bP_0 
&:=
\bP(0) 
\\
\bP_{n+1}(t) 
&:=
\bP_0 + \int_0^t \G \,  \bP_n(\tau) \d \tau
\end{aligned}
\end{equation*}

We proceed by proving that this sequence is bounded (claim \ref{Lemma_ExistenceAndUniquenessOfTheMasterEquation_BoundednessOfTheSolution}), in particular

\begin{equation*}
\begin{aligned}
\|\bP_{n+1}(t) \|_1
&\leq 
\|\bP_{0} \|_1 \, \cdot \, \exp\left(
\int_0^t \| \G \|_{1}^\text{(op)} \d \tau
\right), 
\end{aligned}
\end{equation*}
 and that we can estimate the norm of the difference of two consecutive members of the sequence (claim \ref{Lemma_ExistenceAndUniquenessOfTheMasterEquation_EstimatingDifferenceBetweenTwoConsecutiveElements}): 
 
\begin{equation*}
\begin{aligned}
\| \left( \bP_{n} - \bP_{n-1}\right) (t) \|_1
&\leq
 \frac{\left(\int_0^t \| \G \|_{1}^\text{(op)} \d \tau  \right)^n}{n!}
\end{aligned}
\end{equation*}

From this we can deduce that this sequence is a Cauchy sequence in $ l^1(\Omega) $ (claim \ref{Lemma_ExistenceAndUniquenessOfTheMasterEquation_CauchySequence}), that is 

\begin{equation*}
\begin{aligned}
\| \left( \bP_{n+m} - \bP_{n}\right) (t) \|_1
&\leq
\underbrace{
\sum\limits_{k=n+1}^{n+m} \frac{\left(\int_0^t\| \G \|_{1}^\text{(op)} \d \tau  \right)^k}{k!}
}_{
< \, \infty 
}
\xlongrightarrow{m, n\to \infty} 0, 
\end{aligned}
\end{equation*}

Due to completeness of the space $l^1(\Omega)$, the sequence converges with respect to the $ \| \cdot \|_1$-norm to some function 
\begin{equation*}
\begin{aligned}
\bP \, :\, 
\R_{\geq 0} &\to l^1(\Omega)  \\
t &\mapsto \bP(t)
\end{aligned}
\end{equation*}

This function $\bP(t)$ solves the master equation. To see this, we note that it suffices to show, that the sequence $ \left(\int_0^t \G \, \bP_n(\tau) \d \tau \right)_{n \in \N}  \subseteq l^1(\Omega) $ converges in $l^1(\Omega)$ for $n \to \infty $ to  $\int_0^t \G \, \bP(\tau) \d \tau$, since then we have: 

\begin{equation*}
\begin{aligned}
\bP_n(t)
&=
\bP_0
+
\int_0^t \G \, \bP_n \hspace*{5mm} \Bigl| \lim\limits_{n \to \infty}
\\ 
\underbrace{
\lim\limits_{n \to \infty}\bP_n(t)
}_{
\bP(t)
}
&=
\bP_0
+
\underbrace{
\lim\limits_{n \to \infty}\int_0^t \G \, \bP_n 
}_{
\int_0^t \G \, \bP
}, 
\end{aligned}
\end{equation*}

which is just the integral form of the master equation \eqref{Eq_MasterEquation_IntegralForm}.

To show the convergence of  $ \left(\int_0^t \G \, \bP_n(\tau) \d \tau \right)_{n \in \N} $, we can compute the following:

\begin{equation*}
\begin{aligned}
\left\|
\int_0^t \G \, \bP_n(\tau) \d \tau - \int_0^t \G \, \bP(\tau) \d \tau
\right\|_1
&\leq
\int_0^t
\|
\G 
\bigl(
\bP_n(\tau)
-
\bP(\tau)
\bigr)
\|_1
\d \tau
\\ & \leq
t 
\, \cdot \, 
\sup\limits_{\tau \in [0, t]}
\| \G \|_{1}^\text{(op)}
\|
\bigl(
\bP_n
-
\bP
\bigr)(\tau)
\|_{1}
\\ &\xlongequal{\text{some } \xi \in [0, t]}
t 
\, \cdot \, 
\| \G \|_{1}^\text{(op)}
\, 
\underbrace{
\|
\bigl(
\bP_n(\xi)
-
\bP
\bigr)(\xi)
\|_{1}
}_{
\xlongrightarrow{n \to \infty} 0
}
\xlongrightarrow{n \to \infty}
0, 
\end{aligned}
\end{equation*}

where we used the fact that a continuous function attains its supremum on a compact set. 

In order to show uniqueness of the solution, we assume that there is another function

\begin{equation*}
\begin{aligned}
\tilde{\bP} : \R_{\geq \, 0} &\to l^1(\Omega) \\
t &\mapsto \tilde{\bP}(t)
\end{aligned}
\end{equation*}

solving the master equation. Then we can make the following estimation: 
\begin{equation*}
\begin{aligned}
\|
\bP(t)
- \tilde{\bP}(t)
\|_1
&=
\left\|
\int_0^t
\G
\bigl(
\bP(\tau)
- \tilde{\bP}(\tau)
\bigr)
\d \tau 
\right\|_1
\\ &\leq
\int_0^t 
\underbrace{
\| \G \|_{1}^\text{(op)}
}_{
\geq \, 0
}
\, \cdot \,
\underbrace{
\| \bigl(
\bP(\tau)
- \tilde{\bP}(\tau)
\bigr)
\|_1
}_{
\geq \, 0
}
\d \tau . 
\end{aligned}
\end{equation*}

By the integral version of Grönwall's lemma \ref{Lemma_Groenwall_IntegralVersion}, this implies that $ \| 
\bP(t)
- \tilde{\bP}(t)
\|_1 = 0$, which concludes the proof.


In the next section we will see, that the fact that $\G$ is a generator matrix - and therefore of the form described in equation \eqref{Eq_GeneratorMatrix} - and $\bP_0 \in \ProbStates$ a probability vector, guarantees that the solution $\bP(t) $ will also be a probability vector for all time $t >0$.


\subsection{The solution remains a probability vector}\label{Section_SolutionRemainsAProbabilitySequence}

When $ \bP_0 $ is a probability vector (that is $\bP_0 \in \ProbStates$), then the solution of the master equation $\bP(t \,|\, \bP_0) := \e^{\G \, t} \, \bP_0 $ is also a probability vector for all times $t \geq \, 0 $, provided that $ \| \G \|_{1}^\text{(op)} < \infty $. We divide the proof into three parts:

\begin{lemma}[finite probability flow] \label{Lemma_FiniteProbabilityFlow}

The probability flow is finite for all times $t \geq 0$: 

\begin{equation*}
\begin{aligned}
\sum\limits_{n, m \in \Omega} 
P^{(n)}(t) \, \g_{n \to m}
&=
\sum\limits_{n \in \Omega} \, 
P^{(n)}(t) \, 
\underbrace{
\sum\limits_{m \in \Omega} \, \g_{n \to m}
}_{
\g_{n \to }
}
=
\sum\limits_{n \in \Omega} \,
P^{(n)}(t) \, \g_{n \to }
\\ &\overset{{Hoelder}}{\leq}
\underbrace{
\left\| 
\left(P^{(n)}(t) \right)_{n \in \Omega} 
\right \|_1
}_{1}
\, \cdot \,
\underbrace{
\left\| 
\left( \g_{n \to } \right)_{n \in \Omega} 
\right\|_\infty
}_{
= 
\sup\{
\left( \g_{n \to } \right)_{n \in \Omega} 
\}
}
\\ &\overset{\text{}}{\leq}
\frac{1}{2}
\, 
\| \G \|_{1}^\text{(op)}
\, < \, \infty. 
\end{aligned}
\end{equation*}

\end{lemma}

\begin{lemma}[Boundedness of the solution] \label{Lemma_BoundednessOfTheSolution}

Every entry of the solution of the master equation remains in the interval $ [0,1] $, that is $\bP(t) \in [0,1]^{\Omega} $

\end{lemma}

\begin{proof}

When for some state $i \in\Omega $ and some time $t\geq 0$, the i-th entry of $\bP(t)$ equals one (zero), then the time derivative $ \frac{\d}{\d t} P^{(i)}(t) $ of that entry is negative (positive), since
    
\begin{equation}
\begin{aligned}
P^{(i)}(t) 
=
1 \Longrightarrow \frac{\d}{\d t} P^{(i)}(t) 
&\xlongequal{\eqref{Equation_MasterEquation_components}}
\mathlarger{\sum}\limits_{k \in \Omega \atop k\neq i}^{} \; \big(\,  \underbrace{P^{(k)}(t)}_{0} \, \gamma_{k\to i} - \underbrace{P^{(i)}(t)}_{1} \, \gamma_{i \to k} \, \big)
%
=
\,  -\sum\limits_{k \in \Omega}^{}  \gamma_{i \to k} 
\in (-\infty, 0]
\\
P^{(j)}(t) 
=
\, 0 \Longrightarrow 
\frac{\d}{\d t}  P^{(j)}(t) 
&\xlongequal{\eqref{Equation_MasterEquation_components}}
\, \mathlarger{\sum}\limits_{k \in \Omega \atop k\neq j}^{} \; \big(\,  P^{(k)}(t) \, \gamma_{k\to j} - \underbrace{P^{(j)}(t)}_{0} \, \gamma_{j \to k} \, \big) 
%
=
\, \sum\limits_{k \in \Omega}^{} P^{(k)}(t) \, \gamma_{k\to j}
\overset{\text{lemma } \ref{Lemma_FiniteProbabilityFlow}}{\in} [0, \infty).    
\end{aligned}
\end{equation}

This ensures that every entry of $\bP(t+\epsilon) $ remains in the interval $[0, 1]$: 

\begin{equation*}
\begin{aligned}
P^{(i)}(t) 
=
1 
&\Longrightarrow
\bigl( 
\bP(t+ \epsilon)
\bigr)^{(j)}
=
\bigl( 
\bP(t)
+
\epsilon \, \G \, \bP(t) 
+ \smallO(\epsilon^2)
\bigr)^{(j)}
\\
%
&=
\delta_{i, j}
+
\epsilon \,
\left[
(1-\delta_{i, j}) 
\,
\g_{i \to j}
- \delta_{i,j} \, 
\g_{j \to}
\right]
+ \smallO(\epsilon)
\\ & \xlongequal{j = i}
1 - \epsilon \, \g_{i \to} + \smallO(\epsilon)
\leq
1 \text{   and }
\\ & \xlongequal{j \neq i}
\epsilon\,  \g_{i \to j }
+ \smallO(\epsilon)
\geq
0. 
\\ 
P^{(i)}(t) 
=
0 
&\Longrightarrow
P^{(i)}(t+\epsilon)
=
\bigl( 
\bP(t+ \epsilon)
\bigr)^{(i)}
\\ 
%
&=
\epsilon \,
\sum\limits_{k \in \Omega \atop k\neq i}^{} P^{(k)}(t) \, \gamma_{k\to i}
+ \smallO(\epsilon^2)
\geq%
0. 
\end{aligned}
\end{equation*}

\end{proof}


\begin{lemma}[Interchanging time derivative and index summation] \label{Lemma_InterchangingTimeDerivativeAndIndexSummation} $ $ \\ 

The time derivative of the total probability mass equals the sum of the time derivatives of the individual probabilities:

\begin{equation*}
\begin{aligned}
\frac{\d}{\d t} \sum\limits_{n \in \Omega} P^{(n)}(t)
&=
\sum\limits_{n \in \Omega} \frac{\d}{\d t}  P^{(n)}(t). 
\end{aligned}
\end{equation*}

While this statement statement follows directly from the linearity of the derivative for a finite state space, it involves interchanging limits for the countable, infinite dimensional case.

\end{lemma}

\begin{proof}
\begin{equation*}
\begin{aligned}
&\left|
\frac{
\sum\limits_{n \in \Omega}
P^{(n)}(t+h)
-
\sum\limits_{n \in \Omega}
P^{(n)}(t)
}{
h
}
-
\sum\limits_{n \in \Omega}
\underbrace{
\frac{\d}{\d t}  P^{(n)}(t)
}_{
(\G \bP)^{(n)}(t)
}
\right|
=
\left|
\frac{1}{h}
\sum\limits_{n \in \Omega}
\left(
\underbrace{
\left( 
P^{(n)}(t+h)
-
P^{(n)}(t)
\right)
}_{
\int_t^{t+h}
(\G \, \bP)^{(n)}(\tau)
}
-
\underbrace{
h \, (\G \bP)^{(n)}(t)
}_{
\int_t^{t+h} \d \tau
 (\G \bP)^{(n)}(t)
}
\right)
\right|
\\ &\leq 
\sum\limits_{n \in \Omega}
\frac{1}{h}
\int_t^{t+h} \d \tau 
\left|
\bigl(
\G \bP(\tau )
-
\G \bP(t )
\bigr)^{(n)}
\right|
\xlongequal[\text{convergence}]{\text{monotone}}
\frac{1}{h}
\int_t^{t+h} \d \tau 
\underbrace{
\sum\limits_{n \in \Omega}
\left|
\left(
\G \bP(\tau )
-
\G \bP(t )
\right)^{(n)}
\right|
}_{
\| 
\G \bP(\tau )
-
\G \bP(t )
\|_1
}
\\ &\leq 
\| \G \|_{1}^\text{(op)} \, 
\sup\limits_{\tau \in [t, t + h]}
\| 
\bP(\tau )
-
\bP(t)
\|_1
\xlongrightarrow[\bP \text{  continuous}]{h \to \infty}
0. 
\end{aligned}
\end{equation*}

\end{proof}



\begin{lemma}[The solution remains normalized] $ $ \\

The norm of the solution of the master equation remains constant for all times, that is 

\begin{equation*}
\begin{aligned}
\| \bP(t) \|_1 = \| 
\underbrace{
\bP(t=0)
}_{
\bP_0
} \|_1
\xlongequal{\bP_0 \in \ProbStates}
1.      
\end{aligned}
\end{equation*}

\end{lemma}

\begin{proof}

\begin{equation*}
\begin{aligned}
\frac{\d}{\d t} \| \bP(t) \|_1
&\xlongequal[\bP(t) \in \left(\R_{\geq \, 0}\right)^\Omega]{\text{lemma }\ref{Lemma_BoundednessOfTheSolution}}
\frac{\d}{\d t} \sum\limits_{n \in \Omega} P^{(n)}(t)
\xlongequal{\text{lemma  }\ref{Lemma_InterchangingTimeDerivativeAndIndexSummation}}
\sum\limits_{n \in \Omega}
\underbrace{
\frac{\d}{\d t}  P^{(n)}(t)
}_{
\sum\limits_{m \in \Omega}  
P^{(m)}(t) \, \g_{m \to n}
- 
P^{(n)}(t) \, \g_{n \to m}
}
\\ &=
\sum\limits_{n, m \in \Omega} 
\left( 
P^{(m)}(t) \, \g_{m \to n}
- 
P^{(n)}(t) \, \g_{n \to m}
\right)
=
\underbrace{
\left( 
\sum\limits_{n, m \in \Omega} 
P^{(m)}(t) \, \g_{m \to n}
\right)
}_{
< \, \infty
}
-
\underbrace{
\left( 
\sum\limits_{n, m \in \Omega} 
P^{(n)}(t) \, \g_{n \to m}
\right)
}_{
< \, \infty 
}
\\ &\xlongequal{\text{lemma  } \ref{Lemma_FiniteProbabilityFlow}}
0.  
\end{aligned}
\end{equation*}

\end{proof}


\subsection{The positivity of the solution operator for a strongly connected network }\label{Section_ThePositivityOfTheSolutionOperatorForAStronglyConnectedNetwork}

Given a network $\System = (\Omega, \Edge, \g)$ with two states $i, j \in \Omega$ such that state $i$ is reachable from state $j$, $j \rightsquigarrow i$, then the $i$-$j$-th entry of the solution operator is strictly positive, $ \left(\e^{t \, \G}\right)_{ij}>0 $ for all $t>0$. 

In particular, for a strongly connected network $\System$, the solution operator has only strictly positive entries for all times, making it irreducible, that is $\e^{t \, \G} \in (\R_{>0})^{\Omega \times \Omega}$ for all $t>0$.


\begin{proof}
\end{proof}

Let $i,j \in \Omega$ be states in the network, such that state $i$ is reachable from state $j$, $j \rightsquigarrow i$. Since we know that the limit $\e^{t \, \G} = \lim\limits_{n \to \infty} (\Id + \frac{t \, \G}{n})^n$ exists in the \emph{norm} topology, it also exists pointwise, which means, it suffices to show that that the sequence 

\begin{equation}\label{Def_e_Hoch_t_gamma}
\begin{aligned}
\left(\left(\left(\Id + \frac{t \, \G}{n}\right)^n \right)^{(i,j)}\right)_{n \in \N} 
=
\sum\limits_{k_1 \in \Omega }^{} \cdots \sum\limits_{k_{n-1}\in \Omega}^{ } \left(\Id + \frac{\G \, t}{n} \right)^{(i, \, k_1)} \dotsc \; \left(\Id + \frac{\G \, t}{n} \right)^{(k_{n-1},\, j)}
\end{aligned}
\end{equation}

 has a strictly positive lower bound.

\begin{figure}[H]
\begin{center}
\begin{subfigure}{0.49\textwidth}
 \subcaption{Original network $\System$}
\label{ExampleOfNetwork_vs_AssociatedNetworkOfId+tGamme_n_1}
\vspace*{10mm}
\begin{tikzpicture}[scale=2]
\node[State]   (1) at (0, 0) { $1$ } ;
\node[State]   (2) at (0,-1) { $2$ } ;
\node[State]   (3) at (1.618,-1) { $3$ } ;
\node[State]   (4) at (1.618, 0) { $4$ } ;
\path[Link]             (1) edge node[left ]        {$ {  \g_{1 \to 2} }$} (2);
\path[Link]             (1) edge node[above ]        {$ {  \g_{1 \to 4} }$} (4);
\path[Link]             (2) edge node[below ]        {$ {  \g_{2 \to 3} }$} (3);
\path[Link]             (3) edge node[right ]        {$ {  \g_{3 \to 4} }$} (4);
\path[Link]             (3) edge node[above, sloped ]        {$ {  \g_{3 \to 1} }$} (1);
\end{tikzpicture}
\vspace*{9mm}
\end{subfigure} 
\begin{subfigure}{0.49\textwidth}
  \subcaption{Network $\tilde{\System}_n$ associated to $\Id + \frac{t \, \G}{n}$}
  \label{ExampleOfNetwork_vs_AssociatedNetworkOfId+tGamme_n_2}
\begin{tikzpicture}[scale=2]
\node[State]   (1) at (0, 0) { $1$ } ;
\node[State]   (2) at (0,-1) { $2$ } ;
\node[State]   (3) at (1.618,-1) { $3$ } ;
\node[State]   (4) at (1.618, 0) { $4$ } ;
\path[Link]             (1) edge node[left ]            {$ { \frac{t}{n} \, \g_{1 \to 2} }$} (2);
\path[Link]             (1) edge node[above ]           {$ { \frac{t}{n} \, \g_{1 \to 4} }$} (4);
\path[Link]             (2) edge node[below ]           {$ { \frac{t}{n} \, \g_{2 \to 3} }$} (3);
\path[Link]             (3) edge node[right ]           {$ { \frac{t}{n} \, \g_{3 \to 4} }$} (4);
\path[Link]             (3) edge node[above, sloped ]   {$ { \frac{t}{n} \, \g_{3 \to 1} }$} (1);
\path[MyLoop_UpperLeft]        (1) edge node[ left  ]    {$ { 1 - \frac{t}{n} \, \g_{1 \to}  }$} (1);
\path[MyLoop_LowerLeft]        (2) edge node[ left  ]    {$ { 1 - \frac{t}{n} \, \g_{2 \to} }$} (2);
\path[MyLoop_LowerRight]       (3) edge node[ right ]    {$ { 1 - \frac{t}{n} \, \g_{3 \to} }$} (3);
\path[MyLoop_UpperRight]       (4) edge node[ right ]    {$ { 1 - \frac{t}{n} \, \g_{4 \to} }$} (4);
\end{tikzpicture}

\end{subfigure}
\caption{Illustrating the difference between the original network $\System$ and network $\tilde{S}_n$ associated to the matrix $\Id + \frac{t \, \G}{n}$, with the modified link strength and additional self-loops.  \\ 
The walk $(1,1,2,3,3,4) \in \mathcal{WP} (1\xlongrightarrow{5} 4, \, \tilde{S}_n)\subseteq \, \mathcal{W} (1\xlongrightarrow{5} 4, \tilde{S}_n) $, after removing the self-loops, becomes a \emph{path}, where as the walk $(1,2,3,1,2,3,4) \in \mathcal{W} (1\xlongrightarrow{6} 4, \tilde{S}_n) \backslash \mathcal{WP} (1\xlongrightarrow{6} 4, \tilde{S}_n)$ does not. 
}
\label{ExampleOfNetwork_vs_AssociatedNetworkOfId+tGamme_n} 
\end{center}
\end{figure}

We call $\tilde{\System}_n:= (\Omega, \tilde{\Edge}, \tilde{\g})$ the network associated to the matrix $\Id + \frac{t \, \G}{n}$, provided that $n \in \N$ is large enough, such that all entries are non-negative and $ 1 - \frac{t}{n} \, \g_{i \to}  $ is positive for all $ i \in \Omega $. 
It can be recovered from the original network $\System$, by modifying the links between different states with a factor of $\frac{t}{n}$ and adding a self-loop at every state with a weight of $1 + \frac{t}{n}  \, \G^{(s,s)} $ for all $s \in \Omega$.

We need the following definitions: 

\begin{equation}
\begin{aligned}
\g_{\to}
:=&
\sup\limits_{i \in \Omega } \g_{i \to } \\ 
\mathcal{W} (j\xlongrightarrow{n}i, \, \System)
:=
&\{\text{ all \emph{walks} $ \omega \in \Omega^{n+1} $ of length $n\in \N$ from state $j\in \Omega$ }\\
&\text{ to state $i\in \Omega$ in the network $\System$}  \}\\ 
\Prob (j\xlongrightarrow{n}i, \, \System)
:=&
\{\text{ all \emph{paths} $ \omega \in \Omega^{n+1} $ of length $n\in \N$ from state $j\in \Omega$ } \\
&\text{ to state $i\in \Omega$ in the network $\System$} \} \\ 
\mathcal{WP} (j\xlongrightarrow{n} i, \, \System)
:=
&\{\text{ all \emph{walks} $ \omega \in \Omega^{n+1} $ of length $n\in \N$ from state $j\in \Omega$} \\
&\text{ to state $i\in \Omega$ in the network $\System$, }\\
&\text{ \emph{which become paths} after removing the self-loops} \} \\ 
\Prob (j\rightsquigarrow i, \, \System)
:=&
\bigcup\limits_{n \in \N} \Prob (j\xlongrightarrow{n}i, \, \System) \\
=&\{\text{ all paths of arbitrary length from state $j\in \Omega$} \\
&\text{\hspace*{5mm}                               to state $i\in \Omega$ in the network $\System$ }
\end{aligned}
\end{equation}

Clearly, we have 
\begin{equation*}
\begin{aligned}
\Prob (j\xlongrightarrow{n}i, \, \System) \subseteq \mathcal{WP} (j\xlongrightarrow{n}i, \, \System) \subseteq \mathcal{W} (j\xlongrightarrow{n}i, \, \System).   
\end{aligned}
\end{equation*}

A non-zero summand in the right hand side of Equation~\eqref{Def_e_Hoch_t_gamma} can be interpreted as the weight $\g_{\tilde{\omega}}$ of a walk $\tilde{\omega} \in \, \mathcal{W} (j\xlongrightarrow{n}i, \, \tilde{\System}_n)$ of length $|\tilde{\omega}|=n$ from $j$ to $i$.

When we look only at the weight of a special walk 
\begin{equation*}
\begin{aligned}
 \tilde{\omega} \in \mathcal{WP}(j \xlongrightarrow{|\tilde{\omega}|}i, \tilde{\System}_n) \subseteq \mathcal{W}(j \xlongrightarrow{|\tilde{\omega}|}i, \tilde{\System}_n), 
\end{aligned}
\end{equation*}

then this weight can be separated into the weight of the corresponding \emph{path}  \\
$\omega \in \Prob (j\xlongrightarrow{|\omega|}i, \, \System) $ times a scaling factor $\left( \frac{t}{n}\right)^{|\omega|}$ times the weight of the self-loops: 

\begin{equation}\label{Estimation_gamma_omega_tilde}
\begin{aligned}
\g_{\tilde{\omega}}
&=
\prod_{\alpha=1}^{n} \g_{\tilde{\omega}_{\alpha } \to \,  \tilde{\omega}_{\alpha +1} }
=
\underbrace{
\prod\limits_{k=1}^{|\omega|} \frac{t}{n} \, \g_{\omega_{k} \to \omega_{k+1} }
}_{
= \, \left(\frac{t}{n}\right)^{|\omega|}  \, \g_{\omega} \,
}
\cdot
\underbrace{
\prod\limits_{s \in \{\text{self-loops}(\tilde{\omega})\}  } \left( 1 + \frac{t \, \G_{(s, s)}}{n} \right)
}_{
\geq \, \left(
1 - \frac{t}{n}\, \g_{\to}  
\right)^{|\tilde{\omega}|-|\omega|}
} \geq \\
&\geq
\g_{\omega} \,
\left(\frac{t}{n}\right)^{|\omega|} \, 
\left(
1 - \frac{t}{n}\, \g_{\to}  
\right)^{|\tilde{\omega}|-|\omega|}, 
\end{aligned}
\end{equation}

where the product is taken over all self-loops of $\tilde{\omega}$.

\begin{figure}[H]
\begin{center}
\begin{subfigure}{0.49\textwidth}
  \subcaption{The \emph{walk} \\  $\begin{color}{TUDa_10b}\tilde{\omega} = (1,1,2,3,3,4)\end{color} \in 
  \mathcal{WP}(1\xlongrightarrow{|\tilde{\omega}|=5}4, \tilde{\System}_{n}) $ }
      \label{Fig_1_a}
      \vspace*{10mm}
\begin{tikzpicture}[scale=2]
\node[State]   (1) at (0, 0)     { $1$ } ;
\node[State]   (2) at (0,-1)     { $2$ } ;
\node[State]   (3) at (1.618,-1) { $3$ } ;
\node[State]   (4) at (1.618, 0) { $4$ } ;
\path[color=TUDa_11b,Link]                 (1) edge node[left ]           {$ { \color{TUDa_11b} \frac{t}{n} \, \g_{1 \to 2} }$} (2);
\path[color=black   ,Link]                 (1) edge node[above ]          {$ { \color{black}\frac{t}{n} \, \g_{1 \to 4} }$} (4);
\path[color=TUDa_11b,Link]                 (2) edge node[below ]          {$ { \color{TUDa_11b}\frac{t}{n} \, \g_{2 \to 3} }$} (3);
\path[color=TUDa_11b,Link]                 (3) edge node[right ]          {$ { \color{TUDa_11b}\frac{t}{n} \, \g_{3 \to 4} }$} (4);
\path[color=black   ,Link]                 (3) edge node[above, sloped ]  {$ { \color{black}\frac{t}{n} \, \g_{3 \to 1} }$} (1);
\path[color=TUDa_11b,MyLoop_UpperLeft]   (1) edge node[ left  ]         {$ { \color{TUDa_11b}1 + \frac{t}{n} \, \G^{(1,1)} }$} (1);
\path[color=black,MyLoop_LowerLeft]     (2) edge node[ below  ]         {$ { \color{black}1 + \frac{t}{n} \, \G^{(2,2)} }$} (2);
\path[color=TUDa_11b,MyLoop_LowerRight]  (3) edge node[ right ]         {$ { \color{TUDa_11b}1 + \frac{t}{n} \, \G^{(3,3)} }$} (3);
\path[color=black,MyLoop_UpperRight]    (4) edge node[ above ]         {$ { \color{black}1 + \frac{t}{n} \, \G^{(4,4)} }$} (4);
\end{tikzpicture}
\end{subfigure}
\begin{subfigure}{0.49\textwidth}
  \subcaption{The \emph{path} \\  $\begin{color}{TUDa_9b}{\omega} = (1,2,3,4)\end{color} \in  \Prob(1\xlongrightarrow{|\omega| = 3}4, \tilde{\System}_n)$  \\and \begin{color}{TUDa_1b} self-loops\end{color}.  }
  \label{Fig_1_b}
  \vspace*{5.5mm}
\begin{tikzpicture}[scale=2]
\node[State]   (1) at (0, 0)     { $1$ } ;
\node[State]   (2) at (0,-1)     { $2$ } ;
\node[State]   (3) at (1.618,-1) { $3$ } ;
\node[State]   (4) at (1.618, 0) { $4$ } ;
\path[color=TUDa_9b,Link]                 (1) edge node[left ]           {$ { \color{TUDa_9b} \frac{t}{n} \, \g_{1 \to 2} }$} (2);
\path[color=black   ,Link]                 (1) edge node[above ]          {$ { \color{black}\frac{t}{n} \, \g_{1 \to 4} }$} (4);
\path[color=TUDa_9b,Link]                 (2) edge node[below ]          {$ { \color{TUDa_9b}\frac{t}{n} \, \g_{2 \to 3} }$} (3);
\path[color=TUDa_9b,Link]                 (3) edge node[right ]          {$ { \color{TUDa_9b}\frac{t}{n} \, \g_{3 \to 4} }$} (4);
\path[color=black   ,Link]                 (3) edge node[above, sloped ]  {$ { \color{black}\frac{t}{n} \, \g_{3 \to 1} }$} (1);
\path[color=TUDa_1b,MyLoop_UpperLeft]   (1) edge node[ left  ]         {$ { \color{TUDa_1b}1 + \frac{t}{n} \, \G^{(1,1)} }$} (1);
\path[color=black,MyLoop_LowerLeft]     (2) edge node[ below  ]         {$ { \color{black}1 + \frac{t}{n} \, \G^{(2,2)} }$} (2);
\path[color=TUDa_1b,MyLoop_LowerRight]  (3) edge node[ right ]         {$ { \color{TUDa_1b}1 + \frac{t}{n} \, \G^{(3,3)} }$} (3);
\path[color=black,MyLoop_UpperRight]    (4) edge node[ above ]         {$ { \color{black}1 + \frac{t}{n} \, \G^{(4,4)} }$} (4);
\end{tikzpicture}
\end{subfigure}
\caption{Illustrating the separation of the weight of a \emph{walk} $\tilde{\omega}\in \mathcal{WP}(1\xlongrightarrow{|\tilde{\omega}|=5}4, \tilde{\System}_n)$ into the weight of a \emph{path} $\omega\in \Prob(1\xlongrightarrow{|\omega| = 3}4, \tilde{\System}_n)$ and a product of the weights of self-loops.  \\
For a fixed path $\omega = (1,2,3,4) \in \Prob(1\xlongrightarrow{|\omega|=3}4, \tilde{\System}_{n}) $, there are $\binom{|\tilde{\omega}|=5}{|\omega|=3}=10$ many $\tilde{\omega} \in \mathcal{WP}(1\xlongrightarrow{|\tilde{\omega}|=5}4, \tilde{\System}_{n})$ that include the original path, with only additional self-loops, namely  \\
$\tilde{\omega}_1 = (1,1,1,2,3,4)$, \, 
$\tilde{\omega}_2 = (1,1,2,2,3,4)$, \, 
$\tilde{\omega}_3 = (1,1,2,3,3,4)$, \\ 
$\tilde{\omega}_4 = (1,1,2,3,4,4)$, \,
$\tilde{\omega}_5 = (1,2,2,2,3,4)$, \,
$\tilde{\omega}_6 = (1,2,2,3,3,4)$, \\
$\tilde{\omega}_7 = (1,2,2,3,4,4)$, \,
$\tilde{\omega}_8 = (1,2,3,3,3,4)$, \,
$\tilde{\omega}_9 = (1,2,3,3,4,4)$, \\ and 
$\tilde{\omega}_{10} = (1,2,3,4,4,4)$.  
}
\label{} 
\end{center}
\end{figure}

For a given path $\omega \in \Prob(j \xlongrightarrow{|\omega|}i) $ and some fixed natural number $|\tilde{\omega}| $ greater or equal to $|\omega|$, $ |\tilde{\omega}| \in \N_{\geq \, |\omega|}$, there are $\binom{|\tilde{\omega}|}{|\omega|}$ many walks $\tilde{\omega} \in \mathcal{WP}(j \xlongrightarrow{|\tilde{\omega}|}i) $, which include the original path, but have additional self-loops (out of $ |\tilde{\omega}| $ many transitions, choose the positions of $ |\omega| $ non-trivial ones). 

Now take $n \in \N $ large enough, such that there exists a path $\omega_0 \in \Prob(j \rightsquigarrow i, \System) $ with $ |\omega_0| \leq n $. 
Since we have: 

\begin{equation} \label{Eq_SideCalculation_1}
\begin{aligned}
\frac{n!}{(n-|\omega_0|)! \cdot n^{|\omega_0|}} 
&=
\prod\limits_{k=1}^{|\omega_0|-1} 
\left(
1 - \frac{k}{n}
\right)
\xlongrightarrow{n \to \infty}
1 \text{ , hence }
\frac{n!}{(n-|\omega_0|)! \cdot n^{|\omega_0|}}
\leq 
\frac{1}{2} \text{ , for $n$ large enough } \text{  and } 
\\ \\ 
\left( 1 - \frac{t}{n}\, \g_{\to}   \right)^{n-|\omega_0|}
&\xlongrightarrow{n \to \infty}
\e^{-t \, \g_{\to}} \text{ , hence }
\left( 1 - \frac{t}{n}\, \g_{\to}   \right)^{n-|\omega_0|}
\leq 
\frac{\e^{-t \, \g_{\to}}}{2} \text{ , for $n$ large enough,  }  
\end{aligned}
\end{equation}
we can make the following estimation:

\begin{equation}
\begin{aligned}
\left(\Bigl(\Id + \frac{t \, \G}{n}\Bigr)^n \right)^{(i, j)}
&\text{} \hspace*{4mm}=
\sum\limits_{k_1 \in \Omega }^{} \cdots \sum\limits_{k_{n-1} \in \Omega}^{} \left(\Id + \frac{\G \, t}{n} \right)^{(i, \, k_1)} \cdot \dotsc \cdot \; \left(\Id + \frac{\G \, t}{n} \right)^{(k_{n-1},\, j)} \\
&\text{} \hspace*{4mm}=
\sum\limits_{\tilde{\omega} \in \,  \mathcal{W}\left(j \, \xlongrightarrow{n} \, i, \, \tilde{\System}_n\right)} \, \g_{\tilde{\omega}} \\ 
%
&\overset{\mathcal{W} \, \supseteq \, \mathcal{WP}} \geq
\sum\limits_{\tilde{\omega} \in \,  \mathcal{WP}\left(j \, \xlongrightarrow{n} \, i, \, \tilde{\System}_n\right)} \,
\underbrace{
\hspace*{10mm} \g_{\tilde{\omega}} \hspace*{10mm}
}_{
\hspace*{4mm} \geq \, 
\g_\omega \, \left( \frac{t}{n}\right)^{|\omega|} \, \left( 1 - \frac{t}{n} \, \g_{\to}\right)^{n-|\omega|}
} 
\\
%
&\text{} \hspace*{2mm} \overset{\eqref{Estimation_gamma_omega_tilde}} \geq 
\sum\limits_{
{
\omega \in \,  \Prob\left(j \, \rightsquigarrow \, i, \, \System \right)
}\atop{
|\omega| \leq n
} 
}
\, \binom{n}{|\omega|} \,  \g_\omega \, \left( \frac{t}{n}\right)^{|\omega|} \, \left( 1 - \frac{t}{n} \, \g_{\to}\right)^{n-|\omega|} 
\\ &\text{} \hspace*{4mm}\xlongequal{n \geq |\omega_0 |}
\, \binom{n}{|\omega_0|} \,  \g_{\omega_0} \, \left( \frac{t}{n}\right)^{|\omega_0|} \, \left( 1 - \frac{t}{n} \, \g_{\to}\right)^{n-|\omega_0|} 
\\ 
%
&\text{} \hspace*{4mm}\xlongequal{}
\underbrace{
\frac{n!}{(n-|\omega_0|)! \, n^{|\omega_0|}}
}_{
\geq \,\frac{1}{2}
} 
\, 
 \,  \frac{\g_{\omega_0} \, t^{|\omega_0|}}{|\omega_0|!} \,
\underbrace{ 
 \left( 1 - \frac{t}{n} \, \g_{\to}\right)^{n-|\omega_0|}
 }_{
\geq \, \frac{1}{2} \, \e^{-t \, \g_{\to} }
 }
 \\
&\text{} \hspace*{2mm} \overset{\eqref{Eq_SideCalculation_1}}\geq \, 
\frac{\e^{-t \, \g_{\to} }}{4}  \, 
\frac{\g_{\omega_0} \, t^{|\omega_0|}}{|\omega_0|!} \\
&\text{} \hspace*{4mm}> \, 0. 
\end{aligned}
\end{equation}

In the fourth step, we used both the estimate of Equation~\eqref{Estimation_gamma_omega_tilde} and the fact that every `walk-path' can be separated into a path and self-loops.

This concludes the proof.


\subsection{Uniqueness and strict positivity of stationary solutions for irreducible networks }\label{Section_UniquenessAndStrictPositivityOfStationarySolutionsForIrreducibleNetworks}

\begin{lemma}[The kernel of the generator of an irreducible network]\label{Lemma_StrictPositivityForStationarySolutionsOfStronglyConnectedNetworks} $ $ \\ 

Let $\G$ be as in equation \eqref{Eq_GeneratorMatrix},  $\| \G \|_{1}^\text{(op)} < \infty$ and let the semigroup  $(\e^{t \, \G})_{t \geq 0} $ be irreducible. Then the dimension of the kernel of the generator is at most one-dimensional, that is $\text{dim} (\Kern(\G)) \leq 1$. If the kernel is exactly one-dimensional, it is spanned by an - element-wise - strictly positive eigenvector, that is $ \Kern(\G) = \Span(\bX) $ with $\bX \in (\R_{>0})^{\Omega}$.

\end{lemma}

\begin{proof}
Let $\bX \in l^1(\Omega) \cap \Kern(\G - \lambda \, \Id)$ be an eigenvector of $\G$ to an eigenvalue $\lambda \in \C$, that is $\G \, \bX = \lambda \, \bX $. Then we have

\begin{equation}\label{Eq_OneDimensionalityOfTheKernel}
\begin{aligned}
|\e^{\lambda \, t}| \,  \| \bX \|_1 
&=
\| \e^{\lambda \, t} \bX \|_1 
=
\|\e^{t \, \G} \bX \|_1 
=
\sum\limits_{i \in \Omega }^{} \left|\;  \sum\limits_{j \in \Omega}^{} \left( \e^{t \, \G} \right)^{(i, j)} \, X^{(j)} \; \right| 
\leq
 \\
&\overset{{\color{ElectricPurple} (*)}} \leq
\sum\limits_{i,\, j \in \Omega }^{}
 \left| \, 
 \underbrace{
 \left( \e^{t \, \G} \right)^{(i, j)}
 }_{
 \geq \, 0
 } \, \right| \cdot |X^{(j)}|  
 =
 \sum\limits_{j \in \Omega }^{} |X^{(j)}| \; 
 \underbrace{
 \left(\sum\limits_{i \in \Omega}^{} \left( \e^{t \, \G} \right)^{(i, j)} \right)
 }_{
 =\,1
 } 
 =
 \| \bX \|_1. 
\end{aligned}
\end{equation}

If $\lambda=0$ we have equality in \eqref{Eq_OneDimensionalityOfTheKernel}, but on the other hand, we have equality in step $ \color{ElectricPurple} (*) $, if and only if for all $ i \in \Omega $

\begin{equation} \label{Eq_NonNegativityOfACertainSequence}
\begin{aligned}
\Bigl(\bigl(\e^{t \, \G }\bigr)^{(i, j)} \, X^{(j)} \Bigr)_{j \in \Omega}
&\in
\bigl(\R_{\geq 0} \bigr)^{\Omega} \cup \bigl(\R_{\leq 0} \bigr)^{\Omega}
\end{aligned}
\end{equation}

For an irreducible system, every element of the solution operator $\e^{t \, \G} $ is strictly positive, so we can deduce from equation \eqref{Eq_NonNegativityOfACertainSequence} that 

\begin{equation} \label{Eq_NonNegativityOfElementOfKernel}
\begin{aligned}
\bX
&=
(X^{(j)})_{j \in \Omega }
\in
\bigl(\R_{\geq 0} \bigr)^{\Omega} \cup \bigl(\R_{\leq 0} \bigr)^{\Omega}, 
\end{aligned}
\end{equation}

hence $\Kern(\G) \subseteq \bigl(\R_{\geq 0} \bigr)^{\Omega} \cup \bigl(\R_{\leq 0} \bigr)^{\Omega} $. 
Since $ \left( \bigl(\R_{\geq 0} \bigr)^{\Omega} \cup \bigl(\R_{\leq 0} \bigr)^{\Omega} \right) \backslash \{\bzero\} $ is not path-connected, we can deduce that the dimension of $\Kern(\G)$ must be (at most) one. 

What is left to show, is that every element of $ \bP_* \in \Kern(\G) \cap \ProbStates $ is \emph{strictly positive}, so we have to rule out the case that there are vanishing entries. As we will see, this violates the assumption that the network is irreducible: 

Let us define $S := \{n \in \Omega \,:\, P_*^{(n)} = 0  \} \subseteq \NP(\Omega) $. Then we have for all $s \in S $

\begin{equation} \label{Eq_StrictPositivityOfElementOfKernel}
\begin{aligned}
0
&=
(\G \, \bP)^{(s)}
=
\sum\limits_{j \in \Omega} P^{(j)} \, \g_{j \to s} - 
\underbrace{
P^{(s)}
}_{0}
\, \g_{s \to j}
\\ & =
\sum\limits_{j \in \Omega \backslash S} 
\underbrace{
P^{(j)}
}_{
> \, 0
}
\, \g_{j \to s}
\end{aligned}
\end{equation}

This means that for all $s \in S$ and for all $j \in \Omega \backslash S$ we have $ \g_{j \to s} = 0$. This makes the linear subspace associated to $S^C$ an invariant subspace, since the generator can - after a suitable permutation - be brought into a triangular matrix. This implies $\G \, \bP_0 \in S^C$ and hence $\e^{t \, \G} \, \bP_0 \in S^C$ whenever $\bP_0 \in S^C $.

\begin{figure}[H]
\begin{center}

\begin{minipage}{0.49\textwidth}
\begin{center}
\begin{tikzpicture} 
\node[State](S) at (-1,0) [circle,draw] {$S$} ;

\node[State](SC) at (1,0) [circle,draw] {$S^C$} ;
\path[Link, dotted] (S) edge node[above] {} (SC);
\end{tikzpicture}
\end{center}
\end{minipage}\begin{minipage}{0.49\textwidth}
\begin{equation*}
\begin{aligned}
\G^n
&=
\begin{pmatrix}
* & * \\
\bzero & *
\end{pmatrix}, \, 
\e^{t \, \G}
=
\begin{pmatrix}
* & * \\
\bzero & *
\end{pmatrix}
\\
\bP\left(t \,|\, \bP_0 = \colvec{2}{*}{\bzero}\right)
&=
\underbrace{
\e^{t \, \G}
}_{
\begin{pmatrix}
* & * \\
\bzero & *
\end{pmatrix}
}
\, 
\underbrace{
\bP_0
}_{
\colvec{2}{*}{\bzero}
}
=
\colvec{2}{*}{\bzero}. 
\end{aligned}
\end{equation*}
\end{minipage}
\label{Illustrating_ReducibilityIfOneOfTheEntriesVanishes} 
\caption{The network $ \Omega = S \dot{\cup} S^C $ can be decomposed into two disjoint parts $S$ and $S^C$, where there is no link from $S^C$ to $S$. This results in a block-triangular form of both the generator $\G$, as well as the solution operator $\e^{t \, \G} $, making the network reducible.   }
\end{center}
\end{figure}

\end{proof}






\section{The thermodynamic limit for master equations} \label{Chapter_TheThermodynamicLimitOfTheMasterEquations}


Markov chains are usually treated differently, depending on whether they live on a \emph{finite} state space or a \emph{countable infinite} one. This is mostly due to the fact, that the mathematical methods are entirely different: While one can rely on matrices and knowledge from linear algebra in case of a finite state space, one has to deal with - bounded - operators on Banach spaces, if the state space if countable infinite. 

In this section, we ask the question, of whether and how these two realms - finite and countably infinite - can be linked, or to be more precise, whether it is possible, to approximate a countable infinite system with a finite one. 


\subsection{How to approximate countable, infinite dimensional probability vectors by finite dimensional ones }

For a given countable, infinite dimensional probability vector $ \bP = \bigl(P^{(1)}, P^{(2)}, P^{(3)}, \dots \bigr) \in \ProbStates $, the problem is how to define a sequence $(\bp_n)_{n \in \Omega} \in \ProbStates_n $ such that $ \bp_n \xlongrightarrow{n \to \infty} \bP $, given that they are not even elements of the same vector space ?   

A first idea could be, to `fill the sequence with zeros', as depicted below.

\begin{minipage}{0.49\textwidth}
\begin{equation*}
\begin{aligned}
\\
\bp_2
&=
\bigl(p_2^{(1)}, p_2^{(2)} \bigr)
%
\\
\bp_3
&=
\bigl(p_3^{(1)}, p_3^{(2)}, p_3^{(3)} \bigr)
%
\\
\bp_4
&=
\bigl(p_4^{(1)}, p_4^{(2)}, p_4^{(3)}, p_4^{(4)} \bigr)
\\
\text{ with } 
\bp_n 
&\in
\ProbStates_n
:=
\bigl(\R_{\geq \, 0}\bigr)^{n}  \cap B_{r=1}^{\| \cdot \|_1}(0)
\end{aligned}
\end{equation*}
\end{minipage}\begin{minipage}{0.49\textwidth}
\begin{equation*}
\begin{aligned}
\bP_2
&=
\bp_2 \times \{0\}^{\Omega}
=
\bigl(p_2^{(1)}, p_2^{(2)}, 0, \dots \bigr) 
%
\\
\bP_3
&=
\bp_3 \times \{0\}^{\Omega}
=
\bigl(p_3^{(1)}, p_3^{(2)}, p_3^{(3)}, 0, \dots \bigr)
%
\\
\bP_4
&=
\bp_4 \times \{0\}^{\Omega}
=
\bigl(p_4^{(1)}, p_4^{(2)}, p_4^{(3)}, p_4^{(4)}, 0, \dots \bigr)
\\
\text{ with } 
\bP_n 
&\in
\ProbStates_0
:=
\bigl(\R_{\geq \, 0}\bigr)^{n}  \cap B_{r=1}^{\| \cdot \|_1}(0)
\end{aligned}
\end{equation*}
\end{minipage}

\vspace*{3mm}

This solves the problem of different vector spaces, but as it turns out, \emph{nets} are better suited than \emph{sequences} for describing the transition from finite- to infinite dimensions. 

Therefore, we define the set of all finite subsets of $\Omega$ as $\Fin(\Omega):= \{F \subseteq \Omega \,:\, |F|< \infty \}$. For some fixed $F \in \Fin(\Omega)$ the probability sequence $ \bPF \in \ProbStates_0 $ could look as follows:

\begin{center}
\begin{tabular}{| c | c | }
\hline 
$F$ & $ \bPF $ \\
 \hline  
$ \{1, 2, 3\} $ & $ \bigl(P^{(1)}, P^{(2)}, P^{(3)}, 0, 0, , 0, \dots \bigr)  $ \\
\hline   
$ \{1, 4, 5\} $ & $ \bigl(P^{(1)}, 0, 0, P^{(4)}, P^{(5)}, 0, \dots \bigr) $ \\
\hline
\end{tabular}

\end{center}

This definition makes $\bPF $ a probability sequence, that is effectively finite dimensional.

We note that the element-wise definition $ \left( \PF \right)^{(n)} := \frac{P^{(n)}}{\sum\limits_{f \in F} P^{(f)}}$ requires that the denominator does not vanish. 

We therefore define for some fixed probability sequence $\bP \in \ProbStates $ the set of null sets of the corresponding measure, that is 

\begin{MyDef}[\textbf{the set of null sets}] \label{Def_TheSetOfNullSetsOfP} 
\end{MyDef}
\begin{equation*}
\begin{aligned}
\NP 
&:=
\{ N \subseteq \Omega \,:\, \bigl(\bP\bigr)(N) := \sum\limits_{n \in N} P^{(n)} =0 \}. 
\end{aligned}
\end{equation*}

Then we can define $\bPF$ as follows: 
\begin{MyDef}[\textbf{the probability net $\bPF $ }] \label{Def_bPA} 
\end{MyDef}

Then for any $F \in \Fin(\Omega) \backslash\,  \NP $ we can define $\bPF $ entry-wise as 

\begin{equation}\label{ProbSeq_FinSubNet}\tag{ProbSeq-FinSubNet}
\begin{aligned}
\left( \PF \right)^{(n)} 
&:=
\frac{P^{(n)}}{\sum\limits_{f \in F} P^{(f)}}
\end{aligned}
\end{equation}

However, since the solution $\bP^{F}\bigl( t \,|\, \bPF_0 \bigr) = \e^{t \,\GF} \, \bPF_0 $ depends on the generator, we need a good definition for $\GF $.


\subsection{Definition of finite subnetworks and their corresponding generators}\label{Section_DefinitionOfFiniteSubnetworksAndTheirCorrespondingGenerators}

Following the arguments made above, it makes sense to define $ \GF \in \R^{\Omega \times \Omega } $ as a double sequence, where all but finitely many entries vanish: 


\begin{equation}
\begin{aligned}
\GF \in \R^{\Omega \times \Omega}\text{, with  } \left(\GF\right)^{(i,\,j)} \neq 0 
\Longrightarrow
i, j \in F. 
\end{aligned}
\end{equation}

This ensures that the states of the original network (which was countable, infinite dimensional) are always aligned with the component of the finite subnetwork. 

The first idea on how to define $\GF$ would be to introduce a sharp cutoff, namely 
\begin{equation}
\begin{aligned}
\left(\widetilde{\GF}\right)^{(i,\,j)}
&=
\begin{cases}
(\G)^{(i,\,j)} &\text{ , if } i, j \in F  \\
0 &\text{ , else}. 
\end{cases}
\end{aligned}
\end{equation}

However, this could make $ \widetilde{\GF} $ loose the property of being the generator of a continuous-time Markov chain (in particular, the column-sum would not necessarily equal zero). In the following, we need the definition of \emph{subnetworks}:

\begin{MyDef}[\textbf{Subnetworks}] \label{Def_SubNetworks} 
\end{MyDef}

We call a network $ \System_F = (\Omega_F, \Edge_F) $ a \emph{subnetwork} of $ \System = (\Omega, \Edge) $ if it contains \emph{some} of its states and \emph{all} the original links between those states, that is: 

\begin{equation}
\begin{aligned}
\System_F \subseteq \System 
\iff 
\Omega_F & \subseteq \Omega \text{  and } \\
\Edge_F  & := \{(i,j) \in \Edge \,:\, i,j \in \Omega_F\}. 
\end{aligned}
\end{equation}

\begin{figure}[H]
\begin{center}
\begin{subfigure}{0.33\textwidth}
   \subcaption{}
      \label{Illustratin_SubNetworks_a}
\begin{tikzpicture}[scale=3]
\node[State](1) at (0,0) [circle,draw] {1} ;
\node[State](2) at (1,0) [circle,draw] {2} ;
\node[State](3) at (0.5,-0.866) [circle,draw] {3} ;
\path[Link,bend left] (1) edge node[above] {$\g_{1 \to 2}$} (2);
\path[Link,bend left] (2) edge node[below] {$\g_{2 \to 1}$} (1);
\path[Link] (1) edge node[left,below,sloped]  {$\g_{1 \to 3}$} (3);
\path[Link] (3) edge node[left,below,sloped]  {$\g_{3 \to 2}$} (2);
\end{tikzpicture}
\end{subfigure}\begin{subfigure}{0.33\textwidth}
   \subcaption{}
   \label{Illustratin_SubNetworks_b}
\begin{tikzpicture}[scale=3]
\node[State](1) at (0,0) [circle,draw] {1} ;
\node[State](2) at (1,0) [circle,draw] {2} ;
\path[Link,bend left] (1) edge node[above] {$\g_{1 \to 2}$} (2);
\path[Link,bend left] (2) edge node[below] {$\g_{2 \to 1}$} (1);
\end{tikzpicture}
\end{subfigure} \begin{subfigure}{0.33\textwidth}
   \subcaption{}
   \label{Illustratin_SubNetworks_c}
\begin{tikzpicture}[scale=3]
\node[State](1) at (0,0) [circle,draw] {1} ;
\node[State](2) at (1,0) [circle,draw] {2} ;
\node[State](3) at (0.5,-0.866) [circle,draw] {3} ;
\path[Link] (1) edge node[left,below,sloped]  {$\g_{1 \to 3}$} (3);
\path[Link] (3) edge node[left,below,sloped]  {$\g_{3 \to 2}$} (2);
\end{tikzpicture}
\end{subfigure}
\caption{Illustrating the concept of subnetworks: With Figure~\ref{Illustratin_SubNetworks_a} being the original network, Figure~\ref{Illustratin_SubNetworks_b} is a subnetwork, according to definition \ref{Def_SubNetworks}, while Figure~\ref{Illustratin_SubNetworks_c} is not.  }
\label{Illustratin_SubNetworks} 
\end{center}
\end{figure}
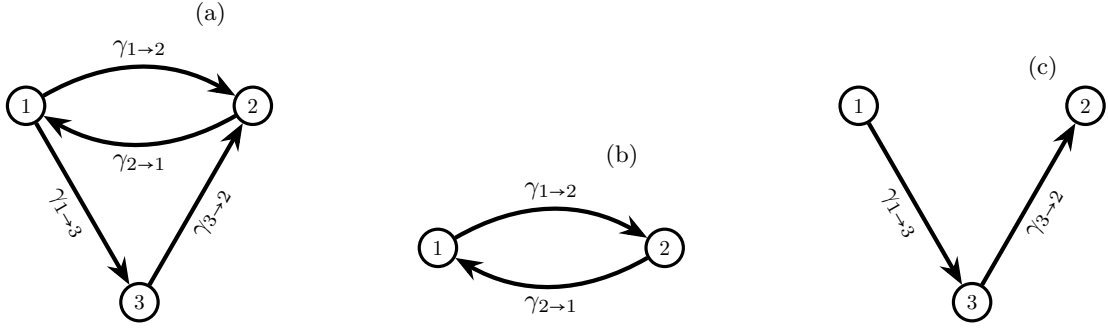

For a finite number of state $F \subseteq \Omega $, we define $\GF$ as the generator of the corresponding \emph{subnetwork} $ \System_F$, that is :

\begin{equation} \label{GenFinSubNet}\tag{GenFinSubNet}
\begin{aligned}
\left(\GF\right)^{(i, \, j)}
&:=
\left(\G^{[\System_F]}\right)^{(i, \, j)}
=
\begin{cases}
(\G)^{(i, \, j)} &\text{ , if } i\neq j, \, i, j \in F  \\
- \sum\limits_{f \in F} \G^{(f, j)} &\text{  , if } i=j, \, i, j \in F  \\
0 &\text{ , else} \\
\end{cases} \\
%
&=
(1 - \delta_{i, j}) \, \G^{(i, j)} \, 1_F(i) \, 1_F(j) - \delta_{i, j} \, \sum\limits_{f \in F} \G^{(f, j)}, 
\end{aligned}
\end{equation}

which ensure that the generator property is being preserved (compare equation \eqref{Eq_MasterEq}):

\begin{equation} \label{Eq_DifferentKindsOfTruncationForTheGenerator}
\begin{aligned}
\G 
&=
\begin{pmatrix}
-\g_{1 \to 2 }-\g_{1 \to 3 }& \g_{2 \to 1} & 0 \\
\g_{1 \to 2 } & -\g_{2 \to 1} & \g_{3 \to2} \\
\g_{1 \to 3} & 0 & -\g_{3 \to2}
\end{pmatrix} &\text{ $\dots$ generator of the network in figure  \ref{Illustratin_SubNetworks_a}} ,  \\
\widetilde{\G^{[\{1, 2\}]}}
&=
\begin{pmatrix}
-\g_{1 \to 2 }-\g_{1 \to 3 }& \g_{2 \to 1} & 0 \\
\g_{1 \to 2 } & -\g_{2 \to 1} & 0 \\
0 & 0 & 0
\end{pmatrix} &\text{   $\dots$ is NOT a generator matrix } \\ 
\G^{[\{1, 2\}]}
&=
\begin{pmatrix}
-\g_{1 \to 2 }& \g_{2 \to 1} & 0 \\
\g_{1 \to 2 } & -\g_{2 \to 1} & 0 \\
0 & 0 & 0
\end{pmatrix} &\text{   $\dots$ is a generator matrix } 
\end{aligned}
\end{equation}


Given an initial probability sequence $\bP_0 \in \ProbStates $ and a finite subset $ F \in \Fin(\Omega) $ with $ \bP_0(F)>0$, we define the initial state $\bPF_0$ - like in equation \eqref{ProbSeq_FinSubNet} - as follows: 

\begin{equation} \label{InitialStateFinSubNet} \tag{InitialStateFinSubNet}
\begin{aligned}
\bigl(
\bPF_0 
\bigr)^{(i)}
&:=
\begin{cases}
    \frac{P_0^{(i)}}{\sum\limits_{f \in F} P^{(f)}_0 } &\text{ , if  } i\in F \\
    0, &\text{ else}. 
\end{cases}
\end{aligned}
\end{equation}

The probability sequence $\bP^F(t\,|\, \bPF_0 ) \in \left(\R_{\geq \, 0} \right)^{\Omega} $ for $ F \in \Fin(\Omega) \backslash \, \NPzero $ is then defined as the solution of the following initial value problem: 

\begin{minipage}{0.35\textwidth}
\begin{equation*} 
\begin{aligned}
\frac{d}{dt} \bP^F(t) &= \GF \, \bP^F(t) \\
 \bP^F(t=0) &= \bPF_0. 
\end{aligned}
\end{equation*}
\end{minipage} $\iff$
\begin{minipage}{0.35\textwidth}
\begin{equation} \label{MEqFinSubNet}\tag{MEqFinSubNet}
\begin{aligned}
\bP^F\left(t\,\bigl|\, \bPF_0 \right)  
=
\e^{t \, \GF } \bPF_0  
\end{aligned}
\end{equation}
\end{minipage}

Since only finitely many entries are unequal to zero, the problem is effectively finite dimensional.

The following table gives some examples of truncations of generators and probability vectors of the original network, given in figure \ref{Example_NetworkWithGenerator}.

\begin{figure}[H]
\begin{center}
\begin{subfigure}{0.49\textwidth}
\begin{tikzpicture}[scale=3]
\node[State](1) at (0,0) [circle,draw] {1} ;
\node[State](2) at (1,0) [circle,draw] {2} ;
\node[State](3) at (0.5,-0.866) [circle,draw] {3} ;
\path[Link,bend left] (1) edge node[above] {$\g_{1 \to 2}$} (2);
\path[Link,bend left] (2) edge node[below] {$\g_{2 \to 1}$} (1);
\path[Link] (1) edge node[left,below,sloped]  {$\g_{1 \to 3}$} (3);
\path[Link] (3) edge node[left,below,sloped]  {$\g_{3 \to 2}$} (2);
\end{tikzpicture}
\end{subfigure} \hspace*{-40mm} \begin{subfigure}{0.49\textwidth}
\begin{equation*}
\begin{aligned}
\G = 
\begin{pmatrix}
-\g_{1 \to 2 }-\g_{1 \to 3 }& \g_{2 \to 1} & 0 \\
\g_{1 \to 2 } & -\g_{2 \to 1} & \g_{3 \to2} \\
\g_{1 \to 3} & 0 & -\g_{3 \to2}
\end{pmatrix}
\end{aligned}
\end{equation*}
\end{subfigure}
\caption{Example of a network with the corresponding generator $\G$ }
\label{Example_NetworkWithGenerator} 
\end{center}
\end{figure}
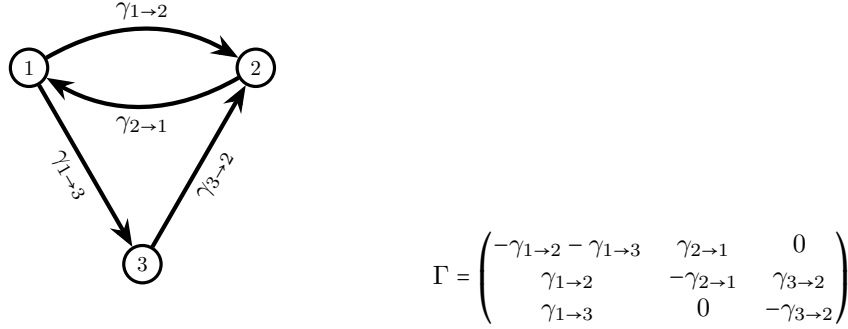

\begin{table}[h!]
\centering
 \begin{tabular}{| c | c | c | c | c |} 
 \hline
$F$ & $\System_F$ & $\GF$ & $\bPF_0$ & $\bP^F_\infty\left( \bPF_0 \right)$  \\ 
 \hline
 $\{1, 2\}$ & 
\begin{tikzpicture}[scale=1.5]
\node[State](1) at (0,0) [circle,draw] {1} ;
\node[State](2) at (1,0) [circle,draw] {2} ;
\path[Link,bend left] (1) edge node[above] {$\g_{1 \to 2}$} (2);
\path[Link,bend left] (2) edge node[below] {$\g_{2 \to 1}$} (1);
\end{tikzpicture}
 &
 $\begin{pmatrix}
-\g_{1 \to 2 }& \g_{2 \to 1} & 0 \\
\g_{1 \to 2 } & -\g_{2 \to 1} & 0 \\
0 & 0 & 0
\end{pmatrix} $
& 
$\frac{1}{P_0^{(1)}+P_0^{(2)}} \colvec{3}{P_0^{(1)}}{P_0^{(2)}}{0}$ 
&
$\frac{1}{\g_{2 \to 1} + \g_{1 \to 2}} \colvec{3}{\g_{2 \to 1}}{\g_{1 \to 2}}{0}$ 
\\ 
\hline
$\{1, 3\}$ & 
\begin{tikzpicture}[scale=1.5]
\node[State](1) at (0,0) [circle,draw] {1} ;
\node[State](3) at (0.5,-0.866) [circle,draw] {3} ;
\path[Link] (1) edge node[left,below,sloped]  {$\g_{1 \to 3}$} (3);
\end{tikzpicture}
 &
 $\begin{pmatrix}
-\g_{1 \to 3 }& 0 & 0 \\
0 & 0 & 0 \\
\g_{1 \to 3 } & 0 & 0
\end{pmatrix} $
& 
$\frac{1}{P_0^{(1)}+P_0^{(3)}} \colvec{3}{P_0^{(1)}}{0}{P_0^{(3)}}$ 
&
$ \colvec{3}{0}{0}{1}$ 
\\
\hline
\\ 
 $\{2, 3\}$ & 
\begin{tikzpicture}[scale=1.5]
\node[State](2) at (1,0) [circle,draw] {2} ;
\node[State](3) at (0.5,-0.866) [circle,draw] {3} ;
\path[Link] (3) edge node[left,below,sloped]  {$\g_{3 \to 2}$} (2);
\end{tikzpicture}
 &
 $\begin{pmatrix}
0 & 0 & 0 \\
0 & 0 & \g_{3 \to 2} \\
0 & 0 & -\g_{3 \to 2}
\end{pmatrix} $
& 
$\frac{1}{P_0^{(2)}+P_0^{(3)}} \colvec{3}{0}{P_0^{(2)}}{P_0^{(3)}}$ 
&
$ \colvec{3}{0}{1}{0}$ 
\\
\hline 
\end{tabular}
\end{table}


\subsection{Approximating the system with finite subnetworks}

Two finite subsets $ A, B \in \Fin(\Omega) $ with $A \cancel{\subseteq} B$ and $ B \cancel{\subseteq} A $ cannot directly be compared, due to the fact that the set of all finite subsets has no \emph{total order} but only a \emph{half order}, the set of probability sequences $ \left(\bPF \right)_{F\in \Fin(\Omega)} $ is no sequence, but a \emph{net} in the topological sense and the thermodynamic limit is a convergence of nets.

We say that the probability vector $ \bP^F \left(t \,\bigl|\, \bPF_0\right) $ converges in the thermodynamic limit to the probability vector probability vector $\bP^{}\left(t \,\bigl|\, \bP_0^{}\right) $ of the countable, infinite dimensional system $\Omega $, if for all $\epsilon $ larger than zero, there exists a finite subset $F_\epsilon \in \Fin(\Omega) $ such that for all finite subsets $F \in \Fin(\Omega)$ larger than $F_\epsilon$, the corresponding probability vector $ \bP^F\left(t \,\bigl|\, \bPF_0\right)$ for the finite subsystem differs from the probability vector of the whole system only by a number $\epsilon$ with respect to the $1$-norm:

\begin{equation} \label{Eq_ThermodLimit} \tag{ThermodLimit}
\begin{aligned}
&\lim\limits_{|F| \to \infty }
\bP^F\left(t \,\bigl|\, \bPF_0\right) 
:=
\lim\limits_{F \in \Fin({\Omega})}
\bP^F\left(t \,\bigl|\, \bPF_0\right) 
=
\bP\left(t \,\bigl|\, \bP_0^{}\right) \overset{\text{def}}\iff \\
&\forall \epsilon>0\;\; \exists \,F_\epsilon \in \Fin({\Omega}) \,:\, \forall F \in \Fin({\Omega}), \, F \supseteq F_\epsilon \;\;\; \left\| \bP^{}\left(t \,\bigl|\, \bP_0^{}\right) - \bP^F\left(t \,\bigl|\, \bPF_0\right) \right\|_1 < \epsilon
\end{aligned}
\end{equation}

For better readability we write ``$\lim\limits_{|F| \to \infty }$'' instead of ``$ \lim\limits_{F \in \Fin({\Omega})} $'', but emphasize that is a convergence of nets.

\begin{thm}[Thermodynamic Limit for master equations]\label{ThermodynamicLimitForMasterEquations} $ $ \\

For bounded generators $(\|\G\|_{1, 1}^\text{(op)} < \infty)$ the solution of the master equation \eqref{Eq_MasterEq} for a countable, infinite dimensional system $ \bP\left(t \,\bigl|\, \bP_0^{}\right) $ at a given time $t>0$ can be approximated by finite subsystems in the thermodynamic limit of definition \eqref{Eq_ThermodLimit}. 
\end{thm}

\begin{proof}


We have \\

\begin{minipage}{0.35\textwidth}
\begin{equation*} 
\begin{aligned}
\frac{d}{dt} \bP^F(t) &= \GF \, \bP^F(t) \\
 \bP^F(t=0) &= \bPF_0. 
\end{aligned}
\end{equation*}
\end{minipage} $\iff$
\begin{minipage}{0.35\textwidth}
\begin{equation*} 
\begin{aligned}
\bP^F\left(t\,\bigl|\, \bPF_0 \right)  
=
\bPF_0 + \int_0^t
\underbrace{
 \frac{d}{d \tau} \, 
\bP^F\left(\tau \,\bigl|\, \bPF_0 \right)
}_{
\GF\, \bP^F\left(\tau \,\bigl|\, \bPF_0 \right)
}
\d \tau  
\end{aligned}
\end{equation*}
\end{minipage}

For two, large enough, finite subnetworks $F_1, \,F_2 \in \Fin(\Omega)$, we have: 

\begin{equation*}
\begin{aligned}
u(t) 
&:=
\Bigl\| 
\bP^{F_1} \left(t \,\bigl|\, \bP_0^{[F_1]} \right)
-
\bP^{F_2}\left(t \,\bigl|\, \bP_0^{[F_2]} \right)
\Bigr\|_1
\\
%
&=
\left\|
\bP_0^{[F_1]} - \bP_0^{[F_2]}
+ \bigintsss_0^t \d \tau
\left(
\G^{[F_1]} \, \bP^{F_1} \left(\tau \,\bigl|\, \bP_0^{[F_1]}\right)
\right) 
- 
\left(
\G^{[F_2]} \, \bP^{F_2}\left(\tau \,\bigl|\, \bP_0^{[F_2]}\right)
\right)
\pm \left( \G^{[F_1]} \, \bP^{F_2} \left(\tau \,\bigl|\, \bP_0^{[F_2]}\right) \right) 
\right \|_1 \leq \\
%
&\leq 
\left\|
\bP_0^{[F_1]} - \bP_0^{[F_2]}
\right \|_1
+
\bigintsss_0^t 
\left(
\underbrace{
\left\|
\left(\G^{[F_1]} - \G^{[F_2]}\right)
\bP^{F_2}\left(\tau \,\bigl|\, \bP_0^{[F_2]}\right)
\right \|_1
}_{
\leq \, 
\left\|
\G^{[F_1]} - \G^{[F_2]}
\right \|_{1, 1}^\text{(op)}
\, 
\left\|
\bP^{F_2}\left(\tau \,\bigl|\, \bP_0^{[F_2]}\right)
\right \|_1
}
\right) \d \tau 
\\
&\hspace*{35mm}+
\bigintsss_0^t 
\left(
\underbrace{
\left \|
\G^{[F_1]} \, \left(\bP^{F_1} \left(\tau \,\bigl|\, \bP_0^{[F_1]}\right) - \bP^{F_2} \left(\tau \,\bigl|\, \bP_0^{[F_2]}\right)\right) 
\right \|_1 
}_{
\leq \, 
\left\| \G^{[F_1]}  \right\|_{1, 1}^\text{(op)} \, \left\| \bP^{F_1} \left(\tau \,\bigl|\, \bP_0^{[F_1]}\right) - \bP^{F_2} \left(\tau \,\bigl|\, \bP_0^{[F_2]}\right)  \right\|_1
} 
\right) \d \tau
\leq 
\\ 
%
&  \overset{\text{lemma }\ref{Lemma_Comparing_1NormOfMatrices_with_Vector1NormOfMatrices}} {\leq }  
\left\|
\bP_0^{[F_1]} - \bP_0^{[F_2]}
\right \|_1
+
\bigintsss_0^t  \left\|
\G^{[F_1]} - \G^{[F_2]}
\right \|_{1, 1}^\text{(op)}
\, 
\underbrace{
\left\|
\bP^{F_2}\left(\tau \,\bigl|\, \bP_0^{[F_2]}\right)
\right \|_1
}_{1} \d \tau
\\
&\hspace*{40mm}+
\bigintsss_0^t  
\left\| \G^{[F_1]}  \right\|_{1, 1}^\text{(op)} \, \left\| \bP^{F_1}\left(\tau \,\bigl|\, \bP_0^{[F_1]}\right) - \bP^{F_2}\left(\tau \,\bigl|\, \bP_0^{[F_2]}\right) \right\|_1 \d \tau 
\\
%
&=
\underbrace{
\left\|
\bP_0^{[F_1]} - \bP_0^{[F_2]}
\right \|_1
+
t \, 
\left\|
\G^{[F_1]} - \G^{[F_2]}
\right \|_{1, 1}^\text{(op)}
}_{
=: \alpha(t) 
}
+
\bigintsss_0^t  
\underbrace{
\left\| \G^{[F_1]}  \right\|_{1, 1}^\text{(op)}
}_{
\beta
}
\,
\underbrace{
\left\| \bP^{F_1}\left(\tau \,\bigl|\, \bP_0^{[F_1]}\right) - \bP^{F_2}\left(\tau \,\bigl|\, \bP_0^{[F_2]}\right) \right\|_1 \d \tau 
}_{
u(\tau)
} 
\\
&\xLongrightarrow[\text{Groenwall}]{\text{lemma \ref{Lemma_Groenwall_IntegralVersion}}} \\ 
&\Bigl\| 
\bP^{F_1} \left(t \,\bigl|\, \bP_0^{[F_1]} \right)
-
\bP^{F_2}\left(t \,\bigl|\, \bP_0^{[F_2]} \right)
\Bigr\|_1
=
u(t)
\leq 
\alpha \; \e^{\int_0^t \beta } =   \\ 
%
&=
\Bigl( \,
\underbrace{
\left\|
\bP_0^{[F_1]} - \bP_0^{[F_2]}
\right \|_1
}_{
\,< \, \epsilon
}
+
t \, 
\underbrace{
\left\|
\G^{[F_1]} - \G^{[F_2]}
\right \|_{1, 1}^\text{(op)} \, 
}_{
\,< \, \epsilon
}
\Bigr) \,\cdot \, 
\underbrace{
\e^{t \, \left\| \G^{[F_1]}  \right\|_{1, 1}^\text{(op)} } 
}_{ 
\leq \, 
\e^{t \, \left\| \G \right\|_{1, 1}^\text{(op)} } 
} \\ 
%
& \leq \, 
\epsilon \, (1+t) \, \e^{t \, \left\| \G \right\|_{1, 1}^\text{(op)} } . 
\end{aligned}
\end{equation*}

We note that this convergence is not uniform in $t$. 

\end{proof}

\begin{rem}
We note that, while $\| \G \|_{1,1}^\text{(op)} < \infty $ is a sufficient criterion, it suffices that $\g_{n \to} \xlongrightarrow{n \to \infty } 0$. 
\end{rem}


\subsection{The appropriate norm for the generator }

A suitable norm $\|\cdot \|_\text{gen} $ for the generator matrix $\G$ must satisfy two requirements: 
\begin{itemize}
\item[i)] The norm ``$\|\cdot \|_\text{gen} $''  must be compatible with the $1$-norm of $\R^{\Omega} $, since that property was used in theorem \ref{ThermodynamicLimitForMasterEquations}. 
\item[ii) ] The generators of finite subnetworks $\GF$ must be able to approximate the full generator $\G$ with respect to $\|\cdot \|_\text{gen} $.  
\end{itemize}

Since the $\| \cdot \|_{1, 1}^\text{(op)}$-norm fulfills both properties (compare lemma \ref{Lemma_Comparing_1NormOfMatrices_with_Vector1NormOfMatrices} and lemma \ref{Lemma_ApproximatingTheFullGeneratorWithGeneratorsOfFiniteSubnetworks} ), we additionally require in the definition of a master equation \eqref{Eq_MasterEq} that 

\begin{equation}
\begin{aligned}
\infty
&>
\| \G \|_{1, 1}^\text{(op)}
:=
\sum\limits_{i, j \in \Omega} |\G^{(i,j)}|
=
2 \, \sum\limits_{j \in \Omega} \g_{j \to }
=
2 \, \| (\g_{j \to })_{j \in \Omega} \|_1. 
\end{aligned}
\end{equation}

This makes $ \G $ both a continuous operator on $l^1(\Omega)$, as well as the generator of the uniformly continuous semigroup  $ \left( \e^{t \, \G} \right)_{t \geq 0} $, which is defined via its power series:  

\begin{equation}\label{Eq_FormOfTheSemiGroupOfTheMasterEquation}
\e^{t \, \G} 
=
\sum\limits_{k \in \N_0} \frac{(t \, \G)^k}{k!}. 
\end{equation}


\subsection{Comparison with other truncation methods}

In the following section we will briefly discuss two other truncation methods, as opposed to our method of finite subsystems, presented in section \ref{Section_DefinitionOfFiniteSubnetworksAndTheirCorrespondingGenerators}. 

The essential - and perhaps philosophical - question, is, what is meant by the term `approximating an unknown object by a sequence of known objects'. The way we understand it, is that the sequence of known objects have properties that are eventually similar to that of the unknown object, not only in the limit case. 

We briefly show why two alternative truncation methods fall short of these requirements.

\subsubsection{The truncation method `sharp cutoff '}

As briefly mentioned above, a natural method for approximating the solution  of linear differential equations on Banach spaces, is introducing a sharp cutoff of the generator, namely

\begin{equation}
\begin{aligned}
\left(
\G_\text{sharp cutoff}^{[\N_{\leq N}]}
\right)^{(i,\,j)}
&:=
\begin{cases}
(\G)^{(i,\,j)} &\text{ , if } i, j \in F  \\
0 &\text{ , else}. 
\end{cases}
\end{aligned}
\end{equation}

This procedure can result in generators, which do not describe master equation of a finite state space, as shown in equation \eqref{Eq_DifferentKindsOfTruncationForTheGenerator}.

In case of a sharp cutoff, we have an outflow of probability of the system, as the following calculation shows: 

\begin{equation}\label{Eq_OutFlowOfProbabilityForSharpCutoff}
\begin{aligned}
\frac{\d}{\d t} \sum\limits_{i \in F} p^{(i)}(t)
&=
\sum\limits_{i \in F} 
\underbrace{
\frac{\d}{\d t} p^{(i)}(t)
}_{
\sum\limits_{a \in F} p^{(a)}(t) \, \g_{a \to i}
-
\sum\limits_{j \in \Omega} p^{(i)}(t) \, \g_{i \to j}
}
\\ &=
\sum\limits_{a \in F} 
\sum\limits_{i \in F} 
p^{(a)}(t) \, \g_{a \to i}
-
\sum\limits_{i \in F} 
p^{(i)}(t) \, \g_{i \to }
\\ &=
\underbrace{
\left(
\sum\limits_{a \in F} p^{(a)}(t) \, \g_{a \to }
-
\sum\limits_{\alpha \in F} p^{(\alpha)}(t) \, \g_{\alpha \to }
\right)
}_{
0
}
-
\sum\limits_{\beta \in \Omega \backslash F} p^{(\beta)}(t) \, \g_{\beta \to }
\\ &\leq
0. 
\end{aligned}
\end{equation}

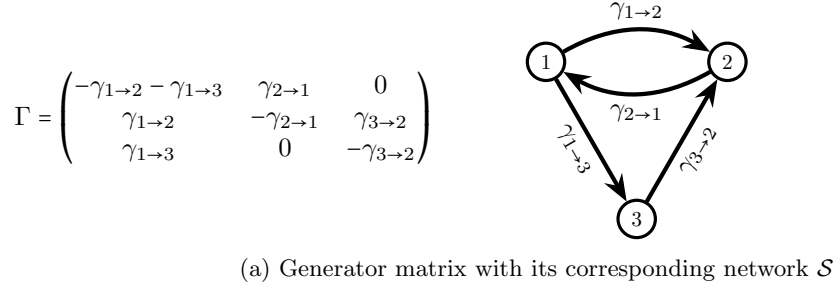
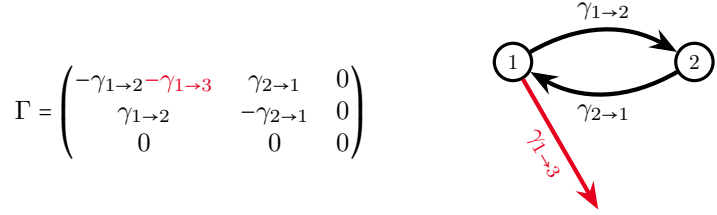
\begin{figure}[H]
\begin{center}
\begin{subfigure}{0.99\textwidth}

\begin{minipage}{0.49\textwidth}
    \begin{equation*} 
    \begin{aligned}
    \G 
    =
    \begin{pmatrix}
    -\g_{1 \to 2 }-\g_{1 \to 3 }& \g_{2 \to 1} & 0 \\
    \g_{1 \to 2 } & -\g_{2 \to 1} & \g_{3 \to2} \\
    \g_{1 \to 3} & 0 & -\g_{3 \to2}
    \end{pmatrix} 
    \end{aligned}
    \end{equation*}
\end{minipage}\begin{minipage}{0.49\textwidth}
\begin{tikzpicture}[scale=2.4]
    \node[State](1) at (0,0) [circle,draw] {1} ;
    \node[State](2) at (1,0) [circle,draw] {2} ;
    \node[State](3) at (0.5,-0.866) [circle,draw] {3} ;
    \path[Link,bend left] (1) edge node[above] {$\g_{1 \to 2}$} (2);
    \path[Link,bend left] (2) edge node[below] {$\g_{2 \to 1}$} (1);
    \path[Link] (1) edge node[left,below,sloped]  {$\g_{1 \to 3}$} (3);
    \path[Link] (3) edge node[left,below,sloped]  {$\g_{3 \to 2}$} (2);
    \end{tikzpicture}
\end{minipage}
\subcaption{Generator matrix with its corresponding network $\System$  }
\end{subfigure}
\begin{subfigure}{0.99\textwidth}

\begin{minipage}{0.49\textwidth}
    \begin{equation*} 
    \begin{aligned}
    \G 
    =
    \begin{pmatrix}
    -\g_{1 \to 2 }{\color{MyRed}-\g_{1 \to 3 } }& \g_{2 \to 1} & 0 \\
    \g_{1 \to 2 } & -\g_{2 \to 1} & 0 \\
    0 & 0 & 0
    \end{pmatrix}
    \end{aligned}
    \end{equation*}
\end{minipage}\begin{minipage}{0.49\textwidth}
        \begin{tikzpicture}[scale=2.4]
        \node[State, black](1) at (0,0) [circle,draw] {1} ;
        \node[State, black](2) at (1,0) [circle,draw] {2} ;
        \node[](3) at (0.5,-0.866) [] {} ;
        \path[Link,bend left, black] (1) edge node[above] {$\g_{1 \to 2}$} (2);
        \path[Link,bend left, black] (2) edge node[below] {$\g_{2 \to 1}$} (1);
        \path[Link, MyRed] (1) edge node[left,below,sloped]  {$\g_{1 \to 3}$} (3);
        \end{tikzpicture}

\end{minipage}

   \subcaption{ The generator matrix \emph{after} a `sharp cutoff' looses the properties of a generator and does no longer correspond to a meaningful network  }



\end{subfigure}

\caption{ Illustrating the fact that a sharp cutoff of the generator results in a matrix, which can not be interpreted as a generator for a network }
\label{} 
\end{center}
\end{figure}

This means that we can non longer interpret the entries $P^{(i)}$ as probabilities, which results in qualitative different behavior between the infinite system and finite cutoffs, no matter how large. 

\subsubsection{The truncation method `condense to some state' }

Another possible truncation method is to consider a finite set $F \in \Fin(\Omega)$ and to `condense' the remaining - infinitely many - states in to a single state `remainder':

\begin{equation} \label{Eq_TruncationMethod_CondensingIntoASingleState} 
\begin{aligned}
\g_{a \to \text{rem}}
&:=
\sum\limits_{\beta \in F^C} \g_{a \to \beta}
\\ 
\g_{\text{rem} \to a}
&:=
\sum\limits_{\beta \in F^C} \g_{ \beta \to a}
\end{aligned}
\end{equation}

While this is feasible and in some cases - as the linear chain on the natural numbers, compare figure \ref{Truncation_CondenseToSingleState_N} - equivalent to our truncation method, it can also create totally different topologies, as depicted in figure \ref{Truncation_CondenseToSingleState_Z}.

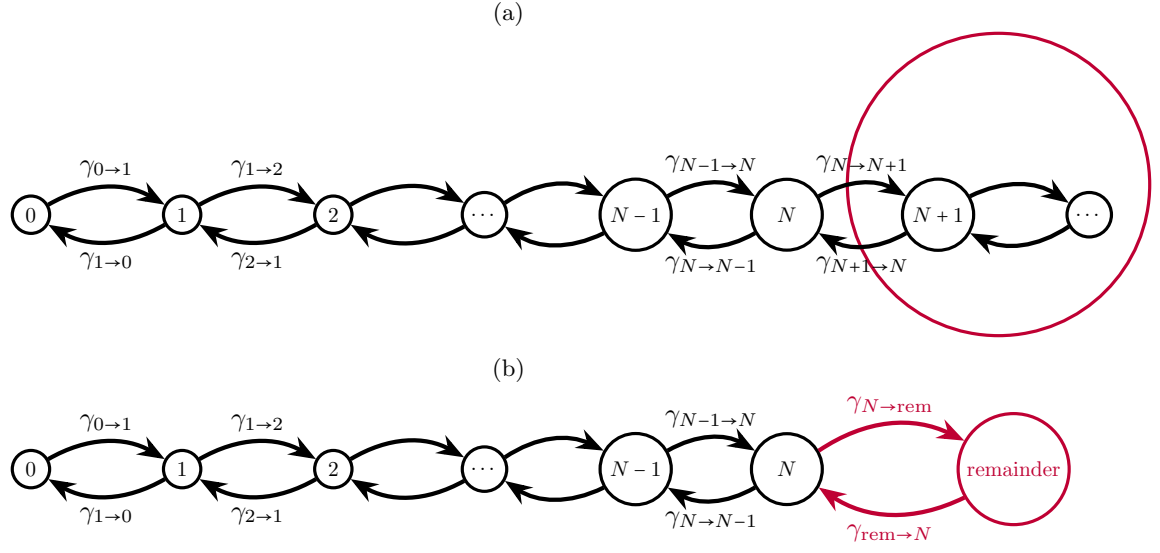
\begin{figure}[H]
\begin{center}
\begin{subfigure}{0.8\textwidth}
   \subcaption{}
      \label{Truncation_CondenseToSingleState_N_0}
\begin{tikzpicture}[scale=2]
\draw[very thick, color=CrimsonGlory] (6.4, 0.2) ellipse (1.0 and 1.0);
\node[State](0)   at (0,0) [circle,draw] {0} ;
\node[State](1)   at (1,0) [circle,draw] {1} ; 
\node[State](2)   at (2,0) [circle,draw] {2} ;
\node[State](3)   at (3,0) [circle,draw] {$\dots$} ;
\node[State](N-1) at (4,0) [circle,draw] {$N-1$} ;
\node[State](N)   at (5,0) [circle,draw] {$\;\;\;N\;\;\;$} ;
\node[State](N+1) at (6,0) [circle,draw] {$N+1$} ;
\node[State](N+2) at (7,0) [circle,draw] {$\dots$} ;
\path[Link,bend left] (0) edge node[above] {$\g_{0 \to 1}$} (1);
\path[Link,bend left] (1) edge node[above] {$\g_{1 \to 2}$} (2);
\path[Link,bend left] (2) edge node[above] {} (3);
\path[Link,bend left] (3) edge node[above] {} (N-1);
\path[Link,bend left] (N-1) edge node[above] {$\g_{N-1 \to N}$} (N);
\path[Link,bend left] (N) edge node[above] {$\g_{N \to N+1}$} (N+1);
\path[Link,bend left] (N+1) edge node[above] {} (N+2);
\path[Link,bend left] (3) edge node[below] {} (2);
\path[Link,bend left] (2) edge node[below] {$\g_{2 \to 1}$} (1);
\path[Link,bend left] (1) edge node[below] {$\g_{1 \to 0}$} (0);
\path[Link,bend left] (N+2) edge node[below] {} (N+1);
\path[Link,bend left] (N+1) edge node[below] {$\g_{N+1 \to N}$} (N);
\path[Link,bend left] (N) edge node[below] {$\g_{N \to N-1}$} (N-1);
\path[Link,bend left] (N-1) edge node[above] {} (3);
\end{tikzpicture}

\end{subfigure}
\begin{subfigure}{0.8\textwidth}
   \subcaption{}
    \label{Truncation_CondenseToSingleState_N_remainder}
\begin{tikzpicture}[scale=2]
\node[State](0)   at (0,0) [circle,draw] {0} ;
\node[State](1)   at (1,0) [circle,draw] {1} ; 
\node[State](2)   at (2,0) [circle,draw] {2} ;
\node[State](3)   at (3,0) [circle,draw] {$\dots$} ;
\node[State](N-1) at (4,0) [circle,draw] {$N-1$} ;
\node[State](N)   at (5,0) [circle,draw] {$\;\;\;N\;\;\;$} ;
\node[State, color=CrimsonGlory](rem) at (6.5,0) [circle,draw, color=CrimsonGlory] {remainder} ;
%
\path[Link,bend left] (0) edge node[above] {$\g_{0 \to 1}$} (1);
\path[Link,bend left] (1) edge node[above] {$\g_{1 \to 2}$} (2);
\path[Link,bend left] (2) edge node[above] {} (3);
\path[Link,bend left] (3) edge node[above] {} (N-1);
\path[Link,bend left] (N-1) edge node[above] {$\g_{N-1 \to N}$} (N);
\path[Link,bend left, color=CrimsonGlory] (N) edge node[above, color=CrimsonGlory] {$\g_{N \to \text{rem}}$} (rem);
%
\path[Link,bend left] (3) edge node[below] {} (2);
\path[Link,bend left] (2) edge node[below] {$\g_{2 \to 1}$} (1);
\path[Link,bend left] (1) edge node[below] {$\g_{1 \to 0}$} (0);
\path[Link,bend left, color=CrimsonGlory] (rem) edge node[below, color=CrimsonGlory] {$\g_{\text{rem} \to N}$} (N);
\path[Link,bend left] (N) edge node[below] {$\g_{N \to N-1}$} (N-1);
\path[Link,bend left] (N-1) edge node[above] {} (3);
\end{tikzpicture}
\end{subfigure}
\caption{Condensing the remainder of a network into a single state, can lead to the same result as the method of finite subsystems  }
\label{Truncation_CondenseToSingleState_N} 
\end{center}
\end{figure}

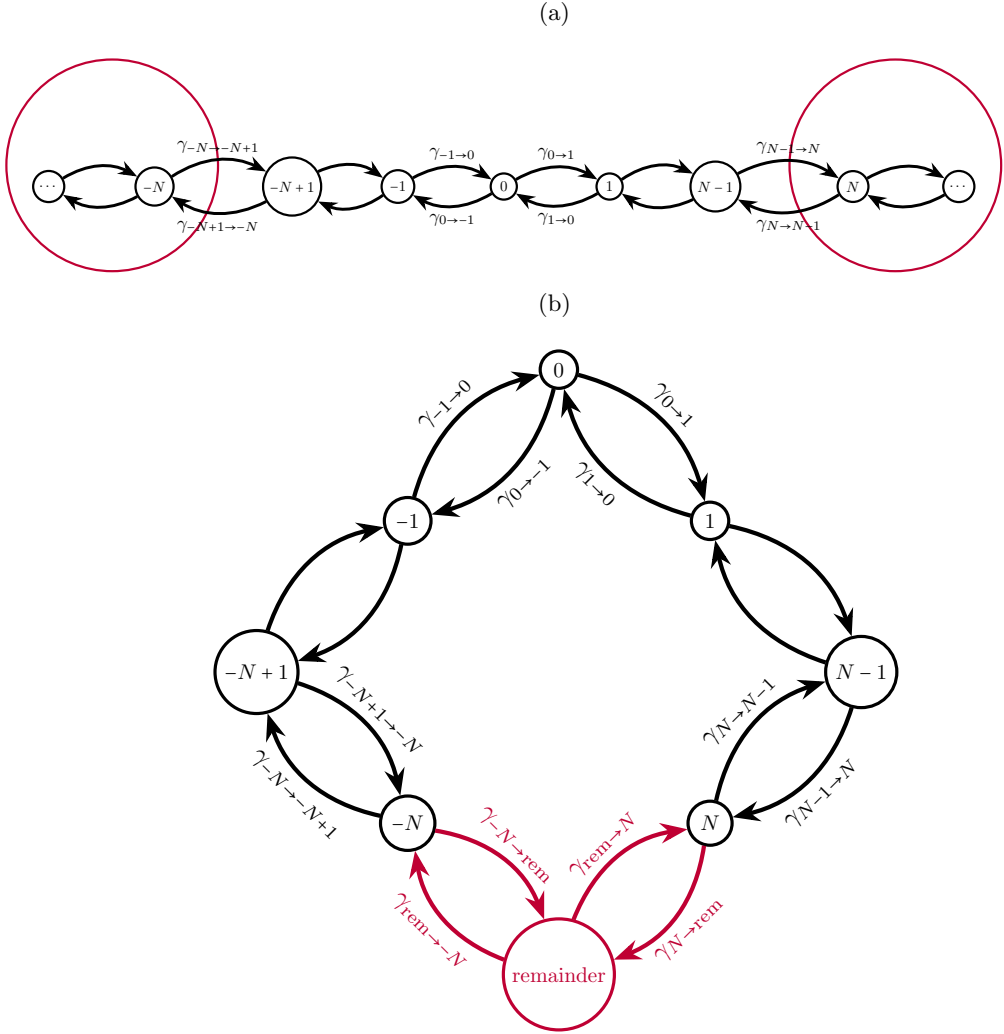
\begin{figure}[H]
\begin{center}
\begin{subfigure}{0.8\textwidth}
   \subcaption{}
      \label{Truncation_CondenseToSingleState_Z_0}
\begin{center}
\hspace*{-15mm}
\scalebox{0.7}{
    \begin{tikzpicture}[scale=2]
\draw[very thick, color=CrimsonGlory] (-3.7, 0.2) ellipse (1.0 and 1.0);
\draw[very thick, color=CrimsonGlory] (3.7, 0.2) ellipse (1.0 and 1.0);
\node[State](-N-1) at (-4.3,0) [circle,draw] {$\dots$} ;
\node[State](-N)   at (-3.3,0) [circle,draw] {$-N$} ;
\node[State](-N+1)   at (-2,0) [circle,draw] {$-N+1$} ; 
\node[State](-1)   at (-1,0) [circle,draw] {$-1$} ;
\node[State](0)    at (0,0) [circle,draw] {$0$} ;
\node[State](1)    at (1,0) [circle,draw] {$1$} ;
\node[State](N-1)    at (2,0) [circle,draw] {$N-1$} ;
\node[State](N)    at (3.3,0) [circle,draw] {$N$} ;
\node[State](N+1)  at (4.3,0) [circle,draw] {$\dots$} ;
\path[Link,bend left] (-N-1) edge node[above] { } (-N);
\path[Link,bend left] (-N)   edge node[above]   {$\g_{-N \to -N+1}$}  (-N+1);
\path[Link,bend left] (-N+1) edge node[above] {}  (-1);
\path[Link,bend left] (-1)   edge node[above] {$\g_{-1 \to 0}$}  (0);
\path[Link,bend left] (0)    edge node[above] {$\g_{0 \to 1}$} (1);
\path[Link,bend left] (1)    edge node[above] {} (N-1);
\path[Link,bend left] (N-1)  edge node[above] {$\g_{N-1 \to N}$ } (N);
\path[Link,bend left] (N)    edge node[above] { } (N+1);
\path[Link,bend left] (N+1)  edge node[below] { } (N);
\path[Link,bend left] (N)    edge node[below]   {$\g_{N \to N-1}$}  (N-1);
\path[Link,bend left] (N-1)  edge node[below] {}  (1);
\path[Link,bend left] (1)   edge node[below] {$\g_{1 \to 0}$}  (0);
\path[Link,bend left] (0)    edge node[below] {$\g_{0 \to- 1}$} (-1);
\path[Link,bend left] (-1)   edge node[below] {} (-N+1);
\path[Link,bend left] (-N+1) edge node[below] {$\g_{-N+1 \to -N}$ } (-N);
\path[Link,bend left] (-N)   edge node[below] { } (-N-1);
\end{tikzpicture}
}

\end{center}
\end{subfigure}
\begin{subfigure}{0.8\textwidth}
   \subcaption{}
   \label{Truncation_CondenseToSingleState_Z_remainder}
\begin{center}
\begin{tikzpicture}[scale=2]
\node[State](-N)   at (-1,1) [circle,draw] {$-N$} ;
\node[State](-N+1)   at (-2,2) [circle,draw] {$-N+1$} ; 
\node[State](-1)   at (-1,3) [circle,draw] {$-1$} ;
\node[State](0)    at (0,4) [circle,draw] {$0$} ;
\node[State](1)    at (1,3) [circle,draw] {$1$} ;
\node[State](N-1)    at (2,2) [circle,draw] {$N-1$} ;
\node[State](N)    at (1,1) [circle,draw] {$N$} ;
\node[State,color=CrimsonGlory](rem)  at (0,0) [circle,draw,color=CrimsonGlory] { remainder} ;
\path[Link,bend left,color=CrimsonGlory] (rem) edge node[below, rotate=-45] {$\g_{ \text{rem} \to -N }$ } (-N);
\path[Link,bend left] (-N)   edge node[below,rotate=-45]   {$\g_{-N \to -N+1}$}  (-N+1);
\path[Link,bend left] (-N+1) edge node[left] {}  (-1);
\path[Link,bend left] (-1)   edge node[above,rotate=45] {$\g_{-1 \to 0}$}  (0);
\path[Link,bend left] (0)    edge node[above,rotate=-45]{$\g_{0 \to 1}$} (1);
\path[Link,bend left] (1)    edge node[left] {} (N-1);
\path[Link,bend left] (N-1)  edge node[below, rotate=45] {$\g_{N-1 \to N}$ } (N);
\path[Link,bend left,color=CrimsonGlory] (N)    edge node[below, rotate=45] {$\g_{N \to \text{rem}}$ } (rem);
\path[Link,bend left,color=CrimsonGlory] (rem)  edge node[rotate=45, above] {$\g_{ \text{rem} \to N }$ } (N);
\path[Link,bend left] (N)    edge node[above, rotate=45]   {$\g_{N \to N-1}$}  (N-1);
\path[Link,bend left] (N-1)  edge node[below] {}  (1);
\path[Link,bend left] (1)    edge node[below,rotate=-45] {$\g_{1 \to 0}$}  (0);
\path[Link,bend left] (0)    edge node[below,rotate=45] {$\g_{0 \to- 1}$} (-1);
\path[Link,bend left] (-1)   edge node[below] {} (-N+1);
\path[Link,bend left] (-N+1) edge node[above,rotate=-45] {$\g_{-N+1 \to -N}$ } (-N);
\path[Link,bend left,color=CrimsonGlory] (-N)   edge node[above, rotate=-45] {$\g_{-N \to \text{rem}}$ } (rem);
\end{tikzpicture}
\end{center}

\end{subfigure}
\caption{The truncation method `condensing the remainder into a single state' can also lead to a totally different topology compared to the original network, as the linear chain on the integers shows}
\label{Truncation_CondenseToSingleState_Z} 
\end{center}
\end{figure}

Since different topologies naturally result in different dynamics - and hence different - long-term behavior, we cannot hope to capture the essential dynamics of the infinite system by a finite, truncated system. The truncation method of condensing the remainder part of a network into a single state seems unfeasible. 
\section{The long-term behavior of an infinite-dimensional master equation }\label{Chapter_TheLongTermBehaviorOfAnInfiniteDimensionalMasterEquation}

Computing the long-term behaviour $t \to \infty $ of the master equation is interesting, because this determines, where a system eventually `is'. For finite dimensional system, this follows directly by applying Perron–Frobenius theorem Gershgorin's circle theorem to the generator matrix. 
However, this theorem is not applicable in the countable, infinite dimensional case. Moreover, we have to specify that we mean by the time limit the convergence in the so-called `strong operator topology', that is $ \lim\limits_{t \to \infty } \|\e^{t \, \G} \bP_0 -  \bP_\infty \|_1 $ for all initial probability states $\bP_0$. We emphasize that this is a strictly weaker version than  the convergence in the operator norm $ \lim\limits_{t \to \infty } \|\e^{t \, \G} -  \e^{\infty \, \G} \|_1^\text{(op)} $, where `$\e^{\infty \, \G}$' denotes the rank one projection operator $ \| \cdot \|_{1} -\lim\limits_{t \to \infty } \e^{t \, \G   } = \bP_* \otimes \bone' $, which can be interpreted as the transition operator at time `$ t=\infty $'.

The pointwise convergence for an irreducible, positive recurrent system was shown in \cite{grimmett2020probability}. Together with lemma \ref{Lemma_ForProbabilityVectors_PointwiseConvergenceImplies1NormConvergence}, this yields the convergence with respect to the $\| \cdot \|_{1} $-norm.

We provide a different proof, which yields some additional results. It can be shown that the space $ l^{1}(\Omega) $ for an irreducible, positive recurrent network can be written as the direct sum of the set of stationary states ($\Kern(\Id - \e^{t \, \G})$) and the closure (with respect to the operator-one-norm) of the image of the identity operator minus the transition matrix ($\overline{ \Image(\Id - \e^{t \, \G}) }$). While this is in general a proper subspace of $l^1(\Omega) $, it coincides with the whole space in the case where the system is both \emph{irreducible} and \emph{positive recurrent}. This means that the minimal assumptions (namely the existence of a stationary solution) are also sufficient to guarantee convergence in the strong operator topology.

Our arguments are based on the following prerequisites: 
Firstly, $(\e^{t \, \G})_{t \geq 0} $ is a \emph{positive} semigroup, meaning that for all $\bX \in l^1(\Omega)$, $\bX \geq 0$ we have $\e^{t \, \G} \, \bX \geq \bzero $ (with `$\bX \geq \bzero $' meaning that $\bX \in (\R_{\geq \, 0})^{\Omega} \backslash \{\bzero \}$). 

Secondly, the time limit $\lim\limits_{t \to \infty} \e^{t \, \G} \, \bP_0 $ for master equation (i.e. CTMC) is closely related to the \emph{Cesaro average} (also called the \emph{mean ergodic limit}) $ \lim\limits_{M \to \infty} \frac{1}{M} \sum\limits_{i=0}^{M-1} Q^{i} \bP_0  $ for DTMC. The so-called \emph{mean ergodic space}, where both these limits exists, is precisely the direct sums described above: $ \lme = \Kern(\Id - Q) \oplus \overline{ \Image(\Id - Q)}.  $


For \emph{positive}, \emph{irreducible} semigroups, we know that there exists a rank one projection $\e^{\infty \, \G} = \bP_* \otimes \bone' $ such that $ \e^{t \, \G} \, \bP_0 \xlongrightarrow[\|\cdot \|_1]{ t \to \infty} \bigl(\bP_* \otimes \bone'\bigr) \cdot \bP_0 = \bP_* $ (see lemma \ref{ConvergenceOfPositiveIrreducibleSemigroupsOnLp} for details). While this theorem is stated in \cite{arendt2020positive} as Corollary 3.7, and some parts can be found in the literature such as \cite{batkai2017positive, arendt1986one}, we are still lacking a direct proof that is accessible to physicists, without diving deep into the mathematical literature. This section is meant to close that gap.

\subsection{The mean ergodic space of $l^1(\Omega)$} \label{Section_TheMeanErgodicSpaceOf_l1Omega}

We start by defining the finite Cesaro mean, as well as the mean ergodic subspace $ \lme $ for a given state space and a corresponding transition matrix of a DTMC. 

\begin{MyDef}[The finite Cesaro mean and the \emph{mean ergodic subspace}] $ $ \\ 
Let $\Omega$ be a countable, infinite state space and $Q \in [0,1]^{\Omega \times \Omega} $ a countable, stochastic matrix, that is  $ \sum\limits_{i \in \Omega} Q^{(i, j)} = 1 $. The \emph{finite Cesaro mean} $\CMean_M$ of the operator $Q$ is defined as 

\begin{equation*} 
\begin{aligned}
\CMean_M (\bX) 
&:=
\CMean_M^{Q} (\bX) 
:=
\frac{1}{M} \, \sum\limits_{i=0}^{M-1} (Q)^{i}( \bX)
&&\text{  for $M \in \N $} 
\end{aligned}
\end{equation*}

We further define the mean ergodic subspace of $l^1(\Omega) $ with respect to $ Q $ as

\begin{equation}
\begin{aligned}
\lme
&:=
\{
\bX \in l^1(\Omega) \,:\, \lim\limits_{M \to \infty} \CMean_M(\bX) \text{ exists with respect to } \| \cdot \|_1 \}, 
\end{aligned}
\end{equation}

where ``m.e'' stands for ``mean ergodic´´. 
\end{MyDef}

Let us now study the mean ergodic space, as well as the corresponding mean ergodic averages. As it turns out, the Cesaro mean projects the mean ergodic subspace onto the set of fixed points of the transition matrix.

\begin{lemma}\label{Lemma_PropertiesOfCMean}
The mean ergodic subspace is closed in $l^1(\Omega)$ and the operator

\begin{equation}
\begin{aligned}
\CMean \,: \,  l^1_{m.e.}&(\Omega, Q)\to l^1(\Omega) \\
    &\bX  \hspace*{6mm} \mapsto \lim\limits_{M \to \infty} \CMean_M(\bX)
\end{aligned}
\end{equation}
 is a projection onto $\Kern(\Id - Q)$, which is the set of fixed points of $Q$.

\end{lemma}

\begin{proof}

First, we note that since $Q$ is stochastic, the operator norm of $Q$, $\CMean_M$ and $\CMean$ are bounded by one: 

\begin{equation}
\begin{aligned}
\| Q \|_{1}^\text{(op)}
&:=
\sup\limits_{\| \bX \|_1 = 1 } 
\underbrace{
\|Q  \, \bX \|_1
}_{
\sum\limits_{i \in \Omega } |(Q  \, \bX)^{(i)}|
}
=
\sup\limits_{\| \bX \|_1 = 1 } 
\sum\limits_{i \in \Omega } \, 
\Bigl| 
\sum\limits_{j \in \Omega }
Q^{(i,j)} \, X^{(j)}
\Bigr |
\leq 
\sup\limits_{\| \bX \|_1 = 1 } \, 
\sum\limits_{j \in \Omega } |X^{(j)}| \; 
\underbrace{
\left(
\sum\limits_{i \in \Omega } Q^{(i,j)} 
\right)
}_{
1
}
\\ &=
\sup\limits_{ \| \bX \|_1 = 1 } 
\underbrace{
\sum\limits_{j \in \Omega } |X^{(j)}|
}_{
\| \bX \|_1
}
=
1, 
\\ 
\| \CMean_M \|_{1}^\text{(op)}
&=
\sup\limits_{\| \bX \|_1 = 1 } 
\left\| \frac{1}{M}
\sum\limits_{i=0}^{M-1} Q^{i} \bX 
\right \| 
\leq
\sup\limits_{\| \bX \|_1 = 1 } 
\frac{1}{M} 
\sum\limits_{i=0}^{M-1} 
\underbrace{
\| Q^{i}  \bX\|
}_{
\leq \| Q  \|^{i} \, \|\bX\|_1 
}
\leq 
1, 
\\ 
\| \CMean \|_{1}^\text{(op)}
&=
\sup\limits_{ \| \bX \|_1 = 1 } 
\|\lim\limits_{M \to \infty} \CMean_M(\bX) \|_1
=
\sup\limits_{ \| \bX \|_1 = 1 } 
\lim\limits_{M \to \infty} 
\underbrace{
\| \CMean_M(\bX) \|_{1}^\text{(op)}
}_{
\leq \, \|\bX\|_1 
}
\leq 
1 . 
\end{aligned}
\end{equation}

Further, we note that $Q$ - and therefore $\CMean_M$ - are a positive operators (meaning $ \bX \geq\bzero  \Rightarrow Q \, \bX \geq\bzero $ and $\CMean_M \, \bX \geq \bzero$) and that the norm for positive vectors is preserved ($ \bX \geq\bzero  \Rightarrow \| \CMean_M \, \bX \| =\| Q \, \bX \| = \| \bX \| $).

Now let $\bX \in \overline{ \lme } $, then we know that there exists a sequence $(\bX_n)_{n\in \N} \subseteq \lme $ converging to $\bX$  and a sequence $ (\bX_{n, *})_{n \in \N} $ such that $ \lim\limits_{M \to \infty} \CMean_M(\bX_n) = \bX_{n, *}$.

First, we show that the sequence $ (\bX_{n, *})_{n \in \N} $  is Cauchy. We can choose the number $m$ and $n$ large enough, such that 
\begin{equation*}
\begin{aligned}
\| \bX_{n, *} - \bX_{m, *}\|_1
&=
\| \bX_{n, *}  \pm \CMean_M (\bX_n) \pm \CMean_M (\bX_m) - \bX_{m, *}\|_1
\\ &\leq 
\underbrace{
\| \bX_{n, *} - \CMean_M (\bX_n) \|_1
}_{
< \, \epsilon
}
+
\underbrace{
\| \CMean_M \|_1 \,
}_{
1
} \, 
\underbrace{ \| \bX_n - \bX_m  \|_1
}_{
< \, \epsilon
}
+
\underbrace{
\| \CMean_M (\bX_m) - \bX_{m, *} \|_1
}_{
< \, \epsilon
}
 \\ &< 3 \, \epsilon, 
\end{aligned}
\end{equation*}

hence the limit $ \bX_* := \lim\limits_{n \to \infty} \bX_{n, *} $ exists in $l^1(\Omega)$. Further, we have:

\begin{equation*}
\begin{aligned}
 \|  \CMean_M (\bX) - \bX_* \|_1
&=
\| \CMean_M (\bX) \pm \CMean_M (\bX_n) \pm \bX_{n, *} - \bX_* \|_1
 \\ & \leq 
\underbrace{
    \| \CMean_M  \|_1
}_{
    \leq \, 1
}
\, \cdot \,
\underbrace{
    \|\bX - \bX_n \|_1
}_{
    < \, \epsilon
}
+
\underbrace{
\| \CMean_M (\bX_n) - \bX_{n, *}  \|_1
}_{< \, \epsilon}
+
\underbrace{
\| \bX_{n,*} - \bX_*\|_1
}_{< \, \epsilon}
 \\ &< 3 \, \epsilon,  
\end{aligned}
\end{equation*}

that is $ \CMean_M (\bX) \xlongrightarrow[\| \cdot \|_1]{M \to \infty} \bX_*$ and hence $\bX \in \lme$. This shows that $ \lme $ is closed in $l^1(\Omega) $.

Now, let us determine the image of $\CMean $. Since $Q$ is a continuous operator, it "commutes" with the Cesaro operator $\CMean$ on $ \lme $:

\begin{equation}\label{Eq_QCommutesWithCMean}
\begin{aligned}
Q \, \CMean(\bX)
&=
Q \, 
\lim\limits_{M \to \infty}
\frac{1}{M} \, \sum\limits_{i=0}^{M-1} (Q)^{i}( \bX)
\xlongequal[\text{continuous}]{\text{$Q$ is }}
\lim\limits_{M \to \infty}
\frac{1}{M} \, \sum\limits_{i=0}^{M-1} (Q)^{i+1}( \bX)
=
\CMean\left( Q \, \bX \right), 
\end{aligned}
\end{equation}

for all $\bX \in \lme $. What is more, is that equation \eqref{Eq_QCommutesWithCMean} tells us that $\lme $ is $Q$-invariant, meaning $Q \, \lme \subseteq \lme $.

Let us now look at the effect of a finite Cesaro operator acting on an element of the image of $\Id - Q$: For $\bX \in l^1(\Omega) $ we have: 
\begin{equation} \label{Eq_CMeanActingOnIdMinusQ_X_1}
\begin{aligned}
\CMean_M \, \bigl((\Id - Q) \, \bX \bigr)
&=
\frac{1}{M} \, 
\underbrace{
\sum\limits_{i=0}^{M-1}  Q^{i}(\Id - Q)
}_{
Q^{0} - Q^{M} 
} \, \bX =
%
\frac{\bX - Q^{M} \, \bX}{M}
\xlongrightarrow{M \to \infty}
\bzero, 
\end{aligned}
\end{equation}

which mean that the mean ergodic subspace of $Q$ contains the image of $\Id - Q$, that is $ \Image(\Id - Q) \subseteq \lme$. 

On the other hand, we have for all $\bX \in \lme $: 
\begin{equation} \label{Eq_CMeanActingOnIdMinusQ_X_2}
\begin{aligned}
\bzero
&\xlongequal{\eqref{Eq_CMeanActingOnIdMinusQ_X_1}}
\underbrace{
\lim\limits_{M \to \infty} \CMean_M
}_{
\CMean 
}
\, \bigl((\Id - Q) \, \bX \bigr)
=
\CMean(\bX) - 
\underbrace{
\CMean(Q \, \bX)
}_{
Q \, \CMean( \bX)
}
\xlongequal{\eqref{Eq_QCommutesWithCMean}}
(\Id - Q ) \, \CMean(\bX), 
\end{aligned}
\end{equation}

which means that the Cesaro average $\CMean(\bX)$ is a fixed point of the operator $Q$, that is $\CMean(\bX) \in \Kern(\Id - Q) $.

It now remains to show that the Cesaro operator is indeed a projection. This can be shown by applying the Cesaro operator twice, and keeping in mind that it is linear, bounded and that its image in $Q$-invariant:

\begin{equation*} 
\begin{aligned}
\CMean^2 \, (\bX) 
&=
\CMean\left(
\lim\limits_{M \to \infty } 
\frac{1}{M} \sum\limits_{i=0}^{M-1} Q^{i} \, \bX 
\right)
\xlongequal[\text{continuous}]{\text{C is }}
\lim\limits_{M \to \infty } 
\underbrace{
\frac{1}{M} \sum\limits_{i=0}^{M-1} 
\overbrace{
\CMean\left(
Q^{i} \, \bX
\right)
}^{
\CMean\left(
\bX
\right)
}
}_{
\CMean\left(
\bX
\right)
}
\xlongequal{\eqref{Eq_CMeanActingOnIdMinusQ_X_2}}
\CMean\left(
\bX
\right). 
\end{aligned}
\end{equation*}

So we can say $\lme $ is a closed, $Q$- invariant subspace, where the Ceasro average is a projection onto the kernel of $\Id - Q$: 

\begin{equation*}
\begin{aligned}
\CMean \,:\, \lme &\to \Kern(\Id - Q) \\
\bX & \mapsto \lim\limits_{M \to \infty} \CMean_M(\bX)
\end{aligned}
\end{equation*}

\end{proof}

We are now in a position to characterize the mean ergodic subspace: While it clearly contains the set of fixed points of $Q$, the previous calculations have shown, that the mean ergodic average vanishes for every vector contained in the image of the identity minus the transition matrix. As it turns, these subspaces uniquely determine the mean ergodic space.

\begin{thm}\label{Thm_ExplicitFormOfTheMeanErgodicSubspaceInl1}
The mean ergodic space of a stochastic operator $Q$ in $l^1(\Omega)$ equals the direct sum of the kernel of $\Id - Q$ and the closure of the image of $\Id - Q$, that is 

\begin{equation}
\begin{aligned}
\lme 
&=
\Kern(\Id - Q) \oplus \overline{ \Image(\Id - Q)}. 
\end{aligned}
\end{equation}
\end{thm}

\begin{proof}
$ $ \\ 
\begin{itemize}
\item["$\supseteq$"]

This part was indirectly shown in lemma \ref{Lemma_PropertiesOfCMean}: \\
The mean ergodic space clearly contains every fixed point of $Q$ (that is $ \Kern(\Id - Q) \subseteq \lme$ ), while equation \eqref{Eq_CMeanActingOnIdMinusQ_X_2} tells us that it also contains the image of $\Id - Q$ ($ \Image(\Id - Q) \subseteq \lme $). Finally, since $\lme $ is closed in $l^1(\Omega) $, the claim follows.

\item["$\subseteq $" ]
Let $\bX \in \lme $, that is $\CMean_M(\bX) \xlongrightarrow[\| \cdot \|_1]{M \to \infty} \bX_* $ for some $\bX_* \in \Kern(\Id - Q) $.

Then we can separate $\bX $ into three parts: One belonging to $\Image(\Id - Q) $, one belonging to $\Kern(\Id - Q) $ and one vanishing:

\begin{equation}\label{Eq_DeterminingTheMeanErgodicSubspace}
\begin{aligned}
\bX
=
\underbrace{
\bX - \CMean_M(\bX) 
}_{
\in  \, \Image(\Id - Q)
}
+
\underbrace{
\CMean_M(\bX) - \bX_*
}_{
\xlongrightarrow{M \to \infty } \, \bzero
}  
+
\underbrace{
\bX_*
}_{
\in \,  \Kern(\Id - Q)
}, 
\end{aligned}
\end{equation}

which implies $\bX \in  \overline{ \Kern(\Id - Q) + \Image(\Id - Q)} $. 

What remains to show is the first under brace, claiming that the image of $\Id - \CMean_M $ is contained in the image of $\Id - Q $. We do this, by consider the following calculation:

\begin{equation*}
\begin{aligned}
\bigl(\Id - Q\bigr) \, &\sum\limits_{i=0}^{M-2} \, 
\left( \frac{M-1-i}{M} \right) \, Q^{i}
=
\underbrace{
\sum\limits_{i=0}^{M-2}
\, 
\left( \frac{M-1-i}{M} \right)
\, Q^{i}
}_{
\sum\limits_{j=0}^{M-2} \left(1 - \frac{j+1}{M} \right) \, Q^{j}
}
-
\underbrace{
\sum\limits_{i=0}^{M-2}\frac{M-1-i}{M} \, Q^{i+1}
}_{
\sum\limits_{j=1}^{M-1} \left(1 - \frac{j}{M} \right) \, Q^{j}
}
\\ &= 
\underbrace{
\left(1 - \frac{1}{M} \right) \, Q^{0}
}_{
\Id - \frac{1}{M} \, Q^{0} 
}
+
\sum\limits_{j=1}^{M-2} 
\underbrace{
\left[ 
\left(1 - \frac{1+j}{M} \right)
-
\left(1 - \frac{j}{M} \right)
\right] 
\, Q^{j}
}_{
\frac{-1}{M} \, Q^{j}
}
- 
\underbrace{
\left(1 - \frac{M-1}{M} \right) 
\, Q^{M-1}
}_{
-\frac{1}{M}
\, Q^{M-1}
}
\\ &= 
\Id -  \frac{1}{M} \, 
\underbrace{
\left( \sum\limits_{i=0}^{M-1} \, Q^{i} \right)
}_{
\CMean_M 
}
=
\Id - \CMean_M \text{  , which implies }
\\ 
\bX - \CMean(\bX)
&=
\left(\Id - Q\right) \, \sum\limits_{i=0}^{M-2} \, 
\left( \frac{M-1-i}{M} \right) \, Q^{i} \bX
\in
\Image(\Id - Q). 
\end{aligned}
\end{equation*}

What remains to show, is that the sum of the two vector spaces $\Kern(\Id - Q) $ and $ \overline{\Image(\Id - Q)} $ in indeed direct. Let $\bX $ be an element of their intersection, that is 
$ \bX \in \Kern(\Id - Q) \cap  \overline{\Image(\Id - Q)} $. We now proceed by showing that the norm of $\bX $ is arbitrary small. 

Since $\bX \in \Kern(\Id - Q) $, we have $\bX = Q \, \bX = \CMean_M(\bX) $. On the other hand, $ \bX \in \overline{\Image(\Id - Q)} $ implies that for every $\epsilon > 0$ there exists an element $\bY_\epsilon \in l^1(\Omega)$ such that $ \|\bX - (\Id - Q) \bY_\epsilon \|_1 < \epsilon$.

Now - for a fixed $\epsilon > 0$ - choose a natural number $M = M_\epsilon \in \N$ such that $ \|\CMean_M (\Id - Q) \bY_\epsilon \| < \epsilon $, which is possible by equation \eqref{Eq_CMeanActingOnIdMinusQ_X_1}. Then we can estimate the norm of $\bX$ to 

\begin{equation}
\begin{aligned}
\| \bX\|_1
&=
\| 
\CMean_M(\bX)
\|_1
=
\|
\CMean_M(\bX) \pm \CMean_M \, (\Id - Q) \, \bY_\epsilon
\|_1
\leq
\underbrace{
    \|
    \CMean_M \left[
    \bX - (\Id - Q) \, \bY_\epsilon
    \right]
    \|_1 
}_{
    \leq \| \CMean_M \|_{1}^\text{(op)} 
    \, 
    \|
    \bX - (\Id - Q) \, \bY_\epsilon
    \|_1
}
+ 
\| 
\CMean_M \, (\Id - Q) \, \bY_\epsilon
\|_1
\\ &\leq 
\underbrace{
    \| \CMean_M \|_{1}^\text{(op)} 
}_{1}
\, 
\underbrace{
    \| \bX - (\Id - Q) \, \bY_\epsilon \|_1
}_{< \, \epsilon }
+
\underbrace{
    \| 
    \CMean_M \, (\Id - Q) \, \bY_\epsilon
    \|_1
}_{
    < \, \epsilon 
}
<
2 \, \epsilon. 
\end{aligned}
\end{equation}

Since $\epsilon > 0$ was arbitrary, this mean that the norm of $\bX $ must be arbitrary small, hence $\bX$ must vanish. This concludes the proof. 

\end{itemize}
\end{proof}


\subsection{The ergodic mean of the transition matrix of an irreducible, positive recurrent Markov chain} \label{Section_TheErgodicMeanOfTheTransitionMatrixOfAnIrreduciblePositiveRecurrentMarkovChain}

So far we characterized the mean ergodic space for arbitrary DTMC, but eventually we are interested in the long-term behavior, for which we need the existence of stationary solutions. Stationary solutions, on the other hand, require - as we have seen in section \ref{Chapter_MarkovChains} - positive recurrence on each strongly connected component. 

It turns out that for networks which are both strongly connected and positive recurrent, all initial states are mean ergodic.

\begin{thm}[For irreducible and positive recurrent networks all states are mean ergodic] \label{Thm_Irreducible_PositiveRecurrent_l1IsMeanErgodic} $ $ \\  

If $Q$ is the stochastic matrix of an \emph{irreducible}, \emph{positive recurrent} Markov chain, then the whole space $l^1(\Omega) $ is mean ergodic with respect to $Q$, that is $ \lme =  l^1(\Omega) $, or, in other words: The Cesaro mean (i.e. the ergodic mean) of a countable stochastic matrix in $ l^1(\Omega ) $ converges strongly, if the underlying Markov chain is both irreducible and positive recurrent.  
    
\end{thm}

\subsection*{basic idea}
Since the network is irreducible and positive recurrent, a unique stationary solution $\bP_{*} \in (\R_{> \, 0})^{\Omega} \cap \Kern(\G) \cap \B_{r=1}^{\| \cdot \|_{1}}(0) $ exists and it is possible - similar to equation \eqref{Eq_DeterminingTheMeanErgodicSubspace} - to write every unit vector $ \bE_{\omega} $ as a sum, consisting of a summand $ \bE_{\omega} - \CMean_{M}(\bE_{\omega}) $ contained in the image of $\Id - Q$, a summand $ \CMean_{M}(\bE_{\omega}) - \bP_* $ converging to zero and third summand $ \bP_* $ contained in the set of stationary states:

\begin{equation*}
\begin{aligned}
\bE_{\omega}
&=
\underbrace{
    \bE_{\omega} - \CMean_{M}(\bE_{\omega})
}_{
\in \, \Image(\Id - Q)
}
+
\underbrace{
    \CMean_{M}(\bE_{\omega}) - \bP_*
}_{
\xlongrightarrow[\| \cdot \|_1]{ ??? } \,  \bzero
}
+
\underbrace{
    \bP_*
}_{
    \in \, \Kern(\Id - Q)
}
\end{aligned}
\end{equation*}

The problem is now, that in order to show that the whole space $ l^1(\Omega) $ is a mean ergodic, we have to assume that the Cesaro mean $\lim\limits_{M \to \infty} \CMean_{M} (\bE_{\omega}) $ of every unit vector $\bE_{\omega}$ converges to $\bP_*$, which would be a circular reasoning. 

This can be circumvented, by showing that for all  $\omega \in \Omega $ the sequence $ \left(\CMean_{M}(\bE_{\omega}) \right)_{M \in \N} \subseteq l^1(\Omega) $ has a convergent \emph{subsequence} $ \left(\CMean_{M_n}(\bE_{\omega}) \right)_{n \in \N}$ in $l^1(\Omega) $. This limit must coincide with $\bP_*$, since: 

\begin{itemize}
\item[(i)]
It lies in the kernel of $\Id-Q$: 

\begin{equation*}
\begin{aligned}
(\Id - Q) \lim\limits_{n \to \infty} 
\underbrace{
    \CMean_{M_n}(\bE_\omega)
}_{
    \frac{1}{M_n} \sum\limits_{k=0}^{M_n-1} Q^{k}(\bE_{\omega}) 
}
\xlongequal{\text{Q continuous}}
\lim\limits_{n \to \infty} 
\underbrace{
    (\Id - Q) \, \sum\limits_{k=0}^{M_n-1} Q^{k}(\bE_{\omega})
}_{
    \frac{\bE_{\omega} - Q^{M_n}\bE_{\omega}}{M_n}
}
=
\bzero
\end{aligned}
\end{equation*}

\item[(ii)] The norm is conserved by the limit $\lim\limits_{n \to \infty} \CMean_{M_n}  $: 

\begin{equation*}
\begin{aligned}
\| \lim\limits_{n \to \infty} \, \CMean_{M_n}(\bE_\omega) \|_1 
\xlongequal{\| \cdot \|_1 \text{ continuous}}
\lim\limits_{n \to \infty} \, 
\underbrace{
    \| \CMean_{M_n}(\bE_\omega) \|_1 
}_{
    1
    }
=
1. 
\end{aligned}
\end{equation*}

\end{itemize}

Since $(\bE_{\omega})_{\omega \in \Omega} $ forms a Schauder basis of $ l^1(\Omega)$, we can conclude that the vector space $l^1(\Omega)$ of summable entries of $\Omega$ equals the sum of the closure of the image- and the kernel of the operator $\Id - Q $, that is: \\
$ l^1(\Omega) = \overline{ \Kern(\Id - Q) + \Image(\Id - Q)} $.



\begin{proof}
Since the Markov chain is irreducible, the dimension of the fixed points of $Q$ (that is the dimension of the kernel of $ \Id - Q $) is at most one (compare lemma \ref{Lemma_StrictPositivityForStationarySolutionsOfStronglyConnectedNetworks}) and due to positive recurrence, there exists a strictly positive probability vector $\bP_* \in \ProbStates \cap (\R_{> \, 0})^{\Omega} $ such that $ \Id - Q = \Span(\bP_*)$ (compare lemma \ref{Lemma_PositiveRecurrenceForCTMC} ).

In order to show the existence of a converging subsequence of $ \left(\CMean_{M}(\bE_{\omega}) \right)_{M \in \N} $, we make use of the fact that the element of the kernel of the generator has strictly positive entries: $\bP_* \in \Kern(\Id - Q) \cap (\R_{>0})^\Omega $. This allows us to define the vector $ \left( \frac{P_*^{(\alpha)}}{P_*^{(\omega)}}\right)_{\alpha \in \Omega } $, which at the same time, is an element-wise upper bound for $\bE_{\omega} $, that is we have

\begin{equation}\label{Eq_CapturingErgodicMean}
\begin{aligned}
\bzero
&\leq \;\;
\bE_{\omega}
&\leq
\frac{\bP_*}{P_*^{(\omega)}} 
&\text{   element-wise, that is } 
\\
0
&\leq \;\;
E_{\omega}^{(\alpha)}
&\leq
\frac{P_*^{(\alpha)}}{P_*^{(\omega)}} 
&\text{   for all  } \alpha \in \Omega. 
\end{aligned}
\end{equation}

We note that equation \eqref{Eq_CapturingErgodicMean} reads the following:

\begin{equation*}
\begin{aligned}
0
&\leq \;\;
\underbrace{
E_{\alpha}^{(\alpha)}
}_{
1
}
&&\leq
\underbrace{
\frac{P_*^{(\alpha)}}{P_*^{(\alpha)}} 
}_{
1
} \text{    for $\alpha=\omega$ and  }
\\ \\ 
0
&\leq \;\;
\underbrace{
E_{\omega}^{(\alpha)}
}_{
0
}
&&\leq
\underbrace{
\frac{P_*^{(\alpha)}}{P_*^{(\omega)}} 
}_{
> \, 0
} \text{  for $\alpha \neq \omega $.  }
\end{aligned}
\end{equation*}

Since $ Q $ is a positive, linear operator (and hence $(Q)^{i} $ is positive, for all natural numbers $i \in \N_0$), we conclude: 

\begin{equation*}
\begin{aligned}
\bzero
\overset{\eqref{Eq_CapturingErgodicMean}}{\leq} 
\left(Q \right)^{i} ( \bE_{\omega})
\overset{\eqref{Eq_CapturingErgodicMean}}{\leq} 
\left(Q \right)^{i} \, 
\left(
\frac{\bP_*}{P_*^{(\omega)}}
\right)
\xlongequal[\text{point of $ Q $}]{\text{$ \bP_*$ is a fixed}}
%
%
\frac{\bP_*}{P_*^{(\omega)}}. 
\end{aligned}
\end{equation*}

By summing over $i$ from $0$ to $M-1$, we can make the same estimation for the corresponding Cesaro-mean: 

\begin{equation*}
\begin{aligned}
\bzero
\leq 
\CMean_{M} \left( \bE_{\omega} \right) = \frac{1}{M} \, \sum\limits_{i=0}^{M-1} \, \left( Q \right)^{i} \, ( \bE_{\omega})
\leq
\frac{\bP_*}{P_*^{(\omega)}}. 
\end{aligned}
\end{equation*}

Since for every $\alpha \in \Omega $ the sequence $ \left(\left( \CMean_{M} \left( \bE_{\omega} \right)  \right)^{(\alpha)} \right)_{M \in \N} $ is contained in the compact interval $ \left[0,  \frac{P_*^{(\alpha)}}{P_*^{(\omega)}} \right] $, we can construct - using a `diagonal argument' - a pointwise converging subsequence of $ \left( \CMean_{M_n} \left( \bE_{\omega} \right)^{} \right)_{n\in\N} $ . 

An alternative way to argue would be to consider the set $ \prod\limits_{\alpha \in \Omega} \left[0,  \frac{P_*^{(\alpha)}}{P_*^{(\omega)}} \right] $, which by Tychonoff's theorem, is compact (in the product topology) as a product of compact spaces. Since every sequence in a compact space contains a convergent subsequence, we get a - pointwise - converging subsequence $  \left( \CMean_{M_n} \left( \bE_{\omega}   \right)^{} \right)_{n\in\N} $ to $ \bP_* \in \ProbStates \cap  \prod\limits_{\alpha \in \Omega} \left[0,  \frac{P_*^{(\alpha)}}{P_*^{(\omega)}} \right] $. 

So far, we have shown the existence of a \emph{pointwise} converging subsequence, that is for all states $\alpha \in \Omega $ and all small numbers $\epsilon >0$, we know of the existence of a natural number $ N_{\alpha, \epsilon} \in \N $ such that for all larger numbers $ n \geq N_{\alpha, \epsilon} $ we have: 

\begin{equation*}
\begin{aligned}
\Bigl|
\bigl(\CMean_{M_n}(\bE_\omega)\bigr)^{(\alpha)} - P_*^{(\alpha)}
\Bigr|
<
\epsilon. 
\end{aligned}
\end{equation*}

What remains of the proof, is to show is that this convergence is not only \emph{pointwise}, but also with respect to the $\| \cdot \|_1$-norm (i.e. that the number $N_{\alpha, \epsilon}$ can be chosen independent of $\alpha $). 

Since both the elements $\CMean_{M_n}(\bE_{\omega}) $ of the sequence, as well as $\bP_{*} $ are probability vectors (that is, $ \CMean_{M_n}(\bE_{\omega}), \bP_{*} \in \ProbStates$ ), this is a direct consequence of lemma \ref{Lemma_ForProbabilityVectors_PointwiseConvergenceImplies1NormConvergence}, which concludes the proof.

\end{proof}



\subsection{The long-term behavior of an irreducible, positive recurrent master equation} \label{Section_TheLongTermBehavioOfAnIrreduciblePositiveRecurrentMasterEquation}

Similar to the discrete-time case we define the subspace where the long-term behavior is well defined. The time limit $\xlongrightarrow[\| \cdot \|_1]{t \to \infty} $ - in analogy to the Cesaro mean - is then again a projection onto the stationary states.

\begin{lemma}
Let $\ltb := \{ \bX \in l^1(\Omega) \,:\, \lim\limits_{t \to \infty } \e^{t \, \G} \bX \text{  exists with respect to } \| \cdot \|_1 \} $. Then $\ltb $ is closed in $l^1(\Omega) $ and the mapping

\begin{equation*}
\begin{aligned} 
 \LL \, : \, \ltb &\to l^1(\Omega) 
\\
\bX &\mapsto \lim\limits_{t \to \infty} \e^{t \, \G} \, \bX 
\end{aligned}
\end{equation*}

is a projection onto $\Kern(\Id - \e^{t \, \G}) = \Kern(\G) $. 
\end{lemma}

\begin{proof}
The closedness of $\ltb $ can be shown in analogy to the proof in lemma \ref{Lemma_PropertiesOfCMean}, by replacing `$ \lim\limits_{M \to \infty} \CMean(\bX) $' by `$ \lim\limits_{t \to \infty} \e^{t \, \G} $' (see lemma \ref{Lemma_ltbClosedIn_l1} in the appendix for details).

The fact that the image of $ \LL $ is contained in the fixed space of $\e^{t \, \G} $ can be seen from the following calculation: For $\bX \in \ltb $ and $\bX_* := \LL(\bX)$, we have 

\begin{equation*}
\begin{aligned} 
\|
\left(
\Id - \e^{\tau \, \G}\right) \LL(\bX) 
\|_1
&=
\|
\underbrace{
\LL(\bX)
}_{
\lim\limits_{t \to \infty } \e^{t \, \G} \, \bX
}
-
\e^{\tau \, \G}  
\underbrace{
\LL(\bX)
}_{
\lim\limits_{t \to \infty } \e^{t \, \G} \, \bX
}
\pm 
\,  \bX_*  \|_1
\leq 
\lim\limits_{t \to \infty }
\left( 
\| \e^{(t+\tau) \, \G} \, \bX - \bX_*\|_1 + \|\bX_* -  \e^{\tau \, \G} \, \bX  \|_1
\right)
\\
%
&=
0
=
\|
\LL(\bX) \left(
\Id - \e^{\tau \, \G}\right) 
\|_1. 
\end{aligned}
\end{equation*}

hence $ \LL(\bX) \in \Kern(\Id - \e^{\tau \, \G}) = \Kern(\G) $ and $\LL(\e^{\tau \, \G} \, \bX) = \LL(\bX) = \e^{\tau \, \G} \,\LL( \bX) $. 

The fact that $ \LL $ is a projection, can be seen as follows:

\begin{equation*}
\begin{aligned} 
\LL^2(\bX)
&=
\lim\limits_{t \to \infty} 
\underbrace{
\e^{t \, \G } \LL
}_{ 
\LL
}
\, (\bX)
= 
\LL(\bX). 
\end{aligned}
\end{equation*}
\end{proof}

Just as we have seen in the discrete-time case, that the two properties of irreducibility and positive recurrence make sure that all states are mean ergodic, they guarantee in the continuous-time case that the time limit exists for all initial states.

\begin{thm}[The time limit for countable, infinite dimensional master equations exists for all networks which are irreducible and positive recurrent] \label{Thm_TheTimeLimitForCountableInfiniteDimensionalMasterEquationsExistsForAllNetworksWhichSAreIrreducibleAndPositiveRecurrent} $ $ \\
If the master equation is both irreducible and positive recurrent, then the limit $\lim\limits_{t \to \infty} \e^{t \, \G} \bX_0 $ exists for all $\bX_0 \in l^1(\Omega) $, that is $\ltb = l^1(\Omega )$. 
\end{thm}

\begin{proof}

For every $ t>0$, the solution operator $Q_t = \e^{t \, \G } $ is a stochastic matrix, which is both irreducible and positive recurrent, if and only if the master equation is both irreducible and positive recurrent.

After applying the theorems \ref{Thm_ExplicitFormOfTheMeanErgodicSubspaceInl1} and \ref{Thm_Irreducible_PositiveRecurrent_l1IsMeanErgodic}, we get: 

\begin{equation*}
\begin{aligned} 
l^1(\Omega)
&\xlongequal[]{\text{thm } \ref{Thm_Irreducible_PositiveRecurrent_l1IsMeanErgodic}}
\underbrace{\Kern(\Id - Q_t)}_{\Kern(\G)} \oplus 
\overline{ \Image(\Id - Q_t)}
\xlongequal[]{\text{thm } \ref{Thm_ExplicitFormOfTheMeanErgodicSubspaceInl1}}
l^1_{m.e.}(\Omega, Q_t)
\end{aligned}
\end{equation*}

Since $\ltb $ is closed in $l^1(\Omega) $ and clearly contains the set of fixed points $\Kern(\Id - Q_t) = \Kern(\G) $, it suffices to show convergence on $ \Image(\Id - Q_t) $.




Now fix $\tau > 0 $ and write $ t = n \, \tau + \tilde{\tau} $. This helps us to make the following estimation: 

\begin{equation*}
\begin{aligned}    
\|
\underbrace{
Q_t \, \left( \Id - Q_\tau \right)
}_{
Q_{t} - Q_{t + \tau}
}
\bX 
\|
&=
\|\left(
Q_{n \, \tau + \tilde{\tau}}
-
Q_{(n+1) \, \tau + \tilde{\tau}}
\right)
\bX 
\|
\\ &=
\|
\left( Q_{\tau}^{n} - Q_{\tau}^{n+1} \right) \, Q_{\tilde{\tau}} \, 
\bX 
\|
\leq 
\|
Q_{\tau}^{n} - Q_{\tau}^{n+1}
\| \, \| 
Q_{\tilde{\tau}} \, \bX 
\|
\end{aligned}
\end{equation*}

In order to show that $ \|Q_{\tau}^{n} - Q_{\tau}^{n+1} \| $ converges  for $ n \to \infty $ (and thus for $t \to \infty $), it suffices - by lemma \ref{Lemma_Katznelson_Tzafriri_1984} - to show that the spectrum of $ Q_{\tau} $ intersects with the unit circle at the point $1$, that is $\sigma(Q_{\tau}) \cap \{ z \in \C \,:\, |z| = 1\}  \subseteq \{1\} $. This can be done by invoking a version of the spectral mapping theorem (compare \ref{Lemma_VersionOfTheSpectralMappingThm}), the fact that the boundary spectrum of the generator of the master equation contains only its spectral bound (lemma \ref{Lemma_BoundarySpectrumContainsOnlyTheSpectralBound}), as well as the fact that this spectral bound vanishes (see lemma \ref{Lemma_GrowthBoundAndSpectralBoundOfPositiveSemigroups}):


\begin{equation*}
\begin{aligned}    
\sigma(\e^{t \, \G}) &\cap \{ z \in \C \,:\, |z| = 1\} 
=
\underbrace{
\left(
\sigma(\e^{t \, \G}) \backslash \{0\}
\right)
}_{
\e^{t \, \sigma(\G)}
}
\cap \, \{ z \in \C \,:\, |z| = 1\} 
\xlongequal{\text{lemma } \ref{Lemma_VersionOfTheSpectralMappingThm} }
\e^{t \, \sigma(\G)} \cap \{ z \in \C \,:\, |z| = 1\} 
\\ &\xlongequal{\text{lemma \ref{Lemma_SpectralMappingAppliedToBoundaryOfTheUnitCircle}}}
\e^{t \, \sigma(\G) \cap i \, \R}
\xlongequal{i \, t \, \R = i \, \R }
\e^{ t \, (\sigma(\G) \cap i \, \R )} 
\xlongequal[s(\G) = 0]{\text{lemma } \ref{Lemma_GrowthBoundAndSpectralBoundOfPositiveSemigroups}}
\e^{ t \, 
\overbrace{
\bigl(\sigma(\G) \cap \{z \in \C \,:\, \Re[z]=s(\G)
\} \bigr)
}^{
\sigma_b(\G)
}} 
\\ &\xlongequal{}
\e^{ t \, \overbrace{\sigma_b(\G)}^{\{0\} } }
\xlongequal{\text{lemma } \ref{Lemma_BoundarySpectrumContainsOnlyTheSpectralBound}}
\e^{\{0\} }
=
\{1\}. 
\end{aligned}
\end{equation*}

This concludes the proof.



\end{proof}

\newpage

\section{The thermodynamic limit of the stationary solutions of the master equation}\label{Chapter_TheThermodynamicLimitOfTheStationarySolutionsOfTheMasterEquations}

\subsection{The (uniform) convergence of the time limit with respect to the system size}

The solution of the master equation of a finite system size approaches a steady state, which - for non irreducible networks - can depend on the initial condition (compare \cite{fernengel2022obtaining}). 

Let us consider a generator $\G$ of a countable, infinite dimensional system, fix an arbitrary initial state $\bP_0 \in \ProbStates $ and a finite set $F \in \Fin(\Omega)$.

In the following, we estimate the difference between the solution of the master equation $ \bP^F(t \,|\, \bPF_0) $ and the steady state $ \bP^F_\infty(\bPF_0) $. When keeping in mind, that the stationary state is an eigenvector of the generator to the eigenvalue $\lambda=0$, we get:

\begin{equation*}
\begin{aligned}
\left\|
\underbrace{
\bP^{F}(t \,|\, \bPF_0)
}_{
\e^{t \, \GF} \bPF_0
} 
-
\bPF_\infty \left(\bPF_0 \right)
\right\|_1    
&=
\sum\limits_{\lambda \in \sigma\left(\GF\right) \backslash \{0\} } \, 
\underbrace{
\e^{t \, \lambda}
}_{
\leq \, \e^{-t \, \Sgap}
}
\,
\underbrace{
\left(
\sum\limits_{d=1}^{g_{\lambda}(\GF)}
\;\; 
\sum\limits_{r=1}^{j_{\lambda, d}(\GF)}
 \mu_{\lambda, d, r} \, (\GF) \, 
\sum\limits_{k=0}^{r-1}
\frac{t^k}{k!}
\right)
}_{
\leq \, c \, t^{j(\lambda)}
} \\ 
&\leq
C \, \e^{-t \, \Sgap } \, t^{\JGF}. 
\end{aligned}
\end{equation*}

where $\sigma $ denotes the spectrum of a matrix, $ g_{\lambda}(\GF) $ the geometric multiplicity, $ j_{\lambda, d}(\GF) \in \N_{\leq \, g_{\lambda}(\GF)}$ the size of Jordan block number $ d $ to the eigenvalue $ \lambda $, and $\mu $ arbitrary coefficients. 

The expression in the parenthesis $ \bigl( \cdots \bigr) $ can be estimated as some constants times $t$ to the power of the size of the largest Jordan block $j(\lambda) $, whereas the exponential term $ \e^{t \, \lambda} $ is less or equal to $ \e^{-t \, S_{g}(\GF)} $, where $ S_{g}(\GF) := \text{dist}\bigl( \Re(\sigma(\GF) \backslash \{0\}), 0 \bigr) := \max\{\Re(\lambda) \,:\, \lambda \in \sigma(\GF)\backslash \{0\} \}   $ is called the \emph{spectral gap}, the minimal distance from the eigenvalues of the generator and the imaginary axis.

The sum over the eigenvalue of the generator can again be estimated to a constant times the time $t$ to the power of the \emph{largest} Jordan block $J(\GF)$ times the exponential function of $t$ times the spectral gap.

While this expression does converge to zero for $t \to \infty $, it does not do so uniformly in the system size $|F|$. 

A sufficient conditions for uniform convergence in the system size would be the following two properties:  

\begin{itemize}
\item[(i)] The size of the Jordan blocks is bounded from above, that is $ \JGF \leq N $ for some $N \in \N $ AND
\item[(ii)] The spectral gap is bounded from below and does therefore not vanish in the thermodynamic limit, that is $\Sgap \geq \epsilon_0$ for some $\epsilon_0 >0$. 
\end{itemize}

With these two requirements met, it is straightforward to see that we can achieve uniform convergence, since for fixed $N \in \N$ and $\epsilon>0$, we can choose a time $T_\epsilon $ such that from thereon out ($t \geq T_\epsilon $) we have 

\begin{equation*}
\begin{aligned}
\left\|
\bP^{F}(t \,|\, \bPF_0)
-
\bP^{F}_\infty \left(\bPF_0 \right)
\right\|_1    
&\leq
C \, \e^{-t \, \Sgap } \, t^{\JGF}
&\leq
C \, \e^{-t \, \epsilon_0 } \, t^{N}
<
\epsilon. 
\end{aligned}
\end{equation*}

But since we lack theorems which guarantee these requirements, we can only prove the existence of the iterated limit 

\begin{equation*}
\begin{aligned}
\lim\limits_{{t \to \infty }\atop {|F| \to \infty } } \,
\bPFt 
\end{aligned}
\end{equation*}

in very specific examples.





\subsection{Sufficient conditions for the convergence of stationary solutions for finite systems in the thermodynamic limit} \label{Section_SufficientConditionsForTheConvergenceInTheThermodynamicLimit}

We start by showing that if the thermodynamic limit for stationary solutions of finite systems exist, then this limit vector is a \emph{stationary} solution for the countable, infinite dimensional system.

\begin{thm}[thermodynamic limit of stationary states]\label{IfThermodynamicLimitOfStationaryStatesExists} $ $ \\ 
Let $ \G $ be the generator of an irreducible master equation with $ \| \G \|^\text{(op)}_{1, 1} < \infty $ and let the limiting states  $\left( \bP_*^F \right)_{F \in \Fin(\Omega)} := \left( \bPInftyF \right)_{F \in \Fin(\Omega)} \in  \ProbStates \cap \Kern(\GF) $ of every finite subsystem $F \in \Fin(\Omega)$ converge for all initial state $\bP_0 \in \ProbStates$ to some probability vector $\bP_*$. Then the master equation is positive recurrent with the stationary and limiting solution $ \bP_\infty = \bP_* = \lim\limits_{|F| \to \infty} \bP_*^F  $. 
\end{thm}

\begin{proof}

First, we note that $\bP_*$ has non-negative entries and is normalized, since both the dual unit vector $\bE_{i}'$, as well as the norm $\| \cdot \|_1$ are continuous: 

\begin{equation*}
\begin{aligned}
\bE_i' \left(\lim\limits_{|F| \to \infty} \bP_*^{F} \right)
&\xlongequal{
\bE_i' \text{   continuous}
}
\lim\limits_{|F| \to \infty}
\underbrace{
\bE_i' \left(
\bP_*^{F}
\right)
}_{
\geq \, 0
}
\geq
0 \text{             and }
\\
\left\|
\lim\limits_{|F| \to \infty} \bP_*^{F} 
\right \|_1
&\xlongequal{\| \cdot \|_1 \text{ cont}}
\lim\limits_{|F| \to \infty} \, 
\underbrace{
\| \bP_*^{F} \|_1
}_{1}
=
1. 
\end{aligned}
\end{equation*}

To see, that $\lim\limits_{|F| \to \infty} \bP_*^{F} $ indeed lies in the kernel of $\G$, we consider

\begin{equation}\label{Eq_Estimating_G_bPInftyA}
\begin{aligned}
\Bigl\| 
\underbrace{
    \G \, \lim\limits_{|F| \to \infty} \bP_*^{F}
}_{
     \lim\limits_{|F| \to \infty}  \G \, \bP_*^{F}
}
\Bigr \|_1
&\xlongequal{\G \text{ continuous}}
\underbrace{
 \left\| 
 \lim\limits_{|F| \to \infty}  \G \,\bP_*^{F}
 \right \|_1
}_{
    \lim\limits_{|F| \to \infty}
    \left\| 
    \G \,\bP_*^{F}
    \right \|_1
}
\xlongequal{\| \cdot \|_1 \text{  continuous}} \\ &=
\lim\limits_{|F| \to \infty}
\left\| 
\underbrace{
    \G \,\bP_*^{F}
    -
    \overbrace{
        \GF \,\bP_*^{F}
    }^{
        \bzero
    }
}_{
    (\G - \GF) \, \bP_*^{F}
}
\right \|_1
=
\lim\limits_{|F| \to \infty}
\underbrace{
    \left\|
    (\G - \GF) \, \bP_*^{F}
    \right\|_1
}_{
    \leq \, \left\|
    (\G - \GF) 
    \right\|_1
    \, \left \| \bP_*^{F} \right \|_1
}
 \\ &\leq 
\lim\limits_{|F| \to \infty}
\left\|
(\G - \GF) 
\right\|_1
\, 
\underbrace{
\left \| \bP_*^{F} \right \|_1
}_{1}
\xlongequal[\ref{Lemma_ApproximatingTheFullGeneratorWithGeneratorsOfFiniteSubnetworks}]{\text{lemma} }
0. 
\end{aligned}
\end{equation}

Since the network was assumed to be irreducible, by lemma \ref{Lemma_StrictPositivityForStationarySolutionsOfStronglyConnectedNetworks}, the stationary solution must be unique, hence the limit $ \lim\limits_{|F| \to \infty} \bP_*^{F} $ does not depend on the initial state $\bP_0$. 

\end{proof}

\begin{rem}
Let $ (F_n)_{n \in \N } \subseteq \Fin(\Omega) $ be an increasing sequence of finite sets (that is $F_{n+1} \subseteq F_n $ for all $n \in \N$ ) such that $\Omega = \bigcup \limits_{n \in \N} F_n$. It is sufficient to check, whether $\left(\bP_\infty^{F_n}(\bP_0^{F_n}) \right)_{n \in \N} $ converges, in order to make the same conclusions as in theorem \ref{IfThermodynamicLimitOfStationaryStatesExists}.  
\end{rem}

\begin{lemma}
\label{ASufficientConditionForTheApproximationOfTheLongTermBehaviorOfProbabilityStates}
[A Sufficient condition for the approximation of the long-term behavior of probability states]
\end{lemma}

Let $\System$ be an irreducible network with only bidirectional links (that is $\g_{\alpha \to \beta} >0 \iff \g_{\beta \to \alpha}> 0 $),  $\bP_0, \bP_* \in \ProbStates $ be probability vectors and $M_0 \in \Fin(\Omega) \backslash (\Null_{\bP_0} \cup \Null_{\bP_*}) $ a strongly connected, finite subnetwork, such that for all strongly connected, finite subnetworks containing $M_0$ ($M \in \Fin(\Omega), \, M \supseteq M_0$, $\System_M $ strongly connected) we have 

\begin{equation}\label{Eq_StationarySolutionLooksLikeCutOff}
\begin{aligned}
\bP^M_\infty(\bPM_0) = \bPM_*, 
\end{aligned}
\end{equation}

then $ \bigl(\bP^M_\infty \bigr)(\bPM_0) \xlongrightarrow{ M \in \Fin(\Omega) \backslash \NPzero } \bP_* $.

\begin{proof}

First, we note that if equation \eqref{Eq_StationarySolutionLooksLikeCutOff} were to hold for all finite sets (instead of only on those finite sets, whose associated network is strongly connected), then the statement would follow from lemma \ref{lemma_ApproximatingProbabilityStates}. What real message from lemma \ref{ASufficientConditionForTheApproximationOfTheLongTermBehaviorOfProbabilityStates} is therefore, that it is enough to check the fulfillment of equation \eqref{Eq_StationarySolutionLooksLikeCutOff} for strongly connected subnetworks. 

Fix $ \epsilon \in \bigl(0, \frac{1}{2} \bigr) $.

\begin{itemize}

\item[(i)] choose $ F_\epsilon^{(1)} \in  \Fin(\Omega)  $ such that for all $ F  \in \Fin(\Omega) \backslash \NPzero, \, F \supseteq F_\epsilon^{(1)} $ we have $ \bigl(\bP_0 \bigr) (F) > 1- \epsilon $ and $ \bigl( \bP_* \bigr) \bigl( F \bigr) >0 $ 

\item[(ii)] 
choose $  F_\epsilon^{(2)} \in  \Fin(\Omega)  $ such that for all $ A  \in \Fin(\Omega) \backslash \NPzero, \, A \supseteq F_\epsilon^{(2)} $ we have $ \| \bP_* - \bPF_* \|_1 < \epsilon $ (this is possible due to lemma \ref{lemma_ApproximatingProbabilityStates})

\item[(iii)] choose $ F_\epsilon :=  F_\epsilon^{(1)} \cup F_\epsilon^{(2)} $ and fix $ F \in \Fin(\Omega) $, $F \supseteq F_\epsilon $. Since there are only bidirectional links, the set $F$ is a disjoint union of minimal absorbing set, that is  $A = \bigcup\limits_{i=0}^{n} M_i $ with $A_i \in \MM_{\System_A} $. We call $ M_0 $ the minimal absorbing set containing $ F_\epsilon $, that is $ M_0 \supseteq F_\epsilon $.

We also have 
\item[(iv)] $ \bP^F_\infty(\bPF_0 \in M_0) = \bP_\infty^{[M_0]}(\bP_0^{[M_0]})$ and 

\item[(v)] $ \bigl(\bPF_0\bigr) \, \bigl( \Omega\backslash M_0 \bigr) = \frac{1-(\bP_0)\,\bigl( M_0 \bigr)}{(\bP_0) \, \bigl( A \bigr)}  \leq \frac{\epsilon}{1-\epsilon} \overset{ \epsilon < 1/2}{<} 2 \, \epsilon $.

\end{itemize}

Then we get:

\begin{equation*}
\begin{aligned}
&\Bigl\| 
\bP_* - 
\underbrace{
\bPF_\infty (\bPF_0)
}_{
\sum\limits_{M \in \MM_{\System_A}} \bigl(\bPF_0(M)\bigr) \bPF_\infty(\bPF_0 \in M)
}
\Bigr\|_1 
\overset{ M_0 \in \MM_{\System_A}}{\leq}
\\ &\leq
\Bigl\| 
\bP_* - 
\overbrace{
\bigl(\bPF_0\bigr) (M_0)
}^{
1-{\color{MyGreen} \bigl(\bPF_0\bigr) \bigl(\Omega \backslash M_0\bigr) } 
} 
\, 
\overbrace{
\bigl( \bPF_\infty \bigr) (\bPF_0 \in M_0) 
}^{
\bP_\infty^{[M_0]}(\bP_0^{[M_0]})
}
\Bigr\|_1
+
\underbrace{
\sum\limits_{M \in \MM_{\System_A} \backslash \{M_0\} } \bigl(\bPF_0\bigr) (M) 
\overbrace{
\| \bPF_\infty(\bPF_0 \in M) \|_1
}^{
1
}
}_{
\color{ElectricPurple} \bigl(\bPF_0\bigr) (\Omega \backslash M_0)
} = 
\\ &\leq
\Bigl\|
\bP_*
-
\underbrace{
\bP_\infty^{[M_0]}(\bP_0^{[M_0]})
}_{
\bigl(
\bP_*
\bigr)^{[M_0]}
}
\Bigr\|_1
+
{\color{MyGreen}  \bigl(\bPF_0\bigr) \bigl(\Omega \backslash M_0 \bigr)  }
\underbrace{
\|\bP_\infty^{[M_0]}(\bP_0^{[M_0]}) \|_1
}_{
1
}
+ 
{  \color{ElectricPurple} \bigl( \bPF_0 \bigr) \bigl(\Omega \backslash M_0\bigr)  } \xlongequal{\text{assumption (iii)}}
\\ &=
\underbrace{
\left\|
\bP_*
-
\bigl(
\bP_*
\bigr)^{[M_0]}
\right\|_1
}_{
\overset{(\text{ii})} < \epsilon
}
+ 2 \, 
\underbrace{
\bigl( \bPF_0\bigr) 
\bigl(
\Omega \backslash M_0
\bigr)  
}_{
\leq \frac{\epsilon}{1-\epsilon} \, \overset{(\text{v})}{<}2\,  \epsilon 
} 
\\ &<
5 \, \epsilon. 
\end{aligned}
\end{equation*}

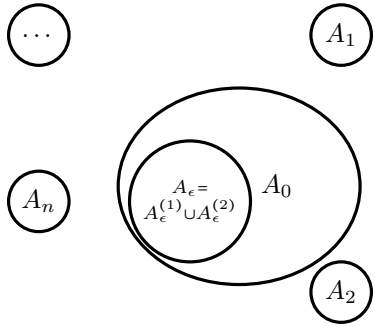
\begin{figure}[H]
\begin{center}
\begin{tikzpicture} 
\draw[very thick] (0.0,0.0) ellipse (0.8 and 0.8);
\draw (0.0, 0.0) node{  $ {A_{\epsilon} = } \atop {A_{\epsilon}^{(1)} \cup A_{\epsilon}^{(2)}}  $ };
\draw[very thick] (0.65, 0.2) ellipse (1.6 and 1.3);
\draw (1.15, 0.2) node{  $A_0 $ };
\draw[very thick] (2.0, 2.2) ellipse (0.4 and 0.4);
\draw (2.0, 2.2) node{  $A_1 $ };
\draw[very thick] (2.0, -1.2) ellipse (0.4 and 0.4);
\draw (2.0, -1.2) node{  $A_2 $ };
\draw[very thick] (-2.0, 2.2) ellipse (0.4 and 0.4);
\draw (-2.0, 2.2) node{  $ \dots $ };
\draw[very thick] (-2.0, 0.0) ellipse (0.4 and 0.4);
\draw (-2.0, 0.0) node{  $ A_n $ };
\end{tikzpicture}
\caption{Illustrating the relationships between the sets $F_\epsilon \subseteq A_0 $,$\dots$, $A_n$ and $A = \bigcup\limits_{i=0}^{n}$.    }
\label{Illustrating_RelationshipsOfSets_A} 
\end{center}
\end{figure}

\end{proof}

\subsection{Explicit expression for the stationary solution in case of detailed balance}\label{Section_ExplicitExpressionForTheStationarySolutionInCaseOfDetailedBalance}

In the following section we want to show, that in case of detailed balance, it is possible to compute the stationary solution $\bP_* $ in terms of transition rates.

\begin{MyDef}[Detailed balance]\label{Def_DetailedBalance} $ $ \\
A probability vector $ \bP_* $ is called \emph{stationary}, if it lies in the kernel of the generator matrix $\G$, that is $\bP_* \in \Kern(\G) $, or component-wise: 

\begin{equation}\label{Eq_StationarySolution}
\begin{aligned}
0
&=
\sum\limits_{j \in \Omega \backslash\{i\} }
\left( P_*^{(j)} \g_{j \to i} - P_*^{(i)} \g_{i \to j} \right)
\end{aligned}
\end{equation}

When in addition to being stationary, each summand $\left( \cdot \right) $ in equation \eqref{Eq_StationarySolution} vanishes, then we say the stationary solution $\bP_*$ exhibits \emph{detailed balance}. 
Since we often encounter so called \emph{`non-normalizable' stationary states} (that is sequences $\bX_* \in (\R_{> \, 0})^{\Omega} \backslash l^{1}(\Omega) $ with $\G \, \bX_* = \bzero$, but $ \|\bX_*\|_1 = \infty$) we relax the definition to some extend, in order to better suit our requirements: 

A master equation is said to satisfy the \emph{generalized detailed balance} condition, if there exists a strictly positive sequence $ \bX_* \in (\R_{> \, 0})^{\Omega} $, such that for all $j \in \Omega\backslash \{i\}$ we have 

\begin{equation} \label{Eq_StationarySolution_DetailedBalance} \tag{gen. det. bal. }
\begin{aligned}
X_*^{(j)} \g_{j \to i}
&=
X_*^{(i)} \g_{i \to j}, 
\end{aligned}
\end{equation}

independent of whether $\bX_*$ is normalizable or not.

\end{MyDef}

\begin{rem}
Even though detailed balance seems to be a property of the \emph{stationary solution $\bP_{*}$ }, equation \eqref{Eq_StationarySolution_DetailedBalance} hints that it is actually a property of the network. And indeed, we will see that it is possible to characterize detailed balance as a network property via the so-called \emph{Kolmogorov criterion} (compare \cite{kelly2011reversibility, tubiblio138082}). 
\end{rem}

\begin{MyDef}[Kolmogorov criterion] \label{Def_KolmogorovCriterion} $ $ \\ 

A network is said to satisfy the \emph{Kolmogorov criterion}, if it is `rotation free', in the sense that a \emph{weight} of a closed walk equals the weight of its reversed walk. To be more precise, let $n \in \N_{\geq \, 2}$ and $\bw:=(\omega_1, \dots, \omega_n) \in \Omega^{n} $ be a path in $\System$. Then the system exhibits detailed balance, if and only if 
\begin{equation}\label{Eq_KolmogorocCriterion}
\begin{aligned}
\g_{\bw_\circlearrowleft}
&:=
\prod\limits_{i=1}^{n-1} \g_{\omega_{i} \to \omega_{i+1}} \cdot \g_{\omega_{n} \to \omega_{1}}
=
\g_{\omega_{1} \to \omega_{n}}\cdot 
\prod\limits_{i=1}^{n-1} \g_{\omega_{i+1} \to \omega_{i}} 
=:
\g_{\bw_\circlearrowright}. 
\end{aligned}
\end{equation}

The network satisfying the Kolmogorov criterion can be interpreted a having no net `circular flows'.

\end{MyDef}

\begin{lemma}[Consequences of the Kolmogorov criterion] \label{Lemma_ConsequencesOfTheKolmogorovCriterion} $ $ \\ 

Consider an irreducible network satisfying the Kolmogorov criterion, two distinct states $\alpha, \beta \in \Omega, \alpha \neq \beta $ and a path $\bw =(\alpha= w_{1}, \dots, w_{n} = \beta) \in \Omega^{n} $. Then the product of the weight of a path, divided by the weight of its inverse path (namely $ \frac{\g_{w}(\alpha \rightsquigarrow \beta)}{\g_{w}(\alpha \leftlsquigarrow \beta)} := \frac{\g_{w_{1} \to w_{2}} \cdot \, \dotsc \, \cdot \g_{w_{n-1} \to w_{n}}}{\g_{w_{1} \gets w_{2}} \cdot \,  \dotsc \, \cdot \g_{w_{n-1} \gets w_{n}}} $) does not depend on the choice of the path $\bw$, but only on its start- and end-point $\alpha$ and $\beta$.

\end{lemma}

\begin{proof}
Let $ \bw_1^{\to} (\alpha \rightsquigarrow \beta) $ and $ \bw_2^{\to} (\alpha \rightsquigarrow \beta) $ be two different path from state $\alpha$ to state $\beta$. 

By Kolmogorov's criterion, we have: 

\begin{equation} 
\begin{aligned}
\g_{\bw_1^{\to}} (\alpha \rightsquigarrow \beta )
\, \cdot \, 
\g_{\bw_2^{\gets}} (\beta \rightsquigarrow \alpha )
&\xlongequal{\text{Kolmogorov}}
\g_{\bw_2^{\to}} (\alpha \rightsquigarrow \beta )
\, \cdot \, 
\g_{\bw_1^{\gets}} (\beta \rightsquigarrow \alpha )
\\ \Longrightarrow 
\frac{
\g_{\bw_1^{\to}} (\alpha \rightsquigarrow \beta )
}{
\g_{\bw_1^{\gets}} (\alpha \leftlsquigarrow \beta )
}
&=
\frac{
\g_{\bw_2^{\to}} (\alpha \rightsquigarrow \beta )
}{
\g_{\bw_2^{\gets}} (\alpha \leftlsquigarrow \beta )
}
\end{aligned}
\end{equation}

Figure \ref{Figure_IllustratingKolmogorovsCriterion} is meant to illustrate this result.

\begin{figure}[H]
\begin{center}
\begin{tikzpicture}[scale=2]
\centering
\node[State](-30)   at (-3,0) [circle,draw] { } ;
\node[State](-20)   at (-2,0) [circle,draw] { } ;
\node[State](-10)   at (-1,0) [circle,draw] { } ;
\node[State](00)    at ( 0,0) [circle,draw]  { } ;
\node[State](10)    at ( 1,0) [circle,draw]  { } ;
\node[State, MyGreen](20)    at ( 2,0) [circle,draw]  { $\boldsymbol{\beta}$ } ;
\node[State](30)    at ( 3,0) [circle,draw]  { } ;
\node[State](-31)  at (-3,1) [circle,draw]  { } ;
\node[State](-21)  at (-2,1) [circle,draw]  { } ;
\node[State](-11)  at (-1,1) [circle,draw]  { } ;
\node[State](01)   at (-0,1) [circle,draw] { } ;
\node[State](11)   at ( 1,1) [circle,draw] { } ;
\node[State](21)   at ( 2,1) [circle,draw] { } ;
\node[State](31)   at ( 3,1) [circle,draw] {  } ;
\node[State](-3-1)  at (-3,-1) [circle,draw] { } ;
\node[State, MyGreen](-2-1)  at (-2,-1) [circle,draw] {$\boldsymbol{\alpha}$ } ;
\node[State](-1-1)  at (-1,-1) [circle,draw] { } ;
\node[State](0-1)  at (0, -1) [circle,draw]  {} ;
\node[State](1-1)  at ( 1,-1) [circle,draw]  {} ;
\node[State](2-1)  at ( 2,-1) [circle,draw]  {} ;
\node[State](3-1)  at ( 3,-1) [circle,draw]  {} ;
\path[Link, bend left, MyBlue, dotted] (-3-1) edge node[above] {} (-2-1);
\path[Link, bend left, MyRed] (-2-1) edge node[above] {} (-1-1);
\path[Link, bend left, MyRed] (-1-1) edge node[above] {} (0-1);
\path[Link, bend left, MyRed] ( 0-1) edge node[above] {} (1-1);
\path[Link, bend left, MyRed] ( 1-1) edge node[above] {} (2-1);
\path[Link, bend left, MyRed] ( 2-1) edge node[above] {} (3-1);
\path[Link, bend left] (-30) edge node[above] {} (-20);
\path[Link, bend left] (-20) edge node[above] {} (-10);
\path[Link, bend left] (-10) edge node[above] {} (00);
\path[Link, bend left] ( 00) edge node[above] {} (10);
\path[Link, bend left] ( 10) edge node[above] {} (20);
\path[Link, bend left, MyRed, dotted] ( 20) edge node[above] {} (30);
\path[Link, bend left, MyBlue] (-31) edge node[above] {} (-21);
\path[Link, bend left, MyBlue] (-21) edge node[above] {} (-11);
\path[Link, bend left, MyBlue] (-11) edge node[above] {} (01);
\path[Link, bend left, MyBlue] ( 01) edge node[above] {} (11);
\path[Link, bend left, MyBlue] ( 11) edge node[above] {} (21);
\path[Link, bend left] ( 21) edge node[above] {} (31);
\path[Link, bend left, MyRed, dotted] (3-1) edge node[above] {} (2-1);
\path[Link, bend left, MyRed, dotted] (2-1) edge node[above] {} (1-1);
\path[Link, bend left, MyRed, dotted] (1-1) edge node[above] {} (0-1);
\path[Link, bend left, MyRed, dotted] (0-1) edge node[above] {} (-1-1);
\path[Link, bend left, MyRed, dotted] (-1-1) edge node[above] {} (-2-1);
\path[Link, bend left, MyBlue] (-2-1) edge node[above] {} (-3-1);
\path[Link, bend left, MyRed] (30) edge node[above] {} (20);
\path[Link, bend left] (20) edge node[above] {} (10);
\path[Link, bend left] (10) edge node[above] {} (00);
\path[Link, bend left] (00) edge node[above] {} (-10);
\path[Link, bend left] (-10) edge node[above] {} (-20);
\path[Link, bend left] (-20) edge node[above] {} (-30);
\path[Link, bend left] (31) edge node[above] {} (21);
\path[Link, bend left, MyBlue, dotted] (21) edge node[above] {} (11);
\path[Link, bend left, MyBlue, dotted] (11) edge node[above] {} (01);
\path[Link, bend left, MyBlue, dotted] (01) edge node[above] {} (-11);
\path[Link, bend left, MyBlue, dotted] (-11) edge node[above] {} (-21);
\path[Link, bend left, MyBlue, dotted] (-21) edge node[above] {} (-31);
\path[Link, bend left, MyBlue] (-3-1) edge node[above] {} (-30);
\path[Link, bend left, MyBlue] (-30) edge node[above] {} (-31);
\path[Link, bend left] (-2-1) edge node[above] {} (-20);
\path[Link, bend left] (-20) edge node[above] {} (-21);
\path[Link, bend left] (-1-1) edge node[above] {} (-10);
\path[Link, bend left] (-10) edge node[above] {} (-11);
\path[Link, bend left] (0-1) edge node[above] {} (00);
\path[Link, bend left] (00) edge node[above] {} (01);
\path[Link, bend left] (1-1) edge node[above] {} (10);
\path[Link, bend left] (10) edge node[above] {} (11);
\path[Link, bend left] (2-1) edge node[above] {} (20);
\path[Link, bend left, MyBlue, dotted] (20) edge node[above] {} (21);
\path[Link, bend left, MyRed] (3-1) edge node[above] {} (30);
\path[Link, bend left] (30) edge node[above] {} (31);
\path[Link, bend left, MyBlue, dotted] (-31) edge node[above] {} (-30);
\path[Link, bend left, MyBlue, dotted] (-30) edge node[above] {} (-3-1);
\path[Link, bend left] (-21) edge node[above] {} (-20);
\path[Link, bend left] (-20) edge node[above] {} (-2-1);
\path[Link, bend left] (-11) edge node[above] {} (-10);
\path[Link, bend left] (-10) edge node[above] {} (-1-1);
\path[Link, bend left] (01) edge node[above] {} (00);
\path[Link, bend left] (00) edge node[above] {} (0-1);
\path[Link, bend left] (11) edge node[above] {} (10);
\path[Link, bend left] (10) edge node[above] {} (1-1);
\path[Link, bend left, MyBlue] (21) edge node[above] {} (20);
\path[Link, bend left] (20) edge node[above] {} (2-1);
\path[Link, bend left] (31) edge node[above] {} (30);
\path[Link, bend left, MyRed, dotted] (30) edge node[above] {} (3-1);
\end{tikzpicture}
\[
\frac{
\g_{{\color{MyRed}\bw_1}^{\to}} (\alpha \rightsquigarrow \beta )
}{
\g_{{\color{MyRed}\bw_1}^{\gets}} (\alpha \leftlsquigarrow \beta )
}
=
\frac{ 
    \begin{tikzpicture}[scale=1]
    \node[State](alpha)  at (0,0) [circle,draw] { } ;
    \node[State](beta)   at (1,0) [circle,draw] { } ;
    \path[Link,MyRed] (alpha) edge node[above] {} (beta);
    \end{tikzpicture}
}{
    \begin{tikzpicture}[scale=1]
    \node[State](alpha)  at (0,0) [circle,draw] { } ;
    \node[State](beta)   at (1,0) [circle,draw] { } ;
    \path[Link,MyRed, dotted] (beta) edge node[above] {} (alpha);
    \end{tikzpicture}
}
\xlongequal[\text{criterion}]{\text{Kolmogorov}}
\frac{ 
    \begin{tikzpicture}[scale=1]
    \node[State](alpha)  at (0,0) [circle,draw] { } ;
    \node[State](beta)   at (1,0) [circle,draw] { } ;
    \path[Link,MyBlue] (alpha) edge node[above] {} (beta);
    \end{tikzpicture}
}{  
    \begin{tikzpicture}[scale=1]
    \node[State](alpha)  at (0,0) [circle,draw] { } ;
    \node[State](beta)   at (1,0) [circle,draw] { } ;
    \path[Link,MyBlue, dotted] (beta) edge node[above] {} (alpha);
    \end{tikzpicture}
}
=
\frac{
\g_{{\color{MyBlue}\bw_2}^{\to}} (\alpha \rightsquigarrow \beta )
}{
\g_{{\color{MyBlue}\bw_2}^{\gets}} (\alpha \leftlsquigarrow \beta )
}
\]
\caption{Illustrating Kolmogorov criterion: The term $ \frac{
\g_{\bw^{\to}} (\alpha \rightsquigarrow \beta )
}{
\g_{\bw^{\gets}} (\alpha \leftlsquigarrow \beta )
} $ is independent of the specific paths ${\color{MyRed} \bw_1   }$ or ${\color{MyBlue} \bw_2   }$. 
}
\label{Figure_IllustratingKolmogorovsCriterion} 
\end{center}
\end{figure}

Hence, the term $ \frac{\g_{\bw^{\to}} (\alpha \rightsquigarrow \beta )}{\g_{\bw^{\gets}} (\alpha \leftlsquigarrow \beta )} $ is independent of the path $\bw$, which enables us to write $ \frac{\g_{\omega_0 \rightsquigarrow \omega }}{\g_{\omega \leftlsquigarrow \omega_0 } } := \frac{\g_{\bw^{\to}} (\omega_{0} \rightsquigarrow \omega )}{\g_{\bw^{\gets}} (\omega_{0} \leftlsquigarrow \omega )}$ for any path $\omega_0 \rightsquigarrow \omega  := (\omega_0 \to \dots \to \omega)$, where we set $\g_{\omega_0 \rightsquigarrow \omega_0 }:=1$ .

\end{proof}


\begin{lemma}[Generalized detailed balance equals Kolmogorov criterion] \label{Lemma_GeneralizedDetailedBalanceEqualsKolmogorovCriterion} $ $ \\ 

A strongly connected network satisfies the generalized detailed balance condition \eqref{Eq_StationarySolution_DetailedBalance}, if and only if it satisfies the Kolmogorov criterion. The strictly positive sequence $\bX_{*} \in (\R_{> \, 0})^{\Omega } $ is then given by 

\begin{equation}
\begin{aligned}
\bX_{*}
&:=
\left(
\frac{\g_{\omega_{0} \rightsquigarrow \omega}}{\g_{\omega_{0} \leftlsquigarrow \omega}}
\right)_{\omega \in \Omega}, 
\end{aligned}
\end{equation}

which is well defined by lemma \ref{Lemma_ConsequencesOfTheKolmogorovCriterion}.

\end{lemma}


\begin{proof} $ $ \\ 
\begin{itemize}
\item["$\Leftarrow$"]

If a network satisfies the Kolmogorov criterion, we have: 

\begin{equation*}
\begin{aligned}
\underbrace{
    X_{*}^{(\alpha)}
}_{
    \frac{\g_{\omega_{0} \rightsquigarrow \alpha}}{\g_{\omega_{0} \leftlsquigarrow \alpha}}
}
\, \g_{\alpha \to  \beta} 
&=
\underbrace{
    \frac{\g_{\omega_{0} \rightsquigarrow \alpha}}{\g_{\omega_{0} \leftlsquigarrow \alpha}}
    \, \cdot \, 
    \frac{\g_{\alpha \to \beta}}{\g_{\alpha \gets \beta}} 
}_{
        \frac{\g_{\omega_{0} \rightsquigarrow \beta}}{\g_{\omega_{0} \leftlsquigarrow \beta}}
}
\, 
\underbrace{
    \g_{\alpha \gets \beta}
}_{
     \g_{\beta \to \alpha}
}
=
\underbrace{
    \frac{\g_{\omega_{0} \rightsquigarrow \beta}}{\g_{\omega_{0} \leftlsquigarrow \beta}}
}_{
    X_{*}^{(\beta)}
} \, \cdot \, \g_{\beta \to \alpha}
=
X_{*}^{(\beta)} \, \cdot \, \g_{\beta \to \alpha}, 
\end{aligned}
\end{equation*}

which means that the network satisfies the generalized detailed balance condition. 

\item["$\Rightarrow$"]

Let $\bw = (w_{1}, \dots, w_{n}) \in \Omega^{n} $ be a walk in $ \System $. If the system satisfies the detailed balance condition, we have a strictly positive sequence $ \bX_{*} \in (\R_{> \, 0})^{\Omega} $ such that 

\begin{equation}
\begin{aligned}
X_{*}^{(w_{1})} \, \g_{w_{1} \to w_{2}} 
&=
\g_{w_{1} \gets w_{2}} \,  X_{*}^{(w_{2})} 
\\
\vdots \hspace*{10mm} &{} \hspace*{10mm}  \vdots 
\\
X_{*}^{(w_{n-1})} \, \g_{w_{n-1} \to w_{n}} 
&=
\g_{w_{n-1} \gets w_{n}} \,  X_{*}^{(w_{n})} 
\\
X_{*}^{(w_{n})} \, \g_{w_{n} \to w_{1}} 
&=
\g_{w_{n} \gets w_{1}} \,  X_{*}^{(w_{1})}  \Bigl| \text{ multiply all equations }
\\
\Longrightarrow \prod\limits_{i=1}^{n} X_{*}^{(w_i)} \, 
\underbrace{
    \g_{w_{1} \to w_{2}} \cdot \dotsc \cdot \g_{w_{n} \to w_{1}} 
}_{
    \g_{\bw_\circlearrowright}
}
&=
\underbrace{
    \g_{w_{1} \gets w_{2}} \cdot \dotsc \cdot \g_{w_{n} \gets w_{1}} 
}_{
    \g_{\bw_\circlearrowleft}
} \, 
\prod\limits_{i=1}^{n} X_{*}^{(w_i)}, 
\end{aligned}
\end{equation}

which means that the network satisfies the Kolmogorov criterion. 
\end{itemize}
\end{proof}
\begin{lemma}[Candidate for stationary solution of irreducible networks with detailed balance]\label{Thm_CandidateForStationarySolutionOfIrreducibleNetworksWithDetailedBalance} $ $ \\

If an irreducible network with generator $ \| \G \|_{1,1}  < \infty $ satisfies the generalized condition of detailed balance \eqref{Eq_StationarySolution_DetailedBalance}, then the sequence $ \bX_*:= \left(\frac{\g_{\omega_0 \rightsquigarrow \omega }}{\g_{\omega \leftlsquigarrow \omega_0 } } \right)_{\omega \in \Omega } \in (\R_{> \, 0})^{\Omega} $ - which is well defined by lemma \ref{Lemma_ConsequencesOfTheKolmogorovCriterion} -  is normalizable ($\bX_{*} \in l^{1} (\Omega)$, that is $ \| \bX_{*} \|_{1} < \infty $ ), if and only if the network is positive recurrent. 

This means that for an irreducible network satisfying the generalized condition for detailed balance \eqref{Eq_StationarySolution_DetailedBalance} the vector $ \frac{\bX_{*}}{\| \bX_{*} \|_{1}} $ is the only `\emph{candidate}' for a stationary solution, in the sense that 

\begin{equation*}
\begin{aligned}
\| \bX_{*} \|_{1} &< \infty \Longrightarrow \frac{\bX_{*}}{\| \bX_{*} \|_{1}} \text{  is the stationary solution of $\G$ }
\\ 
\| \bX_{*} \|_{1} &= \infty \Longrightarrow  \text{  no stationary solution, } \Kern(\G) = \{\bzero\}
\end{aligned}
\end{equation*}

\end{lemma}

\begin{proof}

\begin{itemize}

\item["$\Rightarrow$"] $ $ \\ 

Since $ \bX_{*} \in \left(\R_{>0}\right)^{\Omega} $ and $ \| \bX_{*} \|_{1} < \infty $, the vector $ \frac{\bX_{*}}{\| \bX_{*} \|_{1}} $ is already a probability vector and we only need to show that it is stationary (that is $ \G \, \bX_{*} = \bzero$). 

Let $\Null \Null(i) := \{j \in \Omega \,:\, \g_{i \to j}>0 \} $ denote the set of \emph{nearest neighbors} of state $i \in \Omega $. Since $\bX_{*} \in l^{1} (\Omega)$, $\G \, \bX_{*} $ is well defined and we can compute the result component-wise: 

\begin{equation}\label{Eq_KolmogorocCriterion}
\begin{aligned}
\left( 
\G \bX_*
\right)^{(i) }
&=
\sum\limits_{j \in \Omega\backslash\{i\} }
\underbrace{
    \hspace*{2mm} X_*^{(j)} \hspace*{2mm} 
}_{
     \frac{\g_{\omega_0 \rightsquigarrow j }}{\g_{ \omega_0 \leftlsquigarrow j } }
}
 \g_{j \to i}
-
X_*^{(i)} \g_{i \to j}
\\ &=
\Bigl[
\sum\limits_{j \in \Null \Null(i) }
\underbrace{
    \left(
    \frac{\g_{\omega_0 \rightsquigarrow j }}{\g_{ \omega_0 \leftlsquigarrow j } }
    \cdot 
    \frac{\g_{j \to i}}{\g_{j \gets i}}
    \right) 
}_{
    \frac{\g_{\omega_0 \rightsquigarrow i }}{\g_{\omega_0 \leftlsquigarrow i } }
    =
    X_*^{(i)}
}
\, \cdot \, \g_{i \to j}
\Bigr]
-
X_*^{(i)} \,
\underbrace{
\left( \sum\limits_{j \in \Null \Null(i)} \g_{i \to j} \right) 
}_{
\g_{i \to }
}
\\ &=
X_*^{(i)}
\underbrace{
    \left(
    \sum\limits_{j \in \Null \Null(i) } \g_{i \to j}
    \right)
}_{
    \g_{i \to }
}
-
X_*^{(i)} \g_{i \to } = 0. 
\end{aligned}
\end{equation}

\item["$\Leftarrow$"] $ $ \\

Since the network is irreducible and positive recurrent, by lemma \ref{Lemma_PositiveRecurrenceForCTMC} we know, that there exists a stationary solution $ \bP_{*} \in \Kern(\G) \cap \ProbStates \cap (\R_{> \, 0})^{\Omega} $. Since the network exhibits detailed balance, which is preserved when looking at a strongly connected finite subsystem $ F \in \Fin(\Omega) $, we know that the stationary solution for the finite subsystem is given by $ \bP_{*}^{F} = \bXF_{*}$. 

From section \ref{Chaper_FromMarkovChainsToMasterEquations} in the appendix, we know that it is possible to express the stationary solution at state $\alpha \in \Omega $ as the fraction of the \emph{expected visiting time} $ T^{(\alpha)} $ and the \emph{expected return time} $ \| \bT \|_{1} = \sum\limits_{\alpha \in \Omega} T^{(\alpha)} $ which is true for both finite- and infinite systems: 
\begin{equation*}
\begin{aligned}
\bigl(\bP_{*}^{F}\bigr)^{(\alpha)} 
&=
\bigl(
    \bXF_{*}
\bigr)^{(\alpha)}
=
\frac{
    \bigl( T^{F} \bigr)^{(\alpha)}
}{
    \| \bT^{F} \|_{1} 
}
\\ 
\bigl(\bP_{*} \bigr)^{(\alpha)} 
&=
\frac{
    \bigl( T \bigr)^{(\alpha)}
}{
     \| \bT  \|_{1}
}
\end{aligned}
\end{equation*}

For an increasing sequence of finite, strongly connected subsystems $F_{1} \subseteq F_{2} \subseteq \dots $ with $\bigcup\limits_{n \in \N} F_{n} = \Omega $, both the \emph{expected visiting time} $ \bigl( (T^{F_n})^{(\alpha)} \bigr)_{n \in \N } $ as well as the \emph{expected return time} $ \bigl(\| \bT^{F_n} \|_{1} \bigr)_{n \in \N } $ are monotonously increasing and will converge to the expected visiting / return time of the infinite system, that is 

\begin{equation*}
\begin{aligned}
\bigl(
    \bXF_{*}
\bigr)^{(\alpha)}
=
\frac{
    \bigl( T^{F} \bigr)^{(\alpha)}
}{
    \| \bT^{F} \|_{1} 
}
\xlongrightarrow{|F| \to \infty }
\frac{
    \bigl( T \bigr)^{(\alpha)}
}{
    \| \bT  \|_{1} 
}
=
\bigl(
    \bP_{*}
\bigr)^{(\alpha)}
\end{aligned}
\end{equation*}

This means, we have 

\begin{equation*}
\begin{aligned}
\bXF_{*}
&\xlongrightarrow[\text{pointwise}]{|F| \to \infty }
\frac{\bX_{*}}{\| \bX_{*} \|_{1} } \text{  and }
\\
\bXF_{*}
&\xlongrightarrow[\text{pointwise}]{|F| \to \infty }
\bP_{*}, 
\end{aligned}
\end{equation*}

and hence $ \bP_{*} = \frac{\bX_{*}}{\| \bX_{*} \|_{1} }  $ and $ \| \bX_{*} \|_{1} < \infty $.

\end{itemize}

\end{proof}

\begin{thm}[Thermodynamic limit for irreducible, positive recurrent systems with detailed balance]\label{Theorem_DetailedBalanceApproximatingStationarySolutionByFiniteSubsystems} $ $ \\ 

If a countable, infinite dimensional network of a master equation (satisfying $ \| \G \|_{1,1} < \infty $) is irreducible, positive recurrent and satisfies the (generalized) detailed balance condition of equation \eqref{Eq_StationarySolution_DetailedBalance}, then the stationary solutions of the \emph{finite} dimensional master equation converge in the thermodynamic limit to the stationary solution of the \emph{countable, infinite} dimensional master equations, that is  $\bP^{F}_\infty(\bPF_0) \xlongrightarrow[\| \cdot \|_1]{|F| \to \infty} \bP_* $.

\end{thm}

\begin{proof}
Since the master equation is irreducible and positive recurrent, there exists a unique stationary solution $\bP_* \in \ProbStates \cap \, \Kern(\G) \cap (\R_{> \, 0})^{\Omega}$ such that $\bPt \xlongrightarrow[\| \cdot \|]{t \to \infty}  \bP_\infty(\bP_0) = \bP_*$. 

Since the master equation satisfies the detailed balance condition of equation \eqref{Eq_StationarySolution_DetailedBalance}, we know by theorem \ref{Thm_CandidateForStationarySolutionOfIrreducibleNetworksWithDetailedBalance} the exact form of this stationary solution, namely 

\begin{equation}
\begin{aligned}
\bP_*
&=
\frac{
    \bX_*
}{
    \left\|
    \bX_*
    \right \|_1
}
\\ \text{ with  }
\bX_*
&=
\left(
\frac{
\g_{\omega_0 \rightsquigarrow \omega }
}{
\g_{\omega_0 \leftlsquigarrow \omega }
}
\right)_{\omega \in \Omega}
\in
l^1(\Omega). 
\end{aligned}
\end{equation}

Since the detailed balance condition holds also for finite subsystems, we have: 

\begin{equation*}
\begin{aligned}
\bP_\infty^{F}(\bPF_0)
&=
\frac{
    \left(
    \frac{
    \g_{\omega_0 \rightsquigarrow f }
    }{
    \g_{\omega_0 \leftlsquigarrow f }
    }
    \right)_{f \in \Omega}
}{
    \left\|
    \left(
    \frac{
    \g_{\omega_0 \rightsquigarrow f }
    }{
    \g_{\omega_0 \leftlsquigarrow f }
    }
    \right)_{f \in \Omega}
    \right \|_1
}
=
\bPF_*
\xlongrightarrow[\| \cdot \|_1]{|F| \to \infty}
\bP_*. 
\end{aligned}
\end{equation*}

\end{proof}


\begin{lemma}[Sufficient condition for detailed balance and positive recurrence] \label{Lemma_SufficientConditionForDetailedBalanceAnd positiveRecurrence} $ $ \\

Let $\System$ be a network such that for every link, the reverse link is also present, that is $\g_{\alpha \to \beta}>0 \iff  \g_{\beta \to \alpha}>0 $. Moreover, let the fraction of the link strength between a link and its reverse link, be given by
\begin{equation}\label{Eq_SufficientConditionForDetailedBalance} \tag{suff.cond.det.bal.}
\begin{aligned}
\frac{\g_{\alpha \to \beta}}{\g_{\alpha \gets \beta}}
&=
\frac{q^{\bw_\beta \cdot \bc }}{q^{\bw_\alpha \cdot \bc }}
=
q^{(\bw_{\beta}-\bw_{\alpha}) \cdot \bc }, 
\end{aligned}
\end{equation}

for some positive number $q>0$, some non-negative vector $\bc \in (\R_{\geq \, 0})^{\N} $ and $w := d \circ b$. Then the network fulfills the - generalized - detailed balance condition of equation \eqref{Eq_StationarySolution_DetailedBalance}. 

Here, $d \,:\, \N_0 \to  \{0, 1\}^{\infty }   $ is the dual representation of natural numbers (that is, for $ \bd_n := d(n)$ we have $n = \sum\limits_{i \in \N_0} d^{(i)} \, 2^{i} = \bd \cdot 2^{\N_0}$ ) and $b \,:\, \Omega \to \N_0$ is any bijection.

If, in addition, $\sum\limits_{N \in \N} q^{c^{(N)}} \prod\limits_{i=1}^{N-1} (1+q^{c^{(i)}}) < \infty$, the network is positive recurrent. 

\end{lemma}

\begin{proof}
Let $(\alpha = \omega_0 \to \dots, \dots \omega_n = \beta) $ be any path in $\System$ from state $\alpha$ to state $\beta$. Then the weight of a path, divided by the weight of its corresponding reverse path, depends only on the start - and endpoint: 

\begin{equation*}
\begin{aligned}
\frac{\g_{\alpha \rightsquigarrow \beta}}{\g_{\alpha \leftlsquigarrow \beta}}
&=
\frac{\g_{\omega_0 \to \dots \to \omega_n}}{\g_{\omega_0 \gets \dots \gets \omega_n}}
=
\prod\limits_{i=0}^{n-1}
\underbrace{
    \frac{\g_{\omega_{i} \to \omega_{i+1}}}{\g_{\omega_{i} \gets \omega_{i+1}}}
}_{
    \frac{
    q^{\bc \cdot \bw_{\omega_{i+1}}}
    }{
    q^{\bc \cdot \bw_{\omega_{i}}}
    }
}
=
\frac{
\prod\limits_{i=i}^{n} q^{\bc \cdot \bw_{\omega_{i}}}
}{
\prod\limits_{i=0}^{n-1}  q^{\bc \cdot \bw_{\omega_{i}}}
}
=
\frac{
 q^{\bc \cdot \bw_{\beta}}
}{
 q^{\bc \cdot \bw_{\alpha}}
}
\end{aligned}
\end{equation*}

With $\omega_0 :=b^{-1}(0)$ and $ \frac{\g_{\omega_0 \rightsquigarrow \omega} }{\g_{\omega_0 \leftlsquigarrow \omega}} = \frac{q^{\bc \cdot \bw_{\omega}}}{q^{\bc \cdot \bw_{\omega_0}}} \xlongequal{\bw_{\omega_0} = \bw_{\bzero} = \bzero} \frac{q^{\bc \cdot \bw_{\omega}}}{q^{\bc \cdot \bw_{\bzero}}} = q^{\bc \cdot \bw_{\omega}}$, we get the following expression for $\bX_*$:

\begin{equation}
\begin{aligned}
\bX_*
&=
\left(
\underbrace{
\frac{\g_{\omega_0 \rightsquigarrow \omega} }{\g_{\omega_0 \leftlsquigarrow \omega}}
}_{q^{\bc \cdot \bw_{\omega}}}
\right)_{\omega \in \Omega }
\\ &=
\left(
\bigl( q^{0}, q^{c^{(1)}} \bigr),
\underbrace{
q^{c^{(2)}},q^{c^{(2)}+c^{(1)}}
}_{
q^{c^{(2)}}\,  \bigl( q^{0}, q^{c^{(1)}}\bigr)
},
\underbrace{
q^{c^{(3)}},q^{c^{(3)}+c^{(1)}},q^{c^{(3)}+c^{(2)}},q^{c^{(3)}+c^{(2)}+c^{(1)}}
}_{
q^{c^{(3)}} \, \bigl(q^{0}, q^{c^{(1)}},q^{c^{(2)}},q^{c^{(2)}+c^{(1)}} \bigr) 
}
, \dots\right)
\\ &=
\left(
\underbrace{
    \underbrace{
    \bigl( q^{0}, q^{c^{(1)}} \bigr), 
    q^{c^{(2)}} \, \bigl( q^{0}, q^{c^{(1)}} \bigr)
    }_{
    (1, q^{c^{(2)}}) \otimes (1, q^{c^{(1)}})
    },
    q^{c^{(3)}} \,
    \underbrace{
     \bigl(q^{0}, q^{c^{(1)}},q^{c^{(2)}},q^{c^{(2)}+c^{(1)}} \bigr)
    }_{
    (1, q^{c^{(2)}}) \otimes (1, q^{c^{(1)}})
    }
    }_{
    (1, q^{c^{(3)}}) \otimes1, q^{c^{(2)}}) \otimes (1, q^{c^{(1)}})
    }, \dots
\right)
\\ &=
\bigotimes\limits_{n \in \N}^{\gets} \, (1, q^{c^{(n)}})
:=
\otimes \dots \otimes (1, q^{c^{(3)}}) \otimes (1, q^{c^{(2)}}) \otimes (1, q^{c^{(1)}}). 
\end{aligned}
\end{equation}

This will be discussed in more details in the section about the countable, infinite dimensional hypercube \ref{Chapter_TheInfiniteDimensionalHypercube}.

The network is positive recurrent, if and only if $\bX_*= \left(\frac{ \g_{\omega_0 \rightsquigarrow \omega }
}{\g_{\omega_0 \leftlsquigarrow \omega } } \right)_{\omega \in \Omega} \in l^1(\Omega)$. We can compute the norm of $\bX_*$ to

\begin{equation}\label{Eq_SufficientConditionForPositiveRecurrence}\tag{suff.cond.pos.rec.}
\begin{aligned}
\infty
&\overset{!} >
\| \bX_*\|_1
=
\sum\limits_{\omega \in \Omega} X^{(\omega)}
=
\sum\limits_{\omega \in \Omega} 
\frac{
\g_{\omega_0 \rightsquigarrow \omega}
}{
\g_{\omega_0 \leftlsquigarrow \omega}
}
=
\sum\limits_{\bw \in \{0,1\}^{\infty}} q^{\bw \cdot \bc}
=
\sum\limits_{N \in \N}
\sum\limits_{\bw_{N} \in \{0,1\}^{N-1}\times \{1\} }
q^{\bw_{N} \cdot \bc }
\\ &=
\sum\limits_{N \in \N}
q^{c^{(N)}} \prod\limits_{i=1}^{N-1} (1+q^{c^{(i)}}). 
\end{aligned}
\end{equation}

Clearly, this does not converge, if $ \bc $ is a constant vector, that is $\bc = \text{const} \cdot (1, 1, \dots)$, since $\sum\limits_{N \in \N }q (1+q)^{N-1}  = \infty $. 

It does converge, however, for $ \bc \in \{ (n)_{n \in \N}, (n^2)_{n \in \N}, (2^n)_{n \in \N} \} $, as can be seen, by e.g. applying the ratio test (see fact \ref{Fact_ApplyingRatioTest}).

\end{proof}



\section{A first non-trivial example: the linear chain with one open end }\label{Chapter_AFirstNonTrivialExample_TheLinearChainOnN}

In the following section we focus on  the following generator 

\begin{equation} \label{Eq_GeneratorOfTheLinearChainOnN}
\begin{aligned}
\G 
&=
\begin{pmatrix}
- \g_{0 \to 1} & \g_{1 \to 0}  & 0            & 0      & \dots          & \dots \\ 
\g_{0 \to 1}   & -\g_{1 \to }  & \g_{2 \to 1} &        & \dots          & \dots \\ 
0              &  \g_{1 \to 2} & -\g_{2 \to}  & \ddots &  \vdots        & \vdots \\
0              & 0             & \g_{2 \to 3} & \ddots & \g_{n \to n-1} & \dots  \\ 
\vdots         & \vdots        &  \vdots      & \ddots &  -\g_{n \to }  & \dots \\ 
\vdots         &  \vdots       &  \vdots      & \vdots & \g_{n \to n+1} & \dots  \\
\vdots         & \vdots        & \vdots       & \vdots & \vdots         & \ddots 
\end{pmatrix}, 
\end{aligned}
\end{equation}

with $\g_{n \to } := \sum\limits_{i \in \Omega} \g_{n \to i} $, or component-wise

\begin{equation} \label{Eq_LinearChainOnN0_DefGenerator_Componentwise}
\begin{aligned}
\G^{(m, n)}
&=
\delta_{n,0} \cdot \left(- \g_{0 \to 1} \delta_{m,0} + \g_{0 \to 1} \delta_{m,1} \right)   
\\ +& 
(1 - \delta_{n,0}) 
\cdot \Bigl(
\delta_{m, n-1} \, \g_{n \to n-1}  - \delta_{m,n}\, (\g_{n \to n-1} + \g_{n \to n+1}) + \delta_{m, n+1} \,  \g_{n \to n+1}  \Bigr), \\ 
\end{aligned}
\end{equation}

which corresponds to an infinite long chain with one open end: 

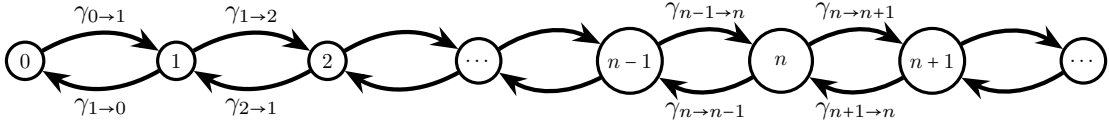
\begin{figure}[H]
\begin{center}
\begin{tikzpicture}[scale=2]
\node[State](0)   at (0,0) [circle,draw] {0} ;
\node[State](1)   at (1,0) [circle,draw] {1} ; 
\node[State](2)   at (2,0) [circle,draw] {2} ;
\node[State](3)   at (3,0) [circle,draw] {$\dots$} ;
\node[State](n-1) at (4,0) [circle,draw] {$n-1$} ;
\node[State](n)   at (5,0) [circle,draw] {$\;\;\;n\;\;\;$} ;
\node[State](n+1) at (6,0) [circle,draw] {$n+1$} ;
\node[State](n+2) at (7,0) [circle,draw] {$\dots$} ;
\path[Link,bend left] (0) edge node[above] {$\g_{0 \to 1}$} (1);
\path[Link,bend left] (1) edge node[above] {$\g_{1 \to 2}$} (2);
\path[Link,bend left] (2) edge node[above] {} (3);
\path[Link,bend left] (3) edge node[above] {} (n-1);
\path[Link,bend left] (n-1) edge node[above] {$\g_{n-1 \to n}$} (n);
\path[Link,bend left] (n) edge node[above] {$\g_{n \to n+1}$} (n+1);
\path[Link,bend left] (n+1) edge node[above] {} (n+2);
\path[Link,bend left] (3) edge node[below] {} (2);
\path[Link,bend left] (2) edge node[below] {$\g_{2 \to 1}$} (1);
\path[Link,bend left] (1) edge node[below] {$\g_{1 \to 0}$} (0);
\path[Link,bend left] (n+2) edge node[below] {} (n+1);
\path[Link,bend left] (n+1) edge node[below] {$\g_{n+1 \to n}$} (n);
\path[Link,bend left] (n) edge node[below] {$\g_{n \to n-1}$} (n-1);
\path[Link,bend left] (n-1) edge node[above] {} (3);
\end{tikzpicture}
\caption{The linear chain with one open end }
\label{Fig_Network_LinearChainOnN0} 
\end{center}
\end{figure}


\subsection{The thermodynamic limit for the linear chain with one open end }\label{Section_TheThermodynamicLimitForTheLinearChainWithOneOpenEnd}

The master equation for the linear chain with one open end satisfies the (generalized) detailed balance condition of equation \eqref{Eq_StationarySolution_DetailedBalance} for

\begin{equation} \label{Eq_CandidateForStationarySolution_LinearChainWithOneOpenEnd}
\begin{aligned}
\bX_* 
&:=
\left( \frac{\g_{0 \rightsquigarrow n}}{\g_{0 \leftlsquigarrow n}} \right)_{n \in \N_0 } = \prod\limits_{i=0}^{n-1} \frac{\g_{i \to i+1}}{\g_{i \gets i+1}}, 
\end{aligned}
\end{equation}

since we have for $n \in \N_0$: 
\begin{equation*}
\begin{aligned}
\underbrace{
X_*^{(n)}
}_{
\frac{\g_{0 \rightsquigarrow n}}{\g_{0 \leftlsquigarrow n}}
}
\, \g_{n \to n\pm 1 }
&=
\underbrace{
\left(
\prod\limits_{i=0}^{n-1} \frac{\g_{i \to i+1}}{\g_{i \gets i+1}}
\right)
 \,
 \frac{\g_{n \to n\pm 1}}{\g_{n \gets n\pm 1}}
}_{
\prod\limits_{i=0}^{n\pm 1} \frac{\g_{i \to i+1}}{\g_{i \gets i+1}}
}
\,
\g_{n\pm 1 \to n}
=
\underbrace{
\left(
\prod\limits_{i=0}^{n\pm 1 -1} \frac{\g_{i \to i+1}}{\g_{i \gets i+1}}
\right)
}_{
X_*^{(n \pm 1)}
}
\,
\g_{n\pm 1 \to n}
=
X_*^{(n \pm 1)} \, \g_{n\pm 1 \to n}. 
\end{aligned}
\end{equation*}

An alternative way to see the (generalized) detailed balance condition is to notice, that since there is exactly one path from the state $0$ to any other state $n$, so clearly $ \frac{\g_{0 \rightsquigarrow n}}{\g_{0 \leftlsquigarrow n}} $ does not depend on the path $0 \rightsquigarrow n$. The following figure is meant to illustrate this:

\begin{figure}[H]
\begin{center}
\begin{equation*}
\begin{aligned}
\bP^{F}_*
&=
\frac{1}{\Zen^{F}}
\colvec{9}
{\g_{T_{\to \state{\bzero}}}}{}
{\g_{T_{\to \state{1}}}}{}
{\vdots}{}
{\g_{T_{\to \state{f}}}}{}
{\vdots}
 \\ &=
\frac{1}{\Zen^{F}}
\colvec{9}{
\underbrace{
    \begin{tikzpicture}[scale=1]
    \centering
    \node[State](0)   at (0,0) [circle,draw] {0} ;
    \node[State](1)   at (1,0) [circle,draw] {1} ; 
    \node[State](dots)   at (2,0) [circle,draw] {$\dots$} ;
    \node[State](f)   at (3,0) [circle,draw] {f} ;
    \node[State](Dots)   at (4,0) [circle,draw] {$\dots$} ;
    \path[Link] (Dots) edge node[above] {} (f);
    \path[Link] (f) edge node[above] {} (dots);
    \path[Link] (dots) edge node[above] {} (1);
    \path[Link] (1) edge node[above] {} (0);
    \end{tikzpicture}
    }_{
    {\g_{T}}_{\to \state{\bzero}}
    }
}{}{
\underbrace{
    \begin{tikzpicture}[scale=1]
    \centering
    \node[State](0)   at (0,0) [circle,draw] {0} ;
    \node[State](1)   at (1,0) [circle,draw] {1} ; 
    \node[State](dots)   at (2,0) [circle,draw] {$\dots$} ;
    \node[State](f)   at (3,0) [circle,draw] {f} ;
    \node[State](Dots)   at (4,0) [circle,draw] {$\dots$} ;
    \path[Link] (Dots) edge node[above] {} (f);
    \path[Link] (f) edge node[above] {} (dots);
    \path[Link] (dots) edge node[above] {} (1);
    \path[Link, CrimsonGlory] (0) edge node[above] {} (1);
    \end{tikzpicture}
    }_{
    \g_{T_{\to \state{1}}}
    }
}{}{\vdots}{}{
\underbrace{
    \begin{tikzpicture}[scale=1]
    \centering
    \node[State](0)   at (0,0) [circle,draw] {0} ;
    \node[State](1)   at (1,0) [circle,draw] {1} ; 
    \node[State](dots)   at (2,0) [circle,draw] {$\dots$} ;
    \node[State](f)   at (3,0) [circle,draw] {f} ;
    \node[State](Dots)   at (4,0) [circle,draw] {$\dots$} ;
    \path[Link] (Dots) edge node[above] {} (f);
    \path[Link,CrimsonGlory] (dots) edge node[above] {} (f);
    \path[Link,CrimsonGlory] (1) edge node[above] {} (dots);
    \path[Link,CrimsonGlory] (0) edge node[above] {} (1);
    \end{tikzpicture}
    }_{
    \g_{T_{\to \state{f}}}
    }
}{}{\vdots}
=
\frac{1}{\Zen^{F}}
\colvec{9}{
\underbrace{
    \begin{tikzpicture}[scale=1]
    \centering
    \node[State](0)   at (0,0) [circle,draw] {0} ;
    \node[State](1)   at (1,0) [circle,draw] {1} ; 
    \node[State](dots)   at (2,0) [circle,draw] {$\dots$} ;
    \node[State](f)   at (3,0) [circle,draw] {f} ;
    \node[State](Dots)   at (4,0) [circle,draw] {$\dots$} ;
    \path[Link] (Dots) edge node[above] {} (f);
    \path[Link] (f) edge node[above] {} (dots);
    \path[Link] (dots) edge node[above] {} (1);
    \path[Link] (1) edge node[above] {} (0);
    \end{tikzpicture}
    }_{
    {\g_{T}}_{\to \state{\bzero}}
    }
    }{}{
    \underbrace{
    \begin{tikzpicture}[scale=1]
    \centering
    \node[State](0)   at (0,0) [circle,draw] {0} ;
    \node[State](1)   at (1,0) [circle,draw] {1} ; 
    \node[State](dots)   at (2,0) [circle,draw] {$\dots$} ;
    \node[State](f)   at (3,0) [circle,draw] {f} ;
    \node[State](Dots)   at (4,0) [circle,draw] {$\dots$} ;
    \path[Link] (Dots) edge node[above] {} (f);
    \path[Link] (f) edge node[above] {} (dots);
    \path[Link] (dots) edge node[above] {} (1);
    \path[Link, MyGreen] (1) edge node[above] {} (0);
    \end{tikzpicture}
    }_{
    {\g_{T}}_{\to \state{\bzero}}
    }
    \, \cdot \, 
    \frac{
    \begin{tikzpicture}[scale=1]
    \centering
    \node[State](0)   at (0,0) [circle,draw] {0} ;
    \node[State](1)   at (1,0) [circle,draw] {1} ; 
    \path[Link, CrimsonGlory] (0) edge node[above] {} (1);
    \end{tikzpicture}
    }{
    \begin{tikzpicture}[scale=1]
    \centering
    \node[State](0)   at (0,0) [circle,draw] {0} ;
    \node[State](1)   at (1,0) [circle,draw] {1} ; 
    \path[Link, MyGreen] (1) edge node[above] {} (0);
    \end{tikzpicture}
    }
}{}{\vdots
}{}{
    \underbrace{
    \begin{tikzpicture}[scale=1]
    \centering
    \node[State](0)   at (0,0) [circle,draw] {0} ;
    \node[State](1)   at (1,0) [circle,draw] {1} ; 
    \node[State](dots)   at (2,0) [circle,draw] {$\dots$} ;
    \node[State](f)   at (3,0) [circle,draw] {f} ;
    \node[State](Dots)   at (4,0) [circle,draw] {$\dots$} ;
    \path[Link] (Dots) edge node[above] {} (f);
    \path[Link,MyGreen] (f) edge node[above] {} (dots);
    \path[Link,MyGreen] (dots) edge node[above] {} (1);
    \path[Link, MyGreen] (1) edge node[above] {} (0);
    \end{tikzpicture}
    }_{
    {\g_{T}}_{\to \state{\bzero}}
    }
    \, \cdot \, 
    \frac{
    \begin{tikzpicture}[scale=1]
    \centering
    \node[State](0)   at (0,0) [circle,draw] {0} ;
    \node[State](1)   at (1,0) [circle,draw] {1} ; 
    \node[State](dots)   at (2,0) [circle,draw] {$\dots$} ;
    \node[State](f)   at (3,0) [circle,draw] {f} ;
    \path[Link, CrimsonGlory] (0) edge node[above] {} (1);
    \path[Link, CrimsonGlory] (1) edge node[above] {} (dots);
    \path[Link, CrimsonGlory] (dots) edge node[above] {} (f);
    \end{tikzpicture}
    }{
    \begin{tikzpicture}[scale=1]
    \centering
    \node[State](0)   at (0,0) [circle,draw] {0} ;
    \node[State](1)   at (1,0) [circle,draw] {1} ; 
    \node[State](dots)   at (2,0) [circle,draw] {$\dots$} ;
    \node[State](f)   at (3,0) [circle,draw] {f} ;
    \path[Link, MyGreen] (1) edge node[above] {} (0);
    \path[Link, MyGreen] (dots) edge node[above] {} (1);
    \path[Link, MyGreen] (f) edge node[above] {} (dots);
    \end{tikzpicture}
    }
}{}{\vdots}
\\ &= 
\frac{1}{\Zen^{F}}
\colvec{9}
{1}{}{
    \frac{
    \begin{tikzpicture}[scale=1]
    \centering
    \node[State](0)   at (0,0) [circle,draw] {0} ;
    \node[State](1)   at (1,0) [circle,draw] {1} ; 
    \path[Link, CrimsonGlory] (0) edge node[above] {} (1);
    \end{tikzpicture}
    }{
    \begin{tikzpicture}[scale=1]
    \centering
    \node[State](0)   at (0,0) [circle,draw] {0} ;
    \node[State](1)   at (1,0) [circle,draw] {1} ; 
    \path[Link, MyGreen] (1) edge node[above] {} (0);
    \end{tikzpicture}
    }
}{}{
    \frac{
    \begin{tikzpicture}[scale=1]
    \centering
    \node[State](0)   at (0,0) [circle,draw] {0} ;
    \node[State](1)   at (1,0) [circle,draw] {1} ; 
    \node[State](2)   at (2,0) [circle,draw] {2} ;
    \path[Link, CrimsonGlory] (0) edge node[above] {} (1);
    \path[Link, CrimsonGlory] (1) edge node[above] {} (2);
    \end{tikzpicture}
    }{
    \begin{tikzpicture}[scale=1]
    \centering
    \node[State](0)   at (0,0) [circle,draw] {0} ;
    \node[State](1)   at (1,0) [circle,draw] {1} ; 
    \node[State](2)   at (2,0) [circle,draw] {2} ;
    \path[Link, MyGreen] (2) edge node[above] {} (1);
    \path[Link, MyGreen] (1) edge node[above] {} (0);
    \end{tikzpicture}
    }
}{}{
    \frac{
    \begin{tikzpicture}[scale=1]
    \centering
    \node[State](0)   at (0,0) [circle,draw] {0} ;
    \node[State](1)   at (1,0) [circle,draw] {1} ; 
    \node[State](2)   at (2,0) [circle,draw] {2} ;
    \node[State](3)   at (3,0) [circle,draw] {3} ;
    \path[Link, CrimsonGlory] (0) edge node[above] {} (1);
    \path[Link, CrimsonGlory] (1) edge node[above] {} (2);
    \path[Link, CrimsonGlory] (2) edge node[above] {} (3);
    \end{tikzpicture}
    }{
    \begin{tikzpicture}[scale=1]
    \centering
    \node[State](0)   at (0,0) [circle,draw] {0} ;
    \node[State](1)   at (1,0) [circle,draw] {1} ; 
    \node[State](2)   at (2,0) [circle,draw] {2} ;
    \node[State](3)   at (3,0) [circle,draw] {3} ;
    \path[Link, MyGreen] (3) edge node[above] {} (2);
    \path[Link, MyGreen] (2) edge node[above] {} (1);
    \path[Link, MyGreen] (1) edge node[above] {} (0);
    \end{tikzpicture}
    }
}{}{\vdots}
\end{aligned}
\end{equation*}

\caption{Illustrating the fact, that the stationary solution for a finite linear chain can be expressed not only in terms of in-trees $\bP^{F}_* = \frac{1}{\Zen^{F}}\left(\g_{T_{\to f}} \right)_{f \in F}$, but also in terms the paths from- and to zero: $\bP^{F}_* = \frac{1}{\Zen^{F}}\left( \frac{\g_{0 \rightsquigarrow f}}{\g_{0 \leftlsquigarrow f}} \right)_{f \in F}$ }
\label{Fig_FiniteLinearChain} 
\end{center}
\end{figure}

By theorem \ref{Thm_CandidateForStationarySolutionOfIrreducibleNetworksWithDetailedBalance}, the only candidate for a stationary solution is $\frac{\bX_*}{\| \bX_* \|_1}$. When $\| \bX_* \|_1 = \infty$, no stationary solution exists, when $\| \bX_* \|_1 < \infty$, then the master equation of the linear chain with one open end is \emph{positive recurrent} (theorem \ref{Thm_CandidateForStationarySolutionOfIrreducibleNetworksWithDetailedBalance}) and its stationary solution can be approximated by corresponding finite subsystems in the thermodynamic limit (theorem \ref{Theorem_DetailedBalanceApproximatingStationarySolutionByFiniteSubsystems}).

When combining these results, we can conclude that for the linear chain with one open end we have: 

\begin{equation}
\begin{aligned}
\lim\limits_{t \to \infty }  \, 
&\lim\limits_{|F| \to \infty } \, 
\bP^{F}(t \,| \, \bPF_0)
= 
\lim\limits_{t \to \infty } \bP(t \,| \, \bP_0) = 
\\ &\xlongequal[\text{positive recurrent}]{\text{irreducible + }}
\bP_*
\xlongequal[\text{detailed balance}]{\text{ thm } \ref{Theorem_DetailedBalanceApproximatingStationarySolutionByFiniteSubsystems}}
\lim\limits_{|F| \to \infty } \, 
\underbrace{
    \bPF_*
}_{
    \bP_\infty^{F}(\bPF_0)
}
\\ &= 
\lim\limits_{|F| \to \infty } \, 
\underbrace{
    \bP_\infty^{F}(\bPF_0)
}_{
    \lim\limits_{t \to \infty }  \, 
    \bP^{F}(t \,| \, \bPF_0)
}
=
\lim\limits_{|F| \to \infty } \, 
\lim\limits_{t \to \infty }  \, 
\bP^{F}(t \,| \, \bPF_0)
\end{aligned}
\end{equation}

\begin{claim}[A simplified version for specific rates] $ $ \\ 

If for all $n \in \N_0$ we have

\begin{equation*}
\begin{aligned}
\g_{n \to n+1} 
&=
\lambda_{n+1} \, x \text{   and } 
\\
\g_{n \gets n+1} 
&=
\lambda_{n} \, (1-x)
\end{aligned}
\end{equation*}

for some $x \in (0,1)$ some sequence $(\lambda_n)_{n \in \N} \in l^1(N_0, \R_{> \, 0}) $ with $\lambda:= \lim\limits_{n \to \infty} \sqrt[n]{\lambda_{n}}$, then $\bX_* = \left( \frac{\lambda_n}{\lambda_0} \, \left(\frac{x}{1-x} \right)^n \right)_{n \in \N_0} $ and we have $\bX_* \in l^1(\N_0)$, if and only if $x \in \left(0, \frac{1}{1+\lambda}\right)$. 
\end{claim}

\begin{proof}
We have 
\begin{equation*}
\begin{aligned}
X_*^{(n)}
&=
\frac{\g_{0 \rightsquigarrow n }}{\g_{0 \leftlsquigarrow n }}
=
\prod\limits_{i=0}^{n-1}
\underbrace{
\frac{\g_{i \to i+1 }}{\g_{i \gets i+1 }}
}_{
\frac{\lambda_{i+1 } \, q}{ \lambda_{i} \, q }
}
=
\frac{\lambda_n}{\lambda_0} \, \left(\frac{x}{1-x} \right)^n \text{  and }
\\
\infty
&>
\| \bX_* \|_1
=
\sum\limits_{n \in \N_0}
X_*^{(n)} 
=
\sum\limits_{n \in \N_0}
\frac{\lambda_n}{\lambda_0} \, \left(\frac{x}{1-x} \right)^n
\\ \xLongleftrightarrow[\text{test}]{\text{root}}
1
&>
\limsup\limits_{n \to \infty}
\sqrt[n]{\frac{\lambda_n}{\lambda_0} \, \left(\frac{x}{1-x} \right)^n}
=
\lambda \, \left(\frac{x}{1-x} \right)
\\ \iff
x
&<
\frac{1}{1+\lambda}
\end{aligned}
\end{equation*}

\end{proof}

\subsection{Computing the kernel of the generator of the linear chain on $\N_0$ `by foot'}
\label{Section_ComputingTheKernelOfTheGeneratorOfTheLinearChainOnNByFoot}

Since the linear chain satisfies the detailed balance condition \eqref{Eq_StationarySolution_DetailedBalance}, we are able to directly write down the only candidate for the stationary solution, as we did in equation \eqref{Eq_CandidateForStationarySolution_LinearChainWithOneOpenEnd}.

With the help of the two lemmata \ref{Lemma_StrictPositivityForStationarySolutionsOfStronglyConnectedNetworks} and \ref{IfThermodynamicLimitOfStationaryStatesExists} we can validate this guess, without computing the kernel of $\G$ directly.

For illustration purposes, we do this now `by foot'. 

\begin{claim}\label{Claim_TildeKern_SpanPStar}
Let $\widetilde{\Kern(\G)} := \{ \bX \in \R^{\N_0} \,:\, \G \bX \text{ exists and }\G \bX = \bzero^{\N_0}\}$ (note that 
$\Kern(\G) = \widetilde{\Kern(\G)} \cap l^1(\N_0)$ ), then $\widetilde{\Kern(\G)} =   \Span(\bX_*) $, where $ \bX_*$ is given by

\begin{equation}\label{PStar_LinearChainOnN}
\begin{aligned}
\bX_* 
&=
\left(
\frac{\g_{0 \rightsquigarrow n}}{\g_{0 \leftlsquigarrow n}}
\right)_{n \in \N }
=
\left(1, \left( \prod\limits_{i=0}^{n-1} \frac{\g_{i \to i+1}}{\g_{i \gets i+1}}\right)_{n \geq \, 1} \right) 
%
%
\\ &=
 \left( 1, \frac{\g_{0 \to 1}}{\g_{1 \to 0}}, \frac{\g_{0 \to 1}}{\g_{2 \to 1}} \frac{\g_{1 \to 2}}{\g_{1 \to 0 }},  \frac{\g_{0 \to 1}}{\g_{3 \to 2}} \, \frac{\g_{1 \to 2} \, \g_{2 \to 3}}{\g_{2 \to 1}\,  \g_{1 \to 0 }} ,  \dots \right). 
=
\left(1, 
\frac{
    \g_{0 \to 1} 
}{
    \g_{0 \gets 1}
}, 
\frac{
    \g_{0 \to 1} \, \g_{1 \to 2}
}{
    \g_{0 \gets 1} \, \g_{1 \gets 2}
}, 
\frac{
    \g_{0 \to 1} \, \g_{1 \to 2} \, \g_{2 \to 3}
}{
    \g_{0 \gets 1} \, \g_{1 \gets 2} \, \g_{2 \gets 3}
}
\right)
\end{aligned}
\end{equation}

where $(a, (b_n)_{n\geq 1}):= (a, b_1, b_2, b_3, \dots )$.

\end{claim}

\begin{proof} [proof of "$\Span(\bX_*) = \widetilde{\Kern(\G)}$"]

Note: we have 

\begin{equation*}
\begin{aligned}
X_*^{(m)}
&=
\frac{\g_{0 \to 1}}{\g_{m \to m-1}} \, 
\prod\limits_{i=1}^{m-2} \frac{\g_{i \to i+1 }}{\g_{i \to i-1}}
\frac{\g_{m-1 \to m}}{\g_{m-1 \to m-2 }}
=
\underbrace{
\left(
\frac{\g_{0 \to 1}}{\g_{m-1 \to m-2 } } \, 
\prod\limits_{i=1}^{m-2} \frac{\g_{i \to i+1 }}{\g_{i \to i-1}}
\right)
}_{
X_*^{(m-1)}
}
\frac{\g_{m-1 \to m}}{\g_{m \to m-1}}
\end{aligned}
\end{equation*}

\begin{itemize}

\item[]
\item["$ \subseteq $": ]

\begin{equation*}
\begin{aligned}
(\G \bX_*)^{(m)}
&\xlongequal{m=0}
-\underbrace{
X_{*}^{(0)}
}_{
1
} \, \g_{0 \to 1} + 
\underbrace{X_*^{(1)}}_{
\frac{\g_{0 \to 1}}{\g_{1 \to 0}}
} \, \g_{1 \to 0} 
=
0 \\
%
&\xlongequal{m \geq 1}
X_{*}^{(m-1)} \, \g_{m-1 \to m} - 
\underbrace{
X_{*}^{(m)}
}_{
X_{*}^{(m-1)} \, \frac{\g_{{m-1} \to m}}{\g_{m \to m-1 }}
} \, {\begin{color}{blue} \g_{m \to} \end{color}}
+ \underbrace{
X_{*}^{(m+1)}
}_{
X_{*}^{(m-1)} \, \frac{\g_{m-1 \to m} }{\g_{m \to m-1 }} \, 
\frac{\g_{m} \to m+1}{
\begin{color}{red} \g_{m+1 \to m} \end{color}
} 
} \, {\begin{color}{red} \g_{m+1 \to m} \end{color} }\\
&=
X_{*}^{(m-1)} \, \g_{{m-1} \to m} 
\underbrace{
\left(1 - \frac{{\begin{color}{blue} \g_{m \to m-1} + \g_{m \to m+1} \end{color}}}{\g_{m \to m-1}} + \frac{\g_{m \to m+1}}{\g_{m \to m-1}} \right)
}_{
0
}
=
0. 
\end{aligned}
\end{equation*}

\item["$\supseteq $": ] Let $\bX \in \widetilde{\Kern(\G)}$, that is $ \G \bX = \bzero$. It suffices to show that $X^{(n+1)} = X^{(n)} \cdot \frac{\g_{n \to n+1}}{\g_{n+1 \to n}}$, which we will prove by induction: 
\begin{itemize}
\item[i-start: ]
\begin{equation*}
\begin{aligned}
0 
&\xlongequal{!}
(\G \bX)^{(0)} = -X^{(0)} \, \g_{0 \to 1} + X^{(1)} \,  \g_{1 \to 0} \\
\Rightarrow 
X^{(1)} 
&=
X^{(0)} \cdot \frac{\g_{0 \to 1}}{\g_{1 \to 0}}
\end{aligned}
\end{equation*}

\item[i-step: ]
\begin{equation*}
\begin{aligned}
0 
&\xlongequal{!}
(\G \bX)^{(m)}
=
X^{(m-1)} \, \g_{m-1 \to m } - 
\underbrace{
X^{(m)} 
}_{X^{(m-1)} \, \frac{\g_{m-1 \to m} }{\g_{m \to m-1}}}
\begin{color}{blue} \g_{m \to } \end{color} + X^{(m+1)} \,  \g_{m+1 \to m} = \\
%
&\xlongequal{\text{i-hypothesis}}
X^{(m-1)}\, \g_{m-1 \to m } 
\underbrace{
\left( 1 - \frac{ \begin{color}{blue}  \g_{m \to m-1} + \g_{m \to m+1} \end{color} }{\g_{m \to m-1}} \right)
}_{
- \frac{\color{red} \g_{m \to m+1}  }{\g_{m \to m-1}}
}+ X^{(m+1)} \g_{m+1 \to m} \\
%
\Rightarrow 
X^{(m+1)}
&=
\frac{1}{\g_{m+1 \to m}}
\underbrace{
\left(
X^{(m-1)} \frac{\g_{m-1 \to m}}{\g_{m \to m-1}}
\right)}_{
X^{(m)}
}
{\color{red} \g_{m \to m+1} } = \\
%
&=
X^{(m)} \, \frac{\g_{m \to m+1}}{\g_{m+1 \to m}}. 
\end{aligned}
\end{equation*}

\end{itemize}

\end{itemize}

\end{proof}

This means, that we have 
\begin{equation}\label{KernOfTheGenOnN}\tag{KernOfTheGenOnN}
\begin{aligned}
\Kern(\G) 
= 
\underbrace{
\widetilde{\Kern(\G) }
}_{
\Span(\bX_*)
} 
\, \cap \,  l^1(\N_0)
=
\Span\left(
\frac{\bX_*}{\|\bX_*\|_1}
\right)
=
\begin{cases}
\Span(\bX_*) &\text{, if } \|\bX_*\| < \infty  \\
\{ \bzero \} &\text{, if } \|\bX_*\| = \infty. 
\end{cases}
\end{aligned}
\end{equation}
\section{The linear  chain with two open ends}\label{Chapter_ASecondExample_TheLinearChainOnZ}

In the following section we focus on  the following generator 

\begin{equation}\label{Eq_GeneratorOfTheLinearChainOnZ}
\begin{aligned}
\G 
&=
\begin{pmatrix}
\ddots         &0              &0              &\vdots        &\vdots        & \vdots         &\vdots  & \vdots         & \vdots \\   
\hdots         &\g_{-n\to-n-1} &0              & 0            &\vdots        & \vdots         &\vdots  & \vdots         & \vdots  \\   
\hdots         &-\g_{-n \to }  &\ddots         & 0            &0             & 0              &\vdots  & \vdots         & \vdots   \\   
\hdots         &\g_{-n\to-n+1} &\ddots         &\g_{-1 \to -2}&0             & 0              &\vdots  & \vdots         & \vdots    \\
\hdots         &0              &\ddots         &-\g_{-1 \to } &\g_{0 \to -1 }& 0              &\vdots  & \vdots         & \vdots     \\
\vdots         &0              &0              &\g_{-1 \to 0} &-\g_{0 \to }  & \g_{1 \to 0}   & \dots  & \vdots         & \vdots \\ 
\vdots         &\vdots         &0              &0             & \g_{0 \to 1} & -\g_{1 \to }   & \dots  & 0              & \vdots  \\ 
\vdots         &\vdots         &\vdots         &0             & 0            &  \g_{1 \to 2}  & \ddots & 0              & \vdots   \\
\vdots         &\vdots         &\vdots         &\vdots        & 0            & 0              & \ddots & \g_{n \to n-1} & \dots     \\ 
\vdots         &\vdots         &\vdots         &\vdots        & \vdots       & 0              & \ddots & -\g_{n \to }   & \dots      \\ 
\vdots         &\vdots         &\vdots         &\vdots        & \vdots       & \vdots         & \vdots & \g_{n \to n+1} & \dots       \\
\vdots         &\vdots         &\vdots         &\vdots        & \vdots       & \vdots         & \vdots & \vdots         & \ddots 
\end{pmatrix} \\
&\text{or component-wise } \\
\G^{(m, n)}
&=
\delta_{m, n-1} \, \g_{n \to n-1}  - \delta_{m,n}\, (\g_{n \to n-1} + \g_{n \to n+1}) + \delta_{m, n+1} \,  \g_{n \to n+1}  \\
&\text{for } m, n \in \Z,  \\ 
&\text{which corresponds to an infinite long chain with \emph{two} open ends: } \\
&
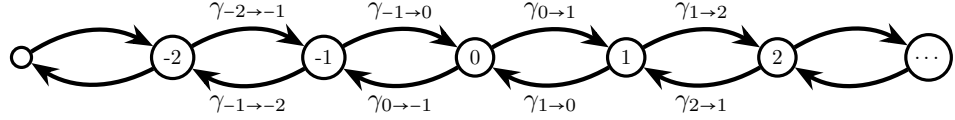
\begin{tikzpicture}[scale=2]
\node[State](-3)   at (-3,0) [circle,draw] {} ;
\node[State](-2)   at (-2,0) [circle,draw] {-2} ; 
\node[State](-1)   at (-1,0) [circle,draw] {-1} ;
\node[State](0)    at (0,0)  [circle,draw] {0} ;
\node[State](1)    at (1,0)  [circle,draw] {1} ;
\node[State](2)    at (2,0)  [circle,draw] {2} ;
\node[State](3)    at (3,0)  [circle,draw] {\dots} ;
\path[Link,bend left] (-3) edge node[above] {}  (-2);
\path[Link,bend left] (-2) edge node[above] {$\g_{-2 \to -1}$}  (-1);
\path[Link,bend left] (-1) edge node[above] {$\g_{-1 \to 0}$}  (0);
\path[Link,bend left] (0)  edge node[above] {$\g_{0 \to 1}$} (1);
\path[Link,bend left] (1)  edge node[above] {$\g_{1 \to 2}$} (2);
\path[Link,bend left] (2)  edge node[above] {} (3);
\path[Link,bend left] (3)  edge node[below] {}  (2);
\path[Link,bend left] (2)  edge node[below] {$\g_{2 \to 1}$}  (1);
\path[Link,bend left] (1)  edge node[below] {$\g_{1 \to 0}$}  (0);
\path[Link,bend left] (0)  edge node[below] {$\g_{0 \to- 1}$} (-1);
\path[Link,bend left] (-1) edge node[below] {$\g_{-1 \to- 2}$} (-2);
\path[Link,bend left] (-2) edge node[below] {} (-3);
\end{tikzpicture}
\end{aligned}
\end{equation}

\subsection{The thermodynamic limit for the linear chain on $\Z $ }\label{Section_TheThermodynamicLimitForTheLinearChainOnZ}

The master equation for the linear chain with \emph{two} open ends satisfies - similarly two the previous section - the (generalized) detailed balance condition of equation \eqref{Eq_StationarySolution_DetailedBalance} for 

\begin{equation} 
\begin{aligned}
\bX_*
:=
\left(
\left(
\prod\limits_{i\in \{0, \dots, z+1 \}}^{}
\frac{
\g_{i\to i-1}
}{
\g_{i\gets i-1}
}
\right)_{z \in \Z_{< \, 0}}, 1, 
\left(
\prod\limits_{i\in \{0, \dots, z-1 \}}^{}
\frac{
\g_{i\to i+1}
}{
\g_{i\gets i+1}
}
\right)_{z \in \Z_{> \, 0}}
\right)
=
\begin{cases}
\prod\limits_{i\in \{0, \dots, z-1 \}}^{}
\frac{
\g_{i\to i+1}
}{
\g_{i\gets i+1}
} &\text{  , if $z>0$} 
\\
\hspace*{7mm}1 &\text{  , if $z=0$} 
\\
\prod\limits_{i\in \{0, \dots, z+1 \}}^{}
\frac{
\g_{i\to i-1}
}{
\g_{i\gets i-1}
} &\text{  , if $z<0$} 
\end{cases}
\end{aligned}
\end{equation}

since we have: 

\begin{equation*}
\begin{aligned}
\underbrace{
X_*^{(z)}
}_{
\prod\limits_{i \in \{0, \dots, z{\color{red}\mp }1\}} \frac{\g_{i \to i+1}}{\g_{i \gets i+1}}
}
\, \g_{z \to z\pm 1 }
&=
\underbrace{
    \left(\prod\limits_{i \in \{0, \dots, z-1\}} \frac{\g_{i \to i+1}}{\g_{i \gets i+1}}
    \right)
    \,
    \frac{
    \g_{z \to z\pm 1 }
    }{
    \g_{z \gets z\pm 1 }
    }
}_{
    \prod\limits_{i \in \{0, \dots, z\pm 1 {\color{red}\mp }1\}} \frac{\g_{i \to i+1}}{\g_{i \gets i+1}}
}
\, \g_{z\pm 1\to z }
=
\underbrace{
    \left(
    \prod\limits_{i \in \{0, \dots, z\pm 1 {\color{red}\mp }1\}} \frac{\g_{i \to i+1}}{\g_{i \gets i+1}}
    \right)
}_{
X_*^{(z \pm 1)}
}
\, \g_{z\pm 1\to z }
 \\ &=
X_*^{(z \pm 1)} \, \g_{z\pm 1\to z } \text{  for } z \in \Z_{\color{red} {>\,  0}\atop{<\,  0} }
\end{aligned}
\end{equation*}

Thus, the linear chain with \emph{two} open ends behaves analogously to the linear chain with \emph{one} open end: 

By theorem \ref{Thm_CandidateForStationarySolutionOfIrreducibleNetworksWithDetailedBalance}, the only candidate for a stationary solution is $\frac{\bX_*}{\| \bX_* \|_1}$. When $\| \bX_* \|_1 = \infty$, no stationary solution exists, when $\| \bX_* \|_1 < \infty$, then the master equation of the linear chain with one open end is \emph{positive recurrent} (theorem \ref{Thm_CandidateForStationarySolutionOfIrreducibleNetworksWithDetailedBalance}) and its stationary solution can be approximated by corresponding finite subsystems in the thermodynamic limit (theorem \ref{Theorem_DetailedBalanceApproximatingStationarySolutionByFiniteSubsystems}).

Similarly to the linear chain with one open end, we can conclude:

\begin{equation}
\begin{aligned}
\lim\limits_{t \to \infty }  \, 
&\lim\limits_{|F| \to \infty } \, 
\bP^{F}(t \,| \, \bPF_0)
= 
\lim\limits_{t \to \infty } \bP(t \,| \, \bP_0) = 
\\ &\xlongequal[\text{positive recurrent}]{\text{irreducible + }}
\bP_*
\xlongequal[\text{detailed balance}]{\text{ thm } \ref{Theorem_DetailedBalanceApproximatingStationarySolutionByFiniteSubsystems}}
\lim\limits_{|F| \to \infty } \, 
\underbrace{
    \bPF_*
}_{
    \bP_\infty^{F}(\bPF_0)
}
\\ &= 
\lim\limits_{|F| \to \infty } \, 
\underbrace{
    \bP_\infty^{F}(\bPF_0)
}_{
    \lim\limits_{t \to \infty }  \, 
    \bP^{F}(t \,| \, \bPF_0)
}
=
\lim\limits_{|F| \to \infty } \, 
\lim\limits_{t \to \infty }  \, 
\bP^{F}(t \,| \, \bPF_0). 
\end{aligned}
\end{equation}

\begin{claim}[A simplified version for specific rates] $ $ \\ 

If for all $ n \in \N_0 $ we have 

\begin{minipage}{0.49\textwidth}
    \begin{equation*}
    \begin{aligned}
    \g_{n \to n+1} 
    &=
    \lambda_{n+1} \, x 
    \\
    \g_{n \gets n+1} 
    &=
    \lambda_{n} \, (1-x)
    \end{aligned}
    \end{equation*}
\end{minipage}\text{   and } \begin{minipage}{0.49\textwidth}
    \begin{equation*}
    \begin{aligned}
    \g_{-n \to -n-1} 
    &=
    \mu_{n+1} \, y 
    \\
    \g_{-n \gets -n-1} 
    &=
    \mu_{n} \, (1-y)
    \end{aligned}
    \end{equation*}
\end{minipage}

for some $x, y>0$, some sequences $(\lambda_n)_{n \in \N}, (\mu_n)_{n \in \N} \in l^1(N_0, \R_{> \, 0}) $ with $\lambda:= \lim\limits_{n \to \infty}\sqrt[n]{\lambda_{n}}$ and $ \mu:= \lim\limits_{n \to \infty}\sqrt[n]{\mu_{n}} $, respectively, then $\bX_* = \left(
\left( \frac{\mu_n}{\mu} \, \left(\frac{y}{1-y} \right)^{-z}\right)_{z \in Z_{< \, 0}}, 1, \left( \frac{\lambda_n}{\lambda_0} \, \left(\frac{x}{1-x} \right)^{z}\right)_{z \in Z_{> \, 0}}\right)  $ and we have $\bX_* \in l^1(\Z)$, if and only if $(x,y) \in \left(0, \frac{1}{1+\lambda}\right) \times \left(0, \frac{1}{1+\mu}\right)$.  
\end{claim}

\begin{proof}
We have 
\begin{equation*}
\begin{aligned}
X_*^{(z)}
&=
\frac{\g_{0 \rightsquigarrow z }}{\g_{0 \leftlsquigarrow z }}
=
\begin{cases}
\prod\limits_{i=0}^{z-1}
\overbrace{
\frac{\g_{i \to i+1 }}{\g_{i \gets i-1 }}
}^{
\frac{\lambda_{i+1 } \, x}{ \lambda_{i} \, x }
}
=
\frac{\lambda_{z}}{\lambda_{0}} \, \left(\frac{q}{1-q} \right)^{z} &\text{  , if $z>0$,  }
\\
\hspace*{7mm} 1  &\text{  , if $z=0$ }
\\
\prod\limits_{i=0}^{z+1}
\underbrace{
\frac{\g_{i \to i-1 }}{\g_{i \gets i-1 }}
}_{
\frac{\mu_{i+1 } \, y}{ \mu_{i} \, y }
}
=
\frac{\mu_{z}}{\mu_{0}} \, \left(\frac{y}{1-y} \right)^{z} &\text{  , if $z<0$,  }
\end{cases}
\end{aligned}
\end{equation*}

and therefore

\begin{equation*}
\begin{aligned}
\infty
&>
\| \bX_* \|_1
=
\sum\limits_{z \in \Z_{< \, 0}}
X_*^{(z)} 
+1+
\sum\limits_{z \in \Z_{> \, 0}}
X_*^{(z)} 
=
\sum\limits_{n \in \N_0}
\frac{\mu_{n}}{\mu_{0}} \, \left(\frac{y}{1-y} \right)^n
+ 1 + 
\sum\limits_{n \in \N_0}
\frac{\lambda_{n}}{\lambda_{0}} \, \left(\frac{x}{1-x} \right)^n
\\ &\xLongleftrightarrow[\text{test}]{\text{root}}
\begin{cases}
1
>
\limsup\limits_{n \to \infty}
\sqrt[n]{\frac{\lambda_n}{\lambda_0} \, \left(\frac{x}{1-x} \right)^n}
=
\lambda \, \left(\frac{x}{1-x} \right)
\text{ and }
\\
1
>
\limsup\limits_{n \to \infty}
\sqrt[n]{\frac{\mu_n}{\mu_0} \, \left(\frac{y}{1-y} \right)^n}
=
\mu \, \left(\frac{y}{1-y} \right)
\end{cases}
\\ \\ &\iff
(x, y)
\in 
\left(0, \frac{1}{1+\lambda}\right) \times \left(0, \frac{1}{1+\mu}\right)
\end{aligned}
\end{equation*}

\end{proof}

\subsection{Computing the kernel of the generator of the linear chain on $\Z$ `by foot'}\label{Section_ComputingTheKernelOfTheGeneratorOfTheLinearChainOnZByFoot}

\begin{claim}\label{Claim_TildeKern_SpanPStar_Z}
Let $ \widetilde{\Kern(\G)} := \{ \bX \in \R^\Z \,:\, \G \bX \text{ exists and }\G \bX = \bzero^{\Z}\}$ be defined as above. Then $ \widetilde{\Kern(\G)} =   \Span(\bX_*, \bY_*) $, where $ \bX_*$ and $ \bY_*$ are given by

\begin{equation}\label{PStar_LinearChainOnZ}
\begin{aligned}
\bX_*
&=
\left(
\prod\limits_{\alpha \leq z-1} \g_{\alpha \to \alpha + 1} 
\, \cdot \, 
\prod\limits_{\beta \geq z+1} \g_{\beta \to \beta - 1}
\right)_{z \in \Z }
%
\\ 
%
\bY_*
&=
\left(
\left( 
(-1)
\sum\limits_{j=1}^{-m} 
\frac{
\prod\limits_{\beta=m+1}^{-j}\g_{\beta \to \beta-1}
}{
\prod\limits_{\alpha=m}^{-j}\g_{\alpha \to \alpha+1}
}
\right)_{m \in \Z_{<0}}
, 0, 
\left( 
\sum\limits_{j=1}^{n}
\frac{
\prod\limits_{\alpha=j}^{n-1} \g_{\alpha \to \alpha + 1} 
}{
\prod\limits_{\beta=j}^{n} \g_{\beta \to \beta - 1}
}
\right)_{n \in \Z_{>0}}
\right) = \\ 
&=
\left(
\left( 
\frac{(-1)}{\g_{m \to m+1}}
\sum\limits_{j=1}^{-m} 
\prod\limits_{\epsilon=m+1}^{-j}
\frac{\g_{\epsilon \to \epsilon-1}}{\g_{\epsilon \to \epsilon+1}}
\right)_{m \in \Z_{<0}}
, 0 , 
\left( 
\frac{1}{\g_{n \to n-1}}
\sum\limits_{j=1}^{n} 
\prod\limits_{\epsilon=j}^{n-1}
\frac{\g_{\epsilon \to \epsilon+1}}{\g_{\epsilon \to \epsilon-1}}
\right)_{n \in \Z_{>0}}
\right)
\end{aligned}
\end{equation}

\end{claim}

\begin{proof}
In order to prove this, we consider the following:

\begin{itemize}
\vspace*{5mm}
\item[$\bX_* \in \widetilde{\Kern(\G)} $]
\begin{equation*}
\begin{aligned}
X_*^{(z-1)} \g_{z-1 \to z} 
+
X_*^{(z+1)} \g_{z+1 \to z} 
&=
\left( 
\prod\limits_{\alpha \leq z-2} {\color{TUDa_9b} \g_{\alpha \to \alpha+1}} \prod\limits_{\beta \geq z} \g_{\beta \to \beta-1}
\right)
{\color{TUDa_9b} \g_{z-1 \to z} }
+
\left(
\prod\limits_{\alpha \leq z} \g_{\alpha \to \alpha+1} \prod\limits_{\beta \geq z+2} {\color{blue} \g_{\beta \to \beta-1} }
\right) 
{\color{blue} \g_{z+1 \to z} } 
\\ &=
\left( 
\prod\limits_{\alpha \leq z-1} 
{\color{TUDa_9b} \g_{\alpha \to \alpha+1}} 
\prod\limits_{\beta \geq z} \g_{\beta \to \beta-1}
\right) 
+
\left(
\prod\limits_{\alpha \leq z} \g_{\alpha \to \alpha+1} \prod\limits_{\beta \geq z+1} {\color{blue} \g_{\beta \to \beta-1} }
\right) 
\\ &=
\underbrace{
\left( 
\prod\limits_{\alpha \leq z-1} \g_{\alpha \to \alpha+1} \prod\limits_{\beta \geq z+1} \g_{\beta \to \beta-1}
\right)
}_{
X_*^{(z)}
}
\, \cdot \, 
\left(\g_{z \to z+1} +  \g_{z \to z-1}  \right) 
\end{aligned}
\end{equation*}

The most direct way to argue would be the following: If $\bX_* \in l^1(\Z) $, then $ \frac{\bX_*}{\left \| \bX_* \right \|_1} \in \Kern(\G) $ and since we know from lemma \ref{Section_UniquenessAndStrictPositivityOfStationarySolutionsForIrreducibleNetworks} that the dimension of the kernel of an irreducible system is at most one dimensional, we can conclude that the kernel of the generator is spanned by 
$ \frac{\bX_*}{\left \| \bX_* \right \|_1} $, or in short: 

\begin{equation*}
\begin{aligned}
\left.
\begin{array}{c c}
    \left.
    \begin{array}{c }
    \bX_* \in \widetilde{\Kern(\G)} \\
    \bX_* \in l^1(\Z) 
    \end{array}     
    \right\} \Longrightarrow &
    \bzero \neq \frac{\bX_*}{\left \| \bX_* \right \|_1} \in \Kern(\G) \\
   & \text{dim}(\Kern(\G)) =1 
\end{array}
\right\} \Longrightarrow
\Kern(\G) = \Span\left(\frac{\bX_*}{\left \| \bX_* \right \|_1} \right)
\end{aligned}
\end{equation*}

For illustration purposes, we will continue to show that while $\bY_* \in \widetilde{\G} $, it  is not normalizable ($\| \bY_* \|_1 = \infty $), 

\item[$\bY_* \in \widetilde{\Kern(\G)} $]

Since $ \bY_* $ looks different on the negative numbers than it does on the positive ones, we divide the proof in to three steps: 

\begin{itemize}
\item[($n > 0$)]
\begin{equation*}
\begin{aligned}
Y_*^{(n-1)} \g_{n-1 \to n}
+
Y_*^{(n+1)} \g_{n+1 \to n}
&=
\sum\limits_{j=1}^{n-1}
\underbrace{
    \left(
    \frac{
    \prod\limits_{\alpha=j}^{n-2} \g_{\alpha \to \alpha + 1} 
    }{
    \prod\limits_{\beta=j}^{n-1} \g_{\beta \to \beta - 1}
    }
    \g_{n-1 \to n}
    \right)
}_{
    \frac{
    \prod\limits_{\alpha=j}^{n-1} \g_{\alpha \to \alpha + 1} 
    }{
    \prod\limits_{\beta=j}^{n-1} \g_{\beta \to \beta - 1}
    }
}
+
\sum\limits_{j=1}^{n+1}
\underbrace{
    \left(
    \frac{
    \prod\limits_{\alpha=j}^{n} \g_{\alpha \to \alpha + 1} 
    }{
    \prod\limits_{\beta=j}^{n+1} \g_{\beta \to \beta - 1}
    }
    \g_{n+1 \to n}
    \right)
}_{
    \frac{
    \prod\limits_{\alpha=j}^{n} \g_{\alpha \to \alpha + 1} 
    }{
    \prod\limits_{\beta=j}^{n} \g_{\beta \to \beta - 1}
    }
}
\\ &=
\sum\limits_{j=1}^{n-1}
\frac{
\prod\limits_{\alpha=j}^{n-1} \g_{\alpha \to \alpha + 1} 
}{
\prod\limits_{\beta=j}^{n-1} \g_{\beta \to \beta - 1}
}
\underbrace{
\left[
1 + \frac{\g_{n \to n+1}}{\g_{n \to n-1}}  
\right]
}_{
\frac{
\g_{n \to n+1} + \g_{n \to n-1}
}{
\g_{n \to n-1}
}
}
+ 
\underbrace{
\left( \frac{\g_{n \to n+1}}{\g_{n \to n-1}} + 1 \right)
}_{
\frac{
\g_{n \to n+1} + \g_{n \to n-1}
}{
\g_{n \to n-1}
}
}
\\ &=
\left( 
\g_{n \to n+1} + \g_{n \to n-1}
\right)
\underbrace{
\left(
\sum\limits_{j=1}^{n}
\frac{
\prod\limits_{\alpha=j}^{n-1} \g_{\alpha \to \alpha + 1} 
}{
\prod\limits_{\beta=j}^{n} \g_{\beta \to \beta - 1}
}
\right)
}_{
Y_*^{(n)}
}
\end{aligned}
\end{equation*}

\item[$(n=0)$]
\begin{equation*}
\begin{aligned}
\underbrace{
Y_*^{(-1)}
}_{
\frac{-1}{\g_{-1 \to 0}}
}
 \, \g_{-1 \to 0} 
+
\underbrace{
Y_*^{(1)}
}_{
\frac{1}{\g_{1 \to 0}}
}
\, \g_{1 \to 0}
&=
0
=
Y_*^{(0)} \, \g_{0 \to}. 
\end{aligned}
\end{equation*}

\item[($m < 0$)]
\begin{equation*}
\begin{aligned}
Y_*^{(m-1)} \g_{m-1 \to m}
+
Y_*^{(m+1)} \g_{m+1 \to n}
&=
(-1)\sum\limits_{j=1}^{-m-1}
\underbrace{
\frac{
\prod\limits_{\beta=m}^{-j} \g_{\beta \to \beta - 1}
}{
\prod\limits_{\alpha=m-1}^{-j} \g_{\alpha \to \alpha + 1} 
}
\g_{m-1 \to m}
}_{
\frac{
\prod\limits_{\beta=m}^{-j} \g_{\beta \to \beta - 1}
}{
\prod\limits_{\alpha=m}^{-j} \g_{\alpha \to \alpha + 1} 
}
}
-
\sum\limits_{j=1}^{-m+1}
\underbrace{
\frac{
\prod\limits_{\beta=m+2}^{-j} \g_{\beta \to \beta - 1}
}{
\prod\limits_{\alpha=m+1}^{-j} \g_{\alpha \to \alpha + 1} 
}
\g_{m+1 \to m}
}_{
\frac{
\prod\limits_{\beta=m+1}^{-j} \g_{\beta \to \beta - 1}
}{
\prod\limits_{\alpha=m+1}^{-j} \g_{\alpha \to \alpha + 1} 
}
} 
\\ &=
(-1)\sum\limits_{j=1}^{-m+1}
\frac{
\prod\limits_{\beta=m+1}^{-j} \g_{\beta \to \beta - 1}
}{
\prod\limits_{\alpha=m+1}^{-j} \g_{\alpha \to \alpha + 1} 
}
\underbrace{
\left[
1 + \frac{\g_{m \to m-1}}{\g_{m \to m+1}}  
\right]
}_{
\frac{
\g_{m \to m-1} + \g_{m \to m+1}
}{
\g_{m \to m+1}
}
}
+ 
\underbrace{
\left[
1 + \frac{\g_{m \to m-1}}{\g_{m \to m+1}}  
\right]
}_{
\frac{
\g_{m \to m-1} + \g_{m \to m+1}
}{
\g_{m \to m+1}
}
}
\\ &=
\left( 
\g_{m \to m-1} + \g_{m \to m+1}
\right)
\underbrace{
\left(
(-1)
\sum\limits_{j=1}^{-m}
\frac{
\prod\limits_{\beta=m+1}^{-j} \g_{\beta \to \beta - 1}
}{
\prod\limits_{\alpha=m}^{-j} \g_{\alpha \to \alpha + 1} 
}
\right)
}_{
Y_*^{(m)}
}
\end{aligned}
\end{equation*}

\item["$ \subseteq $"]

Let $ \bZ_* \in \widetilde{\Kern(\G)} $. We first show that the following quantity

\begin{equation} \label{Eq_DefOf_c}
\begin{aligned}
c^{(z)}
:=
Z_*^{(z)} \, \g_{z \to z+1}
-
Z_*^{(z+1)} \, \g_{z+1 \to z}
\end{aligned}
\end{equation}

is independent of $z \in \Z$: 

\begin{equation*}
\begin{aligned}
0
&=
\left(
\G \, \bZ_* 
\right)^{(z)}
=
Z_*^{(z-1)} \, \g_{z-1 \to z}
- Z_*^{(z)} (\g_{z \to z-1} + \g_{z \to z+1})
+
Z_*^{(z+1)} \, \g_{z+1 \to z} \\
\Longrightarrow 
c^{(z)}
&=
Z_*^{(z)} \, \g_{z \to z+1}
-
Z_*^{(z+1)} \, \g_{z+1 \to z}
=
Z_*^{(z-1)} \, \g_{z-1 \to z}
-
Z_*^{(z)} \, \g_{z \to z-1}
=
c^{(z-1)}. 
\end{aligned}    
\end{equation*}

Let us now fix $z_0 \in \Z$. We can now write $ Z_*^{(z_0 \pm 1)} $ in terms of $Z_*^{(z_0)} $ and $c$

\begin{equation}
\label{ConstantFlow}
\begin{aligned}
c
&=
Z_*^{(z_0)} \, \g_{z_0 \to z_0+1}
-
Z_*^{(z_0+1)} \, \g_{z_0+1 \to z_0 }
\\ \Longrightarrow
Z_*^{(z_0\pm1)}
&=
Z_*^{(z_0)}
\frac{
\g_{z_0 \to z_0\pm1}
}{
\g_{z_0\pm1 \to z_0}
}
\mp
\frac{
c
}{
\g_{z_0\pm1 \to z_0}
}
\end{aligned}
\end{equation}

and by induction, it follows that for all $z \in \Z$ and all $n \in \N$ we have

\begin{equation}\label{Eq_Claim_ExplicitFormOfZ}
\begin{aligned}
Z_*^{(z+n)}
&=
Z_*^{(z)} \, 
\frac{
\prod\limits_{\alpha=z}^{z+n-1} \g_{\alpha \to \alpha + 1}
}{
\prod\limits_{\beta=z+1}^{z+n} \g_{\beta \to \beta - 1}
}
- 
c
\sum\limits_{j=1}^{n}
\frac{
\prod\limits_{\alpha=z+j}^{z+n-1} \g_{\alpha \to \alpha + 1}
}{
\prod\limits_{\beta=z+j}^{z+n} \g_{\beta \to \beta - 1}
}  
\text{     and} \\ 
Z_*^{(z-n)}
&=
Z_*^{(z)} \, 
\frac{
\prod\limits_{\beta=z-n+1}^{z} \g_{\beta \to \beta - 1}
}{
\prod\limits_{\alpha=z-n}^{z-1} \g_{\alpha \to \alpha + 1}
}
+
c
\sum\limits_{j=1}^{n}
\frac{
\prod\limits_{\beta=z-n+1}^{z-j} \g_{\beta \to \beta - 1}
}{
\prod\limits_{\alpha=z-n}^{z-j} \g_{\alpha \to \alpha + 1}
} 
\end{aligned}
\end{equation}

(see lemma \ref{Proof_Eq_Claim_ExplicitFormOfZ} in the appendix for a proof). From this expression we can deduce that $\bZ_* = \frac{Z_*^{(0)}}{X_*^{(0)}} \, \bX_* - c \, \bY_* $ is a linear combination of $\bX_*$ and $\bY_*$, as the calculation in lemma \ref{lemma_ConsequencesFromConstantFlow} shows.

\end{itemize}

\end{itemize}

\end{proof}

We note that $\bY_* \notin l^1(\Z ) $, since for every $n \in \N $ we have

\begin{equation*}
\begin{aligned}
Y_*^{(n)}
&=
\frac{1}{\g_{n \to n-1}}
\sum\limits_{j=1}^{n} 
\prod\limits_{\epsilon=j}^{n-1}
\frac{\g_{\epsilon \to \epsilon+1}}{\g_{\epsilon \to \epsilon-1}}
\overset{j=n}{\geq}
\frac{1}{\g_{n \to n-1}}
\xlongrightarrow{n \to \infty}
\infty \text{      and } 
\\ 
Y_*^{(-n)}
&=
\frac{(-1)}{\g_{-n \to -n+1}}
\sum\limits_{j=1}^{n} 
\prod\limits_{\epsilon=-n+1}^{-j}
\frac{\g_{\epsilon \to \epsilon-1}}{\g_{\epsilon \to \epsilon+1}}
\overset{j=-n}{\leq}
-\frac{1}{\g_{-n \to -n+1}}
\xlongrightarrow{n \to \infty} 
-\infty. 
\end{aligned}
\end{equation*}

This means, that we have 

\begin{equation}\label{KernOfTheGenOnZ}\tag{KernOfTheGenOnZ}
\begin{aligned}
\Kern(\G) 
&= 
\underbrace{
\widetilde{\Kern(\G) }
}_{
\Span(\bX_*, \bY_*)
} 
\, \cap \,  l^1(\Z)
\xlongequal{\bY_* \notin l^1(\Z) }
\Span\left(
\frac{\bX_*}{\|\bX_*\|_1}
\right) 
\end{aligned}
\end{equation}

\subsection{ Differences to the linear chain with one open end } \label{SubSection_ DifferencesToTheLinearChainWithOneOpenEnd}

The two linear chains with one- and two open ends have in common that they admit at most one stationary sequence in $(\R_{> \, 0})^{\Omega} $ (due to the fact that they are both strongly connected and hence irreducible, compare lemma  \ref{Lemma_StrictPositivityForStationarySolutionsOfStronglyConnectedNetworks}). 

However, for the generator of the linear chain with two open ends, there is an additional, non-normalizable sequence with both positive- and negative elements, $\bY_{*} \in (\R)^{\Z} \backslash l^{1} (\Z) $, which satisfies the equation $\G \, \bY_{*} = \bzero $ component-wise.

In this subsection we briefly discuss where this additional sequence appears from, why there is non analogue for the linear chain with one open end.

Since this phenomena can only appear in the countable, infinite dimensional setting, it is instructive to look at a finite subsystem $F_{N} := \{-N, \dots, N\} $ for some $N \in \N $, where we have: 

\begin{equation}
\begin{aligned}
\G^{[F_N]} \, \bY_{*}^{[F_N]}
&\xlongequal[\ref{Lemma_TheAdditionalStationarySequenceForTheLinearChainWithTwoOpenEnds}]{\text{lemma}}
\bE_{-N} - \bE_{N} 
\xlongrightarrow[\text{pointwise}]{N \to \infty}
\bzero. 
\end{aligned}
\end{equation}

This demonstrated nicely what the expression `$ \bY_{*} $ is not normalizable' actually means: It means that it is not a vector, since it is not contained in the vector space $ l^{1}(\Z) $. The expression $ \bY_{*} $ defined in equation \eqref{PStar_LinearChainOnZ} as a whole `does not make sense', even though every element is a finite number. This is indeed a phenomena of \emph{infinite} dimensional vector spaces, since every \emph{finite} dimensional vector with an infinite norm, would have to have a faulty entry, e.g. $(1, \infty, 3) $. 

It also demonstrates the difference between \emph{pointwise} convergence and a convergence with respect to the $1$-norm, since looking at the expression $ \G \, \bY_{*} $ \emph{component-wise} is nothing else but looking at the limit $ p-\lim\limits_{N \to \infty} \, \bigl( \G \, \bY_{*}^{[F_N]} \bigr) $: 

\begin{equation}
\begin{aligned}
\lightning
&=
\G \, 
\underbrace{
    \bigl( p-\lim\limits_{N \to \infty} \bY_{*}^{F_N} \bigr)
}_{
    \bY_{*}
}
 \\ &\mathlarger{\boldsymbol{\neq}}
p-\lim\limits_{N \to \infty} \, 
\underbrace{
    \bigl( \G \, \bY_{*}^{[F_N]} \bigr) 
}_{
    \G^{[F_N]} \, \bY_{*}^{[F_N]} 
}
=
p-\lim\limits_{N \to \infty} \, 
\underbrace{
   \bigl( \G^{[F_N]} \, \bY_{*}^{[F_N]}  \bigr) 
}_{
    \bE_{-N} - \bE_{N} 
}
=
p-\lim\limits_{N \to \infty} \,  \bigl(\bE_{-N} - \bE_{N} \bigr)
=
\bzero. 
\end{aligned}
\end{equation}

So the expression $\bY_{*} $ can be regarded as a artifact of pointwise limit, where the `probability' mass is accumulated at the boundary.

For the linear chain with \emph{one} open end, this does not appear, since in that case, not even the pointwise limit vanishes, since one of the boundaries is contained in all -sufficiently large - finite subsystems: 

\begin{equation}
\begin{aligned}
\G^{[F_N]} \, \bY_{*}^{[F_N]}
&\xlongequal[\ref{Lemma_NoAdditionalStationarySequenceForTheLinearChainWithOneOpenEnd}]{\text{lemma}}
\bE_{0} - \bE_{N} 
\xlongrightarrow[\text{pointwise}]{N \to \infty}
\bE_{0}. 
\end{aligned}
\end{equation}
\section{The countable, infinite dimensional hypercube}\label{Chapter_TheInfiniteDimensionalHypercube}

In this section we have a look at the countable, infinite dimensional hypercube, in the case that the network fulfils the - generalized - detailed balance condition \eqref{Eq_StationarySolution_DetailedBalance}. When on top of that, the conditions of equations \eqref{Eq_SufficientConditionForDetailedBalance} and \eqref{Eq_SufficientConditionForPositiveRecurrence} are satisfied, it is possible to write down an explicit expression for the stationary solution.


The countable, infinite dimensional hypercube is defined as $\System_{[Q_\infty]} =(\Omega_{[Q_\infty]}, \Edge_{[Q_\infty]})$, with  and 

\begin{equation}
\begin{aligned}
\Omega_{[Q_\infty]} 
&:=
\{0, 1\}^\infty = \{ \boldsymbol{\omega} \in \{0,1\}^{\N} \,:\, |\{i \in \N \,:\, \omega^{(i)} \neq 0 \}|< \infty\} \text{ and }
\\
\Edge_{[Q_\infty]} 
&:=
\{(\alpha, \beta) \in (\Omega_{[Q_\infty]})^{2} \,:\, \text{there exists a natural number }i\in \N\,:\, \alpha-\beta = \pm \bE_i\}
\end{aligned}
\end{equation}

It can be interpreted as an infinite set of agents, each being able to occupy two states (spin up / down, or being susceptible / infected), while all but finitely many being in one of theses states. 

While it is always possible to write down the stationary solution of the \emph{finite} dimensional hypercube (compare equation 
\eqref{Eq_StationarySolutionForFiniteDimensionalHypercube}), in order to guarantee that the steady
states remain normalizable as the dimension tends to infinity, we need to make sure that the rates to higher dimensions decline rapidly enough: . 

With the stationary solution of \emph{finite} dimensional hypercubes be given by equation \eqref{Eq_StationarySolutionForFiniteDimensionalHypercube}, viewing $Q_{N}$ as a finite subsystem of $Q_{\infty}$ (with $\bP_*^{[Q_N]} = \bp^{[Q_N]} \times \{0\}^{\infty} $), we can compute the norm of the difference of two stationary solutions of two `consecutive' hypercubes to

\begin{equation}
\begin{aligned}
\Bigl\| 
\underbrace{
    \bP_*^{[Q_{N+1}]}
}_{
    \frac{(1, q^{c^{(N+1)}})}{1+ q^{c^{(N+1)}}} 
    \bigotimes
    \bp_*^{[Q_{N}]} 
    \times \{0\}^{\infty}
}
-
\underbrace{
    \bP_*^{[Q_{N}]}
}_{
    \be_1 \bigotimes \bp_*^{[Q_{N}]} \times \{0\}^{\infty}
}
\Bigr\|_1
&= 
\Bigl\|
\underbrace{
\frac{\left(1, q^{c^{(N+1)}}\right)}{1+ q^{c^{(N+1)}}} - (1,0)
}_{
\frac{(-q^{c^{(N+1)}}, q^{c^{(N+1)}})}{1+ q^{c^{(N+1)}}}
}
\Bigr\|_1
\cdot 
\underbrace{
\|\bp_*^{[Q_{N}]} \|_1 
}_{1}
\\ &= 
\frac{2 \, q^{c^{(N+1)}}}{1+ q^{c^{(N+1)}}}
\end{aligned}
\end{equation}

Now let $N, M \in \N$ be two arbitrary large natural numbers and $ \bP_*^{[Q_{N+M}]},\bP_*^{[Q_{N}]} $ be the stationary solutions of the hypercube $Q_{N+M}$ and $Q_{N}$, repectivly. 

\begin{equation}
\begin{aligned}
\Bigl\| 
\underbrace{
    \bP_*^{[Q_{N+M}]}
}_{
    \bp_*^{[Q_{N+M}]} \times \{0\}^{\infty}
}
-
\underbrace{
    \bP_*^{[Q_{N}]}
}_{
    \bigotimes\limits_{i=1}^{M} \be_1 \otimes \bp_*^{[Q_{N}]} \times \{0\}^{\infty}
}
\Bigr\|_1
&=
\Bigl\| 
\bp_*^{[Q_{N+M}]} - 
 \bigotimes\limits_{i=1}^{M} \be_1 \otimes \bp_*^{[Q_{N}]} 
\pm \sum\limits_{k=1}^{M} \bigotimes_{i=1}^{M-k} \be_1 \bigotimes \bp_*^{[Q_{N+k}]}
\times \{0\}^{\infty}
\Bigr\|_1
\\ & \leq 
\sum\limits_{k=0}^{M-1}
\left\|
\underbrace{
\bigotimes_{i=1}^{M-k} \be_1 \bigotimes \bp_*^{[Q_{N+k}]}
-
\bigotimes_{i=1}^{M-(k+1)} \be_1 \bigotimes \bp_*^{Q_{N+k+1}}
}_{
\bigotimes\limits_{i=1}^{M-(k+1)} \be_1 \bigotimes \left(
\e_1 \bigotimes \bp_*^{[Q_{N+k}]} - \bp_*^{Q_{N+k+1}}
\right) 
}
\right\|_1
\\ &=
\sum\limits_{k=0}^{M-1}
\underbrace{
\left\|
\bigotimes_{i=1}^{M-(k+1)} \be_1
\right\|_1
}_{
1
} \, \cdot \, 
\underbrace{
\left\|
\be_1 \bigotimes \bp_*^{[Q_{N+k}]} - \bp_*^{[Q_{N+k+1}]} 
\right\|_1
}_{
\frac{2 \, q^{c^{(N+k+1)}}}{1+ q^{c^{(N+k+1)}}}
}
\\ &\leq 
2 \, \sum\limits_{k=N+1}^{N+M} q^{c^{(k)}}. 
\end{aligned}
\end{equation}

For the right choice of $\bc$, $\left(\bP_*^{[Q_N]}\right)_{N \in \N} $ is a Cauchy sequence. 

Similarly to equation \eqref{Eq_SufficientConditionForPositiveRecurrence}, we get

\begin{equation*}
\begin{aligned}
\| 
\bP_*^{[Q_{N+M}]}
-
\bP_*^{[Q_{N}]}
\|_1
&=
2 \, \sum\limits_{k=N+1}^{N+M} q^{c^{(k)}} 
\\ &=
\begin{cases}
2 \, \sum\limits_{k=N+1}^{N+M} q^c
=
q^c \, M \xlongrightarrow{N, M \to \infty } \infty & \text{  , if } \bc = c \cdot \bone
\\  \\ 
2 \, \sum\limits_{k=N+1}^{N+M} q^{k}
=
q^{N+1} \sum\limits_{k=0}^{M-1} q^{k} \xlongrightarrow{N, M \to \infty } 0  & \text{  , if } \bc = (n)_{n \in \N_0}. 
\end{cases}
\end{aligned}
\end{equation*}



\section{A network without detailed balance}\label{Chapter_ExampleOfANetworkWithoutDetailedBalance}

In this section we look at a network without detailed balance as an example of a network, where the time limit of the \emph{infinite} dimensional systems exists, but the corresponding limiting state can not be approximated by stationary solutions of \emph{finite} subsystems in the thermodynamic limit. 

\begin{figure}[H]
\begin{center}
\begin{tikzpicture}[scale=2]
\centering
\node[State](-4)   at (-4,0) [circle,draw] {$\dots$} ;
\node[State](-3)   at (-3,0) [circle,draw] {$-3$} ;
\node[State](-2)   at (-2,0) [circle,draw] {$-2$} ;
\node[State](-1)   at (-1,0) [circle,draw] {$-1$} ;
\node[State](0)   at (0,0) [circle,draw] {$0$} ;
\node[State](1)   at (1,0) [circle,draw] {$1$} ;
\node[State](2)   at (2,0) [circle,draw] {$2$} ;
\node[State](3)   at (3,0) [circle,draw] {$3$} ;
\node[State](4)   at (4,0) [circle,draw] {$\dots$} ;
\path[Link] (-4) edge node[above] {} (-3);
\path[Link] (-3) edge node[above] {$\g_{-3 \to -2}$} (-2);
\path[Link] (-2) edge node[above] {$\g_{-2 \to -1}$} (-1);
\path[Link] (-1) edge node[above] {$\g_{-1 \to 0}$} (0);
\path[Link] (0) edge node[above] {$\g_{0 \to 1}$} (1);
\path[Link] (1) edge node[above] {$\g_{1 \to 2}$} (2);
\path[Link] (2) edge node[above] {$\g_{2 \to 3}$} (3);
\path[Link] (3) edge node[above] {} (4);
\path[Link, bend left] (1) edge node[above] {$\g_{1 \to -1}$} (-1);
\path[Link, bend left] (2) edge node[above] {$\g_{2 \to -2}$} (-2);
\path[Link, bend left] (3) edge node[above] {$\g_{3 \to -3}$} (-3);
\path[Link, bend left] (4) edge node[above] {} (-4);
\end{tikzpicture}
\caption{Example of a network without detailed balance}
\label{Fig_ExampleNetwork_NoDetailedBalance} 
\end{center}
\end{figure}
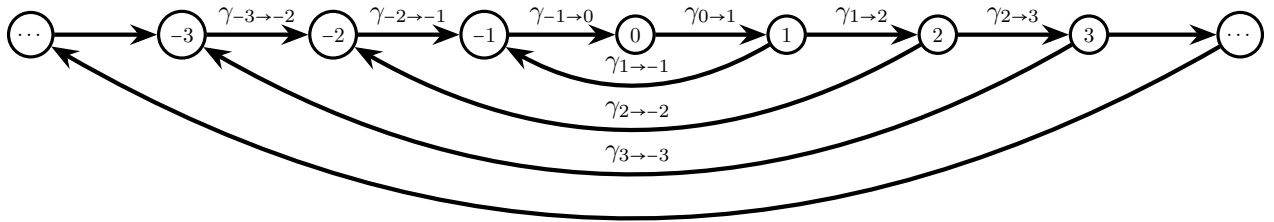

Let $F_N := \{-N, -N+1, \dots, -1, 0, 1, \dots, N\} $. The components $\bigl( \bP_*^{F_N} \bigr)^{(z)} $ of the stationary solution are given by the sum of the weights of all in-trees rooted in state $z$ \cite{mirzaev2013laplacian,fernengel2022obtaining}:

\begin{equation*}
\begin{aligned}
(P_*^{F_N})^{(-n)}
&\propto
\frac{1}{\g_{-n \to -n+1}} \, 
 \prod\limits_{i=-N}^{n-1} \g_{i \to i+1}
\prod\limits_{i=n}^{N-1} 
\underbrace{
    (\g_{i \to i+1} + \g_{i \to -i})
}_{
    \left(1+\frac{\g_{i \to - f(i) }}{\g_{i \to i+1}} \right) \, \g_{i \to i+1}
}
\,  \g_{N \to -N} 
\\
&\xlongequal{a_i = \frac{\g_{i \to - f(i) }}{\g_{i \to i+1}} }
\left(
\prod\limits_{i=-N}^{N-1} \g_{i \to i+1} 
\right) \, \g_{N \to -N} \,
 \frac{
 \left(
    \prod\limits_{i=n}^{N-1} (1 + a_{i} )
\right)
 }{\g_{-n \to -n+1}  } 
\\ &= 
 \frac{
    \left(
    \prod\limits_{i=-N}^{N-1} \g_{i \to i+1} 
    \right) 
    \, \g_{N \to -N} \,
     \left(
        \prod\limits_{i=n}^{N-1} (1 + a_{i} )
    \right)
 }{
  \left(
    \prod\limits_{k=1}^{n-1} (1 - a_{k} )
\right) \, 
\g_{-n \to -n+1}  } &\text{  for $n \in \{1, \dots, N\}$ }  
\\
(P_*^{F_N})^{(n)}
&\propto
\prod\limits_{i=-N}^{n-1} \g_{i \to i+1} \, \prod\limits_{i=n+1}^{N-1} 
 \underbrace{
    (\g_{i \to i+1} + \g_{i \to -i})
}_{
    \left(1+\frac{\g_{i \to - f(i) }}{\g_{i \to i+1}} \right) \, \g_{i \to i+1}
}
\, \g_{N \to -N} 
\\ &=
\left(
\prod\limits_{i=-N}^{N-1} \g_{i \to i+1} 
\right) \, \g_{N \to -N} \,
 \frac{
 \left(
    \prod\limits_{i=n+1}^{N-1} (1 + a_{i} )
\right)
 }{\g_{n \to n+1}  } 
\\ &= 
 \frac{
    \left(
    \prod\limits_{i=-N}^{N-1} \g_{i \to i+1} 
    \right) 
    \, \g_{N \to -N} \,
     \left(
        \prod\limits_{i=1}^{N-1} (1 + a_{i} )
    \right)
 }{
  \left(
    \prod\limits_{k=1}^{n} (1 - a_{k} )
\right) \, 
\g_{n \to n+1}  }
&\text{  for $n \in \{0, \dots, N-1\}$ } 
\\
(P_*^{F_N})^{(N)}
&\propto
\prod\limits_{i=-N}^{N-1} \g_{i \to i+1}
\\ &=
 \frac{
     \left(
    \prod\limits_{i=-N}^{N-1} \g_{i \to i+1} 
    \right) 
    \, \g_{N \to -N} \,
     \left(
        \prod\limits_{i=1}^{N-1} (1 + a_{i} )
    \right)
 }{
 \g_{N \to -N} \,     
 \left(
        \prod\limits_{i=1}^{N-1} (1 + a_{i} )
    \right)
 }  
\end{aligned}
\end{equation*}

\begin{equation*}
\begin{aligned}
\bP_*^{F_N}
&=
\frac{1}{\Zen^{F_N}}
\left(\bzero, 
    \left(
    \frac{1}{\g_{z \to z+1} \, \prod\limits_{k=1}^{|z|-1}(1+a_k)}   \right)_{z \in \{-N, \dots, -1\}},
\left(
    \frac{1}{\g_{z \to z+1} \, \prod\limits_{k=1}^{|z|}(1+a_k)}   
\right)_{z \in \{0, \dots, N-1\}},
\frac{1}{\g_{N \to -N} (1+a)^{N}}, \bzero
\right)
\end{aligned}
\end{equation*}

For a special choice of rates, in particular

\begin{equation}\label{Eq_ExampleNetworkWithoutDetailedBalance_ChoiceOfRates}
\begin{aligned}
\g_{z \to z+1}
&=
q^{|z|} \text{  for all $z \in \Z$ and some $q \in (0,1)$ and }
\\
\g_{k \to -k}
&=
a \, q^{k} \text{ for all $k\in \N$ and some $a>0$}, 
\end{aligned}
\end{equation}

we get:

\begin{equation*}
\begin{aligned}
\bP_*^{F_N}
&=
\frac{1}{\Zen^{F_N}}
\left(\bzero, 
\left( \frac{1+a}{[(1+a)\,q]^{|z|}}  \right)_{z \in \{-N, \dots, -1\}},
\left( \frac{1}{[(1+a)\,q]^{|z|}}   \right)_{z \in \{0, \dots, N-1\}}, 
\frac{1+a}{[(1+a)\,q]^{N} }, \bzero
\right)
\\ &=
\bP_*^{[F_N]} 
+ \frac{a}{[(1+a) \, q]^{N}} \,  \frac{\bE_{N} - \bP^{[F_N]} }{ \Zen^{F_N} }
\xlongrightarrow[\| \cdot \|_1 ]{N \to \infty}
\bP_*
\\ \text{   with }
\bP_*
&=
\frac{1}{\Zen}\,\left(
\left( \frac{1+a}{[(1+a)\,q]^{|z|}}  \right)_{z \in \Z_{< \, 0}},
\left( \frac{1}{[(1+a)\,q]^{|z|}}   \right)_{z \in \Z_{\geq \, 0}} 
\right), 
\\ 
\Zen
&=
\frac{
    (1+a) \, (1+q)
}{
    (1+a) \, q - 1
}. 
\end{aligned}
\end{equation*}

and $\bP_* \in l^1(\Omega) \iff (1+a) \, q > 1$.

It can easily be verified that $\bP_* $ indeed lies in the kernel of the generator matrix:

\begin{equation*}
\begin{aligned}
\bigl( \G \, \bP_* \bigr)^{(z)}
&=
\sum\limits_{s \in \Z} \bP_*^{(s)} \g_{s \to z} - \bP_*^{(z)} \g_{z \to } =
\\ &\xlongequal{z=-n<0}
\underbrace{
    P_*^{(-n-1)} 
}_{
    \frac{(1+a)}{[(1+a) \, q]^{n+1}}
}
\, 
\underbrace{
    \g_{-n-1 \to -n}
}_{ 
    q^{n+1}    
}
+
\underbrace{
    \g_{n \to -n}
}_{
    a \, q^{n}
}
-
\underbrace{
    P_*^{(n)}
}_{
    \frac{1}{[(1+a) \, q]^{n}}
} \, 
\underbrace{
    P_*^{(-n)}
}_{
    \frac{(1+a)}{[(1+a) \, q]^{n}}
}
\, 
\underbrace{
    \g_{-n \to -n+1}
}_{
    q^{n}
} 
\\ & =
\frac{1}{[(1+a) ]^{n}} + \frac{a}{[(1+a) ]^{n}} -  \frac{1+a}{[(1+a) ]^{n}}
=
0, 
\\ \\ &\xlongequal{z=0}
\underbrace{
    P_*^{(-1)}
}_{
    \frac{1}{q}
} \, 
\underbrace{
    \g_{-1 \to 0}
}_{
    q
}
- 
\underbrace{
    P_*^{(0)}
}_{
    1
}
\,
\underbrace{
    \g_{0 \to 1}
}_{
    q^{0}
}
=
0, 
\\ \\ &\xlongequal{z=n>0}
\underbrace{    
    P_*^{(n-1)}
}_{
    \frac{1}{[(1+a) \, q ]^{n-1}}
}
\,
\underbrace{
    \g_{n-1 \to n}
}_{
    q^{n-1}
}
-
\underbrace{
    P_*^{(n)}
}_{
    \frac{1}{[(1+a) \, q ]^{n}}
}
\,
\underbrace{
    \left[ \overbrace{\g_{n \to n+1}}^{q^n} + \overbrace{\g_{n \to -n}}^{a \, q^n} \right]
}_{
    (1+a) \, q^n
}
\\ &= 
\frac{1}{(1+a)^{n-1}} - \frac{(1+a)}{(1+a)^n}
=
0. 
\end{aligned}
\end{equation*}


Hence, we can conclude that the time limit of the network depicted in figure \ref{Fig_ExampleNetwork_NoDetailedBalance} exists (that is $ \bP(t \,|\, \bP_0) \xlongrightarrow[\| \cdot \|_1]{t \to \infty} \bP_*$) due to theorem \ref{Thm_TheTimeLimitForCountableInfiniteDimensionalMasterEquationsExistsForAllNetworksWhichSAreIrreducibleAndPositiveRecurrent}, since it is both \emph{irreducible} and - if the rates are as in equation \eqref{Eq_ExampleNetworkWithoutDetailedBalance_ChoiceOfRates} - \emph{positive recurrent}.  

However, the \emph{thermodynamic} limit $\lim\limits_{|F| \to \infty } \bP_\infty^{F} $ does \emph{not} exist, since for every finite set $ F_{\epsilon} $, it is possible to choose a larger finite set $ F \supseteq F_{\epsilon}, F \in \Fin(\Omega) $ such that the limiting state coincides with some unit vector, $\bP_\infty^{F} = \bE_N $ for some $N \in \N $.


\section{The linear chain with one open end and one trapping state}\label{Chapter_TheLinearChainWithOneOpenEndAndOneTrappingState}

In this section we look at a network, where the thermodynamic limit of finite subsystems exists, but the time limit of the \emph{infinite} dimensional system does not. In other words, the finite subsystems wrongly suggest a candidate for a stationary solution, since they a lacking an essential feature of the countable, infinite dimensional system.  

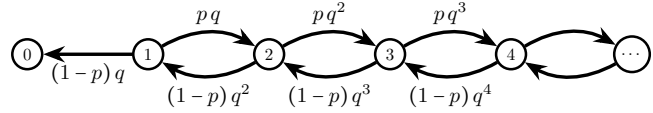
\begin{figure}[H]
\begin{center}
    \scalebox{0.8}{
    \begin{tikzpicture}[scale=2]
    \node[State](0) at (0,0) [circle,draw] {0} ;
    \node[State](1) at (1,0) [circle,draw] {1} ; 
    \node[State](2) at (2,0) [circle,draw] {2} ;
    \node[State](3) at (3,0) [circle,draw] {3} ;
    \node[State](4) at (4,0) [circle,draw] {4} ;
    \node[State](5) at (5,0) [circle,draw] {$\dots$} ;
    \path[Link,bend left] (1) edge node[above] {$p \, q $} (2);
    \path[Link,bend left] (2) edge node[above] {$p \, q^2$} (3);
    \path[Link,bend left] (3) edge node[above] {$p \, q^3$} (4);
    \path[Link,bend left] (4) edge node[above] {} (5);
    \path[Link,bend left] (5) edge node[below] { } (4);
    \path[Link,bend left] (4) edge node[below] {$ (1-p) \, q^4  $ } (3);
    \path[Link,bend left] (3) edge node[below] {$ (1-p) \, q^3  $} (2);
    \path[Link,bend left] (2) edge node[below] {$ (1-p) \, q^2  $} (1);
    \path[Link] (1) edge node[below] {$(1-p) \, q $} (0);
    \end{tikzpicture}
    }
\caption{ Network of a linear chain with an open end on one side and a trapping state on the other, for $q \in (0,1)$ and $p \in (\frac{1}{2}, 1) $.   }
\label{Fig_TheLinearChainWithOneOpenEndAndOneTrappingState} 
\end{center}
\end{figure}

First, we note that the subnetwork $\N $ is \emph{recurrent} for $p \leq \frac{1}{2} $ and \emph{transient} for $p > \frac{1}{2} $ (compare section \ref{Chapter_AFirstNonTrivialExample_TheLinearChainOnN}). 


Now, what is the effect of adding the trapping state $0$ ? First and most obviously, the network is no longer irreducible. The existence of the time limit now depends on the parameter $p \in (0,1)$: For $p \leq \frac{1}{2} $, the subnetwork $ \N $ is \emph{recurrent}, which means every trajectory visits the state $1$ infinitely often. With the trapping state $0$, this means that every trajectory is eventually being captured by $0$: 

\begin{equation*}
\begin{aligned}
\Prob(\text{trajectory is \emph{not} captured after the n-th return})
&=
(1-q_{1 \to 0})^n \xlongrightarrow{n \to \infty} 0. 
\end{aligned}
\end{equation*}

On the other hand, if $p > \frac{1}{2} $, then the state $1$ is visited only finitely often, which means that there is a non-zero chance of escaping `towards infinity', that is $ \bP(t)\xlongrightarrow[\text{pointwise}]{t \to \infty } (p_*, \bzero)  $ with some $p_* \in (0, 1) $ \cite{seneta2006non}. This means that the time limit of the \emph{infinite} dimensional system does \emph{not} exist.

However, for every $\epsilon > 0 $ we can choose a finite set $F_\epsilon \supseteq \{0, \dots, N_\epsilon \} $ such that $\bigl( \bP_0 \bigr) \Bigl(\{0, \dots, N_\epsilon \}\Bigr) \geq 1-\epsilon $, which results in $ \| \underbrace{\bP_*}_{\bE_0} - \bP_\infty^{F} \|_{1} < \epsilon $ for all finite sets $F \supseteq F_\epsilon$. This mean, that the \emph{thermodynamic limit} for \emph{finite} subsystems exists, even though the time limit for the \emph{infinite} dimensional system does not. This is of course due to the fact that in a finite system, the probability is restricted within a finite set and cannot escape `towards infinity'.


\section{Discussion and Conclusion}\label{Chapter_DiscussionAndConclusion}

We studied the differential Chapman-Kolmogorov equation - which is often referred to as the \emph{master equation} for a discrete, countable infinite system. While this is well understood in the finite dimensional case \cite{fernengel2022obtaining, van1992stochastic}, solving a \emph{countable, infinite dimensional} master equation is often avoided and one uses either approximations or entirely different methods.

While section \ref{Chaper_FromMarkovChainsToMasterEquations} serves as a reminder that master equations are closely related to continuous-time Markov chains, section \ref{Chaper_TheSolutionOfTheInfiniteDimensionalMasterEquation} focuses on the actual \emph{solution}, which can be shown to exist, if the rates $\g_{i \to j} $ satisfy certain assumptions. Moreover, it can be shown that - similarly to the finite dimensional case - the solutions remains a probability vector for all times and it is - element wise - strictly positive, for irreducible networks. 

From a physical perspective, one is not only interested in a unique solution, but also, if for a given finite time, one is able to approximate the behavior of an \emph{infinite} dimensional system by a corresponding \emph{finite} one. This \emph{thermodynamic limit} is addressed in section \ref{Chapter_TheThermodynamicLimitOfTheMasterEquations}. 

Another important aspect is the \emph{limiting} behavior of the master equation (see section \ref{Chapter_TheLongTermBehaviorOfAnInfiniteDimensionalMasterEquation}), that is the existence of the limit $t \to \infty $, as this determines the state that a system is after a long time. While this limit is always guaranteed in the finite dimensional case, this is not trivial in the Banach space $ l^1(\Omega) $. We were able to show that the limit $\lim\limits_{t \to \infty} \e^{t \, \G} \, \bP_0$ exists with respect to the $\| \cdot \|_1$-norm, if and only if the \emph{Cesaro mean} (a.k.a. the \emph{ergodic mean}) of the stochastic Markov transition matrix $\e^{\G} $ exists. This mean ergodic space can explicitly be characterized as the direct sum of of its stationary states and the closure of the image of the transition matrix minus the identity: 

\begin{equation*}
\begin{aligned}
l^1_\text{m.e.}(\Omega, \e^{\G} )
&=
\Kern(\Id - \e^{\G}) \oplus \overline{ \Image(\Id - \e^{\G})}. 
\end{aligned}
\end{equation*}

This mean ergodic space coincides with the whole space, if the corresponding system if both \emph{irreducible} and \emph{positive recurrent} (see theorem \ref{Thm_TheTimeLimitForCountableInfiniteDimensionalMasterEquationsExistsForAllNetworksWhichSAreIrreducibleAndPositiveRecurrent}).

Chapter \ref{Chapter_TheThermodynamicLimitOfTheStationarySolutionsOfTheMasterEquations} gives us a sufficient condition for the thermodynamic limit of \emph{stationary} solutions for finite subsystems, which is especially easy to verify if the system satisfies the - generalized - \emph{detailed balance} condition. Various examples are discussed in the sections \ref{Chapter_AFirstNonTrivialExample_TheLinearChainOnN}, \ref{Chapter_ASecondExample_TheLinearChainOnZ}, \ref{Chapter_TheInfiniteDimensionalHypercube}, \ref{Chapter_ExampleOfANetworkWithoutDetailedBalance} and \ref{Chapter_TheLinearChainWithOneOpenEndAndOneTrappingState}. 

A natural question to ask, is if and how it is possible to generalize our results. When we were considering the time limit for countable, infinite dimensional system, we required our system to be both \emph{irreducible} and \emph{positive recurrent}. While positive recurrence is required to guarantee the stationary solution (compare lemma \ref{Lemma_PositiveRecurrenceForCTMC}), the condition of irreducibility can be relaxed to some extend. Since $ P^{(\omega)} (t \,|\, \bP_0) \xlongrightarrow{t \to \infty } 0 $ for every transient state $\omega \in \Omega $ \cite{seneta2006non}, the theorems still work if the network has finitely many transient states. For arbitrary many transient state, further research is needed. 

It is also possible to have more than one strongly connected component (or more precise: more than one \emph{minimal absorbing set} \cite{fernengel2022obtaining}). Instead of requiring positive recurrence for the whole network, we now need to have positive recurrence for every minimal absorbing set. However, we note that minimal absorbing sets need not exist for infinite networks.

Table \ref{Table_AllPossibleCases} summarizes all possible cases, of whether the consecutive limits of the time- and system size exist $\checkmark$, \emph{not} exist $\times$ and coincide.

\begin{table}[H] 
\begin{tabular}{|c|c|c|c|}
\hline
& existence of $ \lim\limits_{t \to \infty}  \lim\limits_{|F| \to \infty} \bP^{F}(t) $ & 
existence of $ \lim\limits_{|F| \to \infty} \lim\limits_{t \to \infty}  \bP^{F}(t) $ & if they exist, do the limits coincide ? \\ 
\hline 
(i) & $\checkmark$ & $\checkmark$ & $\checkmark$ \\ 
\hline
(ii) & $\checkmark$ & $\checkmark$ & $\times$  \\ 
\hline
(iii) & $\times $ & $\times$ &  \\ 
\hline
(iv) & $\checkmark $ & $\times$ &  \\ 
\hline
(v) & $\times $ & $\checkmark$ &  \\ 
\hline
\end{tabular}
\caption{Listing all possible cases of the (non)-existence of the two consecutive limits: time- and system size  }
\label{Table_AllPossibleCases}
\end{table}

While this paper has mainly focused on case (i), that is providing sufficient conditions, which guarantee that both limits exist and commute, there are interesting cases, when they do not. The case (iii) is covered by the linear chain with one open end, in the case that $ \sum\limits_{n \in \N } \frac{\g_{0 \rightsquigarrow n}}{\g_{0 \leftlsquigarrow n} } = \infty $ (see section \ref{Chapter_AFirstNonTrivialExample_TheLinearChainOnN}). 

Section \ref{Chapter_ExampleOfANetworkWithoutDetailedBalance} addresses case (iv): A network without detailed balance, where the probability flows from left to right and is `reshuffled' back to the left side (see figure \ref{Fig_NetworkWithoutDetailedBalance} for a reminder). The time limit for this countable, infinite dimensional system exists for a suitable choice of rates (see equation \eqref{Eq_ExampleNetworkWithoutDetailedBalance_ChoiceOfRates}) - but the thermodynamic limit of the stationary solutions does not, due to topological reasons (compare figure \ref{Fig_Network_FlowingInACircle}).

This view can be challenged, by noting that it is possible to choose another increasing sequence of finite subsets $(F_N)_{N \in \N} \subseteq \Fin(\Omega)$, such that $\bP_\infty^{F_N} \xlongrightarrow{N \to \infty} \bP_\infty $. So one could say, that this is just matter of how we defined our thermodynamic limit. We further believe, that if the time limit for the countable, infinite dimensional system exists, it is always possible to choose such an increasing sequence of finite sets but this is just a hypothesis.

\begin{figure}[H]
\begin{center}
\begin{subfigure}{0.49\textwidth}
   \subcaption{}
\label{Fig_NetworkWithoutDetailedBalance_a}
\scalebox{0.6}{
    \begin{tikzpicture}[scale=2]
    \node[State](-3)  at (-3,0) [circle,draw] {\dots} ;
    \node[State](-2)  at (-2,0) [circle,draw] {$-2$} ;
    \node[State](-1)  at (-1,0) [circle,draw] {$-1$} ;
    \node[State](0)   at (0,0) [circle,draw] {$0$} ;
    \node[State](1)   at (1,0) [circle,draw] {$1$} ;
    \node[State](2)   at (2,0) [circle,draw] {$2$} ;
    \node[State](3)   at (3,0) [circle,draw] {\dots} ;
    %
    \path[Link] (-3) edge node[above] {$ q^{3} $} (-2);
    \path[Link] (-2) edge node[above] {$ q^{2} $} (-1);
    \path[Link] (-1) edge node[above] {$ q^{1} $} (0);
    \path[Link] (0)  edge node[above] {$ q^{0} $} (1);
    \path[Link] (1)  edge node[above] {$ q^{1} $} (2);
    \path[Link] (2)  edge node[above] {$ q^{2} $} (3);
    %
    \path[Link, bend left] (1) edge node[above] {$a \, q^{1}$} (-1);
    \path[Link, bend left] (2) edge node[above] {$ a \, q^{2}$} (-2);
    \path[Link, bend left] (3) edge node[above] {$ a \, q^{3}$} (-3);
    \end{tikzpicture}
}
\end{subfigure}
\begin{subfigure}{0.49\textwidth}
   \subcaption{}
   \label{Fig_NetworkWithoutDetailedBalance_b}
   \begin{equation*}
\begin{pmatrix}
\vdots & \vdots  & \vdots  & \vdots           & \vdots \\
0      & 0       & 0       & 0                & 0  \\
-q^{2} & 0       & 0       & 0                & a \, q^{2}  \\ 
 q^{2} & - q^{1} &         & a \, q^{1}       & 0 \\ 
 0     &  q^{1}  & -q^{0}  & 0                & 0 \\
\vdots & \vdots  &  q^{0}  & -(1+a)  \, q^{1} & 0 \\ 
\vdots & 0       & 0       & q^{1}            &  -(a \, + q^{2}) \\ 
\vdots &  0      &  0      &    0             &  q^{2}
 \end{pmatrix}
\end{equation*}
\end{subfigure} \\ 
\begin{subfigure}{0.99\textwidth}
\centering
    \subcaption{}
   \label{Fig_NetworkWithoutDetailedBalance_c}
    \scalebox{0.6}{
    \begin{tikzpicture}[scale = 1]
    \draw [draw=black] (-5.0,-5.5) rectangle +(13.0,7.5);
    \node (tLessThanInfty)      at (0,1.5)  { \Large $ t< \infty $} ;
    \node (OmegaLessThanInfty)  at (-3.1,0) { \huge $    F \subseteq \Omega  \atop |F| < \infty  $} ;
    \node (OmegaEqualsInfty)    at (-3.1,-3.3)   { \Large $ |\Omega|= \infty $} ;
    \node (tEqualsInfty)        at (4.5,1.5)   {\Large $ ``t= \infty" $} ;
    \draw[->, ultra thick]  (-2.0 , 1.8) -- (-2.0,-4.3) ; 
    \draw[->, ultra thick]  (-4.7 , 1.0) -- ( 6.0, 1.0);  
    \draw( 7.0 , 1.0) node[font=\large] {time limit} ;
    \draw(-2.1 ,-4.6) node[font=\large] {thermodynamic } ;
    \draw(-2.1 ,-5.0) node[font=\large] {limit} ;
    %
    \node (bPFt)    at (0,0)  {$\Large \bP^{F}(t) $} ;
    \node(bPFinfty) at (4.7,0)  {$ \Large \bP^{F}_\infty $};
    \node(bPt)      at (0,-3.4) {$ \Large\bP(t)  $};
    \node(bPinfty)  at (4.3,-3.4) {$ \Large  \bP_\infty   $};
    \node(bP)       at (4.7,-2.7) {$ \Huge  \lightning  $};
    \draw[very thick] (2.3, -2.7) -- (5.9, -3.3) ; 
    \path[Link] (bPFt) edge node[above] {$ t \to \infty $} (bPFinfty);
    \path[Link] (bPt) edge node[above]  {$ t \to \infty $} (bPinfty);
    \path[Link] (bPFt) edge node[left]  {$ |F| \to \infty $} (bPt);
    \path[Link] (bPFinfty) edge node[right]  {$ |F| \to \infty $} (bP);
    \end{tikzpicture}
    }
  \end{subfigure}
  
\caption{ Example of a network without detailed balance \ref{Fig_NetworkWithoutDetailedBalance_a}, together with the corresponding generator \ref{Fig_NetworkWithoutDetailedBalance_b} where the time limit of the countable, infinite dimensional network exists for a suitable choice of rates, but the \emph{thermodynamic} limit does not. This means that the diagram \ref{Fig_NetworkWithoutDetailedBalance_c} does not commute.   }  
\label{Fig_NetworkWithoutDetailedBalance} 
\end{center}
\end{figure}
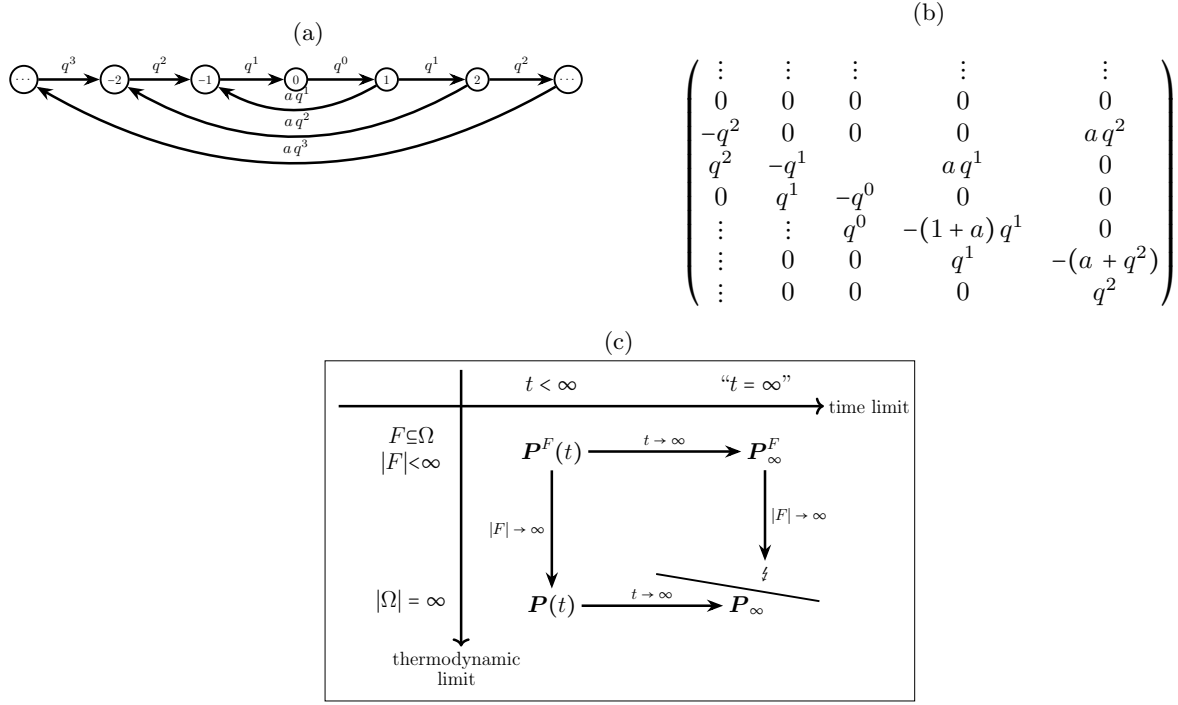

In section \ref{Chapter_TheLinearChainWithOneOpenEndAndOneTrappingState} we discussed case (v): A linear chain with an open end on one side and a trapping state on the other. The thermodynamic limit of the stationary solutions of the finite system always exists ($ \lim\limits_{|F| \to \infty} \lim\limits_{t \to \infty} \bP^{F}(t) = \bE_0 $), but the time limit of the \emph{infinite} network only exist, when the subnetwork $\N \subsetneq \Omega \hat{=} \N_0 $ is \emph{recurrent} (that is, when $p \leq \frac{1}{2} $, compare figure \ref{Fig_TheLinearChainWithAnOpenEndOnOneSideAndATrappingStateOnTheOther}). This has certain similarities with the phase transition in the Ising model, where the two limits (magnetization to zero and system size to infinity) only commute, when the inverse temperature is below a critical values: $\beta \leq \beta_c$ \cite{friedli2018statistical}. 

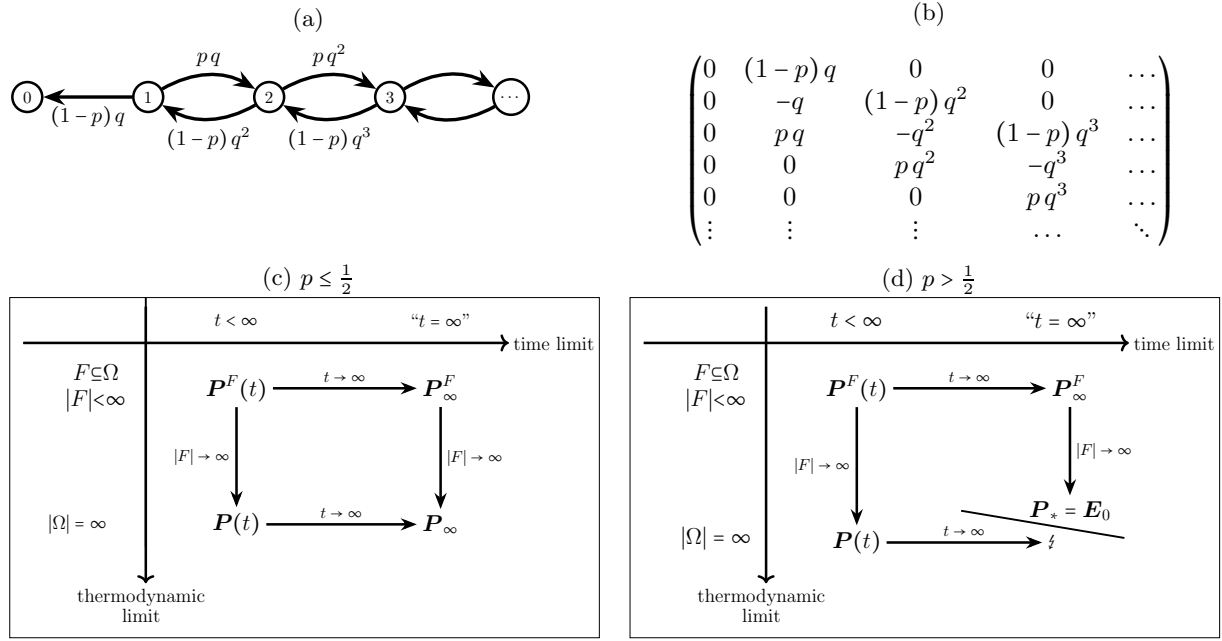
\begin{figure}[H]
\begin{center}
\begin{subfigure}{0.49\textwidth}
   \subcaption{}
      \label{Fig_TheLinearChainWithAnOpenEndOnOneSideAndATrappingStateOnTheOther_a}
    \scalebox{0.8}{
    \begin{tikzpicture}[scale=2]
    \node[State](0) at (0,0) [circle,draw] {0} ;
    \node[State](1) at (1,0) [circle,draw] {1} ; 
    \node[State](2) at (2,0) [circle,draw] {2} ;
    \node[State](3) at (3,0) [circle,draw] {3} ;
    \node[State](4) at (4,0) [circle,draw] {$\dots$} ;
    \path[Link,bend left] (1) edge node[above] {$p \, q $} (2);
    \path[Link,bend left] (2) edge node[above] {$p \, q^2$} (3);
    \path[Link,bend left] (3) edge node[above] {} (4);
    \path[Link,bend left] (4) edge node[below] { } (3);
    \path[Link,bend left] (3) edge node[below] {$ (1-p) \, q^3  $} (2);
    \path[Link,bend left] (2) edge node[below] {$ (1-p) \, q^2  $} (1);
    \path[Link] (1) edge node[below] {$(1-p) \, q $} (0);
    \end{tikzpicture}
    }
\end{subfigure}
\begin{subfigure}{0.49\textwidth}
   \subcaption{}
   \label{Fig_TheLinearChainWithAnOpenEndOnOneSideAndATrappingStateOnTheOther_b}
\begin{equation*}
\begin{pmatrix}
0  & (1-p) \, q & 0            & 0             & \dots \\  
0  & -q         & (1-p) \, q^2 & 0             & \dots \\ 
0  &  p \, q    & -q^2         & (1-p) \, q^3  & \dots  \\ 
0  &  0         & p \, q^2     & -q^3          & \dots   \\
0  &    0       &    0         &  p \, q^3     & \dots  \\ 
\vdots & \vdots    & \vdots                         &  \dots                         & \ddots
\end{pmatrix}
\end{equation*}
\end{subfigure} 
\begin{subfigure}{0.49\textwidth}
   \subcaption{$p\leq \frac{1}{2} $}
      \label{CommutatingDiagrammForLinearChainWithOneOpenEndAndOneTrappingState_a}
    \scalebox{0.6}{
        \begin{tikzpicture} [scale = 1]
        \draw [draw=black] (-5.0,-5.5) rectangle +(13.0,7.5);
        \node (tLessThanInfty)      at (0,1.5)  {$ \mathlarger{\mathlarger{ t< \infty } }$} ;
           \node (OmegaLessThanInfty)  at (-3.1,0) { \huge $    F \subseteq \Omega  \atop |F| < \infty  $} ;
        \node (tEqualsInfty)        at (4.5,1.5)   {$   \mathlarger{\mathlarger{ ``t= \infty" }}$} ;
        \node (OmegaEqualsInfty)    at (-3.5,-3)   {$   \mathlarger{\mathlarger{ |\Omega|= \infty }} $} ;
        \draw[-, very thick]  (-2  , 2.0) -- (-2.0,-3.5);
        \draw[-, very thick] (-4.7 , 1.0) -- ( 6.0, 1.0);
        \draw[->, ultra thick]  (-2.0 , 1.8) -- (-2.0,-4.3) ; 
        \draw[->, ultra thick]  (-4.7 , 1.0) -- ( 6.0, 1.0);  
        \draw( 7.0 , 1.0) node[font=\large] {time limit} ;
        \draw(-2.1 ,-4.6) node[font=\large] {thermodynamic } ;
        \draw(-2.1 ,-5.0) node[font=\large] {limit} ;
        \node (bPAt)    at (0,0)  {$ \Large \bP^{F}(t) $} ;
        \node(bPAinfty) at (4.5,0)  {$  \Large \bP^{F}_\infty  $};
        \node(bPt)      at (0,-3) {$ \Large \bP(t)  $};
        \node(bPinfty)  at (4.5,-3) {$ \Large \bP_\infty   $};
        \path[Link] (bPAt) edge node[above]  {$ t \to \infty $} (bPAinfty);
        \path[Link] (bPt) edge node[above]  {$ t \to \infty $} (bPinfty);
        \path[Link] (bPAt) edge node[left]  {$ |F| \to \infty $} (bPt);
        \path[Link] (bPAinfty) edge node[right]  {$ |F| \to \infty$} (bPinfty);
        \end{tikzpicture}    
        }
\end{subfigure}
\begin{subfigure}{0.49\textwidth}
   \subcaption{$p > \frac{1}{2} $}
   \label{CommutatingDiagrammForLinearChainWithOneOpenEndAndOneTrappingState_b}
    \scalebox{0.6}{
    \begin{tikzpicture}[scale = 1]
    \draw [draw=black] (-5.0,-5.5) rectangle +(13.0,7.5);
    \node (tLessThanInfty)      at (0,1.5)  { \Large $ t< \infty $} ;
    \node (OmegaLessThanInfty)  at (-3.1,0) { \huge $    F \subseteq \Omega  \atop |F| < \infty  $} ;
    \node (OmegaEqualsInfty)    at (-3.1,-3.3)   { \Large $ |\Omega|= \infty $} ;
    \node (tEqualsInfty)        at (4.5,1.5)   {\Large $ ``t= \infty" $} ;
    \draw[->, ultra thick]  (-2.0 , 1.8) -- (-2.0,-4.3) ; 
    \draw[->, ultra thick]  (-4.7 , 1.0) -- ( 6.0, 1.0);  
    \draw( 7.0 , 1.0) node[font=\large] {time limit} ;
    \draw(-2.1 ,-4.6) node[font=\large] {thermodynamic } ;
    \draw(-2.1 ,-5.0) node[font=\large] {limit} ;
    %
    \node (bPFt)    at (0,0)  {$\Large \bP^{F}(t) $} ;
    \node(bPFinfty) at (4.7,0)  {$ \Large \bP^{F}_\infty $};
    \node(bPt)      at (0,-3.4) {$ \Large\bP(t)  $};
    \node(bPinfty)  at (4.3,-3.4) {$ \Huge  \lightning   $};
    \node(bP)       at (4.7,-2.7) {$ \Large \bP_* = \bE_0  $};
    \draw[very thick] (2.3, -2.7) -- (5.9, -3.3) ; 
    \path[Link] (bPFt) edge node[above] {$ t \to \infty $} (bPFinfty);
    \path[Link] (bPt) edge node[above]  {$ t \to \infty $} (bPinfty);
    \path[Link] (bPFt) edge node[left]  {$ |F| \to \infty $} (bPt);
    \path[Link] (bPFinfty) edge node[right]  {$ |F| \to \infty $} (bP);
    \end{tikzpicture}
    }
  \end{subfigure}

\caption{Network \ref{Fig_TheLinearChainWithAnOpenEndOnOneSideAndATrappingStateOnTheOther_a} and generator matrix \ref{Fig_TheLinearChainWithAnOpenEndOnOneSideAndATrappingStateOnTheOther_b} of a linear chain with an open end on one side, and a trapping state on the other, for $p, q \in (0, 1)$.  
The corresponding time limit and the thermodynamic limit are depicted for different values of $p$:  For $p \leq \frac{1}{2} $ (figure \ref{CommutatingDiagrammForLinearChainWithOneOpenEndAndOneTrappingState_a}) the diagram commutes, while for $p > \frac{1}{2} $ (figure \ref{CommutatingDiagrammForLinearChainWithOneOpenEndAndOneTrappingState_b}) the limit for the countable, infinite dimensional system does not exist: $ \lim\limits_{t \to \infty} \lim\limits_{|F| \to \infty}  \bP^{F}(t) = \lightning $, but the thermodynamic limit of the stationary solutions of \emph{finite} systems, does:  $\lim\limits_{|F| \to \infty} \lim\limits_{t \to \infty} \bP^{F}(t) = \bE_0  $.   
}
\label{Fig_TheLinearChainWithAnOpenEndOnOneSideAndATrappingStateOnTheOther} 
\end{center}
\end{figure}

The case that we are still missing is case (ii). The question of whether it is possible that both consecutive limits exist, but do \emph{not} agree, is still open and while we could not find a network displaying this behavior, we also can not rule out that such a network exists.

Another interesting question would be to ask for sufficient condition, of when the iterated limit $\lim\limits_{{t \to \infty} \atop {|F| \to \infty}} \bP^{F}(t \,|\, \bPF_0)$ exists. From a physical point of view, this can be interpreted as gradually increasing the system as time passes. The existence of the iterated limit would mean that no matter how fast or slowly or in which way you increase your system, it would eventually approach a steady state. However, this too is a question for which we need future research. 




\appendix

\section{Acknowledgment }

We thank Robert Haller from the TU Darmstadt and Benno Rumpf from Dallas for the interesting and constructive discussions. 

\section{Nomenclature }

\begin{center}
\begin{longtable}{| c | c | }

\hline 
$\N_{\leq N}$ & $\{1, \dots, N\}$ \\ 
\hline 
$ \ProbStates_n $ & the set of - finite dimensional - probability vectors \\ 
& $ (\R_{\geq 0})^{n} \cap B_{r=1}^{\| \cdot \|_1} (0) $ \\
          \hline 
$ \ProbStates $ & the set of - infinite dimensional - probability vectors \\ 
& $ (\R_{\geq 0})^{\N_0} \cap l^1(\N_0) \cap B_{r=1}^{\| \cdot \|_1} (0) $ \\
          \hline 
$c_{0,0} $      & the set of eventually vanishing sequences  \\
 & $\{X \in \R^\N \,:\, \exists N \in \N\,:\, \forall n\geq N X^{(n)} = 0\}$ \\ 
                    \hline 
$ \ProbStates_0 $ & the set of infinite dimensional probability vectors with a finite support \\
 & $ (\R_{\geq 0})^{\infty} \cap B_{r=1}^{\| \cdot \|_1} (0) $  \\ 
          \hline
$ \bp= \left(p^{(1)}, \dots, p^{(|\Omega|)} \right) $ & 
 $ \in  \ProbStates_{|\Omega|} $, for $ |\Omega| < \infty $ \\
e.g. $\bp_0, \bp_*, \bp_\infty(\bp_0) $ &  finite dimensional probability vectors \\
\hline 
        $ \bP= \left(P^{(n)}\right)_{n \in \N} $ &  $ \in \ProbStates $ \\ 
        e.g. $ \bP_0, \bP_*, \bP_\infty(\bP_0)  $ & 
        infinite dimensional probability vectors \\ 
\hline 
         $ \bp_0, \, \bP_0$ & initial states \\ 
\hline 
          $ \bp_*, \, \bP_*$ & stationary states \\ 
\hline 
          $ \bp_\infty(\bp_0)= \lim\limits_{t\to \infty} \e^{t \, \G} \bp_0 $ & limiting state if exists, \\
          $ \bP_\infty(\bP_0) = \lim\limits_{t\to \infty} \e^{t \, \G} \bP_0 $ & possibly depending on the initial condition   \\          
\hline 
         $\bP_0^{[F]} $  & (see \eqref{InitialStateFinSubNet}) \\
\hline 
         $\GF $  & (see \eqref{GenFinSubNet}) \\
\hline 
          $\bP^{F} (t\,|\, \bP_0^{[F]}) $ & Solution of the following initial value problem (see \eqref{MEqFinSubNet})  \\
          & $\frac{d}{dt} \bP^{F}(t) = \GF \, \bP^{F}(t)$ \\ 
          & $ \bP^{F}(t=0) = \bP^{[F]}_0$ \\ 
\hline 
$ \MM $ & $ \{M \subseteq \Omega \,:\,  M \text{ is a minimal absorbing set}\} $ \\
\hline
$ \g_{i \to } $ & $ \sum\limits_{j \in \Omega} \g_{i \to j} $ \\ 
\hline 
$\Fin(\Omega)$ & the set of all finite subsets of $\Omega$, that is  \\ 
& $\Fin(\Omega) := \{F \subseteq \Omega \,:\, |F|< \infty \}$ \\ 
\hline
$ (\bP) (A)$ & $ \sum\limits_{a \in A} P^{(a)}$ \\ 
\hline 
$ \NP:=\NP(\Omega) $ & the set of all null sets of $\Omega$ with respect to the probability measure $\bP$   \\
& $\NP(\Omega):= \{N \subseteq \Omega \,:\, (\bP)(N)=0 \} $ \\ 
\hline
$ \sigma(A) $ & the spectrum of the operator $A$ \\
& in the finite dimensional case, this is equivalent to the set of eigenvalues  \\ 
\hline
$ g(A, \lambda)$ & geometric multiplicity of the matrix $A$ and the eigenvalue $\lambda \in \sigma(A)$ \\ 
\hline
$ S_{g}(\G) $ & the \emph{spectral gap} of the generator matrix $ \G $ \\
 & $ S_{g}(\G) := \text{dist}\bigl( \Re(\sigma(\G) \backslash \{0\}), 0 \bigr) := \max\{\Re(\lambda) \,:\, \lambda \in \sigma(\G)\backslash \{0\} \}   $\\ 
\hline
$ j(A, \lambda, d)  $ & size of the Jordan block of the matrix $A$ to the eigenvalue $\lambda \in \sigma(A)$ of number $d \in \N_{\leq g(A, \lambda)} $ \\
\hline
$ j(A, \lambda)$ & $ \max\left\{ j(A, \lambda, d) \,:\, d \in \{1, \dots, g(A, \lambda)\} \right\}$ \\
& 
the size of the \emph{largest} Jordan block of the matrix $A$ to the eigenvalue $\lambda \in \sigma(A) $ \\
\hline
$J(A) $  &  $ \max\left\{ j(A, \lambda) \,:\, \lambda \in \sigma(A) \right\}$ \\
& the size of the largest Jordan block of the matrix $A$ \\
\hline
$ \Omega^{\infty} $ & the set of finite sequences of $\Omega$ \\
& $ \Bigl\{ (\omega^{(1)}, \omega^{(2)}, \dots) \in \Omega^{\N} \,:\, |\{i \in \N \,:\, \omega^{(i)}\neq 0 \}|<\infty \Bigr\} $ \\
\hline 
$ \bE_i $ & $ \Bigl( \underbrace{0, \dots, 0}_{i-1 \text{ times }}, 1, 0, 0, \dots \Bigr) $ \\ 
$ \bE_i' \in l^1(\Omega, \R)'  $ & $\bE_i'(\bX) = X^{(i)}$ \\ 
\hline
$ \bone  $ & $(1, 1, \dots ) = \{1\} ^{\Omega} $ \\ 
\hline 
$ \bone' : l^1(\Omega) \to \R  $  &  $\bone'(\bX) = \sum\limits_{i \in \Omega} X^{(i)} $ \\
\hline 
$ \CMean_M (\bX) = \CMean_M^{Q} (\bX) $ &  finite Cesaro mean
$\frac{1}{M} \, \sum\limits_{i=0}^{M-1} (Q)^{i}( \bX)$ 
\\ 
\hline 
 $\lme $ & mean ergodic subspace of $l^1(\Omega )$ where $ \lim\limits_{M \to \infty} \CMean_M(\bX)$ exists \\
 \hline 
$\ltb $ & subspace of $l^1(\Omega) $ where the long-term behavior $\lim\limits_{t \to \infty} \e^{t \, \G} \bX $ exists \\ 
\hline 
$ \bX \geq \bzero $ & $\bX \in (\R_{\geq \, 0})^{\Omega} \backslash \{\bzero\}$ \\
\hline 
$ \bX > \bzero $ & $\bX \in (\R_{> \, 0})^{\Omega} $ \\
\hline 
$\BL(l^1(\Omega))$ & the set of bounded, linear operators on $ l^1(\Omega) $ \\
\hline 
$ \| \bX \|_1 $ &  $ \sum\limits_{i \in \Omega} |X^{(i)}| $ \\ 
\hline 
$ \| A \|_{1}^\text{(op)} $ & $ \sup\limits_{ \|X\|_1 = 1} \| A \bX \|_1 = \sup\limits_{j \in \Omega} \sum\limits_{i \in \Omega} |A^{(i, j)}| $ \\ 
\hline 
$ \| A \|_{1,1}^\text{(op)} $ & $ \sum\limits_{i,j \in \Omega} |A^{(i,j)}| $  \\
\hline 
    \end{longtable}
\end{center}



\section{From Markov chains to master equations}\label{Chaper_FromMarkovChainsToMasterEquations}

Markov chains are special kinds of stochastic processes that satisfy the so-called \emph{Markov property}, which states that the probability for the next event depends only on the current state and not on the previous history. Both the underlying state space $\Omega$ as well as the time can be discrete or continuous. From here on, we will assume, that $ \Omega $ is a \emph{discrete, countable infinite} space. 

\subsection{Discrete-time Markov chains}\label{Sec_DiscreteTimeMarkovchains}

Let $(X_n)_{n \in \N_0}$ be a discrete-time Markov chain on $\Omega$, where the transition probabilities are given by 

\begin{equation}
\begin{aligned}
Q^{(i,j)}
&:=
q_{j \to i}
:=
\Prob\left(X_{n+1} = i \,|\, X_n = j\right) \\
&=
\Prob\left( X_{n+1} = i \,|\, X_n = j, \dots, X_0 = j_0 \right) \text{  (compare \cite{privault2013understanding})}. 
\end{aligned}
\end{equation}

This specifies the transition matrix $Q \in [0,1]^{\Omega \times \Omega}$, which is a column-stochastic, countable matrix that satisfies 

\begin{equation}
\begin{aligned}
Q^{(i,j)} 
&\geq
0 \text{     , for all } i,j \in \Omega \text{  and }
\\
\sum\limits_{i \in \Omega}^{} Q^{(i,j)}
&=
1 \text{     , for all } j \in \Omega  . 
\end{aligned}
\end{equation}

Our system $ \System $ is fully described by a \emph{countable, directed, weighted graph } $ \System = (\Omega, \Edge, q)$ (which we will call a \emph{network}), where the nodes are given by the set of states $ \Omega $ and the edges $ \Edge \subseteq \Omega \times \Omega $ (also called \emph{links}) are a set of \emph{ordered} pairs of states which indicate the transition between these states. The strength of a link is given by its weight function 

\begin{equation}
\begin{aligned}
q :\;\; \Edge &\to [0,1] \\
(j,i) &\mapsto q_{j \to i}. 
\end{aligned}
\end{equation}

When there is no transition from state $j$ to state $i$, the associated transition probability vanishes, $q_{j \to i} = 0 $. 

To keep the notation simple, we do not distinguish between the index $n \in \N $ and the state $\omega_n \in \Omega$ with index $n$: 
\begin{equation}
\begin{aligned}
\omega_n \;\; \, 
\widehat{=}& \,
\hspace*{5mm} n
 \\
\Omega%
=
\{\omega_1, \omega_2. \dotsc,  \} \, 
\widehat{=}& \,
\{1, 2. \dotsc, \}
=
\N. 
\end{aligned}
\end{equation}

We call $ \bP \in (\R)^{\Omega} $ a \emph{probability vector} if $\bP \in [0, 1]^{\Omega} $ and $ \|\bP\|_1=1 $, that is a vector with non-negative entries that sum up to one. When $\bP(n)$ is a probability vector describing the probability distribution of $\Omega$ at the time $n\in \N_0$, then $ \bP(n+1):= Q \, \bP(n)$ is a probability distribution at the next time step $n+1$. 

\begin{MyDef}[\textbf{Stationary solution}]
\end{MyDef}

A probability vector is called \emph{stationary solution} of the Markov chain if its probability distribution does not change after another time step $(Q \, \bP = \bP)  $, that is, if it is an eigenvector of the transition matrix $ Q $ to the eigenvalue $ \lambda=1 $.


\subsection{Continuous-time Markov chains}\label{Sec_ContinuousTimeMarkovchains}

Given a time homogeneous, continuous-time Markov process $(X_t)_{t\geq \, 0} $ on the state space $\Omega $, then we have for all time points $ t_1 < \dots < t_n < t_{n+1}$ and all states $j_1, \dots, j_{n-1}, i, j \in \Omega$

\begin{equation} \label{Def_ContTimeMarkovChains}
\begin{aligned}
\Prob(X_{t_{n+1}}&=i \,|\,X_{t_n} = j, \dots, X_{t_0} = j_0 )
\xlongequal{\text{Markov}}
\Prob(X_{t_{n+1}}=i \,|\,X_{t_n} = j ). 
\end{aligned}
\end{equation}

Time homogeneous means that the transition matrix, whose components are defined as 

\begin{equation} \label{Def_TransitionMatrix_1}
\begin{aligned}
Q^{(i,j)}(t_2,t_1) &:= \Prob(X_{t_2} = i \,|\, X_{t_1} = j)
\text{  for } t_1 < t_2 
\end{aligned}
\end{equation}
depends only on the time difference, that is 

\begin{equation}\label{Def_TimeHomogeneous}
\begin{aligned}
Q^{(i,j)}(t_2, \, t_1) 
\xlongequal{\text{time homogeneous }} 
Q^{(i,j)}(t_2-t_1,0)=: Q^{(i,j)}(t_2-t_1). 
\end{aligned}
\end{equation}

We note that the stochastic matrices $(Q(t))_{t \geq 0} $ satisfy the functional equation $Q(t_1+t_2) = Q(t_2) \cdot Q(t_1) $, making it a semigroup: 

\begin{equation}\label{Equation_ExponentialPropertyOfTimeDependentTransitionMatrix}
\begin{aligned}
\bigl[
Q(t_2) \, &Q(t_1)
\bigr]^{(i, j)}
=
\sum\limits_{k \in \Omega} \, 
\underbrace{
Q^{(i,k)}(t_2)
}_{
\Prob(X_{t_2}=i \,|\, X_0 = k)
} \;\;\; 
\underbrace{
Q^{(k.j)}(t_1)
}_{
\Prob(X_{t_1}=k \,|\, X_0 = j) 
} \\
&\xlongequal{\text{Def } Q}
\sum\limits_{k \in \Omega} 
\underbrace{
\Prob(X_{t_2}=i \,|\, X_0 = k)
}_{
\Prob(X_{t_1 + t_2}=i \,|\, X_{t_1} = k)
}\, 
\Prob(X_{t_1}=k \,|\, X_0 = j) \\
%
&\xlongequal[\text{Markov property}]{\text{time homogeneity + } }
\sum\limits_{k \in \Omega} \, 
\underbrace{
\Bigl(
\Prob(X_{t_1 + t_2}=i \,|\, X_{t_1} = k, X_0 = j)
\Bigr)
}_{
\frac{
\Prob(X_{t_1 + t_2}=i,  X_{t_1} = k, X_0 = j)
}{
\Prob( X_{t_1} = k, X_0 = j)
}
} \;\;\;
\underbrace{
\Bigl(
\Prob(X_{t_1}=k \,|\,  X_0 = j)
\Bigr)
}_{
\frac{
\Prob(X_{t_1}=k ,  X_0 = j)
}{
\Prob( X_0 = j)
}
} \\
%
&=
\sum\limits_{k \in \Omega} \, 
\underbrace{
\left(
\frac{
\Prob(X_{t_1 + t_2}=i,  X_{t_1} = k, X_0 = j)
}{
\Prob( X_0 = j)
}
\right) 
}_{
\Prob(X_{t_1 + t_2}=i,  X_{t_1} = k \,|\,  X_0 = j)
} \\
%
&=
\Prob(X_{t_1 + t_2}=i \,|\,  X_0 = j) 
=
[Q(t_2+t_1)]^{(i, j)}. 
\end{aligned}
\end{equation}


When we additionally assume that the function 
\begin{equation}
\begin{aligned}
Q \;:\; \R_{\geq \, 0} 
&\to
\BL\Bigl(l^1(\Omega ) \Bigr)
\\
t
&\mapsto 
Q_t:= Q(t)
\end{aligned}
\end{equation}

is continuously differentiable, then $(Q_{t})_{t\geq \, 0}$ is even a \emph{uniformly continuous} $C_0$-semigroup, with the bounded generator $\G := \dot{Q}(t=0) $, that is  $Q_t = \e^{t \, \G } $, with $ \| \G \|_{1}^\text{op} < \infty$. We even have $Q(h) \xlongrightarrow[\| \cdot  \|_{1}^\text{op}]{h \to 0+} \Id $, not only pointwise, but also with respect to the operator norm \cite{engel2000one,batkai2017positive}.

When starting st state $j \in \Omega$, the probability for being at a different state $i \in \Omega \backslash\{j\} $ after a small time $h$ grows linearly with $h$, with the proportionality constant $\g_{j \to i} $ called the \emph{rate}. Since probabilities must sum up to one, we get: 

\begin{equation*}
\begin{aligned}
\Prob\left(
X_h = i \,|\, X_0 = j
\right)
=
\begin{cases}
h \, \g_{j \to i}  \hspace*{3mm}+ \smallO(h) &\text{  , if } i\neq j
\\
1-h \, \g_{j \to }+ \smallO(h) &\text{  , if } i= j. 
\end{cases}
\end{aligned}
\end{equation*}

This allows us to compute the entries of the generator matrix, as: 

\begin{equation}\label{Eq_GeneratorMatrix}\tag{generator}
\begin{aligned}
\G^{(i, j)}
&=
\begin{cases}
\g_{j\to i}  &\text{  , for } i \neq j 
\\
-\mathlarger{\sum}\limits_{k \in \Omega}^{ } \g_{j\to k} &\text{  , for } i = j\, .
\end{cases}
\end{aligned}
\end{equation}

Then the function 
\begin{equation*}
\begin{aligned}
\R_{\geq \, 0}
&\to 
l^1(\Omega)
\\
t
&\mapsto
\e^{t \, \G} \, \bP_0 
\end{aligned}
\end{equation*}

is the solution of the following differential equation, called the \emph{mater equation}: 

\begin{equation}\label{Eq_MasterEq}\tag{master eq.}
\begin{aligned}
\dot{\bP}(t) 
&=
\G \, \bP(t)
\\
\bP(t=0) 
&=
\bP_0
\end{aligned}
\end{equation}

Similar to the discrete-time case, we have an associated state transition network, where the states correspond to the nodes and the links (indicating transitions between the states) correspond to the edges. In contrast to the discrete-time case, the weights are only required to be non-negative (in contrast to lying in the interval $ [0,1] $) and we explicitly exclude self-loops: 

\begin{equation}
\begin{aligned}
\g : \;\; \Edge &\to \R_{\geq \, 0} \\
(j,i) &\mapsto \g_{j \to i}, \text{  with } \g_{i \to i } = 0. 
\end{aligned}
\end{equation}

This is due to the fact that for the continuous-time Markov chains we have \emph{transition rates} instead of \emph{transition probabilities}, which means that the probability for remaining in some state is always strictly positive: $ \Prob(X_t=i \,|\, X_0 = i)>0 $ for all times $t \geq 0 $ and all states $ i \in \Omega $.

The solution is given by $ \bP(t \,|\, \bP_0) := \e^{t \, \G } \, \bP_0$, with the initial state $\bP(t=0) = \bP_0 \in [0,1]^{\Omega} $ and the \emph{solution operator}

\begin{equation} \label{SolutionOperator}
\begin{aligned}
\e^{t \, \G } := \mathlarger{\sum}\limits_{k\in \N_0} \frac{\G^k \, t^k}{k!} = \lim\limits_{n\to \infty} \left( \Id + \frac{t \, \G}{n} \right)^n. 
\end{aligned}
\end{equation}
\section{Markov chains: Existence and uniqueness of stationary solutions }\label{Chapter_MarkovChains}

In the following section we look at stationary solutions for both discrete- and continuous-time Markov chains. While an irreducible network guarantees the \emph{uniqueness} of a stationary solution, \emph{positive recurrence} - on an irreducible network - is equivalent to its existence. This is done by writing the \emph{expected visiting numbers} for each state in a vector and showing that this vector is - after normalization - the only candidate for a stationary solution.

\subsection{Discrete-time Markov chains}\label{Sec_DiscreteTimeMarkovChains}

\begin{MyDef}[{Return time for DTMC} (DTMC)] \label{Def_ReturnTimeForDTMC}  $ $ \\

For a discrete-time Markov chain (\emph{DTMC}) $(X_n)_{n \in \N }$ starting from a state $X_0 = \state{\omega_0} \in \Omega $ (that is $\Prob(X_0 = \state{\omega_0}) =1$), we define the \textbf{return time} $T_R \in \N $ as follows:

\begin{equation}\label{Def_ReturnTimeTo_DTMC} \tag{return time, DTMC}
\begin{aligned}
T_R
&:=
T_R^{[DTMC]}\bigl(\state{\omega_0} \bigr)
&:=
\inf\limits_{k \in \N} 
\{
X_k = \state{\omega_0} = X_0 
\} 
\end{aligned}
\end{equation}

with $\inf(\emptyset) := \infty $.  Apparently we have $ X_0 = \state{\omega_0} = X_{T_R} $. 

\end{MyDef}

\begin{MyDef}[{(Positive) recurrence for DTMC} ] \label{Def_PositiveRecurrenceForDTMC}  $ $ \\ 

A DTMC is called \emph{recurrent}, if it almost surely returns to the initial state, that is 

\begin{equation}
\begin{aligned}
1
&=
\Prob(T_R < \infty)
=
\sum\limits_{n \in \N } \Prob\left(T_R = n \,|\, X_{0} = \omega_{0}  \right) 
\end{aligned}
\end{equation}

A DTMC is called \emph{positive recurrent}, if it is recurrent and its expectation value is finite: 

\begin{equation}
\begin{aligned}
\infty
&>
E[T_R,|\, X_{0} = \omega_{0}]
=
\sum\limits_{n \in \N} n \, \cdot \, \Prob\left(T_R = n \,|\, X_{0} = \omega_{0}  \right) 
=
\sum\limits_{n=1 }^{\infty}
\Prob(T_R \geq n) 
\end{aligned}
\end{equation}

\end{MyDef}

\begin{MyDef}[{Expected visiting number for DTMC}] \label{Def_VisitingNumberForDTMC} $ $ \\ 

Further, we define for a DTMC starting at state $ \state{\omega_0} \in \Omega $ the \emph{expected visiting number} $ \NN{\omega} $ to state $\omega \in \Omega$, before the return time $ T_R $, namely 

\begin{equation} \label{Def_DTMC_ExpectedVisitingNumber} \tag{expected visiting number}
\begin{aligned}
\NN{\omega}
&:=
E\left[
\sum\limits_{k=0}^{T_R-1} 1_{\{X_k = \state{\omega} \,|\, X_0 = \state{\omega_0} \}}
\right]
=
E\left[
\sum\limits_{k=1}^{T_R} 1_{\{X_k = \state{\omega} \,|\, X_0 = \state{\omega_0} \}}
\right] 
=
E\left[
\sum\limits_{k \in \N } 1_{\{X_k = \state{\omega}, k \leq T_R \,|\, X_0 = \state{\omega_0} \}}
\right] 
\\ &\xlongequal[\text{convergence}]{\text{monotone}}
\sum\limits_{k \in \N} \,
\underbrace{
E \left[
1_{\{X_k = \state{\omega}, \, T_R \geq k  \,|\, X_0 = \state{\omega_0} \}}
\right ]
}_{
\Prob \left(
X_k = \state{\omega}, \, T_R \geq k  \,|\, X_0 = \state{\omega_0} 
\right)
}
= 
\sum\limits_{k \in \N} \, 
\Prob \left(
X_k = \state{\omega}, \, T_R \geq k  \,|\, X_0 = \state{\omega_0} 
\right)
\geq 
0. 
\end{aligned}
\end{equation}

We point out that for the initial state $\omega_0$ we have $\NN{\omega_0} = 1$. 
\end{MyDef}

The reason for these two definitions is that they appear in the definition of the expected return time:

\begin{lemma}[Expected return time and expected visiting number] \label{Lemma_ExpectedReturnTimeAndExpectedVisitingNumber} $ $ \\ 
The expected return time equals the sum over all states of the expected visiting number of that state, that is $E
\left[
T_R \,|\, X_0 = \state{\omega_0}
\right] = \sum\limits_{\omega \in \Omega} \,
\NN{\omega}
 $. 
\end{lemma}

\begin{proof}
\begin{equation} \label{Eq_DTMC_ExpectedReturnTime} \tag{expected return time}
\begin{aligned}
E
\left[
T_R \,|\, X_0 = \state{\omega_0}
\right]
&=
\sum\limits_{n \in \N } \, 
\underbrace{
    \Prob 
    \left(
    T_R \geq n  \,|\, X_0 = \state{\omega_0}
    \right) 
}_{
    \sum\limits_{\omega \in \Omega}
    \Prob 
    \left(
    T_R \geq n, \, X_n = \state{\omega}  \,|\, X_0 = \state{\omega_0}
    \right) 
} 
=
\underbrace{
    \sum\limits_{n \in \N } \, 
    \sum\limits_{\omega \in \Omega}
}_{
    \sum\limits_{\omega \in \Omega} \, 
    \sum\limits_{n \in \N } 
} 
\, 
\underbrace{
    \Prob 
    \left(
    T_R \geq n, \, X_n = \state{\omega}  \,|\, X_0 = \state{\omega_0}
    \right) 
}_{
    E \left[
    1_{\{
    T_R \geq n, \, X_n = \state{\omega}  \,|\, X_0 = \state{\omega_0}
    \}}
    \right]
}
\\ &\xlongequal[\text{convergence}]{\text{monotone}}
\sum\limits_{\omega \in \Omega} \,
    E 
    \left[
    \sum\limits_{n\in \N }
    1_{
    \{T_R \geq n, 
    X_n = \state{\omega}  \,|\, X_0 = \state{\omega_0}
    \}
    }
    \right]
=
\sum\limits_{\omega \in \Omega} \,
\underbrace{
    E 
    \left[
    \sum\limits_{n=1}^{T_R}
    1_{
    \{
    X_n = \state{\omega}  \,|\, X_0 = \state{\omega_0}
    \}
    }
    \right]
}_{
    \NN{\omega}
}
\\ &= 
\sum\limits_{\omega \in \Omega} \,
\NN{\omega}
= 
\left\|
\underbrace{
\left(
\NN{\omega}
\right)_{\omega \in \Omega }
}_{
\bN 
}
\right\|_1
=
\left\|
\bN
\right\|_1. 
\end{aligned}
\end{equation}
    
\end{proof}

\begin{thm}[Vector of expected visiting numbers is candidate for stationary solution of DTMC ]\label{Thm_VectorOfExpectedVisitingNumbersIsCandidateForStationarySolutionOfDTMC} $ $ \\ 
\end{thm}

For an irreducible DTMC, every entry of $\bN $ is strictly positive and the eigenspace of the transition matrix $Q$ to the eigenvalue $\lambda = 1$ is spanned by $ \frac{\bN}{\| \bN \|_1 } $. This means that $ \frac{\bN}{\| \bN \|_1 } $ is the only possible probability vector that is a stationary solution of the Markov chain, that is 

\begin{equation*}
\begin{aligned}
\Kern(\Id - Q)
&=
\Span\left(\frac{\bN}{\|\bN\|_1}\right) \text{and }
\bN 
\in
(\R_{> \, 0})^{\Omega} \text{  , that is }
\\
\left \{\frac{\bN}{\left \| \bN \right \|_1} \right \}
&= \begin{cases}
\Kern(\Id - Q) \cap \ProbStates \cap (\R_{> \, 0})^{\Omega} 
&\text{, if } \| \bN \|_1 < \infty
  \text{ OR} 
\\ 
\{\bzero\} &\text{, if } \| \bN \|_1 = \infty
\end{cases}
\end{aligned}
\end{equation*}

\begin{proof}

We know from section \ref{Section_UniquenessAndStrictPositivityOfStationarySolutionsForIrreducibleNetworks} that the dimension of the kernel of $Q - 1 \, \cdot \Id $ is at most one dimensional. On the other hand, we have 

\begin{equation}
\begin{aligned}
\left(
Q \, \bN
\right)^{(\beta)}
&=
\sum\limits_{\alpha \in \Omega} 
\underbrace{
Q^{(\beta, \alpha)}
}_{
q_{\alpha \to \beta}
}\, \NN{\alpha}
=
\sum\limits_{\alpha \in \Omega} 
\underbrace{
    \NN{\alpha}
}_{
    \sum\limits_{k \in \N_0 } \, 
    \Prob\left(
    X_k = \state{\alpha}, \, T_R > k \, | \, X_0 = \state{\omega_0}
    \right)
}
\, q_{\alpha \to \beta} = 
\\ &=
\underbrace{
\sum\limits_{\alpha \in \Omega} \, 
\sum\limits_{k \in \N_0 } \, 
}_{
\sum\limits_{k \in \N_0 } \,
\sum\limits_{\alpha \in \Omega} \, , 
}
\Prob\left(
X_k = \state{\alpha}, \, T_R > k \, | \, X_0 = \state{\omega_0}
\right)
\;
\underbrace{
\overbrace{
\Prob \left(
X_{k+1} = \state{\beta} \,|\, X_{k} = \state{\alpha}
\right)
}^{
q_{\alpha \to \beta}
}
}_{
\Prob \left(
X_{k+1} = \state{\beta} \,|\, X_{k} = \state{\alpha}, \, X_0 = \state{\omega_0}, \, T_R>k
\right)
} = 
\\ &\xlongequal[\text{chain}]{\text{Markov}}
\sum\limits_{k \in \N_0 } \,
\sum\limits_{\alpha \in \Omega} \, 
\underbrace{
\Prob\left(
X_k = \state{\alpha}, \, T_R > k \, | \, X_0 = \state{\omega_0}
\right)
\,
\Prob \left(
X_{k+1} = \state{\beta} \,|\, X_{k} = \state{\alpha}, \, X_0 = \state{\omega_0}, \, T_R>k
\right)
}_{
\Prob\left(
X_k = \state{\alpha}, \, X_{k+1} = \state{\beta}, \, T_R > k \, | \, X_0 = \state{\omega_0}
\right)
} = 
\\ &=
\sum\limits_{k \in \N_0 } \,
\underbrace{
\sum\limits_{\alpha \in \Omega} \, 
\Prob\left(
X_k = \state{\alpha}, \, X_{k+1} = \state{\beta}, \, T_R > k \, | \, X_0 = \state{\omega_0}
\right)
}_{
\Prob\left(
X_{k+1} = \state{\beta}, \, 
\underbrace{
T_R > k
}_{
T_R \geq k+1
}
\, | \, X_0 = \state{\omega_0}
\right)
}
= 
\\ &\xlongequal[j = k+1]{\text{index shift}}
\sum\limits_{j \in \N } \, 
\Prob\left(
X_{j} = \state{\beta}, \, 
T_R \geq j
\, | \, X_0 = \state{\omega_0}
\right)
= 
\\ &=
\NN{\beta} \text{, that is }
\\ 
Q \, \bN 
&= 
\bN. 
%
\end{aligned}
\end{equation}

Since the Markov chain is irreducible, there exist two natural numbers $k, m \in \N $ such that 

\begin{align}
\left(Q^m \right)^{(\alpha, \omega_0)}
&=
\Prob\left( X_{m} = \state{\alpha} \,|\, X_{0} = \state{\omega_0} \right) > 0, \label{Q_m_alpha_omega0_0}
\\
\left(Q^k \right)^{(\omega_0, \alpha)}
&=
\Prob\left( X_{k} = \state{\omega_0} \,|\, X_{0} = \state{\alpha} \right) > 0 \label{Q_m_omega0_alpha_0}
\end{align}


Since $ Q \,\bN = \bN $ implies $ Q^k \, \bN = \bN $, we have: 

\begin{equation*}
\begin{aligned}
\NN{\alpha}
&=
\sum\limits_{\omega \in \Omega }
\left(Q^m \right)^{(\alpha, \omega)} \, \NN{\omega}
\geq
\underbrace{
\left(Q^m \right)^{(\alpha, \omega_0)}
}_{
> 0
}
\, 
\underbrace{
\NN{\omega_0}
}_{
1
}
\;
\overset{ \eqref{Q_m_alpha_omega0_0}}>
\;
0 \text{ and }
\\
1
&=
\NN{\omega_0}
=
\sum\limits_{\omega \in \Omega } \, 
\left(Q^k \right)^{(\omega_0, \omega)} \, \NN{\omega}
\geq
\underbrace{
\left(Q^k \right)^{(\omega_0, \alpha)}
}_{
> \, 0 
}\, 
\NN{\alpha}
\\
\Longrightarrow 
\NN{\alpha}
&\overset{\bigg\downarrow } \leq 
\frac{1}{\left(Q^k \right)^{(\omega_0, \alpha)}}
\;
\overset{\eqref{Q_m_omega0_alpha_0}}<\;
\infty. 
\end{aligned}
\end{equation*}

Uniqueness follow from the fact that for irreducible Markov chains we have $\Kern(\Id - Q) \leq 1$ (compare section \ref{Section_UniquenessAndStrictPositivityOfStationarySolutionsForIrreducibleNetworks} ).

\end{proof}

\begin{lemma}
For an irreducible, DTMC the following are equivalent: 

\begin{itemize}
\item[(i)] \emph{All} states in $\Omega $ are \emph{positive recurrent}
\item[(ii)] There exists a \emph{single} state $\state{\omega_0} \in \Omega $ which is positive recurrent 
\item[(iii) ]
There exists a \emph{stationary solution} 
\end{itemize}
\end{lemma}

\begin{proof}
\item[$(i) \Rightarrow (ii)$] clear

\item[$(ii) \Rightarrow (iii)$]
By assumption we have: 

\begin{equation}
\begin{aligned}
\infty
&>
E[T_R \,|\, X_0 = \state{\omega_0}]
\xlongequal{\eqref{Eq_DTMC_ExpectedReturnTime}}
\sum\limits_{\omega \in \Omega} 
\NN{\omega}
=
\| \bN \|_1
\\ 
%
\Longrightarrow&
\left(
\frac{
\NN{\omega}
}{
\sum\limits_{\alpha \in \Omega}  \NN{\alpha}
}
\right)_{\omega \in \Omega }
=
\frac{
 \bN
}{
\left \| 
\bN
\right \|_1
} 
\text{is a stationary probability vector by lemma \ref{Thm_VectorOfExpectedVisitingNumbersIsCandidateForStationarySolutionOfDTMC} }
\end{aligned}
\end{equation}

\item[$(iii) \Rightarrow (i)$]

If there exists an invariant probability distribution, by lemma \ref{Thm_VectorOfExpectedVisitingNumbersIsCandidateForStationarySolutionOfDTMC} it must be of the form $\left( 
\frac{
\bN
}{
\left\| 
\bN 
\right\|_1
} \right)_{\omega \in \Omega } $. Hence, we have for all states $\omega_0 \in \Omega$

\begin{equation}
\begin{aligned}
\infty 
>
\left\| 
\bN
\right\|_1
=
\sum\limits_{\omega \in \Omega}
\NN{\omega}
=
E\left[
T_R \,| \, X_0 = \state{\omega_0} 
\right]. 
\end{aligned}
\end{equation}

\end{proof}


\subsection{Continuous-time Markov chains}\label{Sec_DiscreteTimeMarkovChains}

\begin{MyDef}[Return time for CTMC] \label{Def_ReturnTimeForCTMC}  $ $ \\ 

For a continuous-time Markov chain (\emph{CTMC}) $\left(X_t\right)_{t \geq 0 } $ starting from a state $X_0 = \state{\omega_0}$ (that is $\Prob(X_0 = \state{\omega_0}) =1$), we define the \textbf{return time} $t_R \in \R_{>0} $ as follows:

\begin{equation}\label{Def_ReturnTimeTo_CTMC} \tag{return time, CTMC}
\begin{aligned}
t_R
&:=
t_R^{[CTMC]}\bigl(\state{\omega_0} \bigr)
&:=
\inf\limits_{t \in \R_{>0}} 
\{
X_t = \state{\omega_0} = X_0 \,:\, \exists  \tau \in (0, t) 
\,:\, X_{\tau} \neq \state{\omega_0}  \}. 
\end{aligned}
\end{equation}

Again we have $ X_0 = \state{\omega_0} = X_{t_R} $.

\end{MyDef}

In analogy to the discrete-time case, we define recurrence and positive recurrence for a CTMC. 

\begin{MyDef}[{(Positive) recurrence for CTMC} ] \label{Def_PositiveRecurrenceForCTMC}  $ $ \\ 

A CTMC is called \emph{recurrent}, if it almost surely returns to the initial state, that is 

\begin{equation}
\begin{aligned}
1
&=
\Prob(t_R < \infty)
=
\int\limits_{\R_{\geq 0 }} \rho_{t_R}(\tau) \, \d \tau. 
\end{aligned}
\end{equation}

A CTMC is called \emph{positive recurrent}, if it is recurrent and its expectation value is finite: 

\begin{equation}
\begin{aligned}
\infty
&>
E[t_R \,|\, X_0 = \state{\omega_0} ]
=
\int\limits_{\R_{\geq 0 }} \tau \,  \rho_{t_R}(\tau) \, \d \tau
=
\int\limits_{0 }^{\infty}
\Prob(t_R \geq \tau) \, \d \tau. 
\end{aligned}
\end{equation}

\end{MyDef}

\begin{MyDef}[{Expected visiting time for CTMC}] \label{Def_VisitingNumberForCTMC}  $ $ \\ 

In analogy to the \emph{expected visiting number} in DTMC, we define the \emph{expected visiting time} for CTMC as 

\begin{equation} \label{Def_CTMC_ExpectedVisitingTime} \tag{expected visiting time}
\begin{aligned}
\t{\omega}
&:= 
E\left[
\int\limits_{0}^{t_R} 1_{\{X_t = \state{\omega} \,|\, X_0 = \state{\omega_0} \}} \d \, t
\right]. 
\end{aligned}
\end{equation}

As in the case of (DTMC), the expected visiting number appears in the expected return time: 

\begin{equation} \label{Eq_CTMC_ExpectedReturnTime} 
\begin{aligned}
E 
\left[
t_R
\,|\, X_0 = \state{\omega_0}
\right]
&=
\int\limits_{0}^{\infty} 
\, 
\underbrace{
\Prob 
\left(
t_R > t  \,|\, X_0 = \state{\omega_0}
\right) 
}_{
\sum\limits_{\omega \in \Omega}
\Prob 
\left(
t_R > t, \, X_t = \state{\omega}  \,|\, X_0 = \state{\omega_0}
\right) 
} 
\d \, t
=
\underbrace{
\int\limits_{0}^{\infty} \, \d t \, 
\sum\limits_{\omega \in \Omega}
}_{
\sum\limits_{\omega \in \Omega} \, 
\int\limits_{0}^{\infty} \, \d t \, 
} 
\, 
\underbrace{
\Prob 
\left(
t_R > t, \, X_t = \state{\omega}  \,|\, X_0 = \state{\omega_0}
\right) 
}_{
E \left[
1_{\{
t_R > t, \, X_t = \state{\omega}  \,|\, X_0 = \state{\omega_0}
\}}
\right]
}
\\ &= 
\sum\limits_{\omega \in \Omega} \,
\underbrace{
E 
\left[
\int\limits_{0}^{t_R}
1_{
\{
X_n = \state{\omega}  \,|\, X_0 = \state{\omega_0}
\}
}
\right]
}_{
\t{\omega}
}
= 
\sum\limits_{\omega \in \Omega} \,
\t{\omega}
= 
\left\|
\underbrace{
\left(
\t{\omega}
\right)_{\omega \in \Omega }
}_{
\bT 
}
\right\|_1
\\ &= 
\left\|
\bT
\right\|_1.  
\end{aligned}
\end{equation}

\end{MyDef}

\begin{lemma}[Expected visiting \emph{time} of CTMC and expected visiting \emph{number} of DTMC] \label{Lemma_ExpectedVisitingTimeOfCTMCAndExpectedVisitingNumberOfDTMC}
\end{lemma}

Let $\t{\omega} $ be the expected waiting \textbf{time} of the CTMC and $ \NN{\omega} $ the expected visiting \textbf{number} of the embedded, DTMC. Then we have: 

\begin{equation}
\begin{aligned}
\t{\omega}
&=
%
\NN{\omega} 
\, \cdot \, 
\frac{1}{\g_{\omega \to}}. 
\end{aligned}
\end{equation}

\begin{proof}

\begin{equation} 
\begin{aligned}
\t{\omega}
&:=
E\left[
\underbrace{
\int\limits_{0}^{\infty} 
1_{
\{X_t = \state{\omega},  t_R \geq t \,|\, X_0 = \state{\omega_0} \}
}
\d \, t
}_{
\sum\limits_{k \in \N_0} 
1_{
\{X_{t_k}= \state{\omega},  T_R \geq k \,|\, X_0 = \state{\omega_0} \}
} 
\, \cdot \, (t_{k+1} - t_{k})
}
\right] = 
%
\\ &\xlongequal[\text{convergence}]{\text{monotone}}
%
\sum\limits_{k \in \N_0} \, 
\underbrace{
E\left[
1_{
\{X_{t_k}= \state{\omega},  T_R \geq k \,|\, X_0 = \state{\omega_0} \}
} 
\, \cdot \, (t_{k+1} - t_{k})
\right]
}_{
E\left[
1_{
\{X_{t_k}= \state{\omega},  T_R \geq k \,|\, X_0 = \state{\omega_0} \}
} 
\right]
\, \cdot \, 
E[\tau_{\omega} ]
}
\\ &\overset{(*)}=
\underbrace{
\sum\limits_{k \in \N_0} \, 
E\left[
1_{
\{X_{t_k}= \state{\omega},  T_R \geq k \,|\, X_0 = \state{\omega_0} \}
} 
\right]
}_{
\NN{\omega}
}
\, \cdot \, 
\underbrace{
E[\tau_{\omega} ]
}_{
\frac{1}{\g_{\omega \to }}
}
\\ &= 
\NN{\omega}
\, \cdot \, 
\frac{1}{\g_{\omega \to }}. 
\end{aligned}
\end{equation}

In the first step, we used the fact that the return time $t_R $ to the initial state $\omega_0$ of  the CTMC is in the interval $[t_{k}, t_{k+1}) $, if and only if the return time $T_R$ of the embedded, DTMC equals $k$ ($t_R \in [t_{k}, t_{k+1})  \iff T_R = k$ for some $k \in \N_0 $) and in step (*) we made use of the fact that the waiting time distribution $t_{k+1} - t_{k} $ in a certain state $\state{\omega} $ depends only on that state and is independent on $k \in \N_0$. The expectation value of the waiting time is given by $E[\tau] = \frac{1}{\g_{\omega \to}} $.

\end{proof}


\begin{lemma}[Vector of expected visiting \emph{time} is candidate for stationary solution of CTMC ]\label{Lemma_VectorOfExpectedVisitingTimeIsCandidateForStationarySolutionOfCTMC} $ $ \\ 
\end{lemma}

For an irreducible CTMC, every entry of $\bT $ is strictly positive and the eigenspace of the generator matrix $\G$ to the eigenvalue $\lambda = 0$ is spanned by $ \frac{\bT}{\| \bT \|_1 } $. This means that $ \frac{\bT}{\| \bT \|_1 } $ is the only possible probability vector that is a stationary solution of the Markov chain: 

\begin{equation*}
\begin{aligned}
\G \,  \bT 
&=
\bzero \text{  and }
\bT
\in
(\R_{> \, 0})^{\Omega} \text{  , that is }
\\
\frac{\bT}{\left \| \bT \right \|_1} 
&\in \begin{cases}
\Kern(\G) \cap \ProbStates \cap (\R_{> \, 0})^{\Omega} 
&\text{, if } \| \bT \|_1 < \infty
  \text{ OR} 
\\ 
\{\bzero\} &\text{, if } \| \bT \|_1 = \infty
\end{cases}
\end{aligned}
\end{equation*}


\begin{proof}

We know from section \ref{Section_UniquenessAndStrictPositivityOfStationarySolutionsForIrreducibleNetworks} that the dimension of the kernel of $\G $ is at most one dimensional. On the other hand, we have

\begin{equation}
\begin{aligned}
\left(
\G \, \bT 
\right)^{(\alpha)}
&=
\sum\limits_{\beta \in \Omega}
\G^{(\alpha, \beta)} \, 
\underbrace{
\t{\beta} 
}_{
\frac{
\NN{\beta}
}{
\g_{\beta \to}
}
}
=
\sum\limits_{\beta \in \Omega \atop \beta \neq \alpha}
\NN{\beta} \, 
\underbrace{
\left(
\frac{
\G^{(\alpha, \beta)}
}{
\g_{\beta \to}
}
\right)
}_{
q_{\beta \to \alpha }
}
\, + \, 
\frac{
\NN{\alpha} 
}{
\g_{\alpha \to}
}
 \, 
\underbrace{
\G^{(\alpha, \alpha)}
}_{
-\g_{\alpha \to}
}  
\\ &= 
\underbrace{
\left(
\sum\limits_{\beta \in \Omega \atop \beta \neq \alpha} 
\NN{\beta} \, q_{\beta \to \alpha}
\right)
}_{
\NN{\alpha}
}
\, - \, 
\NN{\alpha}
\\ &= 
0 . 
\\ 
\Longrightarrow 
\G \, \bT 
&=
\bzero. 
\end{aligned}
\end{equation}

Since we have $ \bN \in (\R_{>0})^\Omega  $ and $\left(\frac{1}{\g_{\omega \to}} \right)_{\omega \in \Omega} \in (\R_{>0})^\Omega $, we can conclude that 
\begin{equation}
\begin{aligned}
\t{\omega}
&=
\underbrace{
\NN{\omega}
}_{
\in (0, \infty) 
}
\, \cdot \, 
\underbrace{
\frac{1}{\g_{\omega \to }}
}_{
\in (0, \infty) 
}
\in 
(0, \infty). 
\end{aligned}
\end{equation}

Since we assumed that the Markov chain was irreducible, and we know from section \ref{Section_UniquenessAndStrictPositivityOfStationarySolutionsForIrreducibleNetworks} that an irreducible Markov chain has at most one stationary solution, the kernel of the generator must be spanned by $\left(\frac{\bT}{\| \bT \|_1} \right) $.

\end{proof}

\begin{lemma}[Positive recurrence for CTMC] \label{Lemma_PositiveRecurrenceForCTMC}
For an irreducible, CTMC the following are equivalent: 

\begin{itemize}
\item[(i)] \emph{All} states in $\Omega $ are positive recurrent 
\item[(ii)] There exists a \emph{single} state $\state{\omega_0} \in \Omega $ which is positive recurrent 
\item[(iii) ]
There exists a \emph{stationary solution}  
\end{itemize}
\end{lemma}

\begin{proof} 
\vspace*{5mm}
\begin{itemize}
\vspace*{5mm}
\item[]
\item[$(i) \Rightarrow (ii)$] clear

\item[$(ii) \Rightarrow (iii)$]
By assumption we have: 

\begin{equation}
\begin{aligned}
\infty
&>
E[t_R \,|\, X_0 = \state{\omega_0}]
\xlongequal{\eqref{Eq_CTMC_ExpectedReturnTime}}
\sum\limits_{\omega \in \Omega} 
\t{\omega}
\\
%
&\Longrightarrow 
%
\frac{
\bT
}{
\left \| 
\bT 
\right \|_1
} 
=
\left( 
\frac{
\t{\omega}
}{
\sum\limits_{\alpha \in \Omega} 
\t{\alpha}
} \right)_{\omega \in \Omega }
\text{is a stationary probability vector by lemma  \ref{Lemma_VectorOfExpectedVisitingTimeIsCandidateForStationarySolutionOfCTMC}}
\end{aligned}
\end{equation}

 \item[$(iii) \Rightarrow (i) $]
If there exists an invariant probability distribution, it must be of the form $\left( 
\frac{
\bT 
}{
\left\| 
\bT 
\right\|_1
} \right)_{\omega \in \Omega } $. Hence, we have 

\begin{equation}
\begin{aligned}
\infty 
>
\left\| 
\bT 
\right\|_1
=
\sum\limits_{\omega \in \Omega}
\t{\omega}
=
E\left[
T_R \,| \, X_0 = \state{\omega_0} 
\right]
\end{aligned}
\end{equation}

\end{itemize}

\end{proof}




\section{Lemmata}

\subsection{Existence and uniqueness of the master equation}\label{Section_ExistenceAndUniquenessOfTheMasterEquation}

\begin{lemma}\label{Lemma_ExistenceAndUniquenessOfTheMasterEquation}[Existence and uniqueness of the master equation]
\end{lemma}

\begin{claim}\label{Lemma_ExistenceAndUniquenessOfTheMasterEquation_BoundednessOfTheSolution}

For 
\begin{equation*}
\begin{aligned}
\bP_{n+1}(t) 
&:=
\bP_0 + \int_0^t \G(\tau) \,  \bP_n(\tau) \d \tau
\end{aligned}
\end{equation*}

\begin{equation}
\begin{aligned}
\|\bP_n(t) \|_1
&\leq
\|\bP_0 \|_1 
\, \cdot \, 
\exp{\left[
\int_0^t
\|
\G(\tau)
\|_{1, 1}
\d \tau
\right]}
\end{aligned}
\end{equation}
\end{claim}

\begin{proof}
For the induction start, we get: 

\begin{itemize}
\item[i-start]

\begin{equation}
\begin{aligned}
\|\bP_0 \|_1
&\leq
\|\bP_0 \|_1 
\, \cdot \, 
\underbrace{
\exp{\left[
\int_0^t
\|
\G(\tau)
\|_{1}^\text{(op)}
\d \tau
\right]}
}_{
\geq 1
}. 
\end{aligned}
\end{equation}

\item[i-step]
\begin{equation}
\begin{aligned}
\|\bP_{n+1}(t) \|_1
&\xlongequal{\text{Def}}
\Bigl\|
\bP_0 + \int_0^t \G(\tau) \,  \bP_n(\tau) \d \tau
\Bigr\|_1
\\ &\leq
\|\bP_0 \|_1 
+
\int_0^t 
\underbrace{
    \| \G(\tau) \bP_n(t) \|_{1}
}_{ 
    \leq \, \| \G(\tau) \|_{1}^\text{(op)} \, \| \bP_n(t) \|_{1}
}
\d \tau 
\\ &\leq
\int_0^t 
\| \G(\tau) \|_{1}^\text{(op)}
\, 
\underbrace{
\| \bP_n(\tau) \|_1
}_{
\|\bP_0 \|_1 
\, \cdot \, 
\exp{\left[
\int_0^\tau
\|
\G(\xi)
\|^\text{(op)}_{1}
\d \xi
\right]}
}
\d \tau 
\\ &\xlongequal{\text{i- hypothesis}}
\|\bP_0 \|_1 
\, \cdot \,
\left( 
1 + 
\int_0^t
\underbrace{
\| \G(\tau) \|_{1}^\text{(op)}
\, \cdot \,
\exp{\left[
\int_0^\tau
\|
\G(\xi)
\|_{1}^\text{(op)}
\d \xi
\right]}
}_{
\frac{\d}{\d \tau}
\exp{\left[
\int_0^\tau
\|
\G(\xi)
\|_{1}^\text{(op)}
\d \xi
\right]}
}
\d \tau
\right)
\\ &=
\|\bP_0 \|_1 
\, \cdot \,
\left( 
1 + 
\exp{\left[
\int_0^t
\|
\G(\tau)
\|_{1}^\text{(op)}
\d \tau
\right]}
-
1
\right)
\end{aligned}
\end{equation}

\end{itemize}
\end{proof}

\begin{claim}\label{Lemma_ExistenceAndUniquenessOfTheMasterEquation_EstimatingDifferenceBetweenTwoConsecutiveElements}
\begin{equation*}
\begin{aligned}
\| \left( \bP_{n} - \bP_{n-1}\right) (t) \|_1
&\leq
\frac{\left(\int_0^t \|\G(\tau)\|_{1}^\text{(op)} \d \tau  \right)^n}{n!}
\end{aligned}
\end{equation*}
\end{claim}

\begin{proof}
For the induction start we get: 

\begin{itemize}
\item[i-start]

\begin{equation*}
\begin{aligned}
\| \left( \bP_{1} - \bP_{0}\right) (t) \|_1
&\leq
\int_0^t \d \tau 
\underbrace{
    \| \G(\tau) \, \bP_0 \|_{1}
}_{
    \| \G(\tau)\|_{1}^\text{(op)} \, \|\bP_0 \|_{1}
}
&\leq
\int_0^t
\| \G(\tau) \|_{1}^\text{(op)}
\, \d \tau 
\, 
\underbrace{
\| \bP_0 \|_1
}_{1}
\\ &=
\frac{\left(
\int_0^t
\| \G(\tau) \|_{1}^\text{(op)}
\right)^1}{1!}
\end{aligned}
\end{equation*}

\item[i-step]
\begin{equation*}
\begin{aligned}
\| 
\bigl(\bP_{n+1} - \bP_{n} \bigr)(t)
\|_1
&=
\int_0^t  
\underbrace{
    \left\|
    \G(\tau) \,
    \bigl(\bP_{n} - \bP_{n-1} \bigr)(\tau)
    \right \|_1
}_{
    \leq \, \|\G(\tau) \|_{1}^\text{(op)} \|  \,
\bigl(\bP_{n} - \bP_{n-1} \bigr)(t)
\|_1
}
\d \tau
\\ &=
\int_0^t
\|
\G(\tau)
\|_{1}^\text{(op)}
\, 
\underbrace{
    \|
    \bigl(\bP_{n} - \bP_{n-1} \bigr)(\tau)
    \|_1
}_{
    \leq \, 
    \frac{\left(\int_0^\tau \|\G(\xi)\|_{1}^\text{(op)} \d \xi  \right)^n}{n!}
    }
\d \tau 
\\ & \overset{\text{i-hypothesis}}{\leq} 
\int_0^t
\underbrace{
    \| \G(\tau) \|^\text{(op)}_{1} \, \frac{\left(\int_0^\tau \|\G(\xi)\|_{1}^\text{(op)} \d \xi  \right)^n}{n!}
}_{
    \frac{\d}{ \d \tau} 
    \,
    \frac{\left(\int_0^\tau \| \G(\xi)\|_{1}^\text{(op)} \d \xi  \right)^{n+1}}{(n+1)!}
}
\d \tau 
\\ &=
\frac{\left(\int_0^t \|\G(\xi)\|_{1}^\text{(op)} \d \xi  \right)^{n}}{(n+1)!}. 
\end{aligned}
\end{equation*}

\end{itemize}
\end{proof}

\begin{claim} \label{Lemma_ExistenceAndUniquenessOfTheMasterEquation_CauchySequence}
The sequence $ \bigl( \bP_n(t) \bigr)_{n \in \N} $ is a Cauchy sequence in $ l^p(\Omega) $, that is, there exists a $\bP \in l^p(\Omega) $ such that $ \bP_n(t) \xlongrightarrow[\|\, \|_1]{n \to \infty} \bP(t) $
\end{claim}

\begin{proof}
\begin{equation*}
\begin{aligned}
\left\|
\bigl(\bP_{n+m} - \bP_{n} \bigr)(t)
\right \|_1
&=
\left\|
\sum\limits_{k = n+1}^{n+m}
\bigl(\bP_{k} - \bP_{k-1} \bigr)(t)
\right \|_1
\leq
\sum\limits_{k=n+1}^{n+m}
\left\|
\bigl(\bP_{k} - \bP_{k-1} \bigr)(t)
\right\|_1
\\ 
&\overset{\ref{Lemma_ExistenceAndUniquenessOfTheMasterEquation_EstimatingDifferenceBetweenTwoConsecutiveElements}}{\leq} 
\sum\limits_{k=n}^{n+m-1}
\|
\bP_0 
\|_1
\,
\frac{\left(\int_0^t \|\G(\xi)\|_{1}^\text{(op)} \d \xi  \right)^{k}}{(k)!}
\\ &\leq
\exp{
\left[
\int_0^t \|\G(\xi)\|_{1}^\text{(op)} \d \xi  
\right]
}
<
\infty, \text{   and hence }
\\ 
\left\|
\bigl(\bP_{n+m} - \bP_{n} \bigr)(t)
\right \|_1
&\xlongrightarrow{n, m \to \infty}
0
\end{aligned}
\end{equation*}
\end{proof}


\subsection{Properties of the master equation and its generator} \label{Section_PropertiesOfTheMasterEquationAndItsGenerator}


\begin{lemma}\label{StationaryStatesCoincideWithTheKernelOfTheGenerator}[Connection between stationary states and the kernel of the generator]
\end{lemma}
When $\bp \in \Kern(\G)$, then $\e^{\G \, t} \, \bp = \left( \Id + \sum\limits_{n\geq 1} \frac{(\G \, t)^n}{n!} \right) \, \bp = \bp$, so $\bp$ is stationary. When on the other hand, $\bp$ is stationary, that is $\e^{t \, \G t} \, \bp = \bp$ for all times $t \geq 0$, we have $\G \, \bp = \lim\limits_{t \to 0^+} \frac{\e^{\G \, t} - \Id}{t} \, \bp = \bzero$, so $\bp$ lies in the kernel of $\G$.  


\begin{lemma}
\label{Lemma_Comparing_1NormOfMatrices_with_Vector1NormOfMatrices}
[Compatibility of the $1$-$1$-norm on $\R^{\N \times \N} $ with the $1$-norm on $\R^{\N}$]
\end{lemma}

The operator-$1$-norm $\| A \|_1^\text{(op)}$ of an infinitely large matrix $A \in \R^{\Omega \times \Omega}$ is less or equal than the operator-$1$-$1$-norm $ \|A \|_{1, 1}^\text{(op)} :=  \sum\limits_{i, j \in \N} |A^{(i,j)}| $ of the same matrix, making the $1$-$1$-norm $\| \cdot \|_{1, 1}^\text{(op)}$ on $\R^{\N \times \N}$ compatible with the $1$-norm $\| \cdot \|_1$ on $\R^\Omega$:

\begin{proof}

\begin{equation}\label{Eq_Comparing_1NormOfMatrices_with_Vector1NormOfMatrices}
\begin{aligned}
\| \G \|^\text{(op)}_1
&=
%
\sup\limits_{j \in \N} 
\Bigl\| 
\underbrace{
\G \bE_j
}_{
\left( 
(\g_{j \to k})_{k \in \{1, \dots, j-1\}} , -\g_{j \to}, 
(\g_{j \to k})_{k \in \N_{\geq \, j}}
\right)
}
\Bigr\|_1 \\
%
&=
\sup\limits_{j \in \N} 
\left(
\Bigl(
\sum\limits_{k \in \N\backslash \{j\}} \g_{j \to k}
\Bigr)
+ \g_{j \to }\right) 
=
2 \, \sup\limits_{j \in \N} \g_{j \to }
=
2 \, 
\|
(\g_{j \to })_{j \in \N}
\|_\infty
\leq \\
%
&\leq 
2 \, 
\|
(\g_{j \to })_{j \in \N}
\|_1
=
2 \, \sum\limits_{j \in \N} \g_{j \to }
=
\sum\limits_{j \in \N} \, 
\left(
\sum\limits_{{i \in \N}\atop{i \neq j}} |\G^{(i,j)}|
\right)
+ 
|\G^{(j,j)}|
=
\sum\limits_{i,j \in \N} |\G^{(i, j)}|
= \\
%
&\xlongequal{\text{def}}
\| \G \|_{1, 1}^\text{(op)}. 
\end{aligned}
\end{equation}

This means that the norm $\| \cdot \|_{1, 1}^\text{(op)} $ on $\R^{\Omega \times \Omega}$ is compatible with  $\| \cdot \|_1 $ on $\R^{\Omega }$:

\begin{equation}\label{Eq_CompatibilityOf_Vector1Norm_and_1Norm}
\begin{aligned}
\|\G \, \bX \|_1
\leq
\underbrace{
\|\G  \|_1^\text{(op)}
}_{
\leq \, \|\G  \|_{1, 1}^\text{(op)}
}
\, \cdot \, \| \bX \|_1
\overset{\eqref{Eq_Comparing_1NormOfMatrices_with_Vector1NormOfMatrices}}{ \leq }
\|\G  \|_{1, 1}^\text{(op)}
\, \cdot \, 
\| \bX \|_1
\end{aligned}
\end{equation}

\end{proof}


\begin{lemma}
\label{Lemma_ApproximatingTheFullGeneratorWithGeneratorsOfFiniteSubnetworks}
[Approximating (countable infinitely large) generators with generators of finite subnetworks]
\end{lemma}
The generator $\GF$ of a finite subnetwork $ F \in \Fin(\Omega)$ (defined in \eqref{GenFinSubNet}) converges in the thermodynamic limit to the generator $\G$ of the full, infinitely large network, with respect to the $1$-$1$-norm: $ \GF \xlongrightarrow[\| \cdot \|_1^\text{(op)}]{F \in \Fin(\Omega)} \G $:

\begin{equation}
\begin{aligned}
\left \| \G - \GF \right \|_{1, 1}^\text{(op)}
&=
\sum\limits_{j \in \Omega} 
\left(
\underbrace{
\sum\limits_{{i \in \Omega}\atop{i \neq j}}
\left|
\G^{(i,j)} - (\GF)^{(i,j)}
\right|
}_{
\sum\limits_{i \in \Omega \backslash (F \cup \{j\}) } \g_{j \to i}
}
+
\underbrace{
\left|
\G^{(j,j)} - (\GF)^{(j,j)}
\right|
}_{
\sum\limits_{k \in \Omega \backslash (F \cup \{j\})  } \g_{j \to k}
}
\right)
=
2 \, \sum\limits_{i \in \Omega \backslash F}\;\; \sum\limits_{j \in \Omega \backslash\{i\}}  \g_{j \to i} \xlongrightarrow{F \in \Fin(\Omega)}
0. 
\end{aligned}
\end{equation}

\begin{lemma}
\label{TheOperator1Norm}
[The operator 1-norm]
\end{lemma} 

The operator norm induced by the $1$-norm, can be computed as the supremum over the sums of the absolute values of each column, namely 
\begin{equation*}
\begin{aligned}
\|A\|^\text{(op)}_1
:=
\sup\limits_{\|\bX\|_1 = 1} \|A \, \bX\|_1
\xlongequal{\text{claim}}
\sup\limits_{j \in \N}
\|A \, \bE_j\|_1
= 
\sup\limits_{j \in \N} \, 
\sum\limits_{i \in \N } 
|A^{(i,j)}|
\end{aligned}
\end{equation*}

\begin{proof}
$\sup\limits_{\|\bX\|_1 = 1} \|A \, \bX\|_1 = \sup\limits_{j \in \N}
\|A \, \bE_j\|_1$ \\

\begin{itemize}
\item[$\geq$]
This is clear, since for the sequences $\bE_j = (\overbrace{0, \dots, 0}^{j-1 \text{ times}}, 1, 0^\N)$ we have $\|\bE_j\|_1=1$ 

\item[$\leq$]
\begin{equation*}
\begin{aligned}
\sup\limits_{\|\bX\|_1 = 1} \|A \, \bX \|_1 
&=
\textbf{}
\sum\limits_{i \in \N} 
\Bigl|
\underbrace{
(A \, \bX)^{(i)}
}_{
\sum\limits_{j \in \N} A^{(i,j)} X^{(j)}
} 
\Bigr|
\leq 
\sup\limits_{\|\bX\|_1 = 1}
\sum\limits_{j \in \N} |X^{(j)}| \, \left(  \sum\limits_{i \in \N} |A^{(i,j)}| \right) \leq \\
%
&\overset{\text{Hoelder}}{\leq}
\underbrace{
\sup\limits_{\|\bX\|_1 = 1}  \|\bX\|_1 
}_{1} \, 
\underbrace{
\left\|
\left( \sum\limits_{i \in \N} |A^{(i,j)}| \right)_{j \in \N} 
\right\|_\infty
}_{
\sup\limits_{j \in \N}
\|A \, \bE_j\|_1
}
=
\sup\limits_{j \in \N}
\|A \, \bE_j\|_1. 
\end{aligned}
\end{equation*}
\end{itemize}

\end{proof}

\begin{equation*}
\begin{aligned}
\end{aligned}
\end{equation*}


\begin{lemma}
\label{lemma_ApproximatingProbabilityStates}
[Approximating probability states]
\end{lemma}

For $ \bP \in \ProbStates $ we have $ \bPF \xlongrightarrow[\|\cdot \|_1]{ |F| \to \infty } \bP $, since

\begin{equation*}
\begin{aligned}
\| \bPF - \bP \|_1
&=
\left\| 
\frac{\bP \cdot \bone_F}{\|\bP \cdot \bone_F \|_1} - \bP
\right\|_1
=
\sum\limits_{i \in \Omega } \Bigl| \frac{(\bP \cdot \bone_F)^{(i)}}{\|\bP \cdot \bone_F \|_1} - \bP^{(i)}  \Bigr|  \\
&=
\sum\limits_{f \in F} 
\underbrace{
    \Bigl| \frac{P^{(f)}}{\sum\limits_{i \in F} P^{(i)}} - P^{(f)}  \Bigr| 
}_{
    \frac{P^{(f)}}{\sum\limits_{\beta \in F} P^{(\beta)}} \Bigl| 1 - \sum\limits_{i \in F} P^{(\beta)} \Bigr| 
}
+
\sum\limits_{i \in \Omega \backslash F} P^{(i)}
\\ &=
\underbrace{
\left( \frac{\sum\limits_{b \in F} P^{(b)}}{\sum\limits_{i \in F} P^{(\beta)}} \right)
}_{
1
}
\,\cdot \, \sum\limits_{j \in \Omega \backslash F} P^{(j)} + \sum\limits_{i \in \Omega \backslash F} P^{(i)}
\\ &=
2 \sum\limits_{j \in \Omega \backslash F} P^{(j)} \xlongrightarrow{|F| \to \infty } 0. 
\end{aligned}
\end{equation*}




\subsection{Groenwalls's inequalities} \label{Section_GroenwallssInequalities}

\begin{lemma}
\label{Lemma_Groenwall_IntegralVersion}
[Groenwalls's inequality - differential version]
\end{lemma}
Let $u, \, \beta :[a, b] \to \R $ be continuous function and let $u$ on top be differentiable, If $u'(t) \leq \beta(t) \,\cdot \, u(t) \Rightarrow u(t) \leq u(a) \,\cdot\, \e^{\int_a^t \beta} $. 

\begin{proof}
It is enough to show that

\begin{equation*}
\begin{aligned}
\frac{d}{dt}
\left( 
\frac{u(t)}{\e^{\int_a^t \beta}}
\right)
\leq
0, 
\end{aligned}
\end{equation*}

since then we have 

\begin{equation*}
\begin{aligned}
\left( 
\frac{u(t)}{\e^{\int_a^t \beta}}
\right)
\overset{\text{decreasing}}{\leq}
\left( 
\frac{u(t)}{\e^{\int_a^t \beta}}
\right)\Bigl|_{t=a}
=
\frac{u(a)}{1}, 
\end{aligned}
\end{equation*}

and indeed we have: 

\begin{equation*}
\begin{aligned}
\frac{d}{dt}
\left( 
\frac{u(t)}{ \e^{\int_a^t \beta} }
\right)
&=
\frac{u' \,  \e^{\int_a^t \beta} - u \, \beta \,  \e^{\int_a^t \beta}}{\left( \e^{\int_a^t \beta} \right)^2}
=
\underbrace{
\e^{-\int_a^t \beta}
}_{
\geq \, 0
}
\, \cdot \, 
\underbrace{
\left( u' - \beta \, u \right) 
}_{
\leq \, 0
}
\leq 
\, 0. 
\end{aligned}
\end{equation*}

\end{proof}


\begin{lemma} \label{Lemma_Groenwall_IntegralVersion} 
[Groenwall's inequality - integral version version]
\end{lemma}

Let $I$ be an interval,  $I \in \{[a, b], [a, b)\}$, with $a \in \R, \, b \in \R_{> \, a} \cup \{\infty\}$ and let $\alpha, \beta, \, u : I \to \R $ be continuous functions \textbf{with $\alpha$ being non-decreasing}. Further, let $\beta  \geq 0 $ and $ u(t) \leq \alpha(t) + \int_a^t \beta \, u $ for all $ t \in I $. Then the following holds true:

\begin{equation*}
\begin{aligned}
u(t) 
\leq 
\alpha(t) \,\cdot \,  \e^{\int_a^t \beta}
\end{aligned}
\end{equation*}

\begin{proof}
We have the following estimation for all $s \in [a, s)$: 
\begin{itemize}
\item[(i)]
\begin{equation*}
\begin{aligned}
\frac{\d}{\d s}
\left(
\e^{-\int_a^s \beta} \, \int_a^s \beta u
\right)
&=
\left(
- \beta(s) \, \e^{-\int_a^s \beta} \, \int_a^s \beta \, u
+
\e^{-\int_a^s \beta} \, \beta(s) \, u(s) 
\right) 
=
\e^{-\int_a^s \beta} \, \beta(s) 
\underbrace{
\left[
u(s) - \int_a^s \beta \, u
\right]
}_{
\leq \, \alpha(s)
} \\
%
&\overset{\text{assumption}}{\leq} 
\alpha(s) \, \beta(s) \, \e^{- \int_a^s \beta }
\end{aligned}
\end{equation*}

\item[(ii)]
\begin{equation*}
\begin{aligned}
\left(
\e^{-\int_a^t \beta} \, \int_a^t \beta u
\right)
&=
\bigintssss_a^t \frac{\d }{\d s} 
\underbrace{
\left(
\e^{-\int_a^t \beta} \, \int_a^t \beta u
\right)
}_{
\leq \, \alpha(s) \, \beta(s) \,\e^{-\int_a^s \beta}
}
\d s  \leq \\
%
&\overset{\text{(i)}}{\leq} 
\bigintssss_a^t
\alpha(s) \, \beta(s) \,\e^{-\int_a^s \beta}
\d s 
\end{aligned}
\end{equation*}

\item[iii)]

\begin{equation*}
\begin{aligned}
\int_a^t \beta \, u
=
\e^{ \int_a^t \beta } \,
\underbrace{
\left( 
\e^{-\int_a^t \beta} \, \int_a^t \beta \, u
\right)
}_{
\leq \, \int_a^t \alpha(s) \, \beta(s) \, \e^{-\int_a^s \beta } \d s
}
\overset{\text{(ii)}}{\leq} 
\int_a^t \alpha(s) \, \beta(s) 
\underbrace{
\left(
\e^{\int_a^t \beta} \, \e^{-\int_a^s \beta}
\right)
}_{
\e^{\int_s^t \beta}
}
=
\int_a^t 
\alpha(s) \, \beta(s) \, 
\e^{\int_s^t \beta}
\d s
\end{aligned}
\end{equation*}

\item[iv)]
\begin{equation*}
\begin{aligned}
u(t) 
&\overset{\text{assumption}}{\leq}
\alpha(t)
\,+\, 
\underbrace{
\left(
\int_a^t \beta \, u
\right)
}_{
\leq \, 
\int_a^t 
\alpha(s) \, \beta(s) \, 
\e^{\int_s^t \beta}
\d s
}
\overset{\text{(iii)}}{\leq }
\alpha(t)
+
\int_a^t 
\underbrace{
\alpha(s)
}_{
\leq \, \alpha(t) 
}\, \beta(s) \, 
\e^{\int_s^t \beta}
\d s \leq \\
%
&\overset{\text{$\alpha$ non-decreasing}}\leq 
\alpha(t) \, \cdot \, 
\left[
1 + 
\int_a^t  
\underbrace{
\beta(s) \, 
\e^{-\int_t^s \beta}
}_{
(-1) \, \frac{\d}{\d s}
\e^{-\int_t^s \beta}
}
\d s
\right]
=
\alpha(t) \, \cdot \, 
\left[
1 + 
\underbrace{
\left(
(-1) \,\cdot \,
\e^{-\int_t^s \beta} \right)\Bigl|_{s=a}^{s=t}
}_{
(-1) + 
\e^{\int_a^t \beta}
}
\right] \\
%
=&
\alpha(t) \, \cdot \, \e^{\bigintssss_a^t \beta}
\end{aligned}
\end{equation*}

\end{itemize}
\begin{equation*}
\begin{aligned}
\end{aligned}
\end{equation*}
\end{proof}



\begin{MyDef}
\begin{equation}
\begin{aligned}
\bX^{\pm}
:=
\left(
\max\{0, X^{(\omega)}\}
\right)_{\omega \in \Omega }
\end{aligned}
\end{equation}

In particular, we have: 

\begin{equation}
\begin{aligned}
\bX
&=
\bX^{+} - \bX^{-}
=
\begin{cases}
\bzero &\text{  , if $\bX = \bzero$ }
\\
\pm \| \bX \|_1 \cdot \frac{\bX}{\| \bX \|_1 } &\text{  , if $\bX \in (\R_{\geq \, 0})^{\Omega} \cup (\R_{\leq \, 0})^{\Omega} $ }
\\
\| \bX^{+} \|_1 \cdot \frac{\bX^{+}}{\| \bX^{+} \|_1 }
-
\| \bX^{-} \|_1 \cdot \frac{\bX^{-}}{\| \bX^{-} \|_1 }
&\text{  , else  }
\end{cases}
\end{aligned}
\end{equation}

\end{MyDef}
\subsection{Infinite dimensional semigroups}\label{SectionInfiniteDimensionalSemigroups}


\begin{lemma}[Convergence of positive, irreducible semigroups on $l^p(\Omega)$ \cite{arendt2020positive}]\label{ConvergenceOfPositiveIrreducibleSemigroupsOnLp} $ $ \\ 
Let $(\e^{t \, \G})_{t \geq 0}$ be a bounded, positive and irreducible $C_0$-semigroup on $l^p(\Omega)$ for $p \in [1, \infty)$ with generator $\G$, which is eventually norm continuous. If the kernel of the generator is non-trivial $(\Kern(\G) \neq \{0\})$, then there exist a (element-wise) strictly positive element $\bX_* \in l^p(\Omega), \bX > \bzero $ and a strictly positive functional in the dual space $\bX_*' \in (l^p)'(\Omega)$ such that $ \e^{t \, \G} \bX  \xlongrightarrow[\| \cdot \|_1]{t \to \infty } (\bX_* \bigotimes \bX_*') \bX := \bX_* \cdot  \bX'_*(\bX) $ for all $\bX \in l^p(\Omega) $. 

In the case of  the master equation, this means that 

\begin{equation*}
\begin{aligned}    
\e^{t \, \G} &\xlongrightarrow[\text{strongly}]{t \to \infty} \bP_* \bigotimes \bone' \text{  , that is }
\\
\e^{t \, \G} \, \bX_0 &\xlongrightarrow[\| \cdot \|_1]{t \to \infty} \bP_* \cdot  \bone' (\bX_0) 
\end{aligned}
\end{equation*}

W.l.o.g, we can assume that $\bX_*$ has norm one and call it $\bP_* \in \ProbStates $. The fact that $\bX_*' = \bone'$ follows from normalization: 

\begin{itemize}
\item[$ \bX_0 \in \ProbStates$]

\begin{equation*}
\begin{aligned}    
\left|
\underbrace{
    \sum\limits_{i \in \Omega } X_*'^{(i)} X_0^{(i)}
}_{
    \bX_*'(\bX_0)
}
\right|
&=
\left| \bX_*'(\bX_0) \right| \, \cdot \, 
\underbrace{
    \| \bP_* \|_1
}_{
    1
}
=
\| \bP_* \, \bX_*'(\bX_0) \|_1 
=
\|
\underbrace{
(\bP_*\bigotimes \bone') \bX_0
}_{
\lim\limits_{t \to \infty} \e^{t \, \G} \, \bX_0
}
\|_1 
\\ &=
\lim\limits_{t \to \infty}
\underbrace{
\| \e^{t \, \G} \, \bX_0 \|_1
}_{
1
}
=
1. 
\end{aligned}
\end{equation*}

By choosing $\bX_0 = \bE_{m} $, we can conclude that $X_*'^{(m)} = 1$ for all $m \in \Omega$.

\item[$ \pm\bX_0 \geq \bzero $]
With the limit operator $\LL$ we have:

\begin{equation*}
\begin{aligned}    
\LL(\bX)
&=
\LL \left(
\| \bX \|_1
 \, \frac{\pm \, \bX}{\|\bX \|_1}
 \right)
=
\pm \, \| \bX \|_1
\underbrace{
    \LL \Bigl(
    \overbrace{
    \frac{\bX}{\|\bX \|_1}
    }^{\in \ProbStates }
    \Bigr)
}_{
\bP_*
}
\\ &=
\pm \, \bP_*
\underbrace{
\sum\limits_{i \in \Omega} X^{(i)}
}_{
\bone'(\bX)
}
=
\pm \, \bP_* \cdot \bone'(\bX). 
\end{aligned}
\end{equation*}

\item[$ \bX_0 \in l^p(\Omega ) $]

\begin{equation*}
\begin{aligned}    
\LL(\bX)
&=
\LL \left(
\|\bX_{+}\|_1 \frac{\bX_{+}}{\|\bX_{+}\|_1} 
-
\|\bX_{-}\|_1 \frac{\bX_{-}}{\|\bX_{-}\|_1} 
\right)
=
\|\bX_{+}\|_1 \, \cdot \, 
\underbrace{
\LL \left(
\frac{\bX_{+}}{\|\bX_{+}\|_1} 
\right)
}_{
\bP_*
}
-
\|\bX_{-}\|_1 \, \cdot \, 
\underbrace{
\LL \left(
\frac{\bX_{-}}{\|\bX_{-}\|_1} 
\right)
}_{
\bP_*
}
\\ &=
\bP_* \, \cdot \, 
\sum\limits_{i \in \Omega}
\underbrace{
\left(
X_{+}^{(i)} - X_{-}^{(i)}
\right)
}_{
X^{(i)}
}
=
\bP_* \, \cdot \, \bone'( \bX_0)
\end{aligned}
\end{equation*}

\end{itemize}

\end{lemma}

\begin{lemma}[A version of the spectral mapping theorem]\label{Lemma_VersionOfTheSpectralMappingThm}

For a strongly continuous semigroup $(\e^{t \, \G})_{t \geq 0}$ , which is eventually norm continuous, we have:

\begin{equation*}
\begin{aligned}    
\sigma\left(
\e^{t \, \G}
\right) \backslash \{0\}
&=
\e^{t \, \sigma(\G)}. 
\end{aligned}
\end{equation*}

In particular, we have $\sigma(e^{t \, \G}) \in \left \{\e^{t \, \sigma(\G)}, \e^{t \, \sigma(\G)} \cup \{0\}  \right \} $

Spectral mapping theorems relate the spectrum of a semigroup $ (\e^{t \, \G})_{t \geq 0} $ to that of its generator $A$. They come in all variations and generalizations, the one we are using is from \cite{arendt1986one} (Part A III, Corollary 6.7). In our case, the boundedness of the generator $ \G $ implies that the corresponding semigroup is not only \emph{eventually norm continuous}, but even \emph{uniformly continuous}. 

\end{lemma}

\begin{lemma}\label{Lemma_GrowthBoundAndSpectralBoundOfPositiveSemigroups}[Growth bound and spectral bound of positive semigroups] $ $ \\

The growth bound of a $C_0-$ semigroup is defined as: 

\begin{equation*}
\begin{aligned}
\omega_0(\e^{t \, \G})
&:=
\inf\limits_{\omega \in \R} \{\exists M \geq 1 \,:\, \| \e^{t \, \G} \| \leq M \, \e^{t \, \omega} \text{  for all } t\geq 0 \}, 
\end{aligned}
\end{equation*}

whereas the \emph{spectral bound} is defined as 

\begin{equation*}
\begin{aligned}
s(A)
&:=
\sup \{ \Re[\lambda] \,:\, \lambda \in \sigma(\G) \}
\end{aligned}
\end{equation*}

The space $l^1(\Omega) $ has the property, that 
$ \| \bX + \bY \| = \|\bX\| + \| \bY \|$ for $\bX, \bY \in (\R_{> \, 0})^{\Omega} $), making it a so called \emph{abstract $L$-space}. 

While the spectral bound is in general only less or equal to the growth bound (compare \cite{batkai2017positive}, Proposition 9.33), they coincide for positive $C_0$-semigroups on abstract $L$-spaces (see \cite{batkai2017positive}, Theorem 12.17).

For master equations, we have: 

\begin{equation*}
\begin{aligned}
\| \e^{t \, \G} \|^\text{(op)}_{1}
&\xlongequal{\text{def}}
\sup\limits_{\| \bX \|_1 = 1} 
\underbrace{
\| \e^{t \, \G}  \, \bX \|_{1}
}_{
\sum\limits_{i \in \Omega}
\left| \sum\limits_{j \in \Omega} (\e^{t \, \G})^{(i,j)} \, X^{(j)} \right|
}
\leq 
\sup\limits_{\| \bX \|_1 = 1} 
\sum\limits_{j \in \Omega} |X^{(j)}|
\underbrace{
\sum\limits_{i \in \Omega} (\e^{t \, \G})^{(i,j)}
}_{1}
\\ &=
\sup\limits_{\| \bX \|_1 = 1} 
\underbrace{
\left(\sum\limits_{j \in \Omega} |X^{(j)}| \right)
}_{
| \bX \|_1
}
=
1
=
1 \cdot \e^{t \, 0}. 
\end{aligned}
\end{equation*}

This means that the growth bound (and hence, also the spectral bound) vanish: $ 0 = \omega_0(\e^{t \, \G}) = s(\G) $. 
\end{lemma}


\begin{MyDef}\label{Lemma_ImaginaryAdditiveCyclic}[Imaginary additive cyclic]

A subset $ M\subseteq \C $ is called \emph{imaginary additive cyclic} - or short \emph{cyclic} - if for each $m \in M $ we have $ \Re[m] + i \, \Im[m] \, \Z = \{ \Re[m] + i \, \Im[m] \, z \,:\, z \in \Z\} \subseteq M $. 

Note that a cyclic set can only be bounded, if it is contained in the real line, that is $M \subseteq \R $. 

\end{MyDef}

\begin{lemma}\label{Lemma_BoundarySpectrumContainsOnlyTheSpectralBound}[The boundary spectrum of a positive $C_0$-semigroup with bounded generator contains only the spectral bound] $ $ \\

From \cite{arendt1986one} (part III, "Spectral Theory", Theorem 2.11) we know that the boundary spectrum $ \sigma_b(A) := \{\lambda \in \C \,:\, \Re[\lambda] = s(A)\} $ of a uniformly continuous semigroup is cyclic. Since the spectrum of a bounded operator is again bounded, $\sigma_b(A) $ must be contained in the real line by the comment of definition \ref{Lemma_ImaginaryAdditiveCyclic}. Since the boundary spectrum intersects the real line only at the spectral value, we have $\sigma_b(A) = \{s(A)\} $.

In the case of the master equation, we have: 

\begin{equation}
\begin{aligned}
\sigma_b(\G)
&=
\sigma(\G) \cap \{z \in \C \,:\, \Re[z] = \underbrace{s(\G)}_{0}\}
\xlongequal[\text{and bounded}]{\sigma_b(\G)\text{ is cyclic}}
\{s(\G)\}
\xlongequal{\text{lemma } \ref{Lemma_GrowthBoundAndSpectralBoundOfPositiveSemigroups}}
\{0\}. 
\end{aligned}
\end{equation}

\end{lemma}

\begin{lemma}[lemma of Katznelson-Tzafriri (1984)]\label{Lemma_Katznelson_Tzafriri_1984}
Let $R$ be a linear operator on a Banach space such that $\sup\limits_{n \in \N} \| R^{n}\| < \infty $. Then the intersection of the spectrum of $R$ with the unit circle contains at most the number one, if and only if the difference of two consecutive powers of $R$ converges to zero, that is 

\begin{equation*}
\begin{aligned}
\sigma(R) \cap \{z \in \C \,:\, |z| = 1 \} \subseteq \{1\} 
\iff
\lim\limits_{n \to \infty} \| R^{n} - R^{n+1} \| = 0.
\end{aligned}
\end{equation*}

This Theorem was prooven by Katznelson and Tzafriri in 1984 \cite{katznelson1986power}.

\end{lemma}


\begin{claim}\label{Lemma_SpectralMappingAppliedToBoundaryOfTheUnitCircle}
Let $\emptyset \neq M \subseteq \C $. Then $ \e^{M} \cap \{z \in \C \,:\, |z| = 1 \} = \e^{M \cap i \, \R} $. 

\end{claim}

\begin{proof} 
$ $ \, 
\begin{itemize}
\item["$\subseteq$"]

Let $x \in \e^{M} \cap \{z \in \C \,:\, |z| = 1 \} $. Then there exists a $m \in M $ such that $x = \e^{m} $ and we have: 

\begin{equation*}
\begin{aligned}    
1
&\xlongequal{|x|=1 }
|x|
&\xlongequal{x \in \e^{M} }
|e^{z}|
=
\e^{\Re[z]}, 
\end{aligned}
\end{equation*}

which implies that the real part of $ m $ must vanish and we have $ m = i \, \Im[z] \in M \cap i \, \R $. But this implies $ x = \e^{m} = \e^{i \Im[z]} \in \e^{M \cap i \, \R} $, 

\item["$\supseteq$"]
On the other hand, let $ x \in \e^{M \cap i \, \R }$, which means there exists a $\mu \in \R $ such that $i \, \mu \in M \cap i \, \R $  and $x = \e^{i \, \mu } $. But then we have $ x = \e^{i \, \mu} \in M $, since $i \mu \in M $ and $ | x| = |\e^{i \, \mu}| = 1 $, since $ \mu \in \R $. 

\end{itemize}
\end{proof}


\begin{lemma}\label{Lemma_ltbClosedIn_l1} [$ \ltb $ is closed in $l^1(\Omega )$] $ $ \\ 

First, we note that the limit operator $\LL$ is bounded, since 
\begin{equation*}
\begin{aligned}
\| 
\underbrace{
\LL(\bX)
}_{
\lim\limits_{t \to \infty} \e^{t \, \G } \, \bX 
}
\|_1
&=
\lim\limits_{t \to \infty } 
\underbrace{
\| \e^{t \, \G } \, \bX \|_1
}_{
\leq \, \| \e^{t \, \G } \|_1 \,  \|\bX \|_1
}
\leq 
\lim\limits_{t \to \infty } 
\underbrace{
\| \e^{t \, \G } \|_1
}_{1}
\,  \|\bX \|_1
\leq
\| \bX \|_1. 
\end{aligned}
\end{equation*}

Let $ \bX \in \overline{ \ltb } $, then we know that there exists a sequence $(\bX_n)_{n\in \N} \subseteq \ltb $ converging to $\bX$  and a sequence $ (\bX_{n, *})_{n \in \N} $ such that $ \lim\limits_{t \to \infty} \LL(\bX_n) = \bX_{n, *}$.

First, we show that the sequence $  (\bX_{n, *})_{n \in \N} $  is Cauchy. We can choose the numbers $n$ and $m$ large enough, such that 
\begin{equation*}
\begin{aligned}
\| \bX_{n, *} - \bX_{m, *}\|_1
&=
\| \bX_{n, *}  \pm \LL (\bX_n) \pm \LL (\bX_m) - \bX_{m, *}\|_1
\\ &\leq 
\underbrace{
\| \bX_{n, *} - \LL (\bX_n) \|_1
}_{
< \, \epsilon
}
+
\underbrace{
\| \LL \|_1 \,
}_{
1
} \, 
\underbrace{ \| \bX_n - \bX_m  \|_1
}_{
< \, \epsilon
}
+
\underbrace{
\| \LL (\bX_m) - \bX_{m, *} \|_1
}_{
< \, \epsilon
}, 
\end{aligned}
\end{equation*}

hence the limit $ \bX_* := \lim\limits_{n \to \infty} \bX_{n, *} $ exists in $l^1(\N)$. Further, we have:

\begin{equation*}
\begin{aligned}
 \|  \LL (\bX) - \bX_* \|_1
&=
\| \LL (\bX) \pm \LL (\bX_n) \pm \bX_{n, *} - \bX_* \|_1
 \\ & \leq 
\underbrace{
\| \LL  \|_1
}_{ \leq \, 1}
\, \cdot \,
\underbrace{
\|\bX - \bX_n \|_1
}_{< \, \epsilon}
+
\underbrace{
\| \LL (\bX_n) - (\bX_{n, *}) \|_1
}_{< \, \epsilon}
+
\underbrace{
\| \bX_{n,*} - \bX_*\|_1
}_{
< \, \epsilon
}, 
\end{aligned}
\end{equation*}

that is $ \LL (\bX) \xlongrightarrow[\| \, \|_1]{M \to \infty} \bX_*$ and hence $\bX \in \ltb$. This shows that $ \lme $ is closed in $l^1(\Omega) $.
\end{lemma}


\begin{lemma}[For probability vectors: pointwise convergence implies $\|\cdot \|_{1} $-convergence ] \label{Lemma_ForProbabilityVectors_PointwiseConvergenceImplies1NormConvergence} $ $ \\

Let $\bP_{*} \in \ProbStates $ be a probability vector and $ (\bP_{n})_{ n \in \N } \subseteq \ProbStates $ be a sequence of probability vectors converging pointwise to $\bP_*$, that is $ \bP_{n} \xlongrightarrow[\text{pointwise}]{n \to \infty} \bP_{*} $. Then this convergence is also with rexpect to the $\| \cdot \|_{1}$-norm: 
$ \bP_{n} \xlongrightarrow[\| \cdot \|_{1}]{n \to \infty} \bP_{*} $: 

\end{lemma}

\begin{proof}

Let $F \subseteq \Omega $ be a finite set such that 

\begin{itemize}
\item[(i)] 
$ \Bigl| P_{n}^{(f)} - P_{*}^{(f)} \Bigr| < \frac{\epsilon}{|F|}  $ for all $f \in F $, 

\item[(ii)] 
$ \sum\limits_{f \in F} P_{n}^{(f)} \geq 1 - \epsilon $

\item[(iii)] 
$ \sum\limits_{f \in F} P_{*}^{(f)} \geq 1 - \epsilon $

\end{itemize}

\begin{equation*}
\begin{aligned}
\| \bP_{n} - \bP_{*} \|_{1} 
&=
\sum\limits_{\omega \in \Omega} \Bigl| P_{n}^{(\omega)} - P_{*}^{(\omega)} \Bigr|
\\ &= 
\sum\limits_{ f \in F } \Bigl| P_{n}^{(f)} - P_{*}^{(f)} \Bigr| + 
\sum\limits_{\omega \in \Omega \backslash F } 
\underbrace{
    \Bigl| P_{n}^{(\omega)} - P_{*}^{(\omega)} \Bigr|
}_{
    \leq \, P_{n}^{(\omega)} + P_{*}^{(\omega)}
}
\\ &\leq 
\underbrace{
        \sum\limits_{ f \in F } 
    \overbrace{
        \Bigl| P_{n}^{(f)} - P_{*}^{(f)} \Bigr|
        }^{
        < \frac{\epsilon}{|F|}
        } 
}_{
    < \, \epsilon 
}
+
\underbrace{
    \left( \sum\limits_{\omega \in \Omega \backslash F }  P_{n}^{(\omega)} \right)
}_{
    \leq \, \epsilon
}
+
\underbrace{
    \left( \sum\limits_{\omega \in \Omega \backslash F }  P_{*}^{(\omega)} \right)
}_{
    \leq \, \epsilon
}
\overset{(i)-(iii)}<
\\ &<
3 \, \epsilon. 
\end{aligned}
\end{equation*}

\end{proof}



    



\subsection{Auxiliary calculations}\label{Section_AuxiliaryCalculations}

\begin{lemma} \label{Fact_ApplyingRatioTest}
The expression

\begin{equation}
\begin{aligned}
\sum\limits_{N \in \N} a_{N}
:=
\sum\limits_{N \in \N}
q^{c^{(N)}} \prod\limits_{i=1}^{N-1} (1+q^{c^{(i)}}). 
\end{aligned}
\end{equation}

converges, if $\bc \in \{ (n)_{n \in \N }, (n^2)_{n \in \N }, (2^{n})_{n \in \N } \} $ and diverges, if $\bc = c \cdot \bone$. 

\end{lemma}

\begin{proof}
\begin{equation*}
\begin{aligned}
\frac{a_{N+1}}{a_N}
&=
\frac{q^{c^{(N+1)}}\, \prod\limits_{i=1}^{N} (1+q^{c^{i}}) }{q^{c^{(N)}}\, \prod\limits_{i=1}^{N-1} (1+q^{c^{i}}) }
=
\frac{q^{c^{(N+1)}}\, (1+q^{c^{(N)}}) }{q^{c^{(N)}}}
=
\begin{cases}
1+q^{c}>1 &\text{  , if } \bc = c \, \bone 
\\
q \, (1+q^{N}) \xlongrightarrow{N \to \infty} q < 1 &\text{  , if } \bc = (n)_{n \in \N}
\end{cases}
\end{aligned}
\end{equation*}

\end{proof}


\begin{lemma}[proof of equation \ref{Eq_Claim_ExplicitFormOfZ}] \label{Proof_Eq_Claim_ExplicitFormOfZ}
\end{lemma}

For the induction start we get $ Z_*^{(0)} = Z_*^{(0)} $. For the induction step we consider 

\begin{equation*}
\begin{aligned}
Z_*^{(k+n+1)}
&=
Z_*^{(k+n)} \, \cdot \, 
\frac{
\g_{k+n \to k+n+1}
}{
\g_{k+n+1 \to k+n}
}
+ \frac{c}{\g_{k+n+1 \to k+n}}
 \\ &=
\left[
Z_*^{(k)} \, 
\frac{
\prod\limits_{\alpha=k}^{k+n-1} \g_{\alpha \to \alpha + 1}
}{
\prod\limits_{\beta=k+1}^{k+n} \g_{\beta \to \beta - 1}
}
+ 
c
\sum\limits_{j=1}^{n}
\frac{
\prod\limits_{\alpha=k+j}^{k+n-1} \g_{\alpha \to \alpha + 1}
}{
\prod\limits_{\beta=k+j}^{k+n} \g_{\beta \to \beta - 1}
}  
\right]
\cdot \, 
\frac{
\g_{k+n \to k+n+1}
}{
\g_{k+n+1 \to k+n}
}
+ \frac{c}{\g_{k+n+1 \to k+n}}
 \\ &=
\left[
Z_*^{(k)} \, 
\frac{
\prod\limits_{\alpha=k}^{k+n} \g_{\alpha \to \alpha + 1}
}{
\prod\limits_{\beta=k+1}^{k+n+1} \g_{\beta \to \beta - 1}
}
+
c
\sum\limits_{j=1}^{n+1}
\frac{
\prod\limits_{\alpha=k+j}^{k+n} \g_{\alpha \to \alpha + 1}
}{
\prod\limits_{\beta=k+j}^{k+n+1} \g_{\beta \to \beta - 1}
}  
\right]
 \\
 \text{     and }
 \\ 
Z_*^{(k-n-1)}
&=
Z_*^{(k-n)} 
\, \cdot \, 
\frac{
\g_{k-n \to k-n-1}
}{
\g_{k-n-1 \to k-n}
}
- \frac{c}{\g_{k-n-1 \to k-n}}
 \\ &=
\left[
Z_*^{(k)} \, 
\frac{
\prod\limits_{\beta=k-n}^{k} \g_{\beta \to \beta - 1}
}{
\prod\limits_{\alpha=k-n-1}^{k-1} \g_{\alpha \to \alpha + 1}
}
-
c
\sum\limits_{j=1}^{n}
\frac{
\prod\limits_{\beta=k-n+1}^{k-j}\g_{\beta \to \beta - 1}
}{
\prod\limits_{\alpha=k-n}^{k-j} \g_{\alpha \to \alpha + 1}
}  
\right]
\, \cdot \, 
\frac{
\g_{k-n \to k-n-1}
}{
\g_{k-n-1 \to k-n}
}
- \frac{c}{\g_{k-n-1 \to k-n}}
 \\ &=
Z_*^{(k)} \, 
\frac{
\prod\limits_{\beta=k-n}^{k} \g_{\beta \to \beta - 1}
}{
\prod\limits_{\alpha=k-n-1}^{k-1} \g_{\alpha \to \alpha + 1}
}
-
c
\sum\limits_{j=1}^{n+1}
\frac{
\prod\limits_{\beta=k-n}^{k-j}\g_{\beta \to \beta - 1}
}{
\prod\limits_{\alpha=k-n-1}^{k-j} \g_{\alpha \to \alpha + 1}
}  
\end{aligned}
\end{equation*}


\begin{lemma}\label{lemma_ConsequencesFromConstantFlow}[Proof of $\bZ_* = \frac{Z_*^{(0)}}{X_*^{(0)}} - c \, \tilde{\bP}_* $ ]

\begin{equation*}
\begin{aligned}
Z_*^{(k+n)}
&=
X_k \, 
\frac{
\prod\limits_{\alpha=k}^{k+n-1} \g_{\alpha \to \alpha + 1}
}{
\prod\limits_{\beta=k+1}^{k+n} \g_{\beta \to \beta - 1}
}
- 
c
\sum\limits_{j=1}^{n}
\frac{
\prod\limits_{\alpha=k+j}^{k+n-1} \g_{\alpha \to \alpha + 1}
}{
\prod\limits_{\beta=k+j}^{k+n} \g_{\beta \to \beta - 1}
} 
\\ \xLongrightarrow[\eqref{ConstantFlow}]{k=0} 
Z_*^{(n)}
&=
Z_*^{(0)} \, 
\frac{
\prod\limits_{\alpha=0}^{n-1} \g_{\alpha \to \alpha + 1}
}{
\prod\limits_{\beta=1}^{n} \g_{\beta \to \beta - 1}
} 
\, \cdot \, 
\frac{
\prod\limits_{\alpha \leq -1} \g_{\alpha \to \alpha + 1} \,
\prod\limits_{\beta \geq 1} \g_{\beta \to \beta - 1}
}{
\underbrace{
\prod\limits_{\alpha \leq -1} \g_{\alpha \to \alpha + 1} \,
\prod\limits_{\beta \geq 1} \g_{\beta \to \beta - 1}
}_{
X_*^{(0)}
}
}
- 
c
\underbrace{
\sum\limits_{j=1}^{n}
\frac{
\prod\limits_{\alpha=j}^{n-1} \g_{\alpha \to \alpha + 1}
}{
\prod\limits_{\beta=j}^{n} \g_{\beta \to \beta - 1}
} 
}_{
Y_*^{(n)}
}
\\ &=
\frac{
Z_*^{(0)} \, 
}{
X_*^{(0)}
}
 \, 
\underbrace{
\left(
\prod\limits_{\alpha \leq n-1} \g_{\alpha \to \alpha + 1} \, 
\prod\limits_{\beta \geq n+1}  \g_{\beta \to \beta - 1}
\right)
}_{
X_*^{(n)}
}
- 
c \, 
Y_*^{(n)}
\\ \text{ and } \\ 
Z_*^{(k-n)}
&=
Z_*^{(k)} \, 
\frac{
\prod\limits_{\beta=k-n+1}^{k} \g_{\beta \to \beta - 1}
}{
\prod\limits_{\alpha=k-n}^{k-1} \g_{\alpha \to \alpha + 1}
}
+
c
\sum\limits_{j=1}^{n}
\frac{
\prod\limits_{\beta=k-n+1}^{k-j} \g_{\beta \to \beta - 1}
}{
\prod\limits_{\alpha=k-n}^{k-j} \g_{\alpha \to \alpha + 1}
} 
\\ \xLongrightarrow{k=0}
Z_*^{(-n)}
&=
Z_*^{(0)}
\, 
\frac{
\prod\limits_{\beta=-n+1}^{0} \g_{\beta \to \beta - 1}
}{
\prod\limits_{\alpha=-n}^{-1} \g_{\alpha \to \alpha + 1}
}
\, \cdot \, 
\frac{
\prod\limits_{\alpha \leq -1} \g_{\alpha \to \alpha + 1} \, 
\prod\limits_{\beta \geq 1}  \g_{\beta \to \beta - 1}
}{
\underbrace{
\left(
\prod\limits_{\alpha \leq -1} \g_{\alpha \to \alpha + 1} \, 
\prod\limits_{\beta \geq 1}  \g_{\beta \to \beta - 1}
\right)
}_{
X_*^{(0)}
}
}
+
c \, 
\underbrace{
\sum\limits_{j=1}^{n}
\frac{
\prod\limits_{\beta=-n+1}^{-j} \g_{\beta \to \beta - 1}
}{
\prod\limits_{\alpha=-n}^{-j} \g_{\alpha \to \alpha + 1}
} 
}_{
-Y_*^{(-n)}
}
\\ &=
%
\frac{
Z_*^{(0)}
}{
X_*^{0}
}
\, 
\underbrace{
\left(
\prod\limits_{\alpha\leq -n-1} \g_{\alpha\to \alpha + 1}
\prod\limits_{\beta \geq -n+1} \g_{\beta\to\beta - 1}
\right)
}_{
X_*^{(-n)}
}
-
c \, Y_*^{(-n)}
\\ 
\Longrightarrow 
\bZ_*
&=
\frac{
Z_*^{(0)} \, 
}{
X_*^{(0)}
}\,
\bX_*
-
c \, \tilde{\bP}_*
\in \Span(\bX_*, \tilde{\bP}_*)
\end{aligned}    
\end{equation*}

\end{lemma}


\begin{lemma}[The `additional' stationary sequence for the linear chains with two open ends]\label{Lemma_TheAdditionalStationarySequenceForTheLinearChainWithTwoOpenEnds} $ $ \\ 

Let $\G$ the generator of the linear chain with one open end and $F = \{-N, \dots, -1, 0, 1, \dots, N\} $ be a finite subnetwork for some $N \in \N$. For

\begin{equation*}
\begin{aligned}
\bY_*
&=
\left(
\left( 
(-1)
\sum\limits_{j=1}^{-m} 
\frac{
\prod\limits_{\beta=m+1}^{-j}\g_{\beta \to \beta-1}
}{
\prod\limits_{\alpha=m}^{-j}\g_{\alpha \to \alpha+1}
}
\right)_{m \in \Z_{<0}}
, 0, 
\left( 
\sum\limits_{j=1}^{n}
\frac{
\prod\limits_{\alpha=j}^{n-1} \g_{\alpha \to \alpha + 1} 
}{
\prod\limits_{\beta=j}^{n} \g_{\beta \to \beta - 1}
}
\right)_{n \in \Z_{>0}}
\right) = \\ 
&=
\left(
\left( 
\frac{(-1)}{\g_{m \to m+1}}
\sum\limits_{j=1}^{-m} 
\prod\limits_{\epsilon=m+1}^{-j}
\frac{\g_{\epsilon \to \epsilon-1}}{\g_{\epsilon \to \epsilon+1}}
\right)_{m \in \Z_{<0}}
, 0 , 
\left( 
\frac{1}{\g_{n \to n-1}}
\sum\limits_{j=1}^{n} 
\prod\limits_{\epsilon=j}^{n-1}
\frac{\g_{\epsilon \to \epsilon+1}}{\g_{\epsilon \to \epsilon-1}}
\right)_{n \in \Z_{>0}}
\right)
\end{aligned}
\end{equation*}

we have $ \G^{[F]} \, \bY_{*}^{F} = \bE_{-N} - \bE_{N} \xlongrightarrow[\text{pointwise}]{N \to \infty} \bzero $. 

\end{lemma}

\begin{proof} $ $ \\ 

\begin{itemize}

\item[$n=-N$]

\begin{equation*}
\begin{aligned}
Y_{*}^{(-N+1)} \, &\g_{-N+1 \to -N } - Y_{*}^{(-N)}\, \g_{-N \to -N+1}
= \\ &=
(-1) \, \cdot \, 
    \underbrace{
    \frac{
        \g_{ -N+1 \to -N}
    }{
        \g_{-N+1 \to -N+2}
    } \, 
    \sum\limits_{j=1}^{N-1}
    \prod\limits_{\epsilon=-N+2}^{-j}
    \frac{
        \g_{\epsilon\to \epsilon - 1}
    }{
        \g_{\epsilon\to \epsilon + 1}
    }
}_{
\sum\limits_{j=1}^{N-1}
 \prod\limits_{\epsilon=-N+1}^{-j}
     \frac{
        \g_{\epsilon\to \epsilon - 1}
    }{
        \g_{\epsilon\to \epsilon + 1}
    }
}
+
\underbrace{
\frac{
    \g_{-N \to -N+1}
    }{
    \g_{-N \to -N+1}
    }
}_{1} \, 
\sum\limits_{j=1}^{N} \prod\limits_{\epsilon=-N+1}^{-j} 
\frac{
    \g_{\epsilon\to \epsilon - 1}
}{
    \g_{\epsilon\to \epsilon + 1}
}
\\ &=
 \prod\limits_{\epsilon=-N+1}^{-N} 
\frac{
    \g_{\epsilon\to \epsilon - 1}
}{
    \g_{\epsilon\to \epsilon + 1}
}
\\ &= 
1. 
\end{aligned}
\end{equation*}

\item[$n \in \{-N+1, \dots, -1\}$]
\begin{equation*}
\begin{aligned}
Y_*^{(n-1)} \g_{n-1 \to n}
+
Y_*^{(n+1)} \g_{n+1 \to n}
&=
\sum\limits_{j=1}^{n-1}
\underbrace{
    \left(
    \frac{
    \prod\limits_{\alpha=j}^{n-2} \g_{\alpha \to \alpha + 1} 
    }{
    \prod\limits_{\beta=j}^{n-1} \g_{\beta \to \beta - 1}
    }
    \g_{n-1 \to n}
    \right)
}_{
    \frac{
    \prod\limits_{\alpha=j}^{n-1} \g_{\alpha \to \alpha + 1} 
    }{
    \prod\limits_{\beta=j}^{n-1} \g_{\beta \to \beta - 1}
    }
}
+
\sum\limits_{j=1}^{n+1}
\underbrace{
    \left(
    \frac{
    \prod\limits_{\alpha=j}^{n} \g_{\alpha \to \alpha + 1} 
    }{
    \prod\limits_{\beta=j}^{n+1} \g_{\beta \to \beta - 1}
    }
    \g_{n+1 \to n}
    \right)
}_{
    \frac{
    \prod\limits_{\alpha=j}^{n} \g_{\alpha \to \alpha + 1} 
    }{
    \prod\limits_{\beta=j}^{n} \g_{\beta \to \beta - 1}
    }
}
\\ &=
\sum\limits_{j=1}^{n-1}
\frac{
\prod\limits_{\alpha=j}^{n-1} \g_{\alpha \to \alpha + 1} 
}{
\prod\limits_{\beta=j}^{n-1} \g_{\beta \to \beta - 1}
}
\underbrace{
\left[
1 + \frac{\g_{n \to n+1}}{\g_{n \to n-1}}  
\right]
}_{
\frac{
\g_{n \to n+1} + \g_{n \to n-1}
}{
\g_{n \to n-1}
}
}
+ 
\underbrace{
\left( \frac{\g_{n \to n+1}}{\g_{n \to n-1}} + 1 \right)
}_{
\frac{
\g_{n \to n+1} + \g_{n \to n-1}
}{
\g_{n \to n-1}
}
}
\\ &=
\left( 
\g_{n \to n+1} + \g_{n \to n-1}
\right)
\underbrace{
\left(
\sum\limits_{j=1}^{n}
\frac{
\prod\limits_{\alpha=j}^{n-1} \g_{\alpha \to \alpha + 1} 
}{
\prod\limits_{\beta=j}^{n} \g_{\beta \to \beta - 1}
}
\right)
}_{
Y_*^{(n)}
}
\end{aligned}
\end{equation*}

\item[$(n=0)$]
\begin{equation*}
\begin{aligned}
\underbrace{
Y_*^{(-1)}
}_{
\frac{-1}{\g_{-1 \to 0}}
}
 \, \g_{-1 \to 0} 
+
\underbrace{
Y_*^{(1)}
}_{
\frac{1}{\g_{1 \to 0}}
}
\, \g_{1 \to 0}
&=
0
=
\underbrace{
    Y_*^{(0)}
}_{
    0
}
\, \g_{0 \to}. 
\end{aligned}
\end{equation*}

\item[$n \in \{1, \dots, N-1\}$]
\begin{equation*}
\begin{aligned}
Y_*^{(n-1)} \g_{n-1 \to n}
+
Y_*^{(n+1)} \g_{n+1 \to n}
&=
(-1)\sum\limits_{j=1}^{-n-1}
\underbrace{
\frac{
\prod\limits_{\beta=n}^{-j} \g_{\beta \to \beta - 1}
}{
\prod\limits_{\alpha=n-1}^{-j} \g_{\alpha \to \alpha + 1} 
}
\g_{n-1 \to n}
}_{
\frac{
\prod\limits_{\beta=n}^{-j} \g_{\beta \to \beta - 1}
}{
\prod\limits_{\alpha=n}^{-j} \g_{\alpha \to \alpha + 1} 
}
}
-
\sum\limits_{j=1}^{-n+1}
\underbrace{
\frac{
\prod\limits_{\beta=n+2}^{-j} \g_{\beta \to \beta - 1}
}{
\prod\limits_{\alpha=n+1}^{-j} \g_{\alpha \to \alpha + 1} 
}
\g_{n+1 \to n}
}_{
\frac{
\prod\limits_{\beta=n+1}^{-j} \g_{\beta \to \beta - 1}
}{
\prod\limits_{\alpha=n+1}^{-j} \g_{\alpha \to \alpha + 1} 
}
} 
\\ &=
(-1)\sum\limits_{j=1}^{-n+1}
\frac{
\prod\limits_{\beta=n+1}^{-j} \g_{\beta \to \beta - 1}
}{
\prod\limits_{\alpha=n+1}^{-j} \g_{\alpha \to \alpha + 1} 
}
\underbrace{
\left[
1 + \frac{\g_{n \to n-1}}{\g_{n \to n+1}}  
\right]
}_{
\frac{
\g_{n \to n-1} + \g_{n \to n+1}
}{
\g_{n \to n+1}
}
}
+ 
\underbrace{
\left[
1 + \frac{\g_{n \to n-1}}{\g_{n \to n+1}}  
\right]
}_{
\frac{
\g_{n \to n-1} + \g_{n \to n+1}
}{
\g_{n \to n+1}
}
}
\\ &=
\left( 
\g_{n \to n-1} + \g_{n \to n+1}
\right)
\underbrace{
\left(
(-1)
\sum\limits_{j=1}^{-n}
\frac{
\prod\limits_{\beta=n+1}^{-j} \g_{\beta \to \beta - 1}
}{
\prod\limits_{\alpha=n}^{-j} \g_{\alpha \to \alpha + 1} 
}
\right)
}_{
Y_*^{(n)}
}
\end{aligned}
\end{equation*}

\item[$n=N$]
\begin{equation*}
\begin{aligned}
Y_{*}^{(N-1)} \, &\g_{N-1 \to N } - Y_{*}^{(N)} \, \g_{N \to N-1}
= \\ &=
\underbrace{
    \frac{
        \g_{ N-1 \to N}
    }{
        \g_{N-1 \to N-2}
    } \, 
    \sum\limits_{j=1}^{N-1}
    \prod\limits_{\epsilon=j}^{N-2}
    \frac{
        \g_{\epsilon\to \epsilon + 1}
    }{
        \g_{\epsilon\to \epsilon - 1}
    }
}_{
\sum\limits_{j=1}^{N-1}
 \prod\limits_{\epsilon=j}^{N-1}
     \frac{
        \g_{\epsilon\to \epsilon + 1}
    }{
        \g_{\epsilon\to \epsilon - 1}
    }
}
-
\underbrace{
\frac{
    \g_{N \to N-1}
    }{
    \g_{N \to N-1}
    }
}_{1} \, 
\sum\limits_{j=1}^{N} \prod\limits_{\epsilon=j}^{N-1} 
\frac{
    \g_{\epsilon\to \epsilon + 1}
}{
    \g_{\epsilon\to \epsilon - 1}
}
\\ &=
(-1) \, \cdot \,  \prod\limits_{\epsilon=N}^{N-1} 
\frac{
    \g_{\epsilon\to \epsilon + 1}
}{
    \g_{\epsilon\to \epsilon - 1}
}
\\ &= 
-1. 
\end{aligned}
\end{equation*}

\end{itemize}

\end{proof}


\begin{lemma}[No `additional' stationary sequence for the linear chain with \emph{one} open end]\label{Lemma_NoAdditionalStationarySequenceForTheLinearChainWithOneOpenEnd} $ $ \\

Let $\G$ the generator of the linear chain with one open end and $F = \{0, 1, \dots, N\} $ be a finite subnetwork for some $N \in \N$. For 

\begin{equation*}
\begin{aligned}
\bY_{*}
&=
\left(
0, 
\left( 
\sum\limits_{j=1}^{n}
\frac{
\prod\limits_{\alpha=j}^{n-1} \g_{\alpha \to \alpha + 1} 
}{
\prod\limits_{\beta=j}^{n} \g_{\beta \to \beta - 1}
}
\right)_{n \in \N }
\right) 
=
\left(
 0 , 
\left( 
\frac{1}{\g_{n \to n-1}}
\sum\limits_{j=1}^{n} 
\prod\limits_{\epsilon=j}^{n-1}
\frac{\g_{\epsilon \to \epsilon+1}}{\g_{\epsilon \to \epsilon-1}}
\right)_{n \in \N } 
\right)
\\ &= 
\left(
0, \frac{1}{\g_{1 \to 0 }}, 
\frac{1}{\g_{2 \to 1 } }  \Bigl( \frac{\g_{1 \to 2}}{\g_{1 \to 0}} + 1 \Bigr) , \dots 
\right)
\end{aligned}
\end{equation*}

we have $ \G^{[F]} \, \bY_{*}^{F} = \bE_{0} - \bE_{N} \xlongrightarrow[\text{pointwise}]{N \to \infty} \bE_{0}$. 
\end{lemma}

\begin{proof} $ $ \\

\begin{itemize}

\item[$n=0$]

\begin{equation*}
\begin{aligned}
\bigl(\G^{[F]} \, \bY_{*}^{F}\bigr)^{(n)}
&=
\underbrace{ 
    Y_{*}^{(1)}
}_{
    \frac{1}{\g_{1 \to 0} }
}
\, \g_{1 \to 0}
-
\underbrace{
    Y_{*}^{(0)}
}_{
    0
}\, \g_{0 \to 1}
=
1. 
\end{aligned}
\end{equation*}

\item[$n \in \{1, \dots, N-1\}$]

\begin{equation*}
\begin{aligned}
Y_*^{(n-1)} \g_{n-1 \to n}
+
Y_*^{(n+1)} \g_{n+1 \to n}
&=
\sum\limits_{j=1}^{n-1}
\underbrace{
    \left(
    \frac{
    \prod\limits_{\alpha=j}^{n-2} \g_{\alpha \to \alpha + 1} 
    }{
    \prod\limits_{\beta=j}^{n-1} \g_{\beta \to \beta - 1}
    }
    \g_{n-1 \to n}
    \right)
}_{
    \frac{
    \prod\limits_{\alpha=j}^{n-1} \g_{\alpha \to \alpha + 1} 
    }{
    \prod\limits_{\beta=j}^{n-1} \g_{\beta \to \beta - 1}
    }
}
+
\sum\limits_{j=1}^{n+1}
\underbrace{
    \left(
    \frac{
    \prod\limits_{\alpha=j}^{n} \g_{\alpha \to \alpha + 1} 
    }{
    \prod\limits_{\beta=j}^{n+1} \g_{\beta \to \beta - 1}
    }
    \g_{n+1 \to n}
    \right)
}_{
    \frac{
    \prod\limits_{\alpha=j}^{n} \g_{\alpha \to \alpha + 1} 
    }{
    \prod\limits_{\beta=j}^{n} \g_{\beta \to \beta - 1}
    }
}
\\ &=
\sum\limits_{j=1}^{n-1}
\frac{
\prod\limits_{\alpha=j}^{n-1} \g_{\alpha \to \alpha + 1} 
}{
\prod\limits_{\beta=j}^{n-1} \g_{\beta \to \beta - 1}
}
\underbrace{
\left[
1 + \frac{\g_{n \to n+1}}{\g_{n \to n-1}}  
\right]
}_{
\frac{
\g_{n \to n+1} + \g_{n \to n-1}
}{
\g_{n \to n-1}
}
}
+ 
\underbrace{
\left( \frac{\g_{n \to n+1}}{\g_{n \to n-1}} + 1 \right)
}_{
\frac{
\g_{n \to n+1} + \g_{n \to n-1}
}{
\g_{n \to n-1}
}
}
\\ &=
\left( 
\g_{n \to n+1} + \g_{n \to n-1}
\right)
\underbrace{
\left(
\sum\limits_{j=1}^{n}
\frac{
\prod\limits_{\alpha=j}^{n-1} \g_{\alpha \to \alpha + 1} 
}{
\prod\limits_{\beta=j}^{n} \g_{\beta \to \beta - 1}
}
\right)
}_{
Y_*^{(n)}
} \text{   , hence} 
\\ 
\bigl(\G^{[F]} \, \bY_{*}^{F}\bigr)^{(n)}
&=
0. 
\end{aligned}
\end{equation*}

\item[$n=N $]

\begin{equation*}
\begin{aligned}
\bigl(\G^{[F]} \, \bY_{*}^{F}\bigr)^{(n)}
&= 
Y_{*}^{(N-1)} \, \g_{N-1 \to N}
-
\underbrace{
    Y_{*}^{(N)}
}_{
    0
}\, \g_{N \to N-1}
\\ &= 
\sum\limits_{j=1}^{N-1}
\frac{\prod\limits_{\alpha=j}^{N-1} \g_{\alpha \to \alpha + 1} }{
\prod\limits_{\beta=j}^{N-1} \g_{\alpha \to \beta- 1} 
}
- \frac{\g_{N \to N-1}}{\g_{N \to N-1}} \, 
\sum\limits_{j=1}^{N}
\frac{\prod\limits_{\alpha=j}^{N-1} \g_{\alpha \to \alpha + 1} }{
\prod\limits_{\beta=j}^{N-1} \g_{\alpha \to \beta- 1} 
}
\\ &= 
\sum\limits_{j=1}^{N}
\frac{\prod\limits_{\alpha=j}^{N-1} \g_{\alpha \to \alpha + 1} }{
\prod\limits_{\beta=j}^{N-1} \g_{\alpha \to \beta- 1} 
}
\\ &= 
-1. 
\end{aligned}
\end{equation*}

\end{itemize}

\end{proof}



\subsection{Finite dimensional hypercubes}\label{Section_FiniteDimensionalHypercubes}

The network of an $N$-dimensional hypercube $Q_N$ are given by $\System_{[Q_N]} = (\Omega_{[Q_N]}, \Edge_{[Q_N]})$, with 
\begin{equation}\label{Eq_Network_NDimensionalHyperCube}
\begin{aligned}
\Omega_{[Q_N]} 
&=
\{0,1\}^N
\\
\Edge_{[Q_N]} 
&=
\{(\alpha, \beta) \in (\Omega_{[Q_N]})^{2}  \,:\, \alpha - \beta = \pm \be_i \text{ for some } i \in \{1, \dots, N\} \}
\end{aligned}
\end{equation}

Moreover, we assume that for some states $\alpha \in \Omega $ such that $\alpha + \be_n \in \Omega $ the sufficient conditions for detailed balance \eqref{Eq_SufficientConditionForDetailedBalance} is satisfied, i.e. there exists a sequence $ \bc = (c^{(n)}) \in \Order(n) $ such that

\begin{equation*}
\begin{aligned}
\frac{
    \g_{\alpha \to \alpha + \be_n} 
}{
    \g_{\alpha \gets \alpha + \be_n}
}
&=
q^{c^{(n)}}, \text{ for some } q \in (0,1)
\end{aligned}
\end{equation*}

Then the stationary solution of the master equation on the hypercube is given by 

\begin{equation}\label{Eq_StationarySolutionForFiniteDimensionalHypercube}
\begin{aligned}
\bp^{[Q_N]}
&=
\bigotimes\limits_{i=1}^{N} \frac{(1, q^{c^{(n)}})}{1+q^{c^{(n)}}}
=
\bigotimes\limits_{i=1}^{N}\frac{1}{1+q^{c^{(n)}}} \colvec{2}{1}{q^{c^{(n)}}}
\end{aligned}
\end{equation}

Let us consider the stationary solution of the $N+1$-dimensional hypercube. 
\begin{equation}
\begin{aligned}
\bp_*^{[Q_{N+1}]}
&=
\frac{1}{\Zen_{N+1}} 
\left(
q^{\bw_{N+1} \cdot \bc} 
\right)_{\bw_{N+1} \in \{0,1\}^{N+1} }
=
\frac{1}{\Zen_{N+1}} 
\left(
\bigl(
    q^{\bw_{N} \cdot \bc} 
\bigr)_{\bw_{N} \in \{0,1\}^{N} },
q^{c^{(N+1)}} 
\bigl(
q^{\bw_{N} \cdot \bc} 
\bigr)_{\bw_{N} \in \{0,1\}^{N} }
\right)
\\ &=
\frac{1}{\Zen_{N+1}}
\bigl(1, q^{c^{(N+1)}} \bigr)
\bigotimes 
\bigl(
q^{\bw_{N} \cdot \bc} 
\bigr)_{\bw_{N} \in \{0,1\}^{N} }
=
\frac{(1, q^{c^{(N+1)}})}{1+ q^{c^{(N+1)}}}
\bigotimes \bp_*^{[Q_N]}. 
\end{aligned}
\end{equation}

With $\bp_*^{[Q_{1}]} = \frac{(1,q^{c^{(1)}})}{1+q^{c^{(1)}}}$, the exact form of the stationary solution as shown in equation \eqref{Eq_StationarySolutionForFiniteDimensionalHypercube} follows.


\subsection{Relation between a CTMC and the embedded DTMC }\label{Sec_RelationBetweenACTMCAndTheEmbeddedDTMC}

\begin{MyDef}[recurrence] $ $ \\ 
A (discrete- or continuous-time) Markov chain is called \emph{recurrent}, if it almost surely returns to the initial state, that is if

\begin{equation}
\begin{aligned}
1
&=
\Prob\left(t_R < \infty \,|\, X_0 = \state{i} \right)
\\ &= 
\Prob\left(X_{n \atop t} = \state{i} \text{ for some } {n \in \N \atop t \in \R_{> 0}} \,|\, X_0 = \state{i} \right) 
\\ &= 
\begin{cases}
\sum\limits_{n \in \N } 
\Prob\left(X_{n } = \state{i}  \,|\, X_0 = \state{i} \right)  &\text{, for a DTMC }
\\
\int_0^\infty \rho_{t_R}(t) \d t  &\text{, for a DTMC. } 
\end{cases}
\end{aligned}
\end{equation}

\end{MyDef}

\begin{lemma}[Recurrence for Markov chains] $ $ \\ 

A CTMC is \emph{recurrent}, if and only if the embedded, DTMC is recurrent. 
\end{lemma}

\begin{proof}

\begin{equation}
\begin{aligned}
\Prob\left(
t_R < \infty \,|\, X_0 = \state{\omega_0}
\right)
&=
\int_0^\infty
\underbrace{
\rho_{t_R}(t)
}_{
\sum\limits_{n \in \N_0} \rho_{t_R}(n, t)
}
\d t
=
\int_0^\infty 
\sum\limits_{n \in \N_0}
\rho_{t_R}(n, t)
\d t
\\ &\xlongequal[\text{convergence}]{\text{monotone}}
\sum\limits_{n \in \N_0}
\underbrace{
\int_0^\infty 
\rho_{t_R}(n, t)
\d t
}_{
\Prob 
\left(
X_{t_n} = \state{\omega_0} \, | \, X_0= \state{\omega_0}
\right)
}
\\ &=
\Prob \left(
\tilde{T}_R < \infty \, | \, X_0= \state{\omega_0}
\right), 
\end{aligned}
\end{equation}

where $\tilde{T}_R $ is the return time for the embedded, discrete-time Markov chain. 

\end{proof}

\begin{center}
\hspace*{-5mm}
\begin{tikzpicture} 
\draw[very thick] (0.0, 0.0) --  (0.0,  -6.0) ; 
\draw[very thick] (3.5, 0.0) --  (3.5,  -6.0) ; 
\draw[very thick] (8.5, 0.0) --  (8.5,  -6.0) ; 
\draw[very thick] (13.5, 0.0) -- (13.5, -6.0) ; 
\draw[very thick] (19.0, 0.0) -- (19.0, -6.0) ; 
\draw[very thick] (0.0 , -0.0) -- (19.0, -0.0) ; 
\node[align=left] at (6.0, -0.5)   { DTMC } ;
\node[align=left] at (11.0, -0.5)   { CTMC }   ;
\node[align=left] at (16.0, -0.5)   { embedded DTMC  }   ;
\draw[very thick] (0.0 , -1.0) -- (19.0, -1.0) ; 
\node[align=left] at (1.5, -1.5)   { irreducible } ;
\node[align=left] at (6.0, -1.5)   { strongly \\ connected} ; 
\node[align=left] at (11.0, -1.5)   { strongly \\ connected }   ;
\node[align=left] at (16.5, -1.5)  { strongly \\  connected  }   ;
\draw[very thick] (0.0 , -2.0) -- (19.0, -2.0) ; 
\node[align=left] at (1.5, -2.5)   { recurrent } ;
\node[align=left] at (6.0, -2.5)   {$ \Prob(T_R < \infty) = 1 $  } ;
\node[align=left] at (10.5, -2.5)   { $ \Prob(t_R < \infty) = 1 $ }   ;
\node[align=left] at (16.0, -2.5)   { $ \Prob(\tilde{T}_R < \infty) = 1 $  }   ; 
\draw[very thick] (0.0 , -3.0) -- (19.0, -3.0) ;
\node[align=left] at (1.7, -4.5)   { positive recurrent  } ;
\node[align=left] at (6.0, -4.5)   { 
$ \|\bN \|_1 = E[T_R] < \infty $  \\
$ \frac{\bN}{\| \bN \|_1}  \in \ProbStates $ \\ 
$ \frac{\bN}{\| \bN \|_1}  \in \Kern(\Id - Q) $ \\ 
$ \frac{\bN}{\| \bN \|_1 } > \bzero $ \\ 
} ;
\node[align=left] at (11.0, -4.5)   { 
$  \|\bT \|_1 = E[t_R] < \infty $ \\ 
$ \frac{\bT}{\| \bT \|_1}  \in \ProbStates $ \\ 
$ \frac{\bT}{\| \bT \|_1 }  \in \Kern(\G) $ \\ 
$ \frac{\bT}{\| \bT \|_1 }  > \bzero $ \\ 

}   ;
\node[align=left] at (16.2, -4.5)   { 
$  \| \tilde{\bN} \|_1 = E[\tilde{T_R}] < \infty $ \\
$ \frac{\tilde{\bN}}{\| \tilde{\bN} \|_1}  \in \ProbStates $ \\ 
$ \frac{\tilde{\bN}}{\| \tilde{\bN} \|_1}  \in \Kern(\Id - Q) $ \\ 
$ \frac{\tilde{\bN}}{\| \tilde{\bN} \|_1} > \bzero $ 
}   ; 
\draw[very thick] (0.0 , -6.0) -- (19.0, -6.0) ;
\end{tikzpicture}
\end{center}

\section{Data availability statement}

The authors declare that the data supporting our findings are available within the paper.


\section{Conflict of interests statement}

The authors declare that they have no known competing financial interests or personal relationships that could have appeared to influence the work reported in this paper.

\printbibliography

@phdthesis{tubiblio138082,
           pages = {vii, 159 Seiten},
        language = {en},
         address = {Darmstadt},
          author = {Bernd Michael Fernengel},
           title = {Stationary solutions of classical Markov chains and Lindblad equations},
            year = {2023},
          school = {Technische Universit{\"a}t Darmstadt},
             url = {http://tubiblio.ulb.tu-darmstadt.de/138082/},
             doi = {https://doi.org/10.26083/tuprints-00023979}
}

@article{fernengel2022obtaining,
  title={Obtaining the long-term behavior of master equations with finite state space from the structure of the associated state transition network},
  author={Fernengel, Bernd Michael and Drossel, Barbara},
  journal={Journal of Physics A: Mathematical and Theoretical},
  volume={55},
  number={11},
  pages={115201},
  year={2022},
  publisher={IOP Publishing}
}

@book{konigsberger2006analysis,
  title={Analysis 2},
  author={K{\"o}nigsberger, Konrad},
  year={2006},
  publisher={Springer-Verlag}
}

@article{fischer2004mathematik,
  title={Mathematik f{\"u}r physiker-band 2: Gew{\"o}hnliche und partielle differentialgleichungen, mathematische grundlagen der quantenmechanik. 2},
  author={Fischer, H and Kaul, H},
  journal={Auflage, Teubner},
  year={2004}
}

@book{werner2011funktionalanalysis,
  title={Funktionalanalysis},
  author={Werner, Dirk},
  year={2011},
  publisher={Springer-Verlag}
}

@book{van1992stochastic,
  title={``Stochastic processes in physics and chemistry''},
  author={Van Kampen, Nicolaas Godfried},
  volume={1},
  year={1992},
  publisher={Elsevier}
}

@book{kelly2011reversibility,
  title={``Reversibility and stochastic networks''},
  author={Kelly, Frank P},
  year={2011},
  publisher={Cambridge University Press}
}

@book{honerkamp2012statistical,
  title={``Statistical physics: an advanced approach with applications''},
  author={Honerkamp, Josef},
  year={2012},
  publisher={Springer Science \& Business Media}
}

@article{mirzaev2013laplacian,
  title={Laplacian dynamics on general graphs},
  author={Mirzaev, Inomzhon and Gunawardena, Jeremy},
  journal={Bulletin of mathematical biology},
  volume={75},
  number={11},
  pages={2118--2149},
  year={2013},
  publisher={Springer}
}

@book{bremaud2013markov,
  title={``Markov chains: Gibbs fields, Monte Carlo simulation, and queues''},
  author={Br{\'e}maud, Pierre},
  volume={31},
  year={2020},
  publisher={Springer Science \& Business Media}
}

@book{engel2000one,
  title={``One-parameter semigroups for linear evolution equations''},
  author={Engel, Klaus-Jochen and Nagel et al. , Rainer },
  volume={194},
  year={2000},
  publisher={Springer}
}

@book{douc2018markov,
  title={``Markov chains''},
  author={Douc, Randal and Moulines, Eric and Priouret, Pierre and Soulier, Philippe},
  year={2018},
  publisher={Springer}
}

@article{privault2013understanding,
  title={Understanding markov chains},
  author={Privault, Nicolas},
  journal={Examples and Applications, Publisher Springer-Verlag Singapore},
  volume={357},
  pages={358},
  year={2013},
  publisher={Springer}
}

@book{grimmett2020probability,
  title={Probability and random processes},
  author={Grimmett, Geoffrey and Stirzaker, David},
  year={2020},
  publisher={Oxford university press}
}

@book{norris1998markov,
  title={Markov chains},
  author={Norris, James R},
  number={2},
  year={1998},
  publisher={Cambridge university press}
}

@article{bellman1947boundedness,
  title={The boundedness of solutions of infinite systems of linear differential equations},
  author={Bellman, Richard},
  year={1947}
}

@book{diestel2024graph,
  title={Graph theory},
  author={Diestel, Reinhard},
  year={2024},
  publisher={Springer (print edition); Reinhard Diestel (eBooks)}
}

@article{arendt2020positive,
  title={Positive irreducible semigroups and their long-time behaviour},
  author={Arendt, Wolfgang and Gl{\"u}ck, Jochen},
  journal={Philosophical Transactions of the Royal Society A},
  volume={378},
  number={2185},
  pages={20190611},
  year={2020},
  publisher={The Royal Society Publishing}
}

@book{arendt1986one,
  title={One-parameter semigroups of positive operators},
  author={Arendt, Wolfgang and Grabosch, Annette and Greiner, G{\"u}nther and Moustakas, Ulrich and Nagel, Rainer and Schlotterbeck, Ulf and Groh, Ulrich and Lotz, Heinrich P and Neubrander, Frank},
  volume={1184},
  year={1986},
  publisher={Springer}
}

@article{haag2017modelling,
  title={Modelling with the Master equation},
  author={Haag, G{\"u}nter},
  journal={Solution Methods},
  year={2017},
  publisher={Springer}
}

@book{friedli2018statistical,
  title={Statistical mechanics of lattice systems: a concrete mathematical introduction},
  author={Friedli, Sacha and Velenik, Yvan},
  year={2018},
  publisher={Cambridge University Press}
}

@article{batkai2017positive,
  title={Positive operator semigroups},
  author={B{\'a}tkai, Andr{\'a}s and Fijav{\v{z}}, M Kramar and Rhandi, Abdelaziz},
  journal={Operator Theory: advances and applications},
  volume={257},
  year={2017},
  publisher={Springer}
}

@book{seneta2006non,
  title={Non-negative matrices and Markov chains},
  author={Seneta, Eugene},
  year={2006},
  publisher={Springer Science \& Business Media}
}

@article{katznelson1986power,
  title={On power bounded operators},
  author={Katznelson, Yitzhak and Tzafriri, Lior},
  journal={Journal of Functional Analysis},
  volume={68},
  number={3},
  pages={313--328},
  year={1986},
  publisher={Elsevier}
}

\end{document}